\documentclass[acmlarge]{acmart}

\usepackage{amsmath}

\usepackage{amssymb}
\usepackage{amsthm}
\usepackage{tikz-cd}
\usepackage{ifthen}
\usepackage{xspace}
\usepackage{xparse}
\usepackage{mathpartir}
\usepackage{cleveref}
\usepackage{enumitem}
\usepackage{stmaryrd}
\usepackage{multicol}
\usepackage{tikz-cd}
\usepackage{mathtools}
\usepackage{stackengine}
\usepackage{scalerel}
\usepackage{wasysym}
\usepackage{fontawesome}
\usepackage{ulem}

\usetikzlibrary{matrix}
\usetikzlibrary{calc}
\usetikzlibrary{positioning}
\usetikzlibrary{decorations.pathreplacing}

\DeclareFontFamily{OT1}{pzc}{}
\DeclareFontShape{OT1}{pzc}{m}{it}{<-> s * [1.10] pzcmi7t}{}
\DeclareMathAlphabet{\mathpzc}{OT1}{pzc}{m}{it}

\crefname{figure}{fig.}{fig.}
\Crefname{figure}{Fig.}{Fig.}

\allowdisplaybreaks

\input{macro}

\AtBeginDocument{%
  \providecommand\BibTeX{{%
    \normalfont B\kern-0.5em{\scshape i\kern-0.25em b}\kern-0.8em\TeX}}}

\setcopyright{none}

\begin{document}

\title{An Investigation on Kripke-style Modal Type Theories}

\author{Jason Z. S. Hu}
\email{zhong.s.hu@mail.mcgill.ca}
\affiliation{%
  \department{School of Computer Science}
  \institution{McGill University}
  \streetaddress{McConnell Engineering Bldg. Room 225, 3480 University St.}
  \city{Montr\'eal}
  \state{Qu\'ebec}
  \country{Canada}
  \postcode{ H3A 0E9}
}

\author{Brigitte Pientka}
\email{bpientka@cs.mcgill.ca}
\affiliation{%
  \department{School of Computer Science}
  \institution{McGill University}
  \streetaddress{McConnell Engineering Bldg. Room 107N, 3480 University St.}
  \city{Montr\'eal}
  \state{Qu\'ebec}
  \country{Canada}
  \postcode{ H3A 0E9}
}


\begin{abstract}
  This technical report investigates Kripke-style modal type theories, both simply
  typed and dependently typed. We examine basic meta-theories of the type theories,
  develop their substitution calculi, and give normalization by evaluation
  algorithms. 
\end{abstract}




\maketitle

\section{Introduction}\labeledit{sec:intro}

Modalities are a widely used concept. Certain modalities can readily be implemented in
existing languages, e.g. various (co)monad instances in Haskell
libraries. Unfortunately, not every modality can be implemented in well-known
languages and this problem triggers studies of adding various modalities into
different type
theories~\citep[etc]{davies_modal_2001,pfenning_judgmental_2001,gratzer_implementing_2019,birkedal_modal_2020,kavvos_dual-context_2017,clouston_fitch-style_2018,davies_temporal_2017,nakano_modality_2000}.
Though modalities have been researched quite extensively in the logic communities,
some applications of modalities in programming languages have more specialized
scenarios which allow programmers to assume more specific structures in the modalities
than commonly known as logical systems. On the other hand,
\citet{davies_modal_2001,pfenning_judgmental_2001,kavvos_dual-context_2017,clouston_fitch-style_2018}
give completely general formulation in natural deduction styles that is complete with
respect to their corresponding modal logical systems. 

\subsection{Modality by Axioms}

Traditionally, necessity in modal logics is described by some combination of
axioms. Among all, necessity must satisfy the following axiom:
\begin{align*}
  K: \square (A \to B) \to \square A \to \square B
\end{align*}
The system satisfying this axiom only is also called $K$. In addition, we might admit
the following axiom:
\begin{align*}
  T : \square A \to A
\end{align*}
The system satisfying both $K$ and $T$ is called $T$.  Another common axiom is
\begin{align*}
  4: \square A \to \square \square A
\end{align*}
Admittedly it is strange that the name of an axiom is $4$, but \emph{c'est la
  vie}. If we extend system $T$ with the axiom $4$, then we obtain $S4$. We call the
system which admits only axioms $K$ and $4$ system $K4$. 
There are other axioms but we will not cover them
here. An experienced eye may identify that $K$ effectively is the applicative functor
law and $T$ and $4$ combined are the comonad laws. Therefore, $S4$ describes the
necessity modality as a comonad, so necessity in $S4$ is sometimes called
\emph{a comonadic modality} and $S4$ is in \emph{a comonadic type theory}. 

\subsection{Different Formulations}

In this technical report, we consider the general formulations of necessity modality
in Kripke style (given as the implicit formulation in \citet{davies_modal_2001}). We
distinguish this style from another style, the dual-context or explicit style, given
in the same paper by how contexts are organized. The dual-context style is inspired by
the categorization of two different judgments for a proposition, validity or
truth. In $S4$, validity is a stronger assertion than truth; in particular, if a proposition is
valid, then it is also true, but not necessarily vice versa. This categorization of
propositions corresponds to two different contexts serving different purposes. In a
typing judgment in dual-context style, we have
\begin{mathpar}
  \typing[\Delta \sep \Gamma]t T
\end{mathpar}
where $\Delta$ records valid assumption while $\Gamma$ records the true ones. In
this view, the necessity modality, or $\square$, has the following introduction and
elimination rules in $S4$:
\begin{mathpar}
  \inferrule
  {\typing[\Delta \sep \cdot]t T}
  {\typing[\Delta \sep \Gamma]{\boxit t}{\square T}}

  \inferrule
  {\typing[\Delta \sep \Gamma]{t}{\square T} \\ \typing[\Delta, u : T \sep \Gamma]{t'}{T'}}
  {\typing[\Delta \sep \Gamma]{\letbox u t t'} T}
\end{mathpar}
where $\square T$ can be read as $T$ being valid. In this reading, the
introduction rule says that $T$ is valid if it can be derived only from valid
assumptions. The elimination rule allows use of $\square T$ by pattern matching it and
appends a valid assumption to $\Delta$ in the match body.

Kripke style, on the other hand, takes a different approach. It starts from Kripke
semantics from a logical point of view. Kripke semantics says that assuming an
accessibility relation via which different ``worlds'' are connected, we can model
various modal
logics by considering what structure this relation possesses:
\begin{itemize}
\item If the accessibility relation has no particular structure, then we have system $K$.
\item If the relation is reflexive, then we have system $T$.
\item If the relation is both reflexive and transitive, then we have system $S4$.
\item If the relation is only transitive, not reflexive, then we have system $K4$. 
\end{itemize}
Inspired by this interpretation, the calculus instead uses a stack of contexts
to keep track of assumptions in different worlds accessed through the relation, and
the structure of the relation is revealed in the form of elimination. All four aforementioned systems have
the following introduction and elimination rules for $\square$:
\begin{mathpar}
  \inferrule*
  {\mtyping[\vGamma; \cdot] t T}
  {\mtyping{\boxit t}{\square T}}

  \inferrule*
  {\mtyping t {\square T} \\ |\vDelta| = n}
  {\mtyping[\vGamma; \vDelta]{\unbox n t}{T}}
\end{mathpar}
From Kripke semantic point of view, the introduction rule states that we can follow
the accessibility relation to a new world, in which no assumption is made. On the
other hand, the elimination rule allows us to travel backward to the
already accessed worlds by specifying a proper \tunbox level $n$. This \tunbox level
$n$ is a reflection of Kripke semantics. Its range decides which modal logic the
calculus corresponds to:
\begin{itemize}
\item If $n = 1$, then the system is $T$.\footnote{In this case, we could have omitted
    the \tunbox level completely, but from a theoretical point of view, keeping the
    \tunbox level here helps greatly with uniformity. }
\item If $n \in [0, 1]$, then the system is $T$. 
\item If $n \in \N$ namely all natural numbers, then the system is $S4$.
\item If $n \in [1, +\infty)$ namely all positive numbers, then the system is $K4$. 
\end{itemize}
Notice that $1$ is always in the range. This is because $n = 1$ corresponds to the
axiom $K$, which is the minimal axiom the modality should admit. In this report, we
will apply this observation very frequently.
The fact that the calculus can be parameterized over the range of \tunbox levels is
one major advantage of Kripke style over dual-context style. In dual-context style,
different modal logics have very different formulations, as shown in
\citet{kavvos_dual-context_2017}.

Another advantage of Kripke style is that $\square T$ is considered as a product type,
as opposed to a sum type in dual-context style, where we have to handle commuting
conversions (c.f. \citet{girard_proofs_1989,kavvos_dual-context_2017}). In Kripke
style, this problem naturally goes away. These reasons lead to the choice of study in
this technical report. 

\subsection{Fitch-style Systems}

Some investigations effectively study the same group of systems under the name of ``Fitch
style''. Kripke style and Fitch style are fundamentally the same and only differ
slightly by how they organize contexts. As shown above, Kripke style manages 
context stacks. In Fitch style, a single context is maintained, but in the
context a delimiter symbol (usually $\thelock$) is used. The segment of context
appearing after $\thelock$ is considered assumptions in a new world until the next $\thelock$. For $K$, in Fitch style we
would have the following formulation:
\begin{mathpar}
  \inferrule
  {\typing[\Gamma, \thelock]{t}{T}}
  {\typing{t}{\square T}}

  \inferrule
  {\typing{t}{\square T} \\ \thelock \notin \Gamma'}
  {\typing[\Gamma, \thelock, \Gamma']{\tunbox\ t}{T}}
\end{mathpar}
Compared to the Kripke-style elimination rule of $K$,
\begin{mathpar}
  \inferrule*
  {\mtyping t {\square T}}
  {\mtyping[\vGamma; \Gamma]{\tunbox\ t}{T}}
\end{mathpar}
since the world structure is less clear in Fitch style, a predicate $\thelock \notin
\Gamma'$ is added to ensure all assumptions in $\Gamma'$ are in the same world. This
is unnecessary in Kripke style as the world structure is baked in the syntax of
context stacks. One can
also imagine writing out the elimination rule of $S4$ in Fitch style can also be
annoying as one must count the number of $\thelock$ ``just right''. Therefore, in this
technical report, we use Kripke style in the syntax. However, as we will see in the
constructions later, the line between Kripke and Fitch styles becomes much more blur
in the semantics.

\subsection{Contributions}

To summarize what has been achieved in this technical report, we
\begin{itemize}
\item provide a general scheme for substitutions for all four aforementioned systems,
  which generalizes well with dependent types;
\item provide two normalization by evaluation (NbE) algorithms, one based on a
  presheaf model and the other based on an untyped domain, so that they work for all
  four systems, and thus demonstrate their normalization at once;
\item define a substitution calculus, which, again, works for all four systems;
\item extend Martin-L\"of type theory with modality and explicit substitutions, such
  that the four systems have Martin-L\"of style dependent types;
\item at last, adapt the untyped domain model to dependent types, so that we
  obtain normalization for the dependently typed variants.
\end{itemize}
This technical report contains necessary discussion and proofs. A partial
formalization of certain concepts in Agda\footnote{primarily for simply types} can be found at
\url{https://gitlab.com/JasonHuZS/practice}. 

\section{Simply Typed Modal Type Theory}\labeledit{sec:st:syn}

In this section, we introduce the syntax and judgments for simply typed modal type
theory as per \citet{davies_modal_2001} and its related operations, parameterized by
some unspecified base type $B$. Effectively, we extend simply typed $\lambda$-calculus
with $\square$ modality shown in the previous section. We make the syntax precise below
\begin{alignat*}{2}
  S, T &:=&&\ B \sep \square T \sep S \func T
  \tag{Types, \Typ} \\
  i, j, k, l, m, n & && \tag{Meta-level natural numbers, $\mathbb{N}$} \\
  x, y & && \tag{variables, $\Var$} \\
  s, t, u &:=&&\ x \tag{Terms, $\Exp$} \\
  & && \sep \boxit t \sep \unbox n t
  \tag{box}\\
  & &&\sep \lambda x. t \sep s\ t \tag{functions}  \\
  \Gamma, \Delta, \Phi &:= &&\ \cdot \sep \Gamma, x : T
  \tag{Contexts, \Ctx}\\
  \vGamma, \vDelta &:= &&\ \epsilon \sep \vGamma; \Gamma \tag{Context
    stack, $\vect{\Ctx}$} \\
  w &:= &&\ v \sep \boxit w \sep \lambda x. w
  \tag{Normal form ($\Nf$)} \\
  v &:= &&\ x  \sep v\ w \sep \unbox n v \tag{Neutral form
    ($\Ne$)}
\end{alignat*}
$B$ represents the base type. In order to make our argument as general as possible, we
leave the naming convention abstract. We assume silent $\alpha$ renaming whenever
necessary. We use $\epsilon$ to represent the empty context stack and $\cdot$ to
represent the empty context. For example, $\vGamma; \cdot$ means to append an empty
context on top of the existing context stack $\vGamma$. We write $\epsilon; \cdot$ for
the singleton stack containing the empty context. We overload $;$ for combining two
context stacks as well. For example, we write $\vGamma; \vDelta$ for a combined
context stack where $\vDelta$ is appended on top of $\vGamma$. 

\subsection{Typing Rules}

The typing judgments for $\Exp$ are listed below:
\begin{mathpar}
  
  \inferrule*
  {x : T \in \Gamma}
  {\mtyping[\vGamma; \Gamma] x T}

  \inferrule*
  {\mtyping[\vGamma; \cdot] t T}
  {\mtyping{\boxit t}{\square T}}

  \inferrule*
  {\mtyping t {\square T} \\ |\vDelta| = n}
  {\mtyping[\vGamma; \vDelta]{\unbox n t}{T}}

  \\
  
  \inferrule*
  {\mtyping[\vGamma;\Gamma, x : S]t T}
  {\mtyping[\vGamma;\Gamma]{\lambda x. t}{S \func T}}

  \inferrule*
  {\mtyping t {S \func T} \\ \mtyping s S}
  {\mtyping{t\ s}{T}}
\end{mathpar}
A variable can only refer to bindings in the topmost context in the stack. The introduction
and elimination rules for $\square$ are the same as in \Cref{sec:intro} and by
changing the range of $n$ we obtain different systems. The last two rules for
functions are standard.  Moreover, we stipulate that a context stack being typed in
must have at least length $1$. That means we can always assume there exists a topmost context
which we are working in.

\subsection{Modal Transformation}\labeledit{sec:st:mtrans}

In order to define equivalence between well-typed terms, we need to define a few
operations. 
The \emph{modal transformation} operation, or \emph{MoT}, converts a term in one context stack to another by
re-index the \tunbox levels:
\begin{definition}
  MoT is defined as
  \begin{align*}
    x\{k/l\} &:= x \\
    \boxit t\{k/l\} &:= \boxit{(t\{k/l+1\})} \\
    \unbox n t \{k/l\} &:=
                         \begin{cases}
                           \unbox n {t\{k/l-n\}} & \text{if $n \le l$} \\
                           \unbox{k + n - 1}t &\text{if $n > l$}
                         \end{cases} \\
    \lambda x. t \{k/l\} &:= \lambda x. (t\{k/l\}) \\
    s\ t\{k/l\} &:= (s\{k/l\})\ (t\{k/l\})
  \end{align*}
\end{definition}
In particular, we call the cases where $k = 0$ \emph{modal fusion}\footnote{Since $k$
  comes from an \tunbox level, systems that do not admit the axiom $T$ do not have
  modal fusion. In the rest of this report, we will always assume that modal fusion is
  possible. This will not affect the validity of the discussion for systems that do
  not have modal fusion.}, and the rest \emph{modal weakening}\footnote{On the other
  hand, since $k = 1$ is always possible, all systems have modal weakening.}. The
reason can be seen from the following lemma:
\begin{lemma}\labeledit{lem:mtran-typ}
  If $\mtyping[\vGamma;\Gamma_0;\Delta_0;\cdots;\Delta_l]{t}{T}$, then
  $\mtyping[\vGamma;\Gamma_0; \cdots; (\Gamma_k,
  \Delta_0);\cdots;\Delta_l]{t\{k/l\}}{T}$.
\end{lemma}
\begin{proof}
  Induction on $\mtyping[\vGamma;\Gamma_0;\Delta_0;\cdots;\Delta_l]{t}{T}$.
\end{proof}
To understand this lemma, we should consider some special values.
\begin{itemize}
\item When $k = l = 0$, the lemma states that if $\mtyping[\vGamma;\Gamma_0;\Delta_0]{t}{T}$, then
  $\mtyping[\vGamma;(\Gamma_0, \Delta_0)]{t\{0/0\}}{T}$. Notice that the topmost two
  contexts are fused into one, hence ``modal fusion''.
\item When $k = 2$ and $l = 0$, the lemma states that if $\mtyping[\vGamma;\Gamma_0;\Delta_0]{t}{T}$, then
  $\mtyping[\vGamma;\Gamma_0; \Gamma_1; (\Gamma_2, \Delta_0)]{t\{2/0\}}{T}$. A new
  context $\Gamma_1$ is inserted to the stack and the topmost context is (locally)
  weakened by $\Gamma_2$, hence ``modal weakening''. 
\end{itemize}

MoTs satisfy the following equations:
\begin{lemma}\labeledit{lem:mtran-eq}
  The following holds:
  \begin{itemize}
  \item If $l' < l$, then $t\{k/l\}\{k'/l'\} = t\{k'/l'\}\{k/l+k'-1\}$.
  \item If $l \le l' < l + k$, then $t\{k/l\}\{k'/l'\} = t\{k + k' - 1/l\}$.
  \item If $l + k \le l'$, then $t\{k/l\}\{k'/l'\} = t\{k'/l' - k + 1\}\{k/l\}$.
  \end{itemize}
\end{lemma}
\begin{proof}
  Induction on $t$. The third equation is actually a consequence of the first. 
\end{proof}
The equations cover all possible cases of composing MoTs. Notice that
in the third equation, $l' - k + 1 > l$ and thus can be used to turn an iterative
application of MoTs 
ascending in $l$ to one strictly descending in $l$.

\subsection{Term Substitution}

Just like in typical type theories, we also need an (ordinary) term substitution operation which
replaces a variable with a term:
\begin{definition}
  The term substitution is defined as
  \begin{align*}
    x[u/y] &:= \begin{cases}
      u & \text{if $x = y$} \\
      x & \text{if $x \neq y$}
    \end{cases} \\
    \boxit t [u/y] &:= \boxit{(t[u/y])} \\
    \unbox n t [u/y] &:= \unbox n {(t[u/y])} \\
    \lambda x. t [u/y] &:= \lambda x. (t[u/y])
                         \tag{Notice implicit $\alpha$ renaming here} \\
    s\ t[u/y] &:= (s[u/y])\ (t[u/y])
  \end{align*}
\end{definition}

The following lemma is standard.
\begin{lemma}
  If $\mtyping[\vGamma; \Gamma, x : S, \Gamma'; \vDelta]t T$ and
  $\mtyping[\vGamma;\Gamma]s S$,
  then $\mtyping[\vGamma; \Gamma, \Gamma'; \vDelta]{t[s/x]}T$.
\end{lemma}
\begin{proof}
  Induction on $\mtyping[\vGamma; \Gamma, x : S, \Gamma'; \vDelta]t T$.
\end{proof}

MoTs and term substitutions commute depending on their numerical relation:
\begin{lemma}
  If $x$ is in $n$'th context\footnote{One might find this statement is not precise
    as we use an abstract name representation. Nonetheless, this is possible if one
    uses more concrete name representation, e.g. de Bruijn indices, in which
    term substitutions must be defined with an offset of a world in which the substitutions
    happens. Therefore this $n$ in fact is known from the syntactic level. } and $n \le l$, then
  $t[u/x]\{k/l\} = t\{k/l\}[u\{k/l-n\}/x]$.
\end{lemma}

\begin{lemma}
  If $x$ is in $n$'th context and $n > l$, then $t[u/x]\{k/l\} = t\{k/l\}[u/x]$.
\end{lemma}

\subsection{Term Equivalence}\labeledit{sec:st:equiv-rules}

Given a context and a type, term equivalence forms a PER, which equips with symmetry
and transitivity. Here we employ the $\beta\eta$ equivalence between terms.
\begin{mathpar}
  \inferrule*
  {\mtyequiv s t T}
  {\mtyequiv t s T}

  \inferrule*
  {\mtyequiv s t T \\ \mtyequiv t u T}
  {\mtyequiv s u T}
\end{mathpar}

Congruence rules:
\begin{mathpar}
  
  \inferrule*
  {x : T \in \Gamma}
  {\mtyequiv[\vGamma; \Gamma] x x T}

  \inferrule*
  {\mtyequiv[\vGamma; \cdot]{t}{t'}T}
  {\mtyequiv{\boxit t}{\boxit t'}{\square T}}

  \inferrule*
  {\mtyequiv{t}{t'}{\square T} \\
    |\vDelta| = n}
  {\mtyequiv[\vGamma; \vDelta]{\unbox n t}{\unbox n t'}{T'}}  

  \inferrule*
  {\mtyequiv[\vGamma;(\Gamma, x : S)]{t}{t'}T}
  {\mtyequiv[\vGamma; \Gamma]{\lambda x. t}{\lambda x. t'}{S \func T}}

  \inferrule*
  {\mtyequiv{t}{t'}{S \func T} \\
  \mtyequiv{s}{s'}S}
  {\mtyequiv{t\ s}{t'\ s'}T}
\end{mathpar}

$\beta$ equivalence:
\begin{mathpar}
  \inferrule*
  {\mtyping[\vGamma; \cdot]{t}{T} \\
  |\vDelta| = n}
  {\mtyequiv[\vGamma; \vDelta]{\unbox{n}{(\boxit t)}}{t\{n / 0 \}}{T}}

  \inferrule*
  {\mtyping[\vGamma;(\Gamma, x : S)]t T \\ \mtyping[\vGamma; \Gamma] s S}
  {\mtyequiv[\vGamma; \Gamma]{(\lambda x. t) s}{t[s/x]}{T}}
\end{mathpar}

$\eta$ equivalence:
\begin{mathpar}
  \inferrule*
  {\mtyping{t}{\square T}}
  {\mtyequiv{t}{\boxit{\unbox 1 t}}{\square T}}
  
  \inferrule*
  {\mtyping{t}{S \func T}}
  {\mtyequiv t {\lambda x. (t\ x)}{S \func T}}
\end{mathpar}
Notice that a MoT is used in the $\beta$ equivalence for $\square T$. 

\begin{lemma}[Presupposition]\labeledit{lem:inv-equiv}
  If $\mtyequiv{t}{t'}T$, then $\mtyping t T$ and $\mtyping{t'}T$.
\end{lemma}
\begin{proof}
  Induction on $\mtyequiv{t}{t'}T$.
\end{proof}

\section{Combining Modal Transformation and Substitution}\labeledit{sec:usubst}

In the previous section, we showed that Kripke-style systems have two seemingly unrelated
operations: MoTs and term substitutions. This unfortunately poses
difficulties in representing simultaneous substitutions between context stacks. In typical type
systems, simultaneous substitutions are morphisms between contexts and conceptually just a list of
terms. With context stacks, a list of terms would not be sufficient. Consider how to
represent a simultaneous substitution $\vGamma \To \Delta; \Delta'$. Intuition says that a list of
terms in $\vGamma$ can be used to substitute $\Delta'$, but how do we handle $\Delta$
and in which stack terms substituting $\Delta$
should be well-typed? In this section, we give \emph{unified} substitutions, which unify
MoTs and term substitutions as a single representation. Moreover, unified
substitutions have identity and are closed under composition. Therefore we can 
organize a category with context stacks as objects and unified substitutions as
morphisms, which is very helpful in later developments of normalization. We also show how unified
substitutions are applied to terms.

To obtain enough intuition and arrive at our final solution, let us first consider a
subquestion: can we view MoTs as substitutions?

\subsection{MoTs as Substitutions}

First of all, MoTs are an operation transforming a term in
one context stack to another. In general, when viewed as a morphism, a MoT has the following
type according to \Cref{lem:mtran-typ}:
\begin{align*}
  \{k/l\} : \vGamma;\Gamma_0; \cdots; (\Gamma_k, \Delta_0);\cdots;\Delta_l \To \vGamma;\Gamma_0;\Delta_0;\cdots;\Delta_l
\end{align*}
which is orthogonal to term substitutions defined by 
substituting variables with terms. By \Cref{lem:mtran-eq}, we require MoTs to
satisfy the following laws:
\begin{align*}
  \{k/l\} \circ \{k'/l'\} &= \{k'/l'\} \circ \{k/l+k'-1\}
                            \tag{if $l' < l$} \\
  \{k/l\} \circ \{k'/l'\} &= \{k + k' - 1 / l\}
                            \tag{if $l \le l' < l + k$}
\end{align*}

If we are to represent MoTs as morphisms, then this representation
must be general enough to represent the compositional closure of MoTs.
Hence our next question is what is this compositional closure?

\subsection{An Intuition}

First, let us compare MoTs with term substitutions: consider two adjacent term
substitutions $t[u/x][u'/y]$ where $x$ and $y$ are different (otherwise the second
substitution is voided). Then we can merge these two term substitutions into one by
simultaneously substituting both $x$ and $y$ for $u$ and $u'$, respectively. Can we do
the same for MoTs? If that is possible, then effectively we can merge arbitrary number
of adjacent MoTs into one. Let us first consider some examples.

\subsubsection{An Example for Modal Weakening}

From \Cref{lem:mtran-typ}, we know that $\{1 + k/l\}$ denotes the
morphism $\vGamma;\Gamma_1; \cdots; (\Gamma_k, \Delta_0);\cdots;\Delta_l \To
\vGamma;\Delta_0;\cdots;\Delta_l$. Now consider more concretely,
\begin{align*}
  \{2/3\} :& \epsilon;\Gamma_0; \Gamma_1; (\Gamma_2, \Delta_0) ; \Delta_1; \Delta_2; \Delta_3
             \To \epsilon;\Gamma_0; \Delta_0 ; \Delta_1; \Delta_2; \Delta_3 \\
  \{2/1\} :& \epsilon;\Gamma_0; \Gamma_1; (\Gamma_2, \Delta_0); \Delta_1 ; \Gamma_1'; (\Gamma_2', \Delta_2);
             \Delta_3
             \To \epsilon;\Gamma_0; \Gamma_1; (\Gamma_2, \Delta_0) ; \Delta_1; \Delta_2; \Delta_3
\end{align*}
Where $\epsilon$ is the empty stack. 
Therefore, we derive
\begin{align*}
  \{2/3\} \circ \{2/1\} :
  \epsilon;\Gamma_0; \Gamma_1; (\Gamma_2, \Delta_0); \Delta_1 ; \Gamma_1'; (\Gamma_2', \Delta_2);
  \Delta_3
  \To \epsilon;\Gamma_0; \Delta_0 ; \Delta_1; \Delta_2; \Delta_3
\end{align*}
We obtain from \Cref{lem:mtran-eq} that two modal weakenings can be commuted:
\begin{align*}
  \{2/3\} \circ \{2/1\} = \{2/1\} \circ \{2/4\}
\end{align*}
Notice that effectively these two modal weakenings are somewhat independent: they
modal-weaken $\Delta_0$ and $\Delta_2$, respectively. Thus, if we begin
with the original context stack $\Gamma_0; \Delta_0 ; \Delta_1; \Delta_2; \Delta_3$,
then we can turn $\{2/3\} \circ \{2/1\}$ into a single \emph{simultaneous} modal
weakening operation: $\{(2/3); (2/1)\}$. In general, any composition of modal
weakenings can be represented by $\{\vect{(k_i/l_i)}\}$, where $l_i$ are in strictly
descending order for the sake of obtaining a canonical representation. Thus we can conclude
\begin{align*}
  \{(2/3); (2/1)\} = \{2/3\} \circ \{2/1\} = \{2/1\} \circ \{2/4\}
\end{align*}
or equivalently in commutative diagram:
\begin{center}
  \begin{tikzcd}
    \epsilon;\Gamma_0; \Gamma_1; (\Gamma_2, \Delta_0); \Delta_1 ; \Gamma_1'; (\Gamma_2', \Delta_2);
    \Delta_3
    \ar{d}{\{2/1\}}
    \ar{r}{\{2/4\}}
    \ar{dr}{\{(2/3); (2/1)\}}
    & \epsilon;\Gamma_0; \Delta_0; \Delta_1 ; \Gamma_1'; (\Gamma_2', \Delta_2);
    \Delta_3
    \ar{d}{\{2/1\}}
    \\
    \epsilon;\Gamma_0; \Gamma_1; (\Gamma_2, \Delta_0) ; \Delta_1; \Delta_2; \Delta_3
    \ar{r}{\{2/3\}}
    & \epsilon;\Gamma_0; \Delta_0 ; \Delta_1; \Delta_2; \Delta_3
  \end{tikzcd}
\end{center}

We observe that such simultaneous modal weakening satisfies the following inductive
equations if all $k$'s are at least $1$:
\begin{align*}
  \{\} &= \vect\id \\
  \{(k/l); \vect{(k_i/l_i)}\} &= \{k/l\} \circ \{\vect{(k_i/l_i)}\}
\end{align*}
Note that on the left hand side, in order to write the $k$'s and $l$'s in the
simultaneous form, we have already assumed $l$'s are in descending order. 
These equations allow us to compose two simultaneous modal weakenings by
unpacking them into individual modal weakenings and use \Cref{lem:mtran-eq} to adjust the
order so that all $l$'s are in descending order. The lemma ensures that a strict
descending order is always possible even some $k$'s are $0$. That allows us to
conclude that modal weakening is closed under composition. 

\subsubsection{Examples with Fusion}

The previous observation works very well in systems without the axiom $T$\footnote{Or
  equivalently systems where \tunbox levels are not $0$ or modal fusion is not
  permitted}.  Once modal fusion is involved, complication arrives.  Consider the
following example:
\begin{center}
  \begin{tikzcd}
    \epsilon;\Gamma_0,\Delta,\Gamma_1
    \ar{r}{\{0/0\}}
    \ar{dr}[']{\{0/0\}} &
    \epsilon;\Gamma_0; (\Delta, \Gamma_1)
    \ar{d}{\{1/0\}} \\
    & \epsilon;\Gamma_0; \Gamma_1
  \end{tikzcd}
\end{center}
In this example, $\{0/0\}$ in the diagonal not only fuse $\Gamma_0$ and $\Gamma_1$,
but also introduces an extra (unspecified) $\Delta$. This is because modal fusion
necessarily carries a notion of local weakening. Due to our representation of names, details
involving the length of $\Delta$ is not exposed syntactically in MoTs,
but in other representations, e.g. de Bruijn indices, we will have to handle this
problem explicitly. To capture this complication, we need to allow modal fusion to
perform arbitrary local weakenings.

Even if local weakening caused by modal fusion is not involved, fusion breaks other
assumptions that were safe when only modal weakening are present. For instance,
previously, we stated that individual model weakenings from one simultaneous modal
weakening can be performed in isolation.  Here is a counterexample with the presence
of modal fusion:
\begin{align*}
  \{0/2\} : \epsilon;\Gamma_0; \Gamma_1; (\Gamma_2, \Delta_0, \Delta_1); \Delta_2; \Delta_3
  \To \epsilon;\Gamma_0; \Gamma_1; (\Gamma_2, \Delta_0) ; \Delta_1; \Delta_2; \Delta_3
\end{align*}
Notice that $(\Gamma_2, \Delta_0)$ and $\Delta_1$ are fused. Thus we have
\begin{align*}
  \{2/3\} \circ \{0/2\} :
  \epsilon;\Gamma_0; \Gamma_1; (\Gamma_2, \Delta_0, \Delta_1); \Delta_2; \Delta_3
  \To \epsilon;\Gamma_0; \Delta_0 ; \Delta_1; \Delta_2; \Delta_3
\end{align*}
where $\{2/3\}$ comes from the previous example. 

Even though the numbers in the $l$ position are already in descending order, the
effects of $\{2/3\}$ and $\{0/2\}$ apply to two fused contexts, namely $\Delta_0$ and
$\Delta_1$. This example seems to suggest a complex operation which first does local weakening, then some fusions and
at last modal weakening. 
\begin{definition}
  The \textbf{not-so-good} MoT is defined as follows:
  \begin{align*}
    t\{\vect m, n, k/l\} := t\underbrace{[\uparrow_m]}_{m \in \vect m}\{0/l\}^n\{k+1/l\}
  \end{align*}
  where $\uparrow_m$ is the weakening morphism which weakens the $m$'th context from
  top of the stack: $\uparrow_m : \vGamma ; (\Gamma_m, \Delta); \cdots;\Gamma_0
  \To \vGamma ; \Gamma_m; \cdots;\Gamma_0$. 
\end{definition}
Clearly, this operation is way too complex for a valuable investigation.
Nevertheless, we obtain one useful observation: if modal fusion embeds local
weakenings, a special case of term substitutions, why not add
more complex term substitutions to MoTs?

\subsection{Unified Substitution}\labeledit{sec:usubst}

To best motivate our unified representation of MoTs and term
substitutions, let us take a closer look at an example of multiple applications of
MoTs:
\begin{center}
  \begin{tikzpicture}
    \matrix (m) [matrix of math nodes, row sep=15pt]
    {
      \vGamma = \epsilon; & \Gamma_0; & \Delta_0; & \Delta_1; & (\Gamma_1, \Gamma_2, \Gamma_3); &
      \Delta_2; & \Gamma_4
      \\
      \vDelta = \epsilon; & \Gamma_0; & & & \Gamma_1; \Gamma_2; \Gamma_3; &
      & \Gamma_4 \\
    };
    \draw[<-] (m-2-2)  -- (m-1-2);
    \draw[<-] ($(m-2-5) + (-0.6, 0.3)$) -- ($(m-1-5) + (-0.6, -0.3)$);
    \draw[<-] ($(m-2-5) + (-0.1, 0.3)$) -- ($(m-1-5) + (-0.1, -0.3)$);
    \draw[<-] ($(m-2-5) + (0.5, 0.3)$) -- ($(m-1-5) + (0.5, -0.3)$);
    \draw[<-] (m-2-7) -- (m-1-7);
  \end{tikzpicture}
\end{center}
This diagram represents a composition of a number of MoTs and
contains both modal fusion and modal weakening. For example, we can see that
$\Gamma_4$ is modal-weakened with $\Delta_2$ prepended. Modal fusion happens
consecutively for $\Gamma_1$, $\Gamma_2$ and $\Gamma_3$. These three separate contexts
in the stack are fused into one, and then the resulting context is modal-weakened. The
arrows in the diagram denotes the correspondence of contexts before and after the
simultaneous MoT. Obviously, all these arrows are just some local
weakening typed in the context stacks to the left of the ends of arrows. Moreover, if
we replace these arrows with more complex local substitutions, then all $\Gamma_i$ at the
bottom can be some arbitrary $\Gamma_i'$. Additional bookkeeping is necessary to
maintain the relative positions for the contexts in domain and codomain context
stacks. Achieving all these requirements, we introduce unified substitution defined as follows:
\begin{definition}
  A unified substitution $\vsigma$, or just \textbf{substitution}, between context stacks is defined as
  \begin{align*}
    \sigma, \delta &:= () \sep \sigma, t/x
    \tag{Local substitutions, $\Subst$}\\
     \vsigma, \vdelta &:=  \snil \sigma \sep \sext \vsigma n \sigma
                            \tag{Unified substitutions, $\Substs$} 
  \end{align*}
  \begin{mathpar}
    \inferrule
    {\sigma : \vGamma \To \Gamma}
    {{\snil { \sigma}} : \vGamma \To \epsilon;\Gamma}
    \quad
    \inferrule
    {\vsigma : \vGamma \To \vDelta \quad |\vGamma'| = n \quad \sigma : \vGamma;\vGamma' \To \Delta}
    {\sext \vsigma n \sigma : \vGamma;\vGamma' \To \vDelta;\Delta}
  \end{mathpar}
  where local substitution $\sigma : \vGamma \To \Gamma$ is defined as a list of terms as
  usual. Moreover, the \emph{offset} $n$ is subject to the \tunbox level constraint of the
  system. That is, in system $K$, $n$ must be $1$; in system $T$, $n$ can be $0$ or
  $1$, etc. 
\end{definition}

The core idea of this definition is simple: a local substitution $\sigma$
replaces variables in a context $\Gamma$ and an offset $n$ specifies the
correspondence of contexts. Just as context stacks must be non-empty and consist of
at least one context, a unified substitution must have a topmost local substitution.
By definition, we know that a unified substitution always has a topmost
local substitution.  To illustrate, the morphism in the previous diagram can be represented as
$\varepsilon; \sext {\sext {\sext {\sext \id 3 \wk_1} 0 \wk_2} 0 \wk_3} 2 \id$ where $\id$ is local
identities and
$\wk_i : \Gamma_0;\Delta_0;\Delta_1;(\Gamma_1,\Gamma_2,\Gamma_3) \To {\Gamma_i}$
are appropriate local weakenings. We break down this representation:
\begin{enumerate}
\item We start with $\snil\id : \epsilon ; \Gamma_0 \To \epsilon; \Gamma_0$.
\item We add an offset $3$ and a local weakening $\wk_1$, forming
  $\sext {\snil \id} 3 \wk_1 : 
   \epsilon ; \Gamma_0; \Delta_0;\Delta_1;(\Gamma_1,\Gamma_2,\Gamma_3) \To
  \epsilon; \Gamma_0; \Gamma_1$. The offset $3$ adds three contexts to
  the domain stack ($\Delta_0$, $\Delta_1$ and $\Gamma_1,\Gamma_2,\Gamma_3$). Local weakening
  $\wk_1$ extracts $\Gamma_1$ from $\Gamma_1,\Gamma_2,\Gamma_3$.
\item We extend the substitution to 
$\sext
   {\sext {\snil \id} 3  \wk_1}
    0  \wk_2 : \epsilon ; \Gamma_0; \Delta_0;\Delta_1;(\Gamma_1,\Gamma_2,\Gamma_3) \To
  \epsilon; \Gamma_0; \Gamma_1; \Gamma_2$. Since the offset associated
  with $\wk_2$ is $0$, no context is added to the domain stack. This effectively represents fusion. $\wk_2$ is
  similar to $\wk_1$. 
\item The rest proceeds very similarly.
\end{enumerate}

Clearly, if we change the local substitutions in the example above to contain more
terms, we obtain unified substitutions that also do nontrivial term substitutions.
Subsequently, we may simply write
$\vsigma; \sigma$ instead of $\sext \vsigma 1 \sigma$. In particular,
we will write $\vsigma ; \wk$ instead of $\sext \vsigma 1 \wk$ and 
$\vsigma ; \id$ instead of $\sext \vsigma 1 \id$.  We often omit offsets
that are $1$ for readability.

Next, we shall define some operations on unified substitutions.
In order to define
composition, we describe two essential operations: \emph{truncation}
($\trunc \vsigma n$) drops $n$ substitutions from a unified substitution $\vsigma$;
\emph{truncation offset} ($\Ltotal \vsigma n$) computes when given a unified substitution $\vsigma$ the \emph{total} number of contexts that need to be dropped from the domain context stack, given that we truncate $\vsigma$ by $n$.

The $\Ltotal \vsigma n$ operation takes as input the unified substitution
$\vsigma$ and $n$, the number of local substitutions to be
dropped from $\vsigma$. It computes the sum of $n$ leading
offsets, called \emph{truncation offset}. Let 
$\vsigma = \sext{\sext{\vsigma'}{k_n}{\sigma_n} ; \ldots}{k_1}{\sigma_1}
$, then $\Ltotal {\vsigma}{n} = k_n + \ldots + k_1$ and $\trunc \vsigma n =
\vsigma'$. For the operation to be meaningful, $n$ must be less than $|\Delta|$.
\begin{align*}
  \Ltotal {\_} {\_} &: (\vGamma \To \vDelta) \to \N \to \N \tag{Truncation Offset} \\
  \Ltotal \vsigma 0 &:= 0 \\
  \Ltotal {\sext \vsigma  n  \sigma} {1 + m} &:= n + \Ltotal \vsigma m
\end{align*}

Truncation simply drops $n$ local
substitutions regardless of the offset that is associated with 
each local substitution. 
\begin{align*}
  \trunc {\_} {\_} &: (\vsigma : \vGamma \To \vDelta) \to  (n:\N) 
                     \;\to \trunc{\vGamma} {\Ltotal {\vsigma} n} \To
                     \trunc {\vDelta} n
  \tag{Truncation} \\
  \trunc \vsigma 0 &:= \vsigma \\
  \trunc {(\sext \vsigma m  \sigma)} {1 + n} &:= \trunc {\vsigma} n
\end{align*}
Similar to truncation of substitutions, we rely on truncation of contexts, written as $\trunc{\vGamma}{n}$ which simply drops $n$ contexts from the context stack $\vGamma$, i.e. if $\vGamma = \vGamma'; \Gamma_1; \ldots;\Gamma_n$, then $\trunc \vGamma n = \vGamma'$.
Note that $n$ must satisfy $n < |\vGamma|$, otherwise the operation
would not be meaningful. 
We emphasize that no further restrictions are placed on $n$ and hence
our definitions apply to any of the combinations of Axioms $K$, $T$ and $4$
described in the introduction. 

To organize the category of context stacks, we need to define the identity morphism.
The identity morphism is also a unified substitution:
\begin{align*}
  \vect{\id}_{\vGamma} &: \vGamma \To \vGamma \\
  \vect{\id}_{\vGamma} &:= \sext{\sext{\snil \id}{1}{\id};\cdots}{1}{\id}
\end{align*}
We omit the subscript $\vGamma$ when it can be inferred from the context and the
offset $1$:
\begin{align*}
  \vect{\id} = \snil \id; \cdots; \id
\end{align*}
which matches our intuition even better. 

Next, we will see that MoTs and term substitutions in
the form of \Cref{sec:st:syn} are just special cases of unified substitutions.

\subsubsection{MoTs}

As suggested in the previous section, we are readily able to represent MoTs
in the unified form. For clarity, let us distinguish the
cases for modal weakening and modal fusion.  When it is modal weakening, we know 
MoTs must be of the form $\{k+1/l\}$. Let $|\vDelta| = l$. In this case, 
\begin{align*}
  \{k+1/l\} &: \vGamma;\Gamma_1; \cdots; (\Gamma_{k + 1}, \Delta_0);\vDelta \To
              \vGamma;\Delta_0;\vDelta \\
  \{k+1/l\} &= \varepsilon; 
   \underbrace{\id; \cdots;\id}_{|\vGamma|};  \Uparrow^{k + 1} \wk;
    \underbrace{\id; \cdots; \id}_{|\vDelta|}
\end{align*}
where the offset $k + 1$ on the right adds $\Gamma_1; \cdots; (\Gamma_{k + 1},
\Delta_0)$ to $\vGamma$ in the domain
stack and $\wk : \vGamma;\Gamma_1; \cdots; (\Gamma_{k + 1}, \Delta_0) \To {\Delta_0}$.

Fusion is also easily defined:
\begin{align*}
  \{0/l\} &: \vGamma;(\Gamma_1, \Gamma_2); \vDelta \To \vGamma;\Gamma_1; \Gamma_2; \vDelta \\
  \{0/l\} &= \varepsilon; \underbrace{\id; \cdots; \id}_{|\vGamma|}; \wk_1; \Uparrow^0\! 
              \wk_2; \underbrace{\id; \cdots; \id}_{|\vDelta|}
\end{align*}
where the offset $0$ associated with $\wk_2$ allows us to fuse $\Gamma_1$ and $\Gamma_2$, and
$\wk_i : \vGamma;(\Gamma_1, \Gamma_2) \To {\Gamma_i}$ for $i \in \{1, 2\}$.

\subsubsection{Term Substitution}

Given $\mtyping[\vGamma;(\Gamma, \Gamma')]u U$, and
$\mtyping[\vGamma; (\Gamma, x : U, \Gamma'); \vDelta]t T$, the substituted
term $\mtyping[\vGamma; (\Gamma, \Gamma'); \vDelta]{t[u/x]}T$. Thus we can
give the following definition
\begin{align*}
  [u/x] &: \vGamma; (\Gamma, \Gamma'); \vDelta \To \vGamma; (\Gamma, x : U,
          \Gamma'); \vDelta \\
  [u/x] &= \varepsilon; \underbrace{\id; \cdots; \id}_{|\vGamma|}; \sigma; \underbrace{\id; \cdots; \id}_{|\vDelta|}
\end{align*}
where $\sigma : \vGamma; (\Gamma, \Gamma') \To{ (\Gamma, x : U, \Gamma')}$ is
defined by replacing $x$ with $u$ and keep all other variables as is.

\subsection{Applying Unified Substitution}\labeledit{sec:app-usubst}

Now we have already defined unified substitution, but we have not answered how it can
be applied to terms. The definition is actually very straightforward. 
\begin{definition}
  The simultaneous application of a unified substitution $\vsigma$ is defined as
  \begin{align*}
    x[\valpha; \sigma] &:= \sigma(x) \tag{lookup $x$ in $\sigma$}\\
    (\boxit t)[\vsigma] &:= \boxit{t[\vsigma; ()]} \\
    (\unbox n t)[\vsigma] &:= \unbox{m}{(t[\trunc \vsigma n])} 
                            \tag{note $m := \Ltotal {\vsigma}{n}$}\\
    (\lambda x. t) [\snil \sigma] &:= \lambda x. t[\snil{(\sigma, x/x)}] \\
    (\lambda x. t) [\sext \vsigma k \sigma] &:= \lambda x. t[\sext \vsigma k {(\sigma, x/x)}] \\
    (s\ t)[\vsigma] &:= s[\vsigma]\ u[\vsigma]
  \end{align*}
\end{definition}
There are a few comments worth making:
\begin{enumerate}
\item When we write $\vsigma; ()$, we should have written $\sext\vsigma 1 ()$ but since the
  offset is $1$, we take our liberty to omit it. 
\item Notice that with unified substitution, we no longer need to do case analysis on
  \tunbox level as opposed to \Cref{sec:st:mtrans}, because the result level can be
  computed from the offsets in the unified substitution using the truncation offset operation.
\item Moreover, truncation offset respects the UL constraint imposed by the target
  system. We can quickly verify this:
  \begin{itemize}
  \item In system $K$, since $n$ must be $1$, $\Ltotal \vsigma 1 = 1$ must hold because
    $\vsigma$ also only contains $1$'s.
  \item In system $T$, the sum of zero or one $0$ or $1$ gives $0$ or $1$.
  \item In system $S4$, clearly the sum of any number of natural numbers remains
    a natural number.
  \item In system $K4$, the sum of a positive number of positive numbers is positive.
  \end{itemize}
  Thus we are sure $L$ always gives valid UL for all valid inputs. 
\item In the $\lambda$ case, we extend the topmost substitution by substituting $x$
  for $x$.  This extension resembles the usual
  simultaneous substitution in typical type theories. 
\end{enumerate}

We can prove that the simultaneous application is well-defined:
\begin{lemma}
  If $\mtyping t T$ and $\vsigma : \vDelta \To \vGamma$, then
  $\mtyping[\vDelta]{t[\vsigma]}T$. 
\end{lemma}
\begin{proof}
  Induction on $\mtyping t T$. In the $\tunbox$ case, recall that $\trunc\vsigma n :
  \trunc{\vDelta}{\Ltotal \vsigma n} \To \trunc \vGamma n$. 
\end{proof}

\subsection{Composing Unified Substitution}

Next we shall define composition of unified substitutions, and show that it satisfies the
categorical laws:
\begin{definition}
  The composition of unified substitutions is defined as
  \begin{align*}
    \_ \;\circ \_ \;&: \vGamma' \To \vGamma'' \to \vGamma \To \vGamma' \to \vGamma
                      \To \vGamma'' \\
    (\snil \sigma) \circ \vdelta &:= \snil{(\sigma[\vdelta])} \\
    (\sext \vsigma n \sigma) \circ \vdelta &:= 
                                             \sext {(\vsigma \circ  {(\trunc \vdelta n)})}{\Ltotal \vdelta n}{(\sigma[\vdelta])}
  \end{align*}
  where $\sigma_1[\vsigma_2]$ is the iterative application of $\vsigma_2$ to terms in
  $\sigma_1$ as defined in the previous section.
\end{definition}
The recursive case can be visualized in the following figure:
\begin{center}
  \begin{tikzpicture}
    \matrix (m) [matrix of math nodes, row sep=15pt]
    {
      \vGamma'' & = & \trunc{\vGamma''}{\Ltotal{\vsigma_2} n}; & \underbrace{\vDelta_1}_{m_1} & ; & \cdots & ;
      & \underbrace{\vDelta_n}_{m_n}
      \\
      \vGamma' & = & \trunc{\vGamma'}n; & \Gamma_1' & ; & \cdots & ; & \Gamma'_n \\
      \vGamma & = & \trunc\vGamma 1; & & &  &
      & \Gamma \\
    };
    \draw[->] (m-2-3)  -- (m-3-3) node[midway, right] {$\vsigma_1$};
    \draw[->] (m-1-3)  -- (m-2-3) node[midway, right] {$\trunc{\vsigma_2}n$};
    \draw[decorate, decoration={brace, mirror, raise=-1pt}]
    (m-2-3.south west) -- (m-2-8.south east)
    node[midway] (sigma1) {};
    \draw[->] (sigma1.south) -- (m-3-8) node[pos=.4, right=5pt] {$\sigma_1$};
  \end{tikzpicture}
\end{center}
where $\sum m_i = \Ltotal{\vsigma_2} n$.

With the definition of composition, we can prove the categorical laws for
substitutions. The laws are derived from the following lemma:
\begin{lemma}
  $t[\vsigma \circ \vdelta] = t[\vsigma][\vdelta]$
\end{lemma}
\begin{proof}
  We proceed by induction on $t$. Most cases are immediate by IH. We only consdier the
  case where $t = \unbox n t'$:
  \begin{align*}
    \unbox n t'[\vsigma \circ \vdelta]
    &= \unbox{\Ltotal{\vsigma \circ \vdelta} n}{(t'[\trunc{(\vsigma \circ \vdelta)} n])} \\
    &= \unbox{\Ltotal{\vdelta}{\Ltotal\vsigma n}}{(t'[{(\trunc \vsigma n)} \circ (\trunc{\vdelta}{\Ltotal\vsigma
      n})])}
    \tag{by \Cref{lem:L-comp,lem:trunc-comp}}\\
    &= \unbox{\Ltotal \vsigma n}{(t'[\trunc \vsigma n])}[\vdelta]  \\
    &= \unbox{n}{(t')}[\vsigma][\vdelta] 
  \end{align*}
\end{proof}
\begin{lemma}
  The following equations hold:
  \begin{itemize}
  \item $t[\vect\id][\vsigma] = t[\vsigma]$
  \item $t[\vsigma][\vect\id] = t[\vsigma]$
  \item $t[\vsigma_1][\vsigma_2 \circ \vsigma_3] = t[\vsigma_1 \circ
    \vsigma_2][\vsigma_3]$
  \end{itemize}
\end{lemma}
\begin{proof}
  Immediate from the previous lemma.
\end{proof}
Thus we have established that context stacks and unified substitutions constitute a
categorical structure. 

\subsection{Other Properties of Unified Substitutions}\labeledit{sec:usubst-props}

In this section, we study some other basic properties of operations of unified
substitutions that will be very useful in later development. They are necessary
properties of truncation and truncation offset, and characterize the operations from an extensional
point of view. Let $\vsigma : \vGamma' \To \vGamma$ and $\vdelta : \vGamma'' \To \vGamma'$: 

\begin{lemma}[Distributivity of Addition over Truncation Offset]\labeledit{lem:L-add}
  If $n + m < |\vGamma|$, then 
  $\Ltotal \vsigma {n + m} = \Ltotal {\vsigma} {n} + \Ltotal {\trunc \vsigma n}  m$.
\end{lemma}
\begin{proof}
  We proceed by induction on $n$.
  \begin{itemize}[label=Case]
  \item $n = 0$, then both sides are $\Ltotal \vsigma m$.
  \item $n = 1 + n'$, then
    \begin{align*}
      \Ltotal{\sext{\vsigma}k\sigma}{1 + n' + m}
      &= k + \Ltotal\vsigma {n' + m} \\
      &= k + \Ltotal\vsigma {n'}  + \Ltotal{\trunc\vsigma {n'}}m
        \tag{by IH}\\
      &= \Ltotal{\sext\vsigma k\sigma}{1 + n'}  + \Ltotal{\trunc{(\sext\vsigma k \sigma)}{(1 + n')}}m
    \end{align*}
  \end{itemize}
\end{proof}

\begin{lemma}[Distributivity of Composition over Truncation Offset]\labeledit{lem:L-comp}
  If $n < |\vGamma|$, then $\Ltotal {\vsigma \circ \vdelta} n = \Ltotal {\vdelta} {\Ltotal \vsigma n}$.
\end{lemma}
\begin{proof}
  We proceed by induction on $n$.
  \begin{itemize}[label=Case]
  \item $n = 0$, it is immediate.
  \item $n = 1 + n'$, then
    \begin{align*}
      \Ltotal{(\sext\vsigma  k \sigma) \circ \vdelta}{1 + n'}
      &= \Ltotal{\sext{(\vsigma \circ \trunc{\vdelta}k)}{\Ltotal{\vdelta}k}{(\sigma[\vdelta])}}{1 + n'} \\
      &= \Ltotal\vdelta k + \Ltotal{\vsigma \circ \trunc\vdelta k}{n'} \\
      &= \Ltotal\vdelta k + \Ltotal{\trunc\vdelta k}{\Ltotal{\vsigma}{n'}}
        \tag{by IH} \\
      &= \Ltotal\vdelta{k + \Ltotal\vsigma {n'}}
        \tag{by \Cref{lem:L-add}} \\
      &= \Ltotal\vdelta{\Ltotal{\sext \vsigma k \sigma}{1 + n'}}
    \end{align*}
  \end{itemize}
\end{proof}

\begin{lemma}[Distributivity of Addition over Truncation]\labeledit{lem:trunc-sum}
  If $n + m < |\vGamma|$, then 
$\trunc \vsigma {(n + m)} = \trunc {(\trunc \vsigma n)} m$.
\end{lemma}
\begin{proof}
  Immediate.
\end{proof}

\begin{lemma}[Distributivity of Composition over Truncation]\labeledit{lem:trunc-comp}
  If $n < |\vGamma|$,   
  then $\trunc {(\vsigma \circ \vdelta)} n = (\trunc \vsigma n) \circ (\trunc {\vdelta} {\Ltotal \vsigma n })$.
\end{lemma}
\begin{proof}
  We proceed by induction on $n$.
  \begin{itemize}[label=Case]
  \item $n = 0$, immediate.
  \item $n = 1 + n'$, then
    \begin{align*}
      \trunc{((\sext{\vsigma}k\sigma) \circ \vdelta)}{(1+n')}
      &= \trunc{(\sext{(\vsigma \circ \trunc\vdelta k)}{\Ltotal\vdelta k}{(\sigma[\vdelta])}}{(1+n')} \\
      &= \trunc{(\vsigma \circ \trunc\vdelta k)}{n'} \\
      &= \trunc{\vsigma}{n'} \circ (\trunc{(\trunc{\vdelta}k)}{\Ltotal\vsigma{n'}})
      \tag{by IH} \\
      &= \trunc\vsigma{n'} \circ (\trunc\vdelta{(k + \Ltotal\vsigma{n'})}) \\
      &= \trunc{(\sext\vsigma k \sigma)}{(1 + n')} \circ
        (\trunc\vdelta{\Ltotal{\sext\vsigma k\sigma}{1 + n'}})
    \end{align*}
  \end{itemize}
\end{proof}

\subsection{Comments on Unified Substitution}

It is worth noting that unified substitutions are quite adaptive. Recall that we can
control the range of $\tunbox$ levels to tune which modal logic the system corresponds
to. As long as the \tunbox level constraint is maintained for offsets
in all unified substitutions, then unified substitutions can be used in all four
systems we have been discussing as a representation for simultaneous substitutions.
In fact, the observation of unified substitutions is far reaching, and can be used to
simplify existing work, as to be shown in the next section. 

\section{A Simpler Idempotent Modal Type Theory}

\citet{gratzer_implementing_2019,clouston_fitch-style_2018} discussed Fitch-style
idempotent S4 modal type theories and provided their categorical semantics, substitution
calculi, as well as a normalization algorithm. It appears that many existing
applications of comonadic modal type theories are idempotent S4, so a simpler
formulation definitely helps.  In this section, we briefly discuss
their calculi and define an equivalent formulation with a simpler substitution
calculus.

\subsection{Idempotency in Fitch Style}

This type theory is given by \citet{clouston_fitch-style_2018} and generalized with
dependent types by \citet{gratzer_implementing_2019}. We will only discuss the simply
typed version. The result can be carried over to the dependently typed case in
\citet{gratzer_implementing_2019} as the
fundamental idea remains the same. In this section, we only discuss the core idea. 

In idempotent $S4$, the critical rules are:
\begin{mathpar}
  \inferrule
  {x : T \in \Gamma' \\ \thelock \notin \Gamma'}
  {\typing[\Gamma, \thelock, \Gamma']{x}{T}}

  \inferrule
  {\typing[\Gamma, \thelock]{t}{T}}
  {\typing{\boxit t}{\square T}}

  \inferrule
  {\typing[\Gamma^\thelock]{t}{\square T}}
  {\typing{\unboxf t}{T}}
\end{mathpar}
In the variable rule, according to Kripke semantics, $x$ is looked up in the current
``world'', as ensured by $\thelock \notin \Gamma'$. The introduction rule for $\square$
again simply accesses the next world by inserting a $\thelock$ in the context. The
elimination rule is very different as it does not have an $\tunbox$ level. This is
because with idempotency, all previously accessed worlds can no longer be
distinguished from each other. In other words, we cannot tell whether a world is accessed one world
ago, two worlds ago, or more. In the elimination rule, this is reflected by the operation
$\Gamma^\thelock$, which removes all $\thelock$'s from $\Gamma$. After removal of
$\thelock$'s, all previous worlds are considered as one and visible in the current
world.
Their relative positions are lost.

If we look closer, in the reading of Kripke semantics, the formulation already seems
redundant: despite the fact that we cannot distinguish more than two kinds of worlds
(all of the previous
v.s. the current), syntactically we still maintain more than one $\thelock$ in the
context. Moreover, the substitution calculus given by
\citet{gratzer_implementing_2019} seems quite complex because it needs to handle
more than one $\thelock$'s. For the simplest example, the identity substitution is
typed as
\begin{mathpar}
  \inferrule
  {\Gamma \rhd_\thelock \Delta}
  {\typing{\id}{\Delta}}
\end{mathpar}
where $\Gamma \rhd_\thelock \Delta$ is morally a judgment saying $\Delta$ is $\Gamma$
with more $\thelock$'s. But this judgment has the following rule
\begin{mathpar}
  \inferrule
  { }
  {\Gamma, \thelock, \thelock \rhd_\thelock \Gamma, \thelock}
\end{mathpar}
which makes sense but fails to line up with the immediate intuition. As a result, the eventual
substitution calculus is rather complex and does not have a clear
structure. Hence, in this section, we develop an equivalent system which takes
advantages of certain observations, including unified substitution, and has a simpler
formulation.

\subsection{Idempotency in Kripke Style}

In this section, we develop our target system in Kripke style; there is no technical
reason, as we just want to be consistent with our philosophical choice. The system can
be easily replayed in Fitch style by using $\thelock$ instead of a context stack. Even
though we call it a \emph{stack}, there are only two contexts in it, it is clearer to
write out the contexts explicitly:
\begin{mathpar}
  \inferrule
  {x : T \in \Gamma'}
  {\typing[\Gamma; \Gamma']{x}{T}}

  \inferrule
  {\typing[(\Gamma, \Gamma'); \cdot]{t}{T}}
  {\typing[\Gamma; \Gamma']{\boxit t}{\square T}}

  \inferrule
  {\typing[\cdot; (\Gamma, \Gamma')]{t}{\square T}}
  {\typing[\Gamma; \Gamma']{\unboxf t}{T}}
\end{mathpar}
An experienced eye might spot the similarity between these two contexts and the
dual-context style. They definitely have undeniable relation but they come from
different inspirations. As explained in \Cref{sec:intro}, dual-context formulation
comes from the categorization of validity and truth. In this system, the contexts are
still Kripke-style: for $\Gamma; \Gamma'$, $\Gamma$ denotes assumptions from all
previously accessed world while $\Gamma'$ captures the current one. This view explains
the introduction rule and the elimination rule. In the introduction rule, we just
insert the assumptions $\Gamma'$ in the current world in the conclusion to those in the previous
worlds in the premise, because we are entering a new world and differences between the current world
and previous ones are vanished. Similarly, when eliminating, all assumptions in
previous worlds become visible after being moved to the current world.

Next we do some basic syntactic sanity check. The following lemma is needed to justify the check:
\begin{lemma}\labeledit{lem:mv-prev-assum}
  If $\typing[(\Gamma, x : S); \Gamma']t T$, then $\typing[\Gamma; (x : S, \Gamma')]t
  T$. 
\end{lemma}
\begin{proof}
  Induction on $\typing[(\Gamma, x : S); \Gamma']t T$.
\end{proof}
We check the local soundness for $\square$:
\begin{mathpar}
  \inferrule
  {\inferrule
    {\typing[(\Gamma, \Gamma'); \cdot]{t}{T}}
    {\typing[\cdot; (\Gamma, \Gamma')]{\boxit t}{\square T}}}
  {\typing[\Gamma; \Gamma']{\unboxf{(\boxit t)}}{T}}

  \Longrightarrow

  \typing[\Gamma; \Gamma']{t}{T}
\end{mathpar}
We implicitly applied \Cref{lem:mv-prev-assum} to get $\typing[\Gamma;
\Gamma']{t}{T}$ from $\typing[(\Gamma, \Gamma'); \cdot]{t}{T}$.

We also check the local completeness:
\begin{mathpar}
  \typing[\Gamma; \Gamma']{t}{\square T}

  \Longrightarrow

  \inferrule
  {\inferrule
    {\typing[\cdot; (\Gamma, \Gamma')]{t}{\square T}}
    {\typing[(\Gamma, \Gamma'); \cdot]{\unboxf t}{T}}}
  {\typing[\Gamma; \Gamma']{\boxit{(\unboxf t)}}{\square T}}
\end{mathpar}
Again, we applied \Cref{lem:mv-prev-assum} to get
$\typing[\cdot; (\Gamma, \Gamma')]{t}{\square T}$ from
$\typing[\Gamma; \Gamma']{t}{\square T}$. Thus these rules syntactically make sense,
but we must make sure we are still talking about the same system. 

Comparing two formulations, it is quite easy to see why these two systems ought to be
equivalent. After all, the latter system simply stops maintaining a difference which
the system cannot tell. Moreover, we do not need to define the $\_^\thelock$ operation
because we simply make $\Gamma$ visible by moving it to the current context in the
elimination rule. 

We can formally show their equivalence. It is quite easy after we understand that we only
need to know whether a context is the current one or not. We define a function which
finds the current world in the Fitch-style system:
\begin{align*}
  S(\cdot) &:= (\cdot, \cdot) \\
  S(\Gamma, x : T) &:= (\Delta, (\Delta', x : T))
  \tag{where $(\Delta, \Delta') := S(\Gamma)$} \\
  S(\Gamma, \thelock) &:= (\Gamma^\thelock, \cdot)
\end{align*}
then we have
\begin{lemma}
  If $\typing t T$ in Fitch style and $(\Delta, \Delta') = S(\Gamma)$, then
  $\typing[\Delta; \Delta']t T$ in Kripke style.
\end{lemma}
\begin{proof}
  By induction on the typing judgment.
\end{proof}
The other direction, to some extend, justifies the $\rhd_\thelock$ relation in the
previous section, because of the unnecessary maintenance of extra $\thelock$'s in the
context. We first need a relation between contexts in both styles:
\begin{definition}
  $\Gamma$ (in Fitch style) is more complex than $\Delta$ (in Kripke style), if
  $\Gamma$ is just $\Delta$ with any number of $\thelock$'s inserted.
\end{definition}
Note that this relation literally just adds $\thelock$'s to $\Delta$. After all,
$\Delta$ is in Kripke style and therefore there is no $\thelock$ in it.
\begin{lemma}
  \begin{itemize}
  \item If $\typing[\cdot; \Gamma']t T$ in Kripke style, then $\typing[\Gamma']t T$ in
    Fitch style. 
  \item If $\typing[\Delta; \Gamma']t T$ in Kripke style and $\Gamma$ is more complex
    than $\Delta$, then $\typing[\Gamma, \thelock, \Gamma']t T$ in Fitch style. 
  \end{itemize}
\end{lemma}
\begin{proof}
  By induction on the typing judgment.
\end{proof}

\subsection{Unified Substitution with Idempotency}

Inspired by the idea of unified substitution in \Cref{sec:usubst}, we can also define
a version of idempotent S4. Unlike the ``full blown'' unified substitution, due to
idempotency, there is no longer need to keep track of $\tunbox$ levels and thus the
structure of substitution is largely contracted:
\begin{definition}
  A unified substitution for idempotent S4, $\sigma : \Gamma; \Gamma' \To \Delta;
  \Delta'$, is defined as two local substitutions:
  \begin{enumerate}
  \item $\sigma_1 : \cdot; (\Gamma, \Gamma') \To \Delta$,
  \item $\sigma_2 : \Gamma; \Gamma' \To \Delta'$.
  \end{enumerate}
  We let $\sigma := \sigma_1; \sigma_2$. 
\end{definition}
From the view of unified substitution, $\sigma_2 : \Gamma; \Gamma' \To \Delta'$
represents the top level substitution. However, $\sigma_1 : \cdot; (\Gamma, \Gamma')
\To \Delta$ is curious as it moves $\Gamma$ to the current context in domain. This
treatment is correct and necessary as terms in $\Delta$ must have access to all previous worlds. 

In order to define the application of unified substitution to terms, given
$\sigma : \Gamma; \Gamma' \To \Delta; \Delta'$, we need to supply two morphisms:
$(\Gamma, \Gamma'); \cdot \To (\Delta, \Delta'); \cdot$ and
$\cdot; (\Gamma, \Gamma') \To \cdot; (\Delta, \Delta')$ for substituting $\tbox$ and
$\tunbox$. We define
\begin{align*}
  \widehat\sigma_1 &: \cdot ; (\Gamma, \Gamma') \To (\Delta, \Delta') \\
  \widehat\sigma_1 &:= \sigma_1,\sigma_2
  \tag{concateneating $\sigma_1$ and $\sigma_2$}\\
  \widehat\sigma_2 &: (\Gamma, \Gamma'); \cdot \To \cdot \\
  \widehat\sigma_2 &:= () \\
  \widehat\sigma &: (\Gamma, \Gamma'); \cdot \To (\Delta, \Delta'); \cdot \\
  \widehat\sigma&:= \widehat\sigma_1; \widehat\sigma_2
\end{align*}
Note that we rely on \Cref{lem:mv-prev-assum} to type check $\widehat\sigma_1$. In
particular, since $\sigma_2 : \Gamma; \Gamma' \To \Delta'$, by
\Cref{lem:mv-prev-assum}, we can move $\Gamma$ to the current context, and thus
$\sigma_2 : \cdot; (\Gamma, \Gamma') \To \Delta'$.

Similarly, we have
\begin{align*}
  \widetilde\sigma_1 &: \cdot; (\Gamma, \Gamma') \To \cdot \\
  \widetilde\sigma_1 &:= () \\
  \widetilde\sigma_2 &: \cdot; (\Gamma, \Gamma') \To (\Delta, \Delta') \\
  \widetilde\sigma_2 &:= \sigma_1, \sigma_2 \\
  \widetilde\sigma &: \cdot; (\Gamma, \Gamma') \To \cdot; (\Delta, \Delta') \\
  \widetilde\sigma &:= \widetilde\sigma_1; \widetilde\sigma_2
\end{align*}
So we can define the application of unified substitution as
\begin{align*}
  x[\sigma_1; \sigma_2] &:= \sigma_2(x) \\
  \boxit t[\sigma] &:= \boxit{(t[\widehat{\sigma}])} \\
  \unboxf t[\sigma] &:= \unboxf{(t[\widetilde\sigma])}
\end{align*}
Other cases are omitted.

Clearly, unified substitution has much clearer structure as it simply is a combination
of two local substitutions. We can define the identity substitution as follows:
\begin{align*}
  \id &: \Gamma; \Gamma' \To \Gamma; \Gamma' \\
  \id &:= \wk_1; \wk_2
\end{align*}
where $\wk_1 : \cdot ; (\Gamma, \Gamma') \To \Gamma$ and $\wk_2 : \Gamma; \Gamma' \To
\Gamma'$ are two evident local substitutions.  This definition, again, is much cleaner
and easy to understand.

At last, we define composition:
\begin{align*}
  \_\circ\_ &: \Gamma';\Delta' \To \Gamma'' ; \Delta'' \to \Gamma; \Delta \To \Gamma';
              \Delta' \to \Gamma; \Delta \To \Gamma''; \Delta'' \\
  (\sigma_1; \sigma_2) \circ \delta &:= \sigma_1[\widetilde\delta] ; \sigma_2[\delta]
\end{align*}
Note that
\begin{align*}
  \sigma_1 &: \cdot ; (\Gamma', \Delta') \To \cdot; \Gamma'' \\
  \widetilde\delta &: \cdot ; (\Gamma, \Delta) \To \cdot; (\Gamma', \Delta')
\end{align*}
and $\sigma_2[\delta]$ denote iterative application of substitution as defined above,
so the definition is well-defined.

We can verify that under these definitions of identity and composition, we have a
category of contexts in Kripke style:
\begin{lemma}
  $\id$ and $\_\circ\_$ satisfy categorical laws. 
\end{lemma}
\begin{proof}
  Follow the definition.
\end{proof}
Here we conclude the substitution calculus for idempotent S4. We leave the
generalization to dependent types, an explicit substitution calculus in the style of
\Cref{sec:st:untyped} and a normalization algorithm as future work. 

\section{Presheaf Model and Normalization by Evaluation}\labeledit{sec:presheaf}

In this section, we build a presheaf model for the modal type theory
from which we extract a normalization by evaluation (NbE) algorithm. 

\subsection{Unified Weakening}\labeledit{sec:uweaken}

To understand the Kripke structure of context stacks, we restrict unified
substitutions to only variables. If the resulting sets of unified substitutions also
form a category, then we can use this category as the base category of a presheaf
category, from which we can compute normal forms of terms. Indeed this plan works
out. First let us see how this restricted version of unified substitutions, unified
weakening, is defined:
\begin{definition}
  A unified weakening $\vgamma : \vGamma \To_w \vDelta$ is:
  \begin{align*}
    \vgamma := \varepsilon \sep q(\vgamma) \sep p(\vgamma) \sep \sext
    \vgamma n {}
    \tag{Unified weakenings}
  \end{align*}
  \begin{mathpar}
    \inferrule
    { }
    {\varepsilon: \epsilon ; \cdot \To_w \epsilon ; \cdot}

    \inferrule
    {\vgamma : \vGamma; \Gamma \To_w \vDelta; \Delta}
    {q(\vgamma) : \vGamma; (\Gamma, x : T) \To_w \vDelta; (\Delta, x : T)}

    \inferrule
    {\vgamma : \vGamma; \Gamma \To_w \vDelta; \Delta}
    {p(\vgamma) : \vGamma; (\Gamma, x : T) \To_w \vDelta; \Delta}

    \inferrule
    {\vgamma : \vGamma \To_w \vDelta \\ |\vGamma'| = n}
    {\sextt \vgamma n : \vGamma; \vGamma' \To_w \vDelta; \cdot}
  \end{mathpar}
  In the last rule, the offset $n$ is again parametric and its choice
  determines which modal logic the system corresponds to.
\end{definition}
The $q$ constructor is the identity extension of the weakening
$\vgamma$, while $p$ accommodates weakening of an individual
context. These constructors are typical in the category of
weakenings. To accommodate MoTs, we add to the category of 
weakenings the last rule which modal-transforms a context stack. 
Note that we also write $\vect \id$ for the identity unified weakening.
Since a unified weakening is also a unified substitution after all, its
truncation, truncation offset, composition, application and their properties are all inherited from
unified substitutions. Intuitively, unified weakenings is unified substitutions in which
local substitutions are all local weakenings. This intuition suggests that unified
weakenings are closed under composition:
\begin{lemma}
  If $\vgamma' : \vGamma' \To_w \vGamma''$ and $\vgamma : \vGamma \To_w \vGamma'$, then
  $\vgamma' \circ \vgamma : \vGamma \To_w \vGamma''$. 
\end{lemma}
\begin{proof}
  Induction on $\vgamma' : \vGamma' \To_w \vGamma''$ and analyze $\vgamma : \vGamma
  \To_w \vGamma'$ when necessary. 
\end{proof}
With $\vect \id$ and closure
of composition, context stacks and unified weakenings form a category, called $\WC$, (in fact a
subcategory of the one with unified substitutions as morphisms).
Moreover, unified weakenings are sufficient to encode any MoT.

\subsection{Presheaf Model}

We consider the presheaf category over \WC, namely the category with objects as 
functors $\WC^{op} \To \SetC$ and morphisms as natural transformations. We
interpret both types and contexts as presheaves (denoted as $F$ and $G$), and terms as
natural transformations.

\subsubsection{Interpreting Types and Contexts}

Given two presheaves $F$ and $G$, we form a presheaf exponential $F \hfunc G$ by
\begin{align*}
  F \hfunc G &: \WC^{op} \To \SetC \\
  (F \hfunc G)_{\;\vGamma} &:= \forall \vDelta \To_w \vGamma. F_{\;\vDelta} \to G_{\;\vDelta}
\end{align*}
Note that the right hand side is in fact a set of natural transformations. 
This construction is justified by the Yoneda lemma.

To model $\square$, given a presheaf $F$, we define
\begin{align*}
  \hsquare F &: \WC^{op} \To \SetC \\
  (\hsquare F)_{\;\vGamma} &:= F_{\vGamma; \cdot}
\end{align*}
Alternatively, $\hsquare F = F \circ (- ; \cdot)$.

We then move on to interpret types in the context stack calculus:
\begin{align*}
  \intp{\_} &: \Typ \to \WC^{op} \To \SetC \\
  \intp{B} &:= \Ne\ B \\
  \intp{\square T} &:= \hsquare \intp{T} \\
  \intp{S \func T} &:= \intp{S} \hfunc \intp{T}
\end{align*}
where $\Ne\ T\ \vGamma$ is the set of well-formed neutral terms in $\vGamma$ of type
$T$ as specified by \Cref{sec:st:syn}. Similarly, $\Nf\ T\ \vGamma$ denotes those
well-formed normal terms in $\vGamma$ of type $T$. $\Ne\ T$ is a presheaf as all
neutrals can be applied to unified weakening and so is $\Nf\ T$. We use $a$ or $b$ to
represent an element in $\intp{T}_{\;\vGamma}$. 

We interpret the contexts and context stacks iteratively:
\begin{align*}
  \intp{\_} &: \Ctx \to \WC^{op} \To \SetC \\
  \intp{\cdot} &:= \hat\top \\
  \intp{\Gamma, x : T} &:= \intp{\Gamma} \hat\times \intp{T}
\end{align*}
where $\hat\top$ and $\hat\times$ denotes the pointwise constructions for a fixed terminal object
and products in presheaves, respectively. We use $*$ to represent the only element of the singleton
set of our chosen terminal presheaf. We use $\rho$ to represent a element in
$\intp{\Gamma}_{\;\vGamma}$. 

We interpret context stacks with $\Sigma$ types as we need to know the number of
contexts to truncate from the parameter context stack:
\begin{align*}
  \intp{\_} &: \vect\Ctx \to \WC^{op} \To \SetC \\
  \intp{\epsilon; \Gamma} &:= \hat\top \hat \times \intp{\Gamma} \\
  \intp{\vGamma; \Gamma}_{\;\vDelta} &:= (\Sigma n <
                                     |\vDelta|. \intp{\vGamma}_{\trunc\vDelta n}) \times \intp{\Gamma}_{\;\vDelta}
\end{align*}
We use $\vrho$ to represent a element in $\intp{\vGamma}_{\;\vDelta}$. By definition,
$\intp{\vGamma}_{\;\vDelta}$ must be a product type, in which the second projection
is the interpretation of the topmost context. We will use this fact very frequently in
our later development.

All these interpretations specify the object part. To prove the interpretations do
form functors, we also need to supply the morphism part. Namely given
$\vgamma : \vDelta \To_w \vGamma$ and a target presheaf $F$, we should supply a
function $F_{\;\vGamma} \to F_{\;\vDelta}$.  Given $a \in F_{\;\vGamma}$, we write
$a[\vgamma] : F_{\;\vDelta}$ for the result of applying this function to $a$. We call
this function \emph{monotonicity}. We
intentionally overload square brackets here because monotonicity is just the semantic
version of applying a unified weakening. We use the same notation for all interpretations and
they can be distinguished by the target to which monotonicity applies.
\begin{lemma}
  The interpretation of types $\intp{T}$ is monotone. 
\end{lemma}
\begin{proof}
  Given $\vgamma : \vDelta \To_w \vGamma$, the theorem requires to give a function
  $\intp{T}_{\;\vGamma} \to \intp{T}_{\;\vDelta}$. We do induction on $T$.
  \begin{itemize}[label=Case]
  \item $T = B$, then the goal becomes $\Ne\ B\ \vGamma \to \Ne\ B\ \vDelta$, which is
    just application of unified weakening on neutral terms.
  \item $T = \square T'$, then the goal becomes $\intp{T'}_{\vGamma; \cdot} \to
    \intp{T'}_{\vDelta; \cdot}$. This function is given by
    \begin{align*}
      (a : \intp{T'}_{\vGamma; \cdot}) \mapsto a[\sextt\vgamma 1]
    \end{align*}
    where $\sextt\vgamma 1 : \vDelta; \cdot \To_w \vGamma; \cdot$ by definition of
    unified weakening. 
  \item $T = S \func T'$, then the goal becomes $(\intp{S} \hfunc \intp{T'})_{\;\vGamma}
    \to (\intp{S} \hfunc \intp{T'})_{\;\vDelta}$, or after more expansion:
    \begin{align*}
      (\forall \vDelta' \To_w \vGamma. \intp{S}_{\;\vDelta'} \to \intp{T'}_{\;\vDelta'})
      \to (\forall \vDelta' \To_w \vDelta. \intp{S}_{\;\vDelta'} \to \intp{T'}_{\;\vDelta'})
    \end{align*}
    We give this function as
    \begin{align*}
      (f, \vgamma ': \vDelta' \To_w \vDelta, a : \intp{S}_{\;\vDelta'})
      \mapsto f(\vgamma \circ \vgamma', a)
    \end{align*}
  \end{itemize}
\end{proof}
Note that monotonicity of $\intp{\square T}$ is ambiguous, because $\intp{\square
  T}_{\;\vGamma} = \intp{T}_{\vGamma; \cdot}$. For example, given $a : \intp{\square
  T}_{\;\vGamma}$ and $\vgamma : \vDelta \To_w \vGamma$, both $a[\vgamma]$ and
$a[\sextt\vgamma 1]$ are valid and equal. This ambiguity unfortunately carries through
the whole section but hopefully it does not create more problems than requiring a
small amount of extra attention.

\begin{lemma}
  The interpretation of context $\intp{\Gamma}$ is monotone.
\end{lemma}
\begin{proof}
  The monotonicity follows immediately from that of $\hat\top$, of $\hat\times$ and of $\intp{T}$.
\end{proof}

\begin{lemma}
  The interpretation of context stacks $\intp{\vGamma}$ is monotone.
\end{lemma}
\begin{proof}
  We proceed by induction on $\vGamma$. The first case is immediate. We only
  consider the case where $\vGamma = \vGamma'; \Gamma$. Assuming
  $\vgamma : \vDelta \To_w \vDelta'$, we should define
  $\intp{\vGamma'; \Gamma}_{\;\vDelta'} \to \intp{\vGamma'; \Gamma}_{\;\vDelta}$. We
  further expand the goal:
  \begin{align*}
    (\Sigma n < |\vDelta'|. \intp{\vGamma'}_{\trunc{\vDelta'}n}) \times \intp{\Gamma}_{\;\vDelta'}
    \to (\Sigma n < |\vDelta|. \intp{\vGamma'}_{\trunc\vDelta n}) \times \intp{\Gamma}_{\;\vDelta}
  \end{align*}
  It is immediate to give $\intp{\Gamma}_{\;\vDelta}$ from $\intp{\Gamma}_{\;\vDelta'}$ due to
  the monotonicity of $\intp{\Gamma}$. We then shall show
  $\Sigma n < |\vDelta|. \intp{\vGamma'}_{\trunc\vDelta n}$ from
  $\Sigma n < |\vDelta'|. \intp{\vGamma'}_{\trunc{\vDelta'} n}$.
  \begin{align*}
    &n < |\vDelta'| \tag{assumption} \\
    &\intp{\vGamma'}_{\trunc{\vDelta'}n} \tag{assumption} \\
    &\trunc\vgamma n : \trunc\vDelta{\Ltotal\vgamma n} \To_w \trunc{\vDelta'} n \text{ and } \Ltotal\vgamma n < |\vDelta|
      \tag{truncation law for unified weakening} \\
    &\intp{\vGamma'}_{\trunc\vDelta{\Ltotal\vgamma n}}
      \byIH
  \end{align*}
  We let the first projection of the $\Sigma$ in the conclusion to be $\Ltotal\vgamma n$. To summarize, the proof
  corresponds to the following function
  \begin{align*}
    ((n, \vrho : \intp{\vGamma'}_{\trunc{\vDelta'} n}), \rho : \intp{\Gamma}_{\;\vDelta'})
    \mapsto ((\Ltotal\vgamma n, \vrho[\trunc\vgamma n]), \rho[\vgamma])
  \end{align*}

  This concludes the
  monotonicity of $\intp{\vGamma}$.
\end{proof}

\subsubsection{Functoriality of Interpretations}

These monotonicity proofs constructively return elements in target
sets following given unified weakenings and can be extracted as programs. For rigorousness, we can also verify the
functorial laws so that we can confirm the interpretations are actually all
functors. In particular, given the object and morphism parts specified above, we only
need to show the identity and composition laws. 
\begin{lemma}
  $\intp T$ is a functor. 
\end{lemma}
\begin{proof}
  We first prove that given $a : \intp{T}_{\;\vGamma}$, $a[\vect \id]$ is the
  identity function.
  \begin{itemize}[label=Case]
  \item $T= B$, then it is the same as applying identity weakening to a neutral term,
    so we surely get the same neutral term back.
  \item $T = \square T'$, then $a : \intp{T'}_{\vGamma; \cdot}$ and we can immediately
    resort to the IH.
  \item $T = S \func T'$, then we should show
    \begin{align*}
      (f, \vgamma: \vDelta \To_w \vGamma, a : \intp{S}_{\;\vDelta})
      \mapsto f(\vgamma \circ \vect \id, a)
    \end{align*}
    By categorical laws of $\WC$, we know $\vgamma \circ \vect \id = \vgamma$. Thus
    the definition just becomes an expansion of $f$.
  \end{itemize}

  We then show the composition law. Given $\vgamma : \vGamma' \To_w \vGamma$ and
  $\vgamma' : \vGamma'' \To_w \vGamma'$, we should show that
  $a[\vgamma \circ \vgamma'] = a[\vgamma][\vgamma']$. Again, when
  $T = B$ it reduces to a property for neutral terms and $T = \square T'$ we can apply
  IH immediately. When $T = S \func T'$, then both sides become
  \begin{align*}
    (f, \vgamma'': \vDelta \To_w \vGamma, a : \intp{S}_{\;\vDelta})
    \mapsto f(\vgamma'' \circ (\vgamma' \circ \vgamma), a)
  \end{align*}
  and 
  \begin{align*}
    (f, \vgamma'': \vDelta \To_w \vGamma, a : \intp{S}_{\;\vDelta})
    \mapsto f((\vgamma'' \circ \vgamma') \circ \vgamma, a)
  \end{align*}
  which just reduce to associativity of composition of unified weakenings. 
\end{proof}

\begin{lemma}
  $\intp \Gamma$ is a functor. 
\end{lemma}
\begin{proof}
  The functorial laws follow immediately from those of $\hat \top$, $\hat\times$ and
  $\intp{T}$. 
\end{proof}

\begin{lemma}
  $\intp{\vGamma}$ is a functor. 
\end{lemma}
\begin{proof}
  We only show the case where $\vGamma = \vGamma'; \Gamma$. First we consider the
  identity law. We should show that given $\vrho : \intp{\vGamma; \vGamma}_{\;\vDelta}$,
  $\vrho[\vect \id]$ is still $\vrho$. Expanding the definition, it
  means the following function should be identity
  \begin{align*}
    ((n, \vrho' : \intp{\vGamma'}_{\trunc \vDelta n}), \rho : \intp{\Gamma}_{\;\vDelta})
    \mapsto
    ((\Ltotal{\vect \id} n, \vrho'[\trunc{\vect \id} n]), \rho[\vect \id])
  \end{align*}
  Since $\Ltotal{\vect \id} n = n$ and $\trunc{\vect \id}n = \vect \id$, we have
  $\vrho'[\trunc{\vect \id}n] = \vrho'$ by IH. $\rho = \rho[\vect \id]$ by the
  functoriality of $\intp{\Gamma}$.
  
  We then show the composition law. Given $\vgamma : \vDelta' \To_w \vDelta$ and
  $\vgamma' : \vDelta'' \To_w \vDelta'$, we should show
  $\vrho[\vgamma][\vgamma'] = \vrho[\vgamma \circ
  \vgamma']$. Expanding the definitions and let $((n, \vrho'), \rho) := \vrho$, we have
  \begin{align*}
    \vrho[\vgamma][\vgamma']
    = ((\Ltotal{\vgamma'}{\Ltotal\vgamma n}, \vrho'[\trunc \vgamma n][\trunc{\vgamma'}{\Ltotal \vgamma n}]), \rho[\vgamma][\vgamma'])
  \end{align*}
  and
  \begin{align*}
    \vrho[\vgamma \circ \vgamma']
    = ((\Ltotal{\vgamma \circ \vgamma'}n,\vrho'[\trunc{(\vgamma \circ \vgamma')}n]),
    \rho[\vgamma \circ \vgamma'])
  \end{align*}
  We know that $\Ltotal{\vgamma \circ \vgamma'}n = \Ltotal{\vgamma'}{\Ltotal\vgamma n}$ and
  $\trunc{(\vgamma \circ \vgamma')}n = (\trunc\vgamma n) \circ (\trunc{\vgamma'}{\Ltotal \vgamma n})$ are
  basic properties for unified weakenings. These allow us to prove
  \begin{align*}
    \vrho'[\trunc \vgamma n][\trunc{\vgamma'}{\Ltotal \vgamma n}] = \vrho'[\trunc{(\vgamma \circ \vgamma')}n]
  \end{align*}
  by IH.

  Lastly, $\rho[\vgamma][\vgamma'] = \rho[\vgamma \circ \vgamma']$ holds by
  functoriality of $\intp{\Gamma}$.
\end{proof}

\subsubsection{Normalization}

Next, we shall interpret the terms as natural transformation. Just like the syntactic
case, we need to define $L$ and truncation for the interpretation of context
stacks:
\begin{align*}
  \Ltotal{\_}{\_} &: \intp{\vGamma}_{\;\vDelta} \to \N \to \N \tag{Truncation Offset} \\
  \Ltotal \vrho 0 &:= 0 \\
  \Ltotal {((n, \vrho),\rho)} {1 + m}&:= n + \Ltotal \vrho m  \\[5pt]
  \trunc {\_} {\_} &: (\vrho : \intp{\vGamma}_{\;\vDelta})\; (n:\N) \to
                     \intp{\trunc\vGamma n}_{\trunc \vDelta {\Ltotal \vrho  n}}
  \tag{Truncation}\\
  \trunc \vrho 0 &:= \vrho \\
  \trunc {((n, \vrho), \rho)}{1 + m} &:= \trunc \vrho m
\end{align*}
In these two functions, we implicitly assume the truncation size, $n$, is strictly smaller
than $|\vGamma|$ for the truncation to make sense. With these two functions defined, we
can define the interpretation of terms as natural transformations:
\begin{align*}
  \intp{\_} &:  \mtyping t T \to \intp{\vGamma} \To \intp{T} \\
  \intp{t}_{\;\vDelta} &: \intp{\vGamma}_{\;\vDelta} \to \intp{T}_{\;\vDelta}
  \tag{expanded form}\\
  \intp{x}_{\;\vDelta} ((\_, \rho))
            &:= \rho(x) \tag{where $\rho(x)$ performs projections properly for lookup} \\
  \intp{\boxit t}_{\;\vDelta}(\vrho)
            &:= \intp{t}_{\vDelta; \cdot} ((1, \vrho), *) \tag{$*$ is \emph{the}
              element of the chosen singleton set} \\
  \intp{\unbox n t}_{\;\vDelta}(\vrho)
            &:= \intp{t}_{\;\trunc \vDelta m}(\trunc \vrho n)[\sextt{\vect \id} m]
  \tag{where $m := \Ltotal \vrho n$ and $\sextt{\vect \id} m: \vDelta \To_w \trunc \vDelta m ; \cdot$}\\
  \intp{\lambda x : S. t}_{\;\vDelta}(\vrho)
            &:= (\vgamma : \vDelta' \To_w \vDelta)(a :
              \intp{S}_{\;\vDelta'}) \mapsto
              \intp{t}_{\;\vDelta'} ((\pi, (\rho, a)))
              \tag{where $(\pi, \rho) := \vrho[\vgamma]$}  \\
  \intp{t\ s}_{\;\vDelta}(\vrho) &:= \intp{t}_{\;\vDelta} (\vrho, \vect{\id}_{\;\vDelta}~,~
                                     \intp{s}_{\;\vDelta}(\vrho))    
\end{align*}

We need a reification function to read from values in natural transformations back to
syntactic terms. This is done by defining two natural transformations: reflect goes
from neutral terms to natural transformations, and reify goes from natural transformations to normal
terms.
\begin{align*}
  \downarrow^T &: \intp{T} \To \Nf\ T \\
  \downarrow^B_{\;\vGamma}(a) &:= a \\
  \downarrow^{\square T}_{\;\vGamma}(a : \intp{T}_{\vGamma; \cdot})
               &:= \boxit \downarrow^T_{\vGamma; \cdot}(a) \\
  \downarrow^{S \func T}_{\vGamma; \Gamma}(a)
               &:= \lambda x. \downarrow^T_{\vGamma; (\Gamma, x : S)}(a~(p(\vect \id)~,~{\uparrow^S_{\vGamma; (\Gamma, x : S)}\!(x)}))
                 \tag{where $p(\vect \id) : \vGamma; \Gamma, x{:} S \To_w \vGamma; \Gamma$}\\[5pt]
  \uparrow^T &: \Ne\ T \To \intp{T} \\
  \uparrow^B_{\;\vGamma}(v) &:= v \\
  \uparrow^{\square T}_{\;\vGamma}(v)
               &:= \uparrow^T_{\vGamma; \cdot}({\unbox 1 v}) \\
  \uparrow^{S \func T}_{\vGamma; \Gamma}(v)
               &:= (\vgamma : \vDelta \To_w \vGamma;\Gamma)(a)
                 \mapsto \uparrow^T_{\vDelta}(v[\vgamma]\ \downarrow^S_{\vDelta}(a))
\end{align*}

Finally, we define the identity environment $\intp{\vGamma}_{\;\vGamma}$ as 
\begin{align*}
  \uparrow &: (\vGamma : \vect\Ctx) \to \intp{\vGamma}_{\;\vGamma} \\
  \uparrow^{\epsilon; \cdot} &:= (*, *) \\
  \uparrow^{\vGamma; \cdot} &:= ((1, \uparrow^{\vGamma}), *) \\
  \uparrow^{\vGamma; \Gamma, x : T} &:= (\pi, (\rho, \uparrow^{T}_{\vGamma; \Gamma, x : T}\!\!(x)))\hfill
      \tag{(where $(\pi, \rho) := \uparrow^{\vGamma; \Gamma} [p(\vect \id)]$)}
\end{align*}

This allows us to finally write down the normalization algorithm:
\begin{definition}
  A normalization by evaluation algorithm given $\mtyping t T$ is
  \begin{align*}
    \nbe^T_{\;\vGamma}(t) &:= \downarrow^T_{\;\vGamma} (\intp{t}_{\;\vGamma}(\uparrow^{\vGamma}))
  \end{align*}
\end{definition}

\subsection{Properties of Interpretations}

In this section we consider properties used for correctness proofs below.

\subsubsection{Basic Properties}

Lemmas in this section are semantic counterparts of those in \Cref{sec:usubst-props}.
Let $\vrho \in \intp{\vGamma}_{\;\vDelta}$ and $\vgamma : \vDelta' \To_w \vDelta$:

\begin{lemma}[Distributivity of Addition over Truncation Offset] \labeledit{lem:L-add-sem}
  If $n + m < |\vGamma|$, then $\Ltotal \vrho {(n + m)} = \Ltotal \vrho n +
  \Ltotal{\trunc \vrho n} m$. 
\end{lemma}
\begin{proof}
  This lemma is essentially the semantic counterpart for \Cref{lem:L-add}. The proof
  is also essentially the same. 
\end{proof}

\begin{lemma}[Monotonicity of Truncation Offset]\labeledit{lem:L-mon}
  If $n < |\vGamma|$, then $\Ltotal{\vrho[\vgamma]} n = \Ltotal\vgamma{\Ltotal\vrho n}$.
\end{lemma}
\begin{proof}
  This lemma is essentially the semantic counterpart for \Cref{lem:L-comp}. We shall
  proceed the proof in a similar way:
  \begin{itemize}[label=Case]
  \item $n = 0$, immediate.
  \item $n = 1 + n'$, then $((k, \vrho'), \rho) := \vrho$ and 
    \begin{align*}
      \Ltotal{((k, \vrho'), \rho)[\vgamma]}{1 + n'}
      &= \Ltotal{((\Ltotal\gamma k, \vrho'[\trunc \vgamma k]), \rho[\vgamma])}{(1+n')}\\
      &= \Ltotal\gamma k + \Ltotal{\vrho'[\trunc\vgamma k]}{n'} \\
      &= \Ltotal\gamma k + \Ltotal{\trunc \vgamma k}{\Ltotal{\vrho'}{n'}}
        \byIH\\
      &= \Ltotal\vgamma {k + \Ltotal{\vrho'}{n'}}
        \tag{by \Cref{lem:L-add}} \\
      &= \Ltotal\vgamma {\Ltotal{\vrho'}{(1+n')}}
    \end{align*}
  \end{itemize}
\end{proof}

\begin{lemma}[Monotonicity of Truncation]\labeledit{lem:trunc-mon}
  If $n < |\vGamma|$, then $\trunc{\vrho[\vgamma]}n = (\trunc\vrho n)[\trunc\vgamma{\Ltotal\vrho n}]$.
\end{lemma}
\begin{proof}
  This lemma is essentially the semantic counterpart for \Cref{lem:trunc-comp}. We shall
  proceed the proof in a similar way:
  \begin{itemize}[label=Case]
  \item $n = 0$, immediate.
  \item $n =1+n'$, then $((k, \vrho'), \rho) := \vrho$ and
    \begin{align*}
      \trunc{(((k, \vrho'), \rho)[\vgamma])}{(1 + n')}
      &= \trunc{((\Ltotal \vgamma k, \vrho'[\trunc \vgamma k]), \rho[\vgamma])}{(1 + n')} \\
      &= \trunc{\vrho'[\trunc \vgamma k]}{n'} \\
      &= (\trunc{\vrho'}{n'})[\trunc{\trunc \vgamma k}{\Ltotal{\vrho'}{n'}}]
        \byIH\\
      &= (\trunc{\vrho'}{n'})[\trunc\vgamma{(k + \Ltotal{\vrho'}{n'})}]
      \tag{by \Cref{lem:trunc-sum}} \\
      &= (\trunc\vrho{(1 + n')})[\trunc\vgamma{\Ltotal\vrho{1 + n'}}]
    \end{align*}
  \end{itemize}
\end{proof}

\subsubsection{Naturality of $\intp{t}$}

As we claimed, this interpretation is indeed a natural transformation:
\begin{lemma}
  Given $\mtyping t T$, $\intp{t}$ is a natural transformation from
  $\intp{\vGamma}$ to $\intp{T}$:
  \begin{itemize}
  \item Given $\vrho \in \intp{\vGamma}_{\;\vDelta}$, $\intp{t}_{\;\vDelta}(\vrho) \in \intp{T}_{\;\vDelta}$.
  \item Given $\vgamma : \vDelta \To_w \vDelta'$ and
    $\vrho \in \intp{\vGamma}_{\;\vDelta'}$,
    $\intp{t}_{\;\vDelta'}(\vrho)[\vgamma] = \intp{t}_{\;\vDelta}(\vrho[\vgamma])$:

    \begin{center}
      \begin{tikzcd}[column sep=large]
        \intp{\vGamma}_{\;\vDelta'}
        \ar{r}{\_[\vgamma]}
        \ar{d}{\intp{t}_{\;\vDelta'}}
        & \intp{\vGamma}_{\;\vDelta}
        \ar{d}{\intp{t}_{\;\vDelta}}
        \\
        \intp{T}_{\;\vDelta'}
        \ar{r}{\_[\vgamma]}
        & \intp{T}_{\;\vDelta} 
      \end{tikzcd}
    \end{center} 
  \end{itemize}
\end{lemma}

\begin{proof}
  We proceed by induction on $\mtyping t T$.
  \begin{itemize}[label=Case]
  \item $t= x$, then it works immediately by analysis.
  \item $t = \boxit t'$, then
    \begin{align*}
      \intp{t}_{\;\vDelta'}(\vrho)[\vgamma]
      &= \intp{t'}_{\vDelta' ; \cdot}((1, \vrho), *)[\sextt\vgamma 1] \\
      &= \intp{t'}_{\vDelta ; \cdot}(((1, \vrho), *)[\sextt\vgamma 1])
        \byIH \\
      &= \intp{t'}_{\vDelta ; \cdot}((1, \vrho[\vgamma]), *)
        \tag{by monotonicity of $\intp{\vGamma; \cdot}$} \\
      &= \intp{t}_{\;\vDelta}(\vrho[\vgamma])
    \end{align*}
    
  \item $t = \unbox n t'$, then the left hand side has
    \begin{align*}
      \intp{t}_{\;\vDelta'}(\vrho)[\vgamma]
      &= \intp{t'}_{\trunc{\vDelta'}{\Ltotal\vrho n}}(\trunc\vrho n)[\sextt{\vect \id}{\Ltotal\vrho n}][\vgamma] \\
      &= \intp{t'}_{\trunc{\vDelta'}{\Ltotal\vrho n}}(\trunc \vrho n)[\sextt{\vect \id}{\Ltotal\vrho n} \circ \vgamma]
        \tag{by functoriality}
    \end{align*}
    and the right hand side has
    \begin{align*}
      \intp{t}_{\;\vDelta}(\vrho[\vgamma])
      &= \intp{t'}_{\trunc\vDelta{\Ltotal{\vrho[\vgamma]} n}}(\trunc{\vrho[\vgamma]} n)[\sextt{\vect \id}{\Ltotal\vrho n}] \\
      &= \intp{t'}_{\trunc\vDelta m}(\vrho')[\sextt{\vect \id} m]
        \tag{where $m := \Ltotal\vgamma{\Ltotal\vrho n} = \Ltotal{\vrho[\vgamma]} n$ by
        \Cref{lem:L-mon}} \\
      \tag{$\vrho' := \trunc \vrho n[\trunc\vgamma{\Ltotal\vrho n}] = \trunc{\vrho[\vgamma]}n$
        by \Cref{lem:trunc-mon}} \\
      &= \intp{t'}_{\trunc{\vDelta'}{\Ltotal\vrho n}}(\trunc \vrho
        n)[\trunc\vgamma{\Ltotal\vrho n}][\sextt{\vect \id}m]
      \byIH \\
      &= \intp{t'}_{\trunc{\vDelta'}{\Ltotal\vrho n}}(\trunc \vrho
        n)[(\sextt{\trunc\vgamma{\Ltotal\vrho n}} 1) \circ (\sextt{\vect \id}m)]
         \tag{by functoriality}
    \end{align*}
    Comparing both sides, we realize that we only need to prove
    \begin{align*}
      \sextt{\vect \id}{\Ltotal\vrho n} \circ \vgamma =
      (\sextt{\trunc\vgamma{(\Ltotal\vrho n)}} 1) \circ (\sextt{\vect \id}m)
    \end{align*}
    but they both evaluate to $\sextt{\trunc \vgamma {\Ltotal\vrho n}} m$. Diagrammatically,
    
    \begin{center}
      \begin{tikzcd}[column sep=huge]
        \vDelta
        \ar{r}{\vgamma}
        \ar{d}{\sextt{\vect \id}m}
        & \vDelta'
        \ar{d}{\sextt{\vect \id}{\Ltotal\vrho n}}\\
        \trunc\vDelta m ; \cdot
        \ar{r}{\sextt{\trunc\vgamma{\Ltotal\vrho n}} 1}
        & \trunc{\vDelta'}{\Ltotal\vrho n}; \cdot
      \end{tikzcd}
    \end{center}
    
  \item $t = \lambda x. t'$ and $T = S \func T'$, we should verify $\intp{t}(\vrho) \in
    \intp{T}_{\;\vDelta}$. We assume $\vrho \in \intp{\vGamma}_{\;\vDelta}$, and further
    assume $\vgamma : \vDelta \To_w \vGamma$ and $a : \intp{S}_{\;\vDelta}$ because
    $\intp{T}$ is a presheaf exponential. Then we should check that the naturality
    condition of the presheaf exponential holds. This is given by further assuming
    $\vgamma' : \vDelta' \To_w \vDelta$, and our goal becomes
    \begin{align*}
      \intp{\lambda x. t}_{\;\vDelta}(\vrho)(\vgamma \circ \vgamma', a[\vgamma']) =
      \intp{\lambda x. t}_{\;\vDelta}(\vrho)(\vgamma, a)[\vgamma']
    \end{align*}
    We evaluate both sides:
    \begin{align*}
      \intp{\lambda x. t}_{\;\vDelta}(\vrho)(\vgamma \circ \vgamma', a[\vgamma'])
      &= \intp{t}_{\;\vDelta}((\pi, (\rho, a[\vgamma'])))
        \tag{where $(\pi, \rho) := \vrho[\vgamma \circ \vgamma']$}
    \end{align*}
    \begin{align*}
      \intp{\lambda x. t}_{\;\vDelta}(\vrho)(\vgamma, a)[\vgamma']
      &= \intp{t}_{\;\vDelta}((\pi', (\rho', a)))[\vgamma'] 
        \tag{where $(\pi', \rho') := \vrho[\vgamma]$} \\
      &= \intp{t}_{\;\vDelta}((\pi', (\rho', a))[\vgamma'])
        \byIH
    \end{align*}
    Notice that by the definition of monotonicity of context stacks, $(\pi, (\rho,
    a[\vgamma'])) = (\pi', (\rho', a))[\vgamma']$ and thus the interpretation is indeed
    well-defined. 
    
  \item $t = t'\ s$, immediate from the fact that $\intp{t'}$ gives a presheaf
    exponential and thus this case is immediately
    discharged.
  \end{itemize}
\end{proof}

\subsection{Relation with MoT}

As previously discussed, syntactic MoTs can be uniformly represented as
unified substitutions (even more specifically, unified weakening). In the presheaf
model, MoTs can be modeled by operations manipulating the evaluation environments. In order to state
the relation, we need the following operation, which intuitively shrinks the
interpretation of a context stack based on a syntactic MoT $\{k/l\}$:
\begin{definition}
  The $\shrink$ operation rearranges $\intp{\vGamma'}_{\;\vDelta}$ to
  $\intp{\vGamma''}_{\;\vDelta}$ according to $\{k/l\} : \vGamma' \To_w \vGamma''$. More concretely,
  if $\{k/l\} : \vGamma' \To_w \vGamma''$ where
  $\vGamma' = \vGamma; \Gamma_0; \cdots ; (\Gamma_k, \Delta_0); \cdots ; \Delta_l$ and
  $\vGamma'' = \vGamma; \Gamma_0; \Delta_0; \cdots ; \Delta_l$, then $\shrink$
  converts $\intp{\vGamma'}_{\;\vDelta}$ to $\intp{\vGamma''}_{\;\vDelta}$:
  \begin{align*}
    \shrink(\_, k, l)
    &: \intp{\vGamma'}_{\;\vDelta} \to \intp{\vGamma''}_{\;\vDelta} \\
    \shrink((\pi, \rho \in \intp{\Gamma_0, \Delta_0}_{\;\vDelta}), 0, 0)
    &:= ((0, (\pi, \rho_1)), \rho_2)
      \tag{where $\rho_1 \in \intp{\Gamma_0}_{\;\vDelta}$ and $\rho_2 \in
      \intp{\Delta_0}_{\;\vDelta}$ are extracted from $\rho$} \\
    \shrink((\pi, \rho), 1 + k', 0)
    &:= ((\Ltotal{\pi, \rho}{(1+k')}, \trunc{(\pi, \rho)}{(1+k')}), \rho_2)
      \tag{where $\rho_2$ is extracted from $\rho$ per above} \\
    \shrink(((n, \vrho), \rho), k, 1 + l')
    &:= ((n, \shrink(\vrho, k, l')), \rho)
  \end{align*}
\end{definition}
We can check that this function is well-defined and it decreases by $l$. Next we shall
investigate its interaction with truncation and truncation offset. We consistently use $\vGamma'$
and $\vGamma''$ in the next few lemmas as in the definition above.

\begin{lemma}\labeledit{lem:L-shrink-1}
  Given $\vrho \in \intp{\vGamma'}_{\;\vDelta}$ and $n \le l$,
  $\Ltotal{\shrink(\vrho, k, l)} n = \Ltotal\vrho n$. 
\end{lemma}
\begin{lemma}\labeledit{lem:trunc-shrink-1}
  Given $\vrho \in \intp{\vGamma'}_{\;\vDelta}$ and $n \le l$,
  $\trunc{\shrink(\vrho, k, l)}n = \shrink(\trunc \vrho n, k, l - n)$. 
\end{lemma}
\begin{proof}[Proof of \Cref{lem:L-shrink-1,lem:trunc-shrink-1}]
  Induction on $n$. By inverting $n \le l$, we know we always in the last case of
  $\shrink$, except when $n = 0$, where the equations hold automatically.
\end{proof}

\begin{lemma}\labeledit{lem:L-shrink-2}
  Given $\vrho \in \intp{\vGamma'}_{\;\vDelta}$ and $n > l$,
  $\Ltotal{\shrink(\vrho, k, l)} n = \Ltotal\vrho{(k + n - 1)}$. 
\end{lemma}
\begin{proof}
  Based on \Cref{lem:L-shrink-1}, it is sufficient to consider the case where $l =
  0$, and thus $n > 0$ and $n = 1 + m$ for some $m$. 

  Let $(\pi, \rho) := \vrho$. We consider the case whether $k = 0$:
  \begin{itemize}[label=Case]
  \item $k = 0$, then
    \begin{align*}
      \Ltotal{\shrink(\vrho, 0, 0)} n
      &= \Ltotal{((0, (\pi, \rho_1)), \rho_2)}{1 + m} \\
      &= \Ltotal{(\pi, \rho_1)}m\\[5pt]
      \Ltotal\vrho{n - 1} &= \Ltotal{(\pi, \rho)}m
    \end{align*}
    We know that for both sides, only $\pi$ impacts the value and thus they are equal.
    
  \item $k \neq 0$, then
    \begin{align*}
      \Ltotal{\shrink(\vrho, k, 0)}n
      &= \Ltotal{((\Ltotal \vrho k, \trunc \vrho k), \rho_2)}{1 + m} \\
      &= \Ltotal \vrho k + \Ltotal{\trunc \vrho k}m \\
      &= \Ltotal\vrho{k + m} \tag{\Cref{lem:L-add-sem}}
    \end{align*}
  \end{itemize}
\end{proof}

However, if we think $\trunc{\shrink(\vrho, k, l)}n = \trunc{\vrho}{(k + n - 1)}$ given the
preconditions of the previous lemma, then we are wrong. For a counterexample, if
$\vrho \in \intp{\vGamma; (\Gamma_0, \Delta_0)}_{\;\vDelta}$, then
$\shrink(\vrho, 0, 0) \in \intp{\vGamma; \Gamma_0; \Delta_0}_{\;\vDelta}$. Moreover, if
we take $n = 1$, then $\trunc{\shrink(\vrho, 0, 0)}1 \in \intp{\vGamma;
  \Gamma_0}_{\;\vDelta}$, while $\trunc\vrho 0 = \vrho$ so they are not even in the same set and cannot be compared!
Fortunately, this is the only group of counterexamples, and we can refine our equation a
bit further. 
\begin{lemma}\labeledit{lem:trunc-shrink-2}
  Given $\vrho \in \intp{\vGamma'}_{\;\vDelta}$, $n > l$ and $k \neq 0$,
  $\trunc{\shrink(\vrho, k, l)}n = \trunc\vrho{(k + n - 1)}$. 
\end{lemma}
\begin{proof}
  If we know $k \neq 0$, then we can compare both sides as they can be verified in the
  same set. The proof goes similarly by induction on $l$ and compute accordingly. 
\end{proof}

If we truncate a bit more, then we can say something when $k = 0$:
\begin{lemma}\labeledit{lem:trunc-shrink-3}
  Given $\vrho \in \intp{\vGamma'}_{\;\vDelta}$, $n > l + 1$,
  $\trunc{\shrink(\vrho, 0, l)}n = \trunc\vrho{(n - 1)}$. 
\end{lemma}
\begin{proof}
  By truncating one more, we avoided the counterexample. The proof again goes very
  similarly. 
\end{proof}

We show that $\shrink$ and monotonicity commute:
\begin{lemma}\labeledit{lem:shrink-mon-comm}
  Given $\vrho \in \intp{\vGamma'}_{\;\vDelta}$ and $\vgamma : \vDelta' \To_w \vDelta$,
  $\shrink(\vrho[\vgamma], k, l) = \shrink(\vrho, k, l)[\vgamma]$. 
\end{lemma}
\begin{proof}
  We proceed by induction on $l$ and consider the case of $k$. The proof is just
  straightforward calculation.
  \begin{itemize}[label=Case]
  \item $l = 0$, $k = 0$, let $(\pi, \rho) := \vrho$,
    \begin{align*}
      \shrink(\vrho[\vgamma], 0, 0)
      &= ((0, (\pi, \rho_1)[\vgamma]), \rho_2[\vgamma]) \\
      &= ((0, (\pi, \rho_1), \rho_2))[\vgamma] \\
      &= \shrink(\vrho, 0, 0)[\vgamma]
    \end{align*}
    
  \item $l = 0$, $k \neq 0$, let $(\pi, \rho) := \vrho$,
    \begin{align*}
      \shrink(\vrho[\vgamma], k, 0)
      &= ((\Ltotal{\vrho[\vgamma]} k, \trunc{\vrho[\vgamma]}k), \rho_2[\vgamma]) \\
      &= ((\Ltotal\vgamma{\Ltotal \vrho k}, \trunc \vrho k[\trunc\vgamma{\Ltotal \vrho k}]),
        \rho_2[\vgamma])
        \tag{by \Cref{lem:L-mon,lem:trunc-mon}}\\
      &= ((\Ltotal \vrho k, \trunc \vrho k), \rho_2)[\vgamma] \\
      &= \shrink(\vrho, k, 0)[\vgamma]
    \end{align*}
    
  \item $l = 1 + l'$, let $((n, \vrho'), \rho) := \vrho$,
    \begin{align*}
      \shrink(\vrho[\vgamma], k, l)
      &= \shrink(((\Ltotal \vgamma n, \vrho'[\trunc \vgamma n]), \rho[\vgamma]), k,
        1+l') \\
      &= ((\Ltotal \vgamma n, \shrink(\vrho'[\trunc \vgamma n], k, l')), \rho[\vgamma])
      \\
      &= ((\Ltotal \vgamma n, \shrink(\vrho', k, l')[\trunc \vgamma n]), \rho[\vgamma])
      \byIH \\
      &= ((n, \shrink(\vrho', k, l')), \rho)[\vgamma] \\
      &= \shrink(\vrho, k, 1 + l')[\vgamma]
    \end{align*}
  \end{itemize}
\end{proof}

Then we shall state the relation between MoT and $\intp{t}$:
\begin{theorem}
  If $\mtyping[\vGamma'']t T$, then given
  $\vrho \in \intp{\vGamma'}_{\;\vDelta}$, $\intp{t\{k/l\}}_{\;\vDelta}(\vrho) =
  \intp{t}_{\;\vDelta}(\shrink(\vrho, k, l))$.
\end{theorem}
\begin{proof}
  We proceed by induction on $t$.
  \begin{itemize}[label=Case]
  \item $t = x$, immediate by analyzing all cases of $\shrink$. 
  \item $t = \boxit t'$,
    \begin{align*}
      \intp{t\{k/l\}}_{\;\vDelta}(\vrho)
      &= \intp{t'\{k/l+1\}}_{\vDelta; \cdot}((1, \vrho), \cdot) \\
      &= \intp{t'}_{\vDelta; \cdot}(\shrink(((1, \vrho), \cdot), k, l + 1))
        \byIH\\
      &= \intp{t'}_{\vDelta; \cdot}((1, \shrink(\vrho, k, l + 1)), \cdot) \\
      &= \intp{t}_{\;\vDelta}(\shrink(\vrho, k, l))
    \end{align*}
    
  \item $t = \unbox n t'$,
    \begin{itemize}[label=Subcase]
    \item $n \le l$, let $n' = \Ltotal{\shrink(\vrho, k, l)} n$ and $\vrho' =
      \shrink(\vrho, k, l)$. 
      \begin{align*}
        \intp{t\{k/l\}}_{\;\vDelta}(\vrho)
        &= \intp{t'\{k/l-n\}}_{\vDelta|_{\Ltotal\vrho n}}(\trunc \vrho n)[\sextt{\vect \id}{\Ltotal\vrho n}] \\
        &= \intp{t'}_{\trunc\vDelta{\Ltotal\vrho n}}(\shrink(\trunc \vrho n, k, l - n))[\sextt{\vect \id}{\Ltotal\vrho n}]
          \byIH \\
        &= \intp{t'}_{\trunc\vDelta{n'}}(\trunc{\vrho'} n)[\sextt{\vect \id}{n'}]
          \tag{by \Cref{lem:L-shrink-1,lem:trunc-shrink-1}} \\
        &= \intp{t}_{\;\vDelta}(\shrink(\vrho, k, l))
      \end{align*}
      
    \item $n > l$, let $n' = \Ltotal{\shrink(\vrho, k, l)}n$ and $\vrho' =
      \shrink(\vrho, k, l)$. 
      \begin{align*}
        \intp{t\{k/l\}}_{\;\vDelta}(\vrho)
        = \intp{t'}_{\trunc\vDelta{(\Ltotal\vrho{k + n - 1})}}(\trunc\vrho{(k + n - 1)})[\sextt{\vect \id}{\Ltotal \vrho{k + n - 1}}]
      \end{align*}
      Notice that there is an implicit syntactic weakening here for $t'$ for the
      interpretation to make sense. 
      \begin{align*}
        \intp{t}_{\;\vDelta}(\shrink(\vrho, k, l))
        = \intp{t'}_{\trunc\vDelta{n'}}(\trunc{\vrho'}n)[\sextt{\vect \id}{n'}]
      \end{align*}
      By \Cref{lem:L-shrink-2}, we can show $n' = \Ltotal\vrho{k + n - 1}$ but as
      previously discussed, we cannot have $\trunc\vrho{(k + n - 1)} = \trunc{\vrho'}n$ in
      general, so we have not done here. In fact, counterexamples for this equation
      only exist when $k = 0$ and $n = l + 1$.

      Nevertheless, we observe that
      $\intp{t'}_{\trunc\vDelta{n'}}(\trunc\vrho{(k + n - 1)}) =
      \intp{t'}_{\trunc\vDelta{n'}}(\trunc{\vrho'}n)$ holds. Notice that $\trunc\vrho{(k + n - 1)} =
      \trunc\vrho l \in
      \intp{\vGamma; (\Gamma_0, \Delta_0)}_{\trunc\vDelta{n'}}$ only
      has more values than $\trunc{\vrho'}n = \trunc{\vrho'}{(l + 1)} \in \intp{\vGamma; \Gamma_0}_{\trunc\vDelta{n'}}$ in the topmost context.
      $\intp{t'}_{\trunc\vDelta{n'}}(\trunc\vrho{(k + n - 1)}) =
      \intp{t'}_{\trunc\vDelta{n'}}(\trunc{\vrho'}n)$ can be seen as weakening as $t'$ does not
      refer to the extra values in $\trunc \vrho n$. 
    \end{itemize}
    
  \item $t = \lambda x : S. t'$, then
    \begin{align*}
      \intp{t\{k/l\}}_{\;\vDelta}(\vrho)
      &= (\vDelta', \vgamma : \vDelta' \To_w \vDelta, a : \intp{S}_{\;\vDelta'})
        \mapsto \intp{t'\{k/l\}}_{\;\vDelta'}((\pi, (\rho, a)))
        \tag{where $(\pi, \rho) := \vrho[\vgamma]$} \\
      &= (\vDelta', \vgamma : \vDelta' \To_w \vDelta, a : \intp{S}_{\;\vDelta'})
        \mapsto \intp{t'}_{\;\vDelta'}(\shrink((\pi, (\rho, a)), k, l)) \\
      &= (\vDelta', \vgamma : \vDelta' \To_w \vDelta, a : \intp{S}_{\;\vDelta'})
        \mapsto \intp{t'}_{\;\vDelta'}(\pi', (\rho', a))
      \tag{by \Cref{lem:shrink-mon-comm}, where $(\pi', \rho') := \shrink(\vrho, k,
        l)[\vgamma]$}\\
      &= \intp{t}_{\;\vDelta}(\shrink(\vrho, k, l))
    \end{align*}
    
  \item $t = t'\ s$, immediate by IH.
  \end{itemize}
\end{proof}

\begin{corollary}\labeledit{cor:modal-trans}
  If $\mtyping[\trunc \vGamma n; \cdot] t T$, then given
  $\vrho \in \intp{\vGamma}_{\;\vDelta}$,
  $\intp{t\{n/0\}}_{\;\vDelta}(\vrho) = \intp{t}_{\;\vDelta}((\Ltotal\vrho n, \trunc \vrho n), *)$.
\end{corollary}
\begin{proof}
  Specialize the previous theorem. Also notice
  $\shrink(\vrho, n, 0) = ((\Ltotal\vrho n, \trunc \vrho n), *)$. 
\end{proof}

\subsection{Relation with Term Substitutions}

Next we consider the relation between term substitution and evaluation
environments. This relation is much more classical than the previous one.

\begin{definition}
  The $\inser$ operation inserts a semantic value into an interpretation of context
  stacks. More formally, given $\vrho \in \intp{\vGamma; (\Gamma, \Gamma');
    \vGamma'}_{\;\vDelta}$ where $|\vGamma'| = n$, and $a \in
  \intp{T}_{\trunc\vDelta{\Ltotal\vrho n}}$, then $\inser(\vrho, n, a) \in \intp{\vGamma;
    (\Gamma, x : T, \Gamma'); \vGamma'}_{\;\vDelta}$.
  \begin{align*}
    \inser((\pi, \rho), 0, a) &:= (\pi, \rho')
                              \tag{where $\rho'$ is $\rho$ with $a$ inserted in the proper spot} \\
    \inser(((m, \vrho'), \rho), 1 + n', a) &:= ((m, \inser(\vrho', n', a)), \rho)
  \end{align*}
\end{definition}

\begin{lemma}\labeledit{lem:insert-L}
  Given $\vrho \in \intp{\vGamma; (\Gamma, \Gamma'); \vGamma'}_{\;\vDelta}$ where
  $|\vGamma'| = n$, and $a \in \intp{T}_{\trunc\vDelta{\Ltotal\vrho n}}$, then
  $\Ltotal{\inser(\vrho, n, a)} m = \Ltotal\vrho m$.
\end{lemma}
\begin{proof}
  This is immediate because $\inser$ does not touch the numbers inside. 
\end{proof}

\begin{lemma}\labeledit{lem:insert-trunc-1}
  Given $\vrho \in \intp{\vGamma; (\Gamma, \Gamma'); \vGamma'}_{\;\vDelta}$ where
  $|\vGamma'| = n$, $a \in \intp{T}_{\trunc\vDelta{\Ltotal\vrho n}}$ and $m \le n$, then
  $\trunc{\inser(\vrho, n, a)} m = \inser(\trunc\vrho m, n - m, a)$. 
\end{lemma}

\begin{lemma}\labeledit{lem:insert-trunc-2}
  Given $\vrho \in \intp{\vGamma; (\Gamma, \Gamma'); \vGamma'}_{\;\vDelta}$ where
  $|\vGamma'| = n$, $a \in \intp{T}_{\trunc\vDelta{\Ltotal\vrho n}}$ and $m > n$, then
  $\trunc{\inser(\vrho, n, a)}m = \trunc\vrho m$. 
\end{lemma}
\begin{proof}
  The previous two lemmas should be immediate considering the behavior of $\inser$. 
\end{proof}

\begin{lemma}\labeledit{lem:insert-mon}
  Given $\vrho \in \intp{\vGamma; (\Gamma, \Gamma'); \vGamma'}_{\;\vDelta}$ where
  $|\vGamma'| = n$, $a \in \intp{T}_{\trunc\vDelta{\Ltotal\vrho n}}$ and
  $\vgamma : \vDelta' \To_w \vDelta$, then
  $\inser(\vrho, n, a)[\vgamma] = \inser(\vrho[\vgamma], n, a[\trunc\vgamma{\Ltotal\vrho n}])$.
\end{lemma}
\begin{proof}
  We proceed by induction on $n$.
  \begin{itemize}[label=Case]
  \item $n = 0$, immediate.
  \item $n = 1 + n'$, let $((m, \vrho'), \rho) := \vrho$:
    \begin{align*}
      \inser(\vrho, n, a)[\vgamma]
      &= ((m, \inser(\vrho', n', a)), \rho)[\vgamma] \\
      &= ((\Ltotal\vgamma m, \inser(\vrho', n', a)[\trunc\vgamma m]), \rho[\vgamma]) \\
      &=((\Ltotal\vgamma m, \inser(\vrho'[\trunc\vgamma m], n', a[\trunc{\trunc\vgamma m}{\Ltotal{\vrho'}{n'}}])), \rho[\vgamma])
         \byIH \\
      &=\inser(((\Ltotal\vgamma m, \vrho'[\trunc\vgamma m]), \rho[\vgamma]), 1 + n', a[\trunc\vgamma{(m + \Ltotal{\vrho'}{n'})}]) \\
      &=\inser(\vrho[\vgamma], 1 + n', a[\trunc\vgamma{\Ltotal\vrho{(1 + n')}}])
    \end{align*}
  \end{itemize}
\end{proof}

\begin{lemma}\labeledit{lem:insert-weaken}
  If $\mtyping[\vGamma; (\Gamma, \Gamma'); \vGamma']{t}{T}$, given $\vrho \in \intp{\vGamma;
    (\Gamma, \Gamma'); \vGamma'}_{\;\vDelta}$, $|\vGamma'| = n$ and $a \in
  \intp{T}_{\trunc\vDelta{\Ltotal\vrho n}}$, then $\intp{t}_{\;\vDelta}(\vrho) =
  \intp{t}_{\;\vDelta}(\inser(\vrho, n, a))$.
\end{lemma}
\begin{proof}
  This is immediate because $t$ does not bind $x$. This can be seen by a
  straightforward induction.
\end{proof}

\begin{theorem}
  If $\mtyping[\vGamma; (\Gamma, x : S, \Gamma'); \vGamma']{t}{T}$ and
  $\mtyping[\vGamma; (\Gamma, \Gamma')]{s}{S}$, 
  Given $\vrho \in \intp{\vGamma; (\Gamma, \Gamma'); \vGamma'}_{\;\vDelta}$ where
  $|\vGamma'| = n$,
  then
  $\intp{t[s/x]}_{\;\vDelta}(\vrho) = \intp{t}_{\;\vDelta}(\inser(\vrho, n,
  \intp{s}_{\trunc\vDelta{\Ltotal\vrho n}}(\trunc \vrho n)))$. 
\end{theorem}
\begin{proof}
  We proceed by induction on $t$ and analyze $n$. Let
  $\vrho' := \inser(\vrho, n, \intp{s}_{\trunc\vDelta{\Ltotal\vrho n}}(\trunc \vrho n))$.
  \begin{itemize}[label=Case]
  \item $t = y$,
    \begin{itemize}[label=Subcase]
    \item $n = 0$ and $y = x$, immediate. Both sides equal
      $\intp{s}_{\;\vDelta}(\vrho)$. 
    \item $n = 0$ and $y \neq x$, immediate. Both sides look up $\vrho$. 
    \item $n > 0$, immediate.
    \end{itemize}
  \item $t = \boxit t'$, then
    \begin{align*}
      \intp{t[s/x]}_{\;\vDelta}(\vrho)
      &= \intp{t'[s/x]}_{\vDelta; \cdot}((1, \vrho); *) \\
      &= \intp{t'}_{\vDelta; \cdot}(\inser(((1, \vrho); *), 1 + n,
        \intp{s}_{\trunc\vDelta{\Ltotal{((1, \vrho); *)}{1 + n}}}(\trunc{((1, \vrho); *)}{(1 + n)})))
        \byIH\\
      &= \intp{t'}_{\vDelta; \cdot}((1, \inser(\vrho, n, \intp{s}_{\trunc\vDelta{\Ltotal\vrho n}}(\trunc \vrho n))), *) \\
      &= \intp{\boxit t'}_{\;\vDelta}(\vrho')
    \end{align*}
    
  \item $t = \unbox m t'$, then
    \begin{align*}
      \intp{t[s/x]}_{\;\vDelta}(\vrho)
      &= \intp{t'[s/x]}_{\trunc\vDelta{\Ltotal\vrho m}}(\trunc\vrho m)[\sextt{\vect \id}{\Ltotal\vrho m}]
    \end{align*}
    Now we need to compare $m$ and $n$ to see if the substitution stays.
    \begin{itemize}[label=Subcase]
    \item $m \le n$, then the substitution is still effective,
      \begin{align*}
        \intp{t[s/x]}_{\;\vDelta}(\vrho) 
        &=\intp{t'}_{\trunc\vDelta{\Ltotal\vrho m}}(\inser(\trunc\vrho m, n - m,
          \intp{s}_{\trunc{(\trunc\vDelta{\Ltotal\vrho m})}{\Ltotal{\trunc\vrho m}{n - m}}}(\trunc{\trunc\vrho m}{(n - m)})))[\sextt{\vect \id}{\Ltotal\vrho m}]
          \byIH \\
        &=\intp{t'}_{\trunc\vDelta{\Ltotal\vrho m}}(\inser(\trunc\vrho m, n - m,
          \intp{s}_{\trunc\vDelta{\Ltotal\vrho n}}(\trunc \vrho n)))[\sextt{\vect \id}{\Ltotal\vrho m}]
           \tag{by \Cref{lem:L-add-sem}} \\
        &=\intp{t'}_{\trunc\vDelta{\Ltotal\vrho m}}(\trunc{\vrho'}m)[\sextt{\vect \id}{\Ltotal\vrho m}]
           \tag{by \Cref{lem:insert-trunc-1}} \\
        &=\intp{\unbox m t'}_{\;\vDelta}(\vrho')
           \tag{notice $\Ltotal\vrho m = \Ltotal{\vrho'} m$ per \Cref{lem:insert-L}}
      \end{align*}
      
    \item $m > n$, then the substitution is truncated. 
      \begin{align*}
        \intp{t[s/x]}_{\;\vDelta}(\vrho)
        &= \intp{t'}_{\trunc\vDelta{\Ltotal\vrho m}}(\trunc\vrho m)[\sextt{\vect \id}{\Ltotal\vrho m} ] \\
        &= \intp{t'}_{\trunc\vDelta{\Ltotal{\vrho'} m}}(\trunc{\vrho'}m)[\sextt{\vect \id}{\Ltotal{\vrho'} m}]
          \tag{by \Cref{lem:insert-L,lem:insert-trunc-2}} \\
        &= \intp{\unbox m t'}_{\;\vDelta}(\vrho')
      \end{align*}
    \end{itemize}
    
  \item $t = \lambda y : S'. t'$, then
    \begin{align*}
      \intp{t[s/x]}_{\;\vDelta}(\vrho)
      &= (\vDelta', \vgamma, a : \intp{S'}_{\;\vDelta'}) \mapsto \intp{t'[s/x]}_{\;\vDelta'}((\pi, (\rho, a)))
        \tag{where $(\pi, \rho) := \vrho[\vgamma]$} \\
      &= (\vDelta', \vgamma, a : \intp{S'}_{\;\vDelta'})
        \mapsto \intp{t'}_{\;\vDelta'}(\inser((\pi, (\rho, a )), n, \intp{s}_{\trunc{\vDelta'}{\Ltotal{(\pi, (\rho, a))}n}}(\trunc{(\pi, (\rho, a))}n)))
        \byIH
    \end{align*}
    We will clear up this equation a little. First,
    \begin{align*}
      \Ltotal{(\pi, (\rho, a))}n = \Ltotal{(\pi, \rho)}n = \Ltotal{\vrho[\vgamma]}n
    \end{align*}
    Next,
    \begin{align*}
      &\ \intp{s}_{\trunc{\vDelta'}{\Ltotal{\vrho[\vgamma]} n}}(\trunc{(\pi, (\rho, a))}n) \\
      =&\ \intp{s}_{\trunc{\vDelta'}{\Ltotal{\vrho[\vgamma]} n}}(\trunc{(\pi, \rho)}n)
      \tag{by \Cref{lem:insert-weaken}} \\
      =&\ \intp{s}_{\trunc{\vDelta'}{\Ltotal{\vrho[\vgamma]} n}}(\trunc{\vrho[\vgamma]}n) \\
      =&\ \intp{s}_{\trunc{\vDelta'}{\Ltotal\vgamma {(\Ltotal\vrho n)}}}(\trunc \vrho n[\trunc\vgamma{\Ltotal\vrho n}])
        \tag{by \Cref{lem:L-mon,lem:trunc-mon}} \\
      =&\ \intp{s}_{\trunc\vDelta{\Ltotal\vrho n}}(\trunc \vrho n)[\trunc\vgamma{\Ltotal\vrho n}]
         \tag{by naturality}
    \end{align*}
    Last, we can pop the $a$ to the top of the environment because $x$ must occur
    before $y$. This concludes the following equation:
    \begin{align*}
      \intp{t[s/x]}_{\;\vDelta}(\vrho)
      &= (\vDelta', \vgamma, a)
        \mapsto \intp{t'}_{\;\vDelta'}((\pi', (\rho', a)))
        \tag{where $(\pi', \rho') := \inser(\vrho[\vgamma], n, \intp{s}_{\trunc\vDelta{\Ltotal\vrho n}}(\trunc \vrho n)[\trunc\vgamma{\Ltotal\vrho n}])$}
    \end{align*}
    However,
    \begin{align*}
      (\pi', \rho') &= \inser(\vrho[\vgamma], n,
                      \intp{s}_{\trunc\vDelta{\Ltotal\vrho n}}(\trunc \vrho n)[\trunc\vgamma{\Ltotal\vrho n}])
      \\
                    &= \inser(\vrho, n, \intp{s}_{\trunc\vDelta{\Ltotal\vrho n}}(\trunc \vrho n))[\vgamma]
                      \tag{by \Cref{lem:insert-mon}}
    \end{align*}
    Therefore
    \begin{align*}
      \intp{t[s/x]}_{\;\vDelta}(\vrho)
      &= (\vDelta', \vgamma, a)
        \mapsto \intp{t'}_{\;\vDelta'}((\pi', (\rho', a))) \\
      &= \intp{t}_{\;\vDelta}(\inser(\vrho, n, \intp{s}_{\trunc\vDelta{\Ltotal\vrho n}}(\trunc \vrho n)))
    \end{align*}
    
  \item $t = t'\ s$, immediate by IH.
  \end{itemize}
\end{proof}

\begin{corollary}\labeledit{cor:subst}
  If $\mtyping[\vGamma; (\Gamma, x : S)]{t}{T}$ and
  $\mtyping[\vGamma; \Gamma]{s}{S}$, 
  Given $(\pi, \rho) \in \intp{\vGamma; \Gamma}_{\;\vDelta}$,
  then
  $\intp{t[s/x]}_{\;\vDelta}((\pi, \rho)) = \intp{t}_{\;\vDelta}((\pi, (\rho,
  \intp{s}_{\;\vDelta}((\pi, \vrho)))))$. 
\end{corollary}

\subsection{Completeness}

We shall move on to prove the completeness of the NbE algorithm given above. The
theorem states
\begin{theorem}\labeledit{thm:nbe-completeness}
  (completeness) If $\mtyequiv{t}{t'}{T}$, then $\nbe^T_{\;\vGamma}(t) = \nbe^T_{\;\vGamma}(t')$. 
\end{theorem}

\subsubsection{What We Should Do}

Ideally, we should prove the completeness theorem through an partial equivalence relation(PER)
between values in the semantics as defined below. However, we take a shortcut here
because we can see that in fact the PER we are about to define induces
equality. Thus, we can just directly go for semantic equality and therefore the
completeness theorem automatically falls into one piece. For the sake of documentary completeness,
we sketch the PER below:
\begin{definition}
  $\Eq_{\;\vGamma}(T)$ is a relation between $\intp{T}_{\;\vGamma}$:
  \begin{align*}
    \Eq_{\;\vGamma}(T) &\subseteq \intp{T}_{\;\vGamma} \times \intp{T}_{\;\vGamma} \\
    \Eq_{\;\vGamma}(B) &:= \{ (a, b) \sep a = b \} \\
    \Eq_{\;\vGamma}(\square T) &:= \{ (a, b) \sep \forall \vgamma : \vDelta \To_w
                               \vGamma. (a[\vgamma], b[\vgamma]) \in
                               \Eq_{\;\vDelta}(T) \} \\
    \Eq_{\;\vGamma}(S \func T) &:= \{ (a, b) \sep \forall \vgamma : \vDelta \To_w
                               \vGamma, (a', b') \in \Eq_{\;\vDelta}(S). (a(\vgamma,
                               a'), b(\vgamma, b')) \in \Eq_{\;\vDelta}(T) \}
  \end{align*}
\end{definition}
We can verify that this relation is indeed a PER.  We generalize $\Eq$ to both
contexts and context stacks.
\begin{definition}
  $\Eq_{\;\vGamma}(\Gamma)$ is a relation between $\intp{\Gamma}_{\;\vGamma}$:
  \begin{align*}
    \Eq_{\;\vGamma}(\Gamma) &\subseteq \intp{\Gamma}_{\;\vGamma} \times \intp{\Gamma}_{\;\vGamma} \\
    \Eq_{\;\vGamma}(\cdot) &:= \{ (*, *) \} \\
    \Eq_{\;\vGamma}(\Gamma, x : T) &:= \{ ((\rho, a), (\rho', b)) \sep (\rho, \rho') \in
                                   \Eq_{\;\vGamma}(\Gamma) \tand (a, b) \in \Eq_{\;\vGamma}(T) \}
  \end{align*}
\end{definition}

\begin{definition}
  $\Eq_{\;\vDelta}(\vGamma)$ is a relation between $\intp{\vGamma}_{\;\vDelta}$:
  \begin{align*}
    \Eq_{\;\vDelta}(\vGamma)
    &\subseteq \intp{\vGamma}_{\;\vDelta} \times \intp{\vGamma}_{\;\vDelta} \\
    \Eq_{\;\vDelta}(\epsilon; \Gamma)
    &:= \{ ((*, \rho), (*, \rho')) \sep (\rho, \rho') \in \Eq_{\;\vDelta}(\Gamma) \} \\
    \Eq_{\;\vDelta}(\vGamma; \Gamma)
    &:= \{(((n, \vrho), \rho), ((n, \vrho'), \rho'))
      \sep (\vrho, \vrho') \in \Eq_{\trunc\vDelta n}(\vGamma) \tand (\rho, \rho') \in \Eq_{\;\vDelta}(\Gamma) \}
  \end{align*}
\end{definition}
At last, we would work with the following notion of semantic typing:
\begin{definition}
  \begin{align*}
    \msemtyeq{t}{t'}{T} :=
    \forall \vDelta', (\vrho, \vrho') \in
    \Eq_{\;\vDelta}(\vGamma). (\intp{t}_{\;\vDelta}(\vrho), \intp{t'}_{\;\vDelta}(\vrho'))
    \in Eq_{\;\vDelta}(T)
  \end{align*}
\end{definition}

\subsubsection{PER Rules}

Instead, we simply work with equality in the semantics, so we actually will work with
the following notion:
\begin{definition}
  We define semantic typing and semantic typing equivalence as follows:
  \begin{align*}
    \msemtyeq{t}{t'}{T} :=
    \mtyping{t}{T} \wedge \mtyping{t'}{T} \wedge \forall \vDelta, \vrho \in \intp{\vGamma}_{\;\vDelta}. \intp{t}_{\;\vDelta}(\vrho) = \intp{t'}_{\;\vDelta}(\vrho)
  \end{align*}
  \begin{align*}
    \msemtyp{t}{T} := \msemtyeq{t}{t}{T}
  \end{align*}
\end{definition}

We proceed to prove the semantic counterpart for the syntactic rules in
\Cref{sec:st:equiv-rules}.

\begin{lemma}
  \begin{mathpar}
    \inferrule*
    {\msemtyeq s t T}
    {\msemtyeq t s T}
  \end{mathpar}
\end{lemma}
\begin{lemma}
  \begin{mathpar}
    \inferrule*
    {\msemtyeq s t T \\ \msemtyeq t u T}
    {\msemtyeq s u T}
  \end{mathpar}
\end{lemma}
\begin{proof}
  We use the fact that $\intp{\vGamma}_{\;\vDelta}$ and $\intp{T}_{\;\vDelta}$ are PER. 
\end{proof}

\subsubsection{Congruence Rules}

\begin{lemma}
  \begin{mathpar}
    \inferrule*
    {x : T \in \Gamma}
    {\msemtyeq[\vGamma; \Gamma] x x T}
  \end{mathpar}
\end{lemma}
\begin{proof}
  Immediate because we are doing the same lookup.
\end{proof}

\begin{lemma}
  \begin{mathpar}
    \inferrule*
    {\msemtyeq[\vGamma; \cdot]{t}{t'}T}
    {\msemtyeq{\boxit t}{\boxit t'}{\square T}}
  \end{mathpar}
\end{lemma}
\begin{proof}
  In fact, since $\intp{\square T}_{\;\vDelta} = \intp{T}_{\vDelta; \cdot}$, the premise
  immediately induces the conclusion.
\end{proof}

\begin{lemma}
  \begin{mathpar}
    \inferrule*
    {\msemtyeq{t}{t'}{\square T} \\
      |\vDelta| = n}
    {\msemtyeq[\vGamma; \vDelta]{\unbox n t}{\unbox n t'}{T'}}  
  \end{mathpar}
\end{lemma}
\begin{proof}
  Assume $\vrho \in \intp{\vGamma; \vDelta}_{\;\vDelta'}$, then
  \begin{align*}
    H_1: &\ \msemtyeq{t}{t'}{\square T}
           \tag{assumption} \\
         &\ \trunc \vrho n \in \intp{\vGamma}_{\trunc{\vDelta'}{\Ltotal\vrho n}} \\
         &\ \intp{t}_{\trunc{\vDelta'}{\Ltotal\vrho n}}(\trunc \vrho n) = \intp{t'}_{\trunc{\vDelta'}{\Ltotal\vrho n}}(\trunc \vrho n)
      \tag{by $H_1$}
  \end{align*}
  Since $\intp{\unbox n t}_{\;\vDelta}(\vrho) = \intp{t}_{\trunc\vDelta{\Ltotal\vrho n}}(\trunc \vrho n)[\sextt{\vect \id}{\Ltotal\vrho n}]$, we can conclude
  $\intp{\unbox n t}_{\;\vDelta}(\vrho) = \intp{\unbox n t'}_{\;\vDelta}(\vrho)$. 
\end{proof}

\begin{lemma}
  \begin{mathpar}
    \inferrule*
    {\msemtyeq[\vGamma;(\Gamma, x : S)]{t}{t'}T}
    {\msemtyeq[\vGamma; \Gamma]{\lambda x. t}{\lambda x. t'}{S \func T}}
  \end{mathpar}
\end{lemma}
\begin{lemma}
  \begin{mathpar}
    \inferrule*
    {\msemtyeq{t}{t'}{S \func T} \\
      \msemtyeq{s}{s'}S}
    {\msemtyeq{t\ s}{t'\ s'}T}
  \end{mathpar}
\end{lemma}
\begin{proof}
  Immediate by premises.
\end{proof}

\subsubsection{$\beta$ Rules}

\begin{lemma}
  \begin{mathpar}
    \inferrule*
    {\msemtyp[\vGamma; \cdot]{t}{T} \\
      |\vDelta| = n}
    {\msemtyeq[\vGamma; \vDelta]{\unbox{n}{(\boxit t)}}{t\{n / 0 \}}{T}}
  \end{mathpar}
\end{lemma}
\begin{proof}
  Assume $\vrho \in \intp{\vGamma; \vDelta}_{\;\vDelta'}$, then 
  \begin{align*}
    \intp{\unbox{n}{(\boxit t)}}_{\;\vDelta'}(\vrho)
    &= \intp{\boxit t}_{\trunc{\vDelta'}{\Ltotal\vrho n}}(\trunc \vrho n)[\sextt{\vect \id}{\Ltotal\vrho n}] \\
    &= \intp{t}_{\trunc{\vDelta'}{\Ltotal\vrho n}; \cdot}((1, \trunc \vrho n), *)[\sextt{\vect \id}{\Ltotal\vrho n}] \\
    &= \intp{t}_{\;\vDelta'}(((1, \trunc \vrho n), *)[\sextt{\vect \id}{\Ltotal\vrho n}])
      \tag{by naturality} \\
    &= \intp{t}_{\;\vDelta'}(((\Ltotal\vrho n, \trunc \vrho n), *)) \\
    &= \intp{t\{n/0\}}_{\;\vDelta'}(\vrho)
      \tag{by \Cref{cor:modal-trans}}
  \end{align*}
\end{proof}

\begin{lemma}
  \begin{mathpar}
    \inferrule*
    {\msemtyp[\vGamma;(\Gamma, x : S)]t T \\ \msemtyp[\vGamma; \Gamma] s S}
    {\msemtyeq[\vGamma; \Gamma]{(\lambda x. t) s}{t[s/x]}{T}}
  \end{mathpar}
\end{lemma}
\begin{proof}
  Assume $\vrho \in \intp{\vGamma; \Gamma}_{\;\vDelta}$, let $(\pi, \rho) := \vrho$, then
  \begin{align*}
    \intp{(\lambda x. t) s}_{\;\vDelta}(\vrho)
    &= \intp{t}_{\;\vDelta}((\pi, (\rho, \intp{s}_{\;\vDelta}(\vrho)))) \\
    &= \intp{t[s/x]}_{\;\vDelta}(\vrho)
      \tag{by \Cref{cor:subst}}
  \end{align*}
\end{proof}

\subsubsection{$\eta$ Rules}

\begin{lemma}
  \begin{mathpar}
    \inferrule*
    {\msemtyp{t}{\square T}}
    {\msemtyeq{t}{\boxit{(\unbox 1 t)}}{\square T}}
  \end{mathpar}
\end{lemma}
\begin{proof}
  Assume $\vrho \in \intp{\vGamma}_{\;\vDelta}$, then
  \begin{align*}
    \intp{\boxit{(\unbox 1 t)}}_{\;\vDelta}(\vrho)
    &= \intp{t}_{\trunc{\vDelta; \cdot} 1}(\trunc{((1, \vrho), *)} 1)[\sextt{\vect \id}{\Ltotal{((1, \vrho), *)}{1}}] \\
    &= \intp{t}_{\;\vDelta}(\vrho)[\sextt{\vect \id}1] \\
    &= \intp{t}_{\;\vDelta}(\vrho)
  \end{align*}
  In the last equation, since $\intp{t}_{\;\vDelta}(\vrho) : \intp{T}_{\vDelta; \cdot}$,
  $\sextt{\vect \id} 1 : \vDelta; \cdot \To_w \vDelta; \cdot$ is just identity. 
\end{proof}

\begin{lemma}
  \begin{mathpar}
    \inferrule*
    {\msemtyp{t}{S \func T}}
    {\msemtyeq t {\lambda x. (t\ x)}{S \func T}}
  \end{mathpar}
\end{lemma}
\begin{proof}
  Assume $\vrho \in \intp{\vGamma}_{\;\vDelta}$, then
  \begin{align*}
    \intp{\lambda x. (t\ x)}_{\;\vDelta}(\vrho)
    &= (\vDelta', \vgamma, a \in \intp{S}_{\;\vDelta'}) \mapsto
      \intp{t}_{\;\vDelta'}((\pi, (\rho, a)), \vect\id, a)
      \tag{where $(\pi, \rho) :=\vrho[\vgamma] = \vrho[\vgamma \circ \vect
      \id]$}
    \\
    &= (\vDelta', \vgamma, a \in \intp{S}_{\;\vDelta'}) \mapsto
      \intp{t}_{\;\vDelta'}(\vrho[\vgamma], \vect\id, a)
      \tag{by \Cref{lem:insert-weaken}} \\
    &= (\vDelta', \vgamma, a \in \intp{S}_{\;\vDelta'}) \mapsto
      \intp{t}_{\;\vDelta}(\vrho, \vgamma \circ \vect\id, a)
      \tag{by naturality} \\
    &= \intp{t}_{\;\vDelta}
  \end{align*}
\end{proof}

\subsubsection{Fundamental Theorem}

We have finished all the semantic typing rules. Thus we can establish semantic typing
and semantic typing equivalence.
\begin{theorem}
  (fundamental)
  \begin{itemize}
  \item If $\mtyping t T$, then $\msemtyp t T$.
  \item If $\mtyequiv{t}{t'}T$, then $\msemtyeq{t}{t'}T$.
  \end{itemize}
\end{theorem}

We next finish the completeness proof. 
\begin{proof}[Proof of \Cref{thm:nbe-completeness}]
  Given $\mtyequiv{t}{t'}T$, we know $\msemtyeq{t}{t'}T$. We plug in
  $\uparrow^{\vGamma}$, so we obtain
  $\intp{t}_{\Gamma}(\uparrow^{\vGamma}) = \intp{t'}_{\Gamma}(\uparrow^{\vGamma})$, as
  well as
  $\downarrow^T_{\;\vGamma} (\intp{t}_{\;\vGamma}(\uparrow^{\vGamma})) =
  \downarrow^T_{\;\vGamma}(\intp{t'}_{\Gamma}(\uparrow^{\vGamma}))$.
  
\end{proof}

\subsection{Soundness and Gluing Model}

\subsubsection{Gluing Model}

A gluing model glues the syntax and the semantics of terms, from which we can extract
the soundness theorem. We write $t \sim a \in \glu{T}_{\;\vGamma}$ in place of $(t,
a)\in \glu{T}_{\;\vGamma}$. The model is parameterized by a context stack
$\vGamma$, so the model moves along unified weakenings, hence forming a Kripke model. 
\begin{align*}
  \glu{T}_{\;\vGamma} &\subseteq \Exp \times \intp{T}_{\;\vGamma} \\
  \glu{B}_{\;\vGamma} &:= \{(t, a) \sep \mtyequiv{t}{a}{B} \} \\
  \glu{\square T}_{\;\vGamma} &:= \{(t, a) \sep \mtyping{t}{\square T} \tand \forall
                              \vDelta. \unbox{|\vDelta|}{t} \sim a[\sextt{\vect \id}{|\vDelta|}] \in \glu{T}_{\vGamma; \vDelta} \} \\ 
  \glu{S \func T}_{\;\vGamma} &:= \{(t, a) \sep \mtyping{t}{S \func T} \tand \forall
                              \vgamma : \vDelta \To_w \vGamma, s \sim b \in
                              \glu{S}_{\;\vDelta}. t[\vgamma]\ s \sim a(\vgamma, b) \in
                              \glu{T}_{\;\vDelta} \} 
\end{align*}

\begin{lemma}
  If $t \sim a \in \glu{T}_{\;\vGamma}$, then $\mtyping t T$. 
\end{lemma}

\begin{lemma}\labeledit{lem:glue-mon}
  If $t \sim a \in \glu{T}_{\;\vGamma}$, given $\vgamma : \vDelta \To_w \vGamma$, then
  $t[\vgamma] \sim a[\vgamma] \in \glu{T}_{\;\vDelta}$. 
\end{lemma}
\begin{proof}
  We proceed by induction on $T$.
  \begin{itemize}[label=Case]
  \item $T = B$, immediate because $a \in \Ne\ T\ \vGamma$.
  \item $T = \square T'$, then
    \begin{align*}
      H_1: &\ \mtyping{t}{\square T'} \tag{by assumption} \\
      H_2: &\ \forall \vDelta'. \unbox{|\vDelta'|}{t} \sim a[\sextt{\vect \id}{|\vDelta'|}] \in \glu{T'}_{\vGamma; \vDelta'}
             \tag{by assumption} \\
           &\ \mtyping[\vDelta]{t[\vgamma]}{\square T'}
             \tag{by application of unified weakening} \\
           &\ \text{assume }\vDelta' \\
      H_3: &\ \unbox{|\vDelta'|}{t} \sim a[\sextt{\vect \id}{|\vDelta'|}] \in \glu{T'}_{\vGamma; \vDelta'}
             \tag{by $H_2$} \\ 
           &\ \text{let } \vgamma' := \vgamma; \underbrace{\id; \cdots;
             \id}_{|\vDelta'|} : \vDelta; \vDelta' \To_w \vGamma; \vDelta' \\
           &\ (\unbox{|\vDelta'|}{t})[\vgamma'] \sim a[\sextt{\vect \id}{|\vDelta'|}][\vgamma'] \in \glu{T'}_{\vDelta; \vDelta'} 
             \tag{by IH on $H_3$ with $\vgamma'$}
    \end{align*}
    We consider the term and the value:
    \begin{align*}
      (\unbox{|\vDelta'|}{t})[\vgamma']
      &= \unbox{|\vDelta'|}{(t[\trunc{\vgamma'}{\;|\vDelta'|}])} \\
      &= \unbox{|\vDelta'|}{(t[\vgamma])} 
    \end{align*}
    \begin{align*}
      a[\sextt{\vect \id}{|\vDelta'|}][\vgamma']
      &= a[\sextt{\vect \id}{|\vDelta'|} \circ \vgamma'] \\
      &= a[\sextt{(\vect \id \circ \trunc{\vgamma'}{\;|\vDelta'|})}{|\vDelta'|}] \\
      &= a[\sextt\vgamma{|\vDelta'|}] \\
      &= a[(\vgamma; \id) \circ (\sextt{\vect \id}{|\vDelta'|})] \\
      &= a[\vgamma][\sextt{\vect \id}{|\vDelta'|}]
    \end{align*}
    Then we can conclude
    \begin{align*}
      &\unbox{|\vDelta'|}{(t[\vgamma])} \sim a[\vgamma][\sextt{\vect \id}{|\vDelta'|}])
      \in \glu{T'}_{\vDelta; \vDelta'} \\
      &t[\vgamma] \sim a[\vgamma] \in \glu{\square T'}_{\;\vDelta}
        \tag{by abstraction}
    \end{align*}
    
  \item $T = S \func T'$, then
    \begin{align*}
      H_1: &\ \mtyping{t}{S \func T'} \tag{by assumption} \\
      H_2: &\ \forall \vgamma' : \vDelta' \To_w \vGamma, s \sim b \in
             \glu{S}_{\;\vDelta'}. t[\vgamma']\ s \sim a(\vgamma', b) \in \glu{T'}_{\;\vDelta'}
             \tag{by assumption} \\
      &\ \text{assume }\vgamma' : \vDelta' \To_w \vDelta, s \sim b \in
        \glu{S}_{\;\vDelta'} \\
      H_3: &\ t[\vgamma \circ \vgamma']\ s \sim a(\vgamma \circ \vgamma', b) \in
             \glu{T'}_{\;\vDelta'}
             \tag{by $H_2$} \\
      &\ t[\vgamma][\vgamma']\ s \sim a[\vgamma](\vgamma', b) \in
        \glu{T'}_{\;\vDelta'} \\
      &\ t[\vgamma] \sim a[\vgamma] \in \glu{S \func T'}_{\;\vDelta}
      \tag{by abstraction} 
    \end{align*}
  \end{itemize}
\end{proof}

\begin{lemma}\labeledit{lem:glue-resp-equiv}
  If $t \sim a \in \glu{T}_{\;\vGamma}$ and $\mtyequiv{t}{t'}{T}$, then $t'~ a \in
  \glu{T}_{\;\vGamma}$. 
\end{lemma}
\begin{proof}
  We prove by induction on $T$.
  \begin{itemize}[label=Case]
  \item $T = B$, immediate by transitivity.
  \item $T = \square T'$, then 
    \begin{align*}
      &\ \forall \vDelta. \unbox{|\vDelta|}{t} \sim a[\sextt{\vect \id}{|\vDelta|}] \in
      \glu{T'}_{\vGamma; \vDelta}
      \tag{by assumption} \\
      &\ \text{assume }\vDelta \\
      &\ \unbox{|\vDelta|}{t} \sim a[\sextt{\vect \id}{|\vDelta|}] \in
      \glu{T'}_{\vGamma; \vDelta} \\
      &\ \mtyequiv[\vGamma; \vDelta]{\unbox{|\vDelta|}t}{\unbox{|\vDelta|}t'}{T'}
        \tag{by congruence} \\
      &\ \unbox{|\vDelta|}{t'} \sim a[\sextt{\vect \id}{|\vDelta|}] \in
        \glu{T'}_{\vGamma; \vDelta}
        \byIH \\
      &\ t' \sim a \in \glu{\square T'}_{\;\vGamma}
    \end{align*}
    
  \item $T = S \func T'$, then
    \begin{align*}
      &\ \forall \vgamma : \vDelta \To_w \vGamma,
        s \sim b \in \glu{S}_{\;\vDelta}. t[\vgamma]\ s \sim a(\vgamma,
        b) \in \glu{T'}_{\;\vDelta} \tag{by assumption} \\
      &\ \text{assume }\vgamma : \vDelta \To_w \vGamma,
        s \sim b \in \glu{S}_{\;\vDelta} \\
      &\ \mtyequiv[\vDelta]{t[\vgamma]\ s}{t'[\vgamma]\ s}{T'}
        \tag{by congruence} \\
      &\ t'[\vgamma]\ s \sim a(\vgamma, b) \in \glu{T'}_{\;\vDelta}
        \byIH \\
      &\ t' \sim a \in \glu{S \func T'}_{\;\vGamma}
    \end{align*}
  \end{itemize}
\end{proof}

One last lemma we should prove is that the gluing model implies equivalence after
reification. We need to prove two lemmas mutually:
\begin{lemma}\labeledit{lem:ne-glue}
  If $\mtyping t T$ and $t$ is neutral, then $t \sim \uparrow^T_{\;\vGamma}(t) \in
  \glu{T}_{\;\vGamma}$. 
\end{lemma}
\begin{proof}
  We proceed by induction on $T$.
  \begin{itemize}[label=Case]
  \item $T= B$, immediate.
  \item $T = \square T'$, then assuming $\vDelta$, we shall prove
    \begin{align*}
      \unbox{|\vDelta|}{t} \sim \uparrow^{T'}_{\vGamma; \cdot}(\unbox 1 t)[\sextt{\vect \id}{|\vDelta|}] \in \glu{T'}_{\vGamma; \vDelta}
    \end{align*}
    We perform the following reason on the syntactic term:
    \begin{align*}
      \unbox{|\vDelta|}{t}
      &\approx \unbox{|\vDelta|}{(\boxit{(\unbox 1 t)})}
        \tag{$\eta$ rule} \\
      &\approx (\unbox 1 t)\{|\vDelta|/ 0\}
        \tag{$\beta$ rule} \\
      &\approx \unbox 1 t [\sextt{\vect \id}{|\vDelta|}]
        \tag{as unified weakening}
    \end{align*}
    Thus the goal becomes
    \begin{align*}
      \unbox 1 t [\sextt{\vect \id}{|\vDelta|}] \sim \uparrow^{T'}_{\vGamma; \cdot}(\unbox 1 t)[\sextt{\vect \id}{|\vDelta|}]) \in \glu{T'}_{\vGamma; \vDelta}
    \end{align*}
    by \Cref{lem:glue-resp-equiv}.

    Now we reason forwardly:
    \begin{align*}
      &\unbox 1 t \sim \uparrow^{T'}_{\vGamma; \cdot} (\unbox 1 t) \in \glu{T}_{\;\vGamma}
      \byIH \\
      &\unbox 1 t [\sextt{\vect \id}{|\vDelta|}] \sim \uparrow^{T'}_{\vGamma; \cdot}(\unbox 1 t)[\sextt{\vect \id}{|\vDelta|}] \in \glu{T'}_{\vGamma; \vDelta}
      \tag{by \Cref{lem:glue-mon}}
    \end{align*}
    
  \item $T = S \func T'$, then assume $\vgamma : \vDelta \To_w \vGamma$ and $s \sim b :
    \glu{S}_{\;\vDelta}$,
    \begin{align*}
      \uparrow^{S \func T'}_{\;\vGamma}(t)(\vgamma, b)
      &= \uparrow^{T'}_{\;\vDelta}(t[\vgamma]\ \downarrow^S_{\;\vDelta}(b)) 
    \end{align*}
    By IH, we have
    \begin{align*}
      t[\vgamma]\ \downarrow^S_{\;\vDelta}(b) \sim \uparrow^{T'}_{\;\vDelta}(t[\vgamma]\
      \downarrow^S_{\;\vDelta}(b))
      \in \glu{T'}_{\;\vDelta}
    \end{align*}
    We further obtain from \Cref{lem:glue-resp-equiv,lem:glue-nf}
    \begin{align*}
      t[\vgamma]\ s \sim \uparrow^{S \func T'}_{\;\vGamma}(t)(\vgamma, b)
      \in \glu{T'}_{\;\vDelta}
    \end{align*}
    Therefore $(t, \uparrow^{S \func T'}_{\;\vGamma} t) \in \glu{S \func T'}_{\;\vGamma}$.
  \end{itemize}
\end{proof}

\begin{lemma}\labeledit{lem:glue-nf}
  If $t \sim a \in \glu{T}_{\;\vGamma}$, then
  $\mtyequiv{t}{\downarrow^T_{\;\vGamma}(a)}{T}$. 
\end{lemma}
\begin{proof}
  We proceed by induction on $T$.
  \begin{itemize}[label=Case]
  \item $T = B$, immediate.
  \item $T = \square T'$, then
    \begin{align*}
      H_1: &\ \forall \vDelta. \unbox{|\vDelta|}{t} \sim
             a[\sextt{\vect \id}{|\vDelta|}] \in
             \glu{T'}_{\vGamma; \vDelta}
             \tag{by assumption} \\
           &\ \unbox 1 t \sim a[\sextt{\vect \id}{|\vDelta|}] \in \glu{T'}_{\vGamma; \cdot}
             \tag{by $H_1$ with $\vDelta$ a singleton containing an empty context} \\
           &\ \unbox 1 t \sim a \in \glu{T'}_{\vGamma; \cdot}
             \tag{$\sextt{\vect \id}1 = \vect \id : \vGamma; \cdot \To_w \vGamma; \cdot$} \\
           &\ \mtyequiv{\unbox 1 t}{\downarrow^{T'}_{\vGamma; \cdot}(a)}{T'}
             \byIH\\
           &\ \mtyequiv{\boxit{(\unbox 1 t)}}{\boxit (\downarrow^{T'}_{\vGamma;
             \cdot}(a))}{\square T'}
             \tag{by congruence}\\
           &\ \mtyequiv{t}{\boxit{(\downarrow^{T'}_{\vGamma; \cdot}(a))}}{T'}
             \tag{by $\eta$ and transitivity}
    \end{align*}
  \item $T = S \func T'$, let $\vGamma';\Gamma := \vGamma$, then
    \begin{align*}
      H_2: &\ \forall \vgamma : \vDelta \To_w \vGamma,
        s \sim b \in \glu{S}_{\;\vDelta}. t[\vgamma]\ s \sim a(\vgamma,
             b) \in \glu{T'}_{\;\vDelta} \tag{by assumption} \\
           &\ x \sim \uparrow^S_{\vGamma'; (\Gamma, x : S)}(x) \in \intp{S}_{\vGamma'; (\Gamma, x : S)}
             \tag{by \Cref{lem:ne-glue}} \\
           &\ t[p(\vect \id)]\ x \sim a(p(\vect \id), \uparrow^S_{\vGamma'; (\Gamma, x : S)}(x)) \in
             \glu{T'}_{\vGamma'; (\Gamma, x : S)}
             \tag{by $H_2$ and let $\vgamma$ be $p(\vect \id) : \vGamma'; (\Gamma, x
             : S) \To_w \vGamma'; \Gamma$} \\
           &\ \mtyequiv[\vGamma'; (\Gamma, x : S)]{t\ x}{\downarrow^{T'}_{\vGamma';
             (\Gamma, x : S)}(a(p(\vect \id), \uparrow^S_{\vGamma'; (\Gamma, x
             : S)}(x)))}{T'}
             \byIH \\
           &\ \mtyequiv{\lambda x.t\ x}
             {\lambda x.\downarrow^{T'}_{\vGamma';
             (\Gamma, x : S)}(a(p(\vect \id), \uparrow^S_{\vGamma'; (\Gamma, x
             : S)}(x)))}{S \func T'}
             \tag{by congruence} \\
           &\ \mtyequiv{t}
             {\lambda x.\downarrow^{T'}_{\vGamma';
             (\Gamma, x : S)}(a(p(\vect \id), \uparrow^S_{\vGamma'; (\Gamma, x
             : S)}(x)))}{S \func T'}
             \tag{by $\eta$ and transitivity}
    \end{align*}
  \end{itemize}
\end{proof}

After setting up all the technical lemmas, we should formulate the semantic typing
judgment. But first, we will need to generalize the gluing model to contexts and
context stacks.
\begin{align*}
  \glu{\Gamma}_{\;\vDelta} &\subseteq \vDelta \To \Gamma \times \intp{\Gamma}_{\;\vDelta} \\
  \glu{\cdot}_{\;\vDelta} &:= \{((), *)\} \\
  \glu{\Gamma, x : T}_{\;\vDelta} &:= \{ ((\sigma, t/x), (\rho, a)) \sep \sigma \sim \rho \in \glu{\Gamma}_{\;\vDelta}
                                  \tand t \sim a \in \glu{T}_{\;\vDelta} \} \\\\
  \glu{\vGamma}_{\;\vDelta} &\subseteq \vDelta \To \vGamma \times \intp{\vGamma}_{\;\vDelta} \\
  \glu{\epsilon; \Gamma}_{\;\vDelta} &:= \{ ((\varepsilon; \sigma), (*, \rho)) \sep \sigma \sim \rho
                                  \in \glu{\Gamma}_{\;\vDelta} \} \\
  \glu{\vGamma; \Gamma}_{\;\vDelta} &:= \{ ((\vsigma; n; \sigma), ((n, \vrho), \rho))
                                    \sep \vsigma \sim \vrho \in
                                    \glu{\vGamma}_{\trunc \vDelta n} \tand \sigma \sim \rho \in \glu{\Gamma}_{\;\vDelta} \}
\end{align*}

We need the follow lemmas to understand operations on $\vsigma \sim \vrho \in
\glu{\vGamma}_{\;\vDelta}$.
\begin{lemma}\labeledit{lem:glue-L}
  If $\vsigma \sim \vrho \in \glu{\vGamma}_{\;\vDelta}$, then $\Ltotal \vsigma n = \Ltotal\vrho n$.
\end{lemma}
\begin{proof}
  Immediate by induction on $n$. 
\end{proof}

\begin{lemma}\labeledit{lem:glue-trunc}
  If $\vsigma \sim \vrho \in \glu{\vGamma}_{\;\vDelta}$, then $\trunc \vsigma n \sim \trunc \vrho n \in
  \glu{\trunc \vGamma n}_{\trunc \vDelta{\Ltotal \vsigma n}}$.
\end{lemma}
\begin{proof}
  Immediate by induction on $n$. 
\end{proof}

\begin{lemma}\labeledit{lem:glue-stack-mon}
  If $\vsigma \sim \vrho \in \glu{\vGamma}_{\;\vDelta}$, given $\vgamma : \vDelta' \To_w
  \vDelta$,  then
  $\vsigma \circ \vgamma \sim \vrho[\vgamma] \in \glu{\vGamma}_{\;\vDelta'}$. 
\end{lemma}
\begin{proof}
  Immediate by induction on $\vGamma$ and apply \Cref{lem:glue-mon} when proper. 
\end{proof}

Finally, we can define the judgment:
\begin{definition}
  \begin{align*}
    \mSemtyp t T := \forall \vsigma \sim \vrho \in
    \glu{\vGamma}_{\;\vDelta}. t[\vsigma] \sim \intp{t}_{\;\vDelta}(\vrho) \in \glu{T}_{\;\vDelta}
  \end{align*}
\end{definition}

\subsubsection{Semantic Typing Rules}

\begin{lemma}
  \begin{mathpar}
    \inferrule*
    {x : T \in \Gamma}
    {\mSemtyp[\vGamma; \Gamma] x T}
  \end{mathpar}
\end{lemma}
\begin{proof}
  Immediate.
\end{proof}

\begin{lemma}
  \begin{mathpar}
    \inferrule*
    {\mSemtyp[\vGamma; \cdot] t T}
    {\mSemtyp{\boxit t}{\square T}}
  \end{mathpar}
\end{lemma}
\begin{proof}
  Assume $\vsigma \sim \vrho \in \glu{\vGamma}_{\;\vDelta}$ and $\vDelta'$, then
  \begin{align*}
    H_1: &\ \mSemtyp[\vGamma; \cdot] t T
           \tag{by assumption} \\
         &\ (\sext\vsigma{|\vDelta'|}{()}) \sim ((|\vDelta'|, \vrho), *) \in \glu{\vGamma; \cdot}_{\vDelta; \vDelta'}
    \\
         &\ t[\sext\vsigma{|\vDelta'|}{()}] \sim \intp{t}_{\vDelta; \vDelta'}(((|\vDelta'|, \vrho), *)) \in \glu{T}_{\vDelta; \vDelta'}
           \tag{by $H_1$}
  \end{align*}
  We also know
  \begin{align*}
    t[\sext{\vsigma}{|\vDelta'|}{()}]
    &= t[\vsigma; ()][\sext{\vect \id}{|\vDelta'|}{()}] \\
    &\approx \unbox{|\vDelta'|}{(\boxit{(t[\vsigma; ()])})} \\
    &\approx \unbox{|\vDelta'|}{(\boxit{t}[\vsigma])}
  \end{align*}
  Then
  \begin{align*}
    &\ \unbox{|\vDelta'|}{(\boxit{t}[\vsigma])} \sim \intp{\boxit t}_{\;\vDelta}(\vrho)[\sext{\vect \id}{|\vDelta'|}{()}] \in \glu{T}_{\vDelta; \vDelta'}
      \tag{by \Cref{lem:glue-resp-equiv}} \\
    &\ \boxit{t}[\vsigma] \sim \intp{\boxit t}_{\;\vDelta}(\vrho) \in \glu{\square
      T}_{\;\vDelta}
      \tag{by abstraction}
  \end{align*}
\end{proof}

\begin{lemma}
  \begin{mathpar}
    \inferrule*
    {\mSemtyp t {\square T} \\ |\vDelta| = n}
    {\mSemtyp[\vGamma; \vDelta]{\unbox n t}{T}}
  \end{mathpar}
\end{lemma}
\begin{proof}
  Immediate by definition of $\glu{\square T}$ and \Cref{lem:glue-trunc}. 
\end{proof}

\begin{lemma}
  \begin{mathpar}
    \inferrule*
    {\mSemtyp[\vGamma;\Gamma, x : S]t T}
    {\mSemtyp[\vGamma;\Gamma]{\lambda x. t}{S \func T}}
  \end{mathpar}
\end{lemma}
\begin{proof}
  Assume $\vsigma \sim \vrho \in \glu{\vGamma; \Gamma}_{\;\vDelta}$, $\vDelta'$,
  $\vgamma : \vDelta' \To_w \vDelta$ and $s \sim b \in \glu{S}_{\;\vDelta'}$. Let
  $(\valpha; \sigma) := \vsigma$, $(\pi, \rho) := \vrho$ and
  $(\vbeta; \gamma) := \vgamma$. Then we should prove that
  \begin{align*}
    (\lambda x. (t[\valpha'; \sigma', x/x]))\ s \sim
    \intp{t}_{\;\vDelta'}((\pi', (\rho', b))) \in \glu{T}_{\;\vDelta'}
  \end{align*}
  where $(\valpha'; \sigma') := \vsigma \circ \vgamma$ and $(\pi', \rho') := \vrho[\vgamma]$. 

  We need to construct $\glu{\vGamma;\Gamma, x : S}_{\;\vDelta'}$ from what we have. We
  claim
  \begin{align*}
    (\valpha'; \sigma', s/x) \sim (\pi', (\rho', b)) \in \glu{\vGamma;\Gamma, x : S}_{\;\vDelta'}
  \end{align*}
  This is because
  $\vsigma \circ \vgamma \sim \vrho[\vgamma] \in \glu{\vGamma;
    \Gamma}_{\;\vDelta'}$ by \Cref{lem:glue-stack-mon}. Therefore
  \begin{align*}
    t[\valpha'; \sigma', s/x] \sim \intp{t}_{\;\vDelta'}((\pi', (\rho', b))) \in \glu{T}_{\;\vDelta'}
  \end{align*}
  Comparing the goal and what we have, we just need to reason about the syntactic
  term, and they are just one $\beta$ rule away. 
\end{proof}

\begin{lemma}
  \begin{mathpar}
    \inferrule*
    {\mSemtyp t {S \func T} \\ \mSemtyp s S}
    {\mSemtyp{t\ s}{T}}
  \end{mathpar}
\end{lemma}
\begin{proof}
  Immediate by just applying the premises. 
\end{proof}

\subsubsection{Fundamental Theorem}

Since we have proven the semantic typing rules, we can conclude the fundamental
theorem:
\begin{theorem}
  (fundamental) If $\mtyping t T$, then $\mSemtyp t T$. 
\end{theorem}

We also need to relate $\vect \id$ and $\uparrow^{\vGamma}$:
\begin{lemma}
  $\vect \id \sim \uparrow^{\vGamma} \in \glu{\vGamma}_{\;\vGamma}$.
\end{lemma}
\begin{proof}
  Immediate by induction on $\vGamma$ and apply \Cref{lem:ne-glue}. 
\end{proof}

\begin{theorem}
  (soundness) If $\mtyping t T$, then $\mtyequiv{t}{\nbe^T_{\;\vGamma}(t)}T$. 
\end{theorem}
\begin{proof}
  Applying the fundamental theorem, we have $t[\vect \id] \sim
  \intp{t}_{\;\vGamma}(\uparrow^{\vGamma})\in \glu{T}_{\;\vGamma}$.  We conclude the goal
  by applying \Cref{lem:glue-resp-equiv,lem:glue-nf}. 
\end{proof}

\section{Normalization by Evaluation in Untyped Domain}\labeledit{sec:st:untyped}

In the previous section, we developed an NbE algorithm based on a presheaf
model. However, there are some difficulties for extending the presheaf model to
dependent types because of the complexity of
bookkeeping. \citet{kaposi_normalisation_2017} shows how one can extract an NbE
algorithm from a category with families (CwF) model, but the proof is conceivably more
complex than \citet{abel_normalization_2013} and their formulation uses inverse
abstractions as the elimination for dependent functions instead of applications as it is
typically done. 
Therefore, we choose to follow \citet{abel_normalization_2013}, and our first goal is to develop an NbE algorithm based on an untyped
domain. To simplify certain proofs, we will first extend the calculus in
\Cref{sec:st:syn} with explicit substitution, and then interpret the extended calculus
to an untyped domain, from which we extract the NbE algorithm.

\subsection{Explicit Substitution Calculus}\labeledit{sec:st:subst-calc}

We formulate a version of the calculus in \Cref{sec:st:syn} with explicit
substitution. In this section, we complete the details of this calculus.
\begin{alignat*}{2}
  s, t, u &:=&&\ \cdots \sep t[\vect \sigma] \tag{application of unified substitution} \\
  \vsigma, \vdelta &:= &&\ \vect I \sep \vsigma, t/x \sep \wk_x
  \sep \sextt\vsigma n  \sep \vsigma \circ \vdelta \tag{Unified substitution, \Substs}
\end{alignat*}

In addition to the existing typing rules, we add
\begin{mathpar}
  \inferrule
  {\mtyping[\vDelta]t T \\ \mtyping{\vsigma}{\vDelta}}
  {\mtyping{t[\vsigma]}{T}}
\end{mathpar}
The following judgments specify well-formed unified substitutions:
\begin{mathpar}
  \inferrule
  { }
  {\mtyping{\vect I}{\vGamma}}

  \inferrule
  {\mtyping{\vsigma}{\vGamma'; \Gamma} \\ \mtyping t T}
  {\mtyping{\vsigma, t/x}{\vGamma';(\Gamma, x : T)}}

  \inferrule
  { }
  {\mtyping[\vGamma; (\Gamma, x : T)]{\wk_x}{\vGamma;\Gamma}}

  \inferrule
  {\mtyping{\vsigma}{\vDelta} \\ |\vGamma'| = n}
  {\mtyping[\vGamma; \vGamma']{\sextt\vsigma n}{\vDelta; \cdot}}

  \inferrule
  {\mtyping[\vGamma']{\vsigma}{\vGamma''} \\ \mtyping{\vdelta}{\vGamma'}}
  {\mtyping{\vsigma \circ \vdelta}{\vGamma''}}
\end{mathpar}
The truncation offset function is defined recursively on the syntax of $\vsigma$:
\begin{align*}
  \Ltotal \vsigma 0 &:= 0 \\
  \Ltotal {\vect I} {1 + n} &:= 1 + n \\
  \Ltotal {(\vsigma, t/x)} {1 + n} &:= \Ltotal {\vsigma} {1 + n} \\
  \Ltotal {\wk_x} {1 + n} &:= \vect I \\
  \Ltotal {\sextt \vsigma m} {1 + n} &:= m + \Ltotal \vsigma n \\
  \Ltotal {\vsigma \circ \vdelta} {1 + n} &:= \Ltotal \vdelta {\Ltotal \vsigma {1 +
                                              n}} 
\end{align*}
The following lemma states the
condition when truncation offset is well-defined:
\begin{lemma}
  If $\mtyping \vsigma \vDelta$ and $n < |\vDelta|$, then $\Ltotal\vsigma n <
  |\vGamma|$. 
\end{lemma}
\begin{proof}
  Induction on $n$ and $\mtyping \vsigma \vDelta$.
\end{proof}
In particular, it ensures the recursive cases always hit the defined cases. 

We can also define truncation in a similar manner:
\begin{align*}
  \trunc {\vsigma} 0&:=  \vsigma \\
  \trunc {\vect I} {1 + n}&:= \vect I \\
  \trunc {(\vsigma, t/x)} {1 + n} &:= \trunc {\vsigma} {1 + n} \\
  \trunc {\wk_x} {1 + n} &:= \vect I \\
  \trunc {(\sextt \vsigma m)} {1 + n} &:= \trunc \vsigma {n} \\
  \trunc {(\vsigma \circ \vdelta)} {1 + n} &:= (\trunc \vsigma {1 + n})
                                             \circ (\trunc \vdelta {\Ltotal \vsigma {1+n}})
\end{align*}
We can prove a similar lemma that quantifies the well behavior of truncation:
\begin{lemma}
  If $\mtyping \vsigma \vDelta$ and $n < |\vDelta|$, then 
  $\mtyping[\trunc\vGamma{\Ltotal\vsigma n}]{\trunc\vsigma n}{\trunc\vDelta n}$.
\end{lemma}
\begin{proof}
  Induction on $n$ and $\mtyping \vsigma \vDelta$.
\end{proof}

We reexamine the properties in \Cref{sec:usubst-props} and all of them hold. 

The equational theory of the newly added components all enjoy congruence and PER laws,
which I will not write out here. I will focus on the more interesting rules.
Let us first consider various application of unified substitution to terms:
\begin{mathpar}
  \inferrule
  {\mtyping t T}
  {\mtyequiv{t[\vect I]}{t}{T}}

  \inferrule
  {\mtyping[\vGamma']{\vsigma}{\vGamma''} \\ \mtyping[\vGamma]{\vdelta}{\vGamma'} \\
  \mtyping[\vGamma'']{t}{T}}
  {\mtyequiv{t[\vsigma \circ \vdelta]}{t[\vsigma][\vdelta]}{T}}
\end{mathpar}

Then we look specifically into variables:
\begin{mathpar}
  \inferrule
  {\mtyping{\vsigma}{\vGamma'; \Gamma} \\ \mtyping t T}
  {\mtyequiv{x[\vsigma, t/x]}{t}{T}}
 
  \inferrule
  {\mtyping{\vsigma}{\vGamma'; \Gamma} \\ \mtyping t T' \\ y : T \in \Gamma}
  {\mtyequiv{y[\vsigma, t/x]}{y[\vsigma]}{T}}

  \inferrule
  {y : T \in \Gamma}
  {\mtyequiv[\vGamma; (\Gamma, x : T')]{y[\wk_x]}{y}{T}}
\end{mathpar}

The we consider applying unified substitutions other terms:
\begin{mathpar}
  \inferrule
  {\mtyping{\vsigma}{\vDelta} \\ \mtyping[\vDelta; \cdot]{t}{T}}
  {\mtyequiv{\boxit t[\vsigma]}{\boxit{(t[\sextt\vsigma 1])}}{\square T}}

  \inferrule
  {\mtyping[\vDelta]{t}{\square T} \\ \mtyping[\vGamma; \vGamma']{\vsigma}{\vDelta;
      \vDelta'} \\\\
    |\vDelta'| = n \\ |\vGamma'| = \Ltotal\vsigma n}
  {\mtyequiv[\vGamma; \vGamma']{\unbox n t[\vsigma]}{\unbox {\Ltotal\vsigma n}{(t[\trunc\vsigma n])}}{T}}

  \inferrule
  {\mtyping[\vGamma; \Gamma]{\vsigma}{\vDelta;\Delta} \\
  \mtyping[\vDelta; \Delta, x : S]{t}{T}}
  {\mtyequiv[\vGamma; \Gamma]{\lambda x. t[\vsigma]}{\lambda x. (t[(\vsigma \circ
      \wk_x), x / x])}{S \func T}}

  \inferrule
  {\mtyping{\vsigma}{\vDelta} \\
    \mtyping[\vDelta]{s}{S \func T} \\
  \mtyping[\vDelta]{t}{S}}
  {\mtyequiv{s\ t[\vsigma]}{(s[\vsigma])\ (t[\vsigma])}{T}}
\end{mathpar}

The advantage of this explicit substitution is that the $\beta$ equivalences have
fixed syntax:
\begin{mathpar}
  \inferrule*
  {\mtyping[\vGamma; \cdot]{t}{T} \\
  |\vDelta| = n}
  {\mtyequiv[\vGamma; \vDelta]{\unbox{n}{(\boxit t)}}{t[\sextt{\vect I}n]}{T}}

  \inferrule*
  {\mtyping[\vGamma;(\Gamma, x : S)]t T \\ \mtyping[\vGamma; \Gamma] s S}
  {\mtyequiv[\vGamma; \Gamma]{(\lambda x. t) s}{t[\vect I, s/x]}{T}}
\end{mathpar}
The $\eta$ expansion is simpler:
\begin{mathpar}
  \inferrule*
  {\mtyping{t}{\square T}}
  {\mtyequiv{t}{\boxit{(\unbox 1 t)}}{\square T}}

  \inferrule*
  {\mtyping{t}{S \func T}}
  {\mtyequiv{t}{\lambda x. t[\wk_x]\ x}{S \func T}}
\end{mathpar}
We make the local weakening explicit in the function case. 

Next, we consider the interaction between different unified substitutions, or the
algebraic axiomization of unified substitutions. We first begin with categorical laws:
\begin{mathpar}
  \inferrule
  {\mtyping{\vsigma}{\vDelta}}
  {\mtyequiv{\vsigma \circ \vect I}{\vsigma}{\vDelta}}

  \inferrule
  {\mtyping{\vsigma}{\vDelta}}
  {\mtyequiv{\vect I \circ \vsigma}{\vsigma}{\vDelta}}

  \inferrule
  {\mtyping[\vGamma'']{\vsigma''}{\vGamma'''} \\ \mtyping[\vGamma']{\vsigma'}{\vGamma''} \\ \mtyping{\vsigma}{\vGamma'}}
  {\mtyequiv{(\vsigma'' \circ \vsigma') \circ \vsigma}{\vsigma'' \circ (\vsigma' \circ \vsigma)}{\vGamma'''}}
\end{mathpar}
We next consider composition:
\begin{mathpar}
  \inferrule
  {\mtyping[\vGamma']{\vsigma}{\vGamma''; \Gamma} \\ \mtyping[\vGamma']{t}{T} \\ \mtyping{\vdelta}{\vGamma'}}
  {\mtyequiv{\vsigma, t/x \circ \vdelta}{(\vsigma \circ \vdelta), t[\vdelta]/x}{\vGamma''; (\Gamma, x : T)}}
  
  \inferrule
  {\mtyping{\vsigma}{\vGamma'} \\ \mtyping[\vGamma'']{\vdelta}{\vGamma;\vDelta} \\
    |\vDelta| = n}
  {\mtyequiv[\vGamma'']{(\sextt\vsigma n) \circ \vdelta}{\sextt{(\vsigma \circ \trunc\vdelta n)}{\Ltotal \vdelta n} }{\vGamma';\cdot}}
\end{mathpar}

At last, we consider the extensionality principles which stipulates how syntax should
interact:
\begin{mathpar}
  \inferrule
  {\mtyping[\vGamma']{\vsigma}{\vGamma; \Gamma} \\ \mtyping[\vGamma']{t}{T}}
  {\mtyequiv[\vGamma']{\wk_x \circ (\vsigma, t/x)}{\vsigma}{\vGamma; \Gamma}}

  \inferrule
  {\mtyping[\vGamma']{\vsigma}{\vGamma; (\Gamma, x : T)}}
  {\mtyequiv[\vGamma']{\vsigma}{(\wk_x \circ \vsigma), (x[\vsigma]/x)}{\vGamma; (\Gamma, x : T)}}

  \inferrule
  {\mtyping{\vsigma}{\vDelta; \cdot} \\ |\vDelta| > 0 \\ \Ltotal\vsigma 1 = n}
  {\mtyequiv{\vsigma}{\sextt{\trunc \vsigma 1}{n}}{\vDelta; \cdot}}
\end{mathpar}

After defining the equivalence between unified substitutions, we shall examine truncation offset and
truncation respect the equivalence:
\begin{lemma}
  If $\mtyequiv{\vsigma}{\vsigma'}\vDelta$ and $n < |\vDelta|$, then $\Ltotal\vsigma n =
  \Ltotal{\vsigma'} n$.
\end{lemma}
\begin{lemma}
  If $\mtyequiv{\vsigma}{\vsigma'}\vDelta$ and $n < |\vDelta|$, then
  $\mtyequiv[\trunc\vGamma{\Ltotal\vsigma n}]{\trunc\vsigma n}{\trunc{\vsigma'}n}{\trunc\vDelta n}$.
\end{lemma}
\begin{proof}
  Induction on $n$ and $\mtyequiv{\vsigma}{\vsigma'}\vDelta$.
\end{proof}

\subsection{Untyped Domain}\labeledit{sec:st:domain}

In this section, we define the untyped domain in which we will operate and define the
NbE algorithm:
\begin{alignat*}{2}
  z & && \tag{Domain variables, $\N$} \\
  a, b &:=&&\ \Lambda(\vrho, x, t) \sep \tbox(a) \sep \uparrow^T(c )
  \tag{Domain terms, $D$}\\
  c &:= &&\ z  \sep c\ d \sep \tunbox(k, c)
  \tag{Neutral domain terms, $D^{\Ne}$} \\
  d &:= &&\ \downarrow^T(a)
  \tag{Normal domain terms, $D^{\Nf}$} \\
  \Env &:= &&\ \Var \rightharpoonup D \\
  \Envs &:= &&\ \N \to \N \times \Env \\
  \rho & && \tag{Local evaluation environment, $\Env$} \\
  \vect{\rho} & && \tag{(Global) evaluation environment, $\Envs$} 
\end{alignat*}
The idea is that we interpret the syntactic terms into the domain $D$ and perform
normalization in $\D$. Once it is done, we convert the value from $D$ back to
syntactic normal form using two readback functions. Type-directed $\eta$ expansion is performed
during the readback process.

In $D$, we distinguish local evaluation environments $\Env$, which are partial
functions mapping a syntactic variable to a domain value, and global evaluation
environments $\Envs$, which are functions mapping an $\tunbox$ level to an offset of \tunbox level
and an \Env. The purpose of the offsets is very similar to those in unified
substitutions, which are used to keep track of \tunbox levels for terms. In $D$, we
take one step further by modelling infinite context stacks. Thus a stream of offsets
and local environment
are just a function returning tuples. We of course will not need an
infinite stack and the typing judgments will only make use of some finite
prefix. Nevertheless, it is convenient to avoid thinking about finiteness in an untyped setting.
Notice that function abstractions in $D$, $\Lambda$, are modeled by three parameters:
the evaluation environment it captures, the syntactic variable and its
syntactic function body. This trick is defunctionalization~\citep{reynolds_definitional_1998} which uses data to
represent functions.

At last, it is worth mentioning that $\uparrow$ and $\downarrow$ become \emph{syntax}
in the untyped domain, as opposed to reflection and reification \emph{functions} in
the presheaf model. The roles of reflection and reification are taken by the readback
functions to be defined very shortly. 

Next we define some operations on $\Env$ and $\Envs$. First, since $\Env$ is partial,
we have $\emp : \Env$ which is just undefined everywhere.

We can define our first $\Envs$ using $\emp$:
\begin{align*}
  \empenv &: \Envs \\
  \empenv(\_) &:= (1, \emp)
\end{align*}
It is the empty environment which is just undefined everywhere. 
We can extend an $\Envs$ in various ways. We overload the function $\ext$ to do so. Which
concrete instance to apply depends on the arguments. First, we can insert an $\Env$ to
an $\Envs$:
\begin{align*}
  \ext &: \Envs \to \N \to \Env \to \Envs \\
  \ext(\vrho, n, \rho)(0) &:= (n, \rho) \\
  \ext(\vrho, n, \rho)(1 + m) &:= \vrho(m)
\end{align*}
If we leave the environment out, we stipulate it to be $\emp$. 
\begin{align*}
  \ext &: \Envs \to \N \to \Envs \\
  \ext(\vrho, n) &:= \ext(\vrho, n, \emp) 
\end{align*}
If we further leave the natural number unspecified, we stipulate to be $1$:
\begin{align*}
  \ext &: \Envs \to \Envs \\
  \ext(\vrho) &:= \ext(\vrho, 1) 
\end{align*}
We should also be able to insert a $D$ to the top of the stack and bind $x$ to it:
\begin{align*}
  \ext &: \Envs \to \Var \to D \to \Envs \\
  \ext(\vrho, x, a)(0) &:= (k, \rho[x \mapsto a])  \tag{where $(k, \rho) := \vrho(0)$} \\
  \ext(\vrho, x, a)(1 + m) &:= \vrho(1 + m)
\end{align*}
where $\rho[x \mapsto a]$ is the same as $\rho$ except that when given $x$ it returns
$a$. We write the iterative applications of this version of $\ext$ in just one pair of
parentheses:
\begin{align*}
  \ext(\vrho, x_1, a_1, \cdots, x_i, a_i) :=
  \ext(\cdots \ext(\vrho, x_1, a_1), \cdots, x_i, a_i)
\end{align*}

We use the $\drop$ operation to drop a binding in the topmost $\Env$:
\begin{align*}
  \drop &: \Envs \to \Var \to \Envs \\
  \drop(\vrho, x)(0) &:= (k, \rho')
                    \tag{where $(k, \rho) := \vrho(0)$ and $\rho'$ is $\rho$ with $x$
                    removed from its domain} \\
  \drop(\vrho, x)(1+n) &:= \vrho(1 + n)
\end{align*}

We also define the familiar truncation and truncation offset  on $\vrho$:
\begin{align*}
  \trunc {\_} {\_} &: \Envs \to \N \to \Envs \\
  (\trunc \vrho n)(m) &:= \vrho(n + m) \\
  \Ltotal {\_} {\_} &: \Envs \to \N \to \N \\
  \Ltotal \vrho 0 &:= 0 \\
  \Ltotal \vrho {1 + n} &:= m + \Ltotal {\trunc \vrho 1} n
                          \tag{where $\vrho(0) := (m, \_)$}
\end{align*}

The following lemma holds:
\begin{lemma}\labeledit{lem:L-add-envs}
  $\Ltotal\vrho{n + m} = \Ltotal\vrho n + \Ltotal{\trunc\vrho n} m$
\end{lemma}
\begin{proof}
  Induction on $n$. 
\end{proof}

\subsection{Untyped Modal Transformations}\labeledit{sec:st:ut-mtrans}

To model $\beta$ reduction for $\square T$, we shall define untyped modal
transformations (\UnMoTs),
ranging over $\kappa$,  in
$D$. We even want to take a step further by considering a unified representation of
compositional closure of modal transformation in $D$. It turns out that $\N \to \N$
is a sufficient representation, which models a unified weakening from an infinite
context stack to another. 

We first define operations for this representation of modal transformation:
\begin{align*}
  \sextt{\_}{\_} &: (\N \to \N) \to \N \to \N \to \N \\
  (\sextt \kappa n)\;(0) &:= n \\
  (\sextt \kappa n)\;(m) &:= \kappa(1 + m) \\
  \trunc {\_}{\_} &: (\N \to \N) \to \N \to \N \to \N \\
  ({\trunc \kappa n})~(m) &:= \kappa(n + m)\\
  \Ltotal {\_} {\_} &: (\N \to \N) \to \N \to \N \\
  \Ltotal \kappa 0 &:= 0 \\
  \Ltotal \kappa {1 + n} &:= \kappa(0) + \Ltotal {\trunc \kappa 1} n  \\
  \_\circ\_ &: (\N \to \N) \to (\N \to \N) \to \N \to \N \\
  (\kappa \circ \kappa')(0) &:= \Ltotal{\kappa'}{\kappa(0)} \\
  (\kappa \circ \kappa')(1 + n) &:= ((\trunc\kappa 1) \circ (\trunc{\kappa'}{\kappa(0)}))(n)
\end{align*}
We also need an identity modal transformation:
\begin{align*}
  \vone &: \N \to \N \\
  \vone &:= \_ \mapsto 1
\end{align*}
We next define application of \UnMoTs to $D$ and $\Envs$:
\begin{align*}
  \tbox(a)[\kappa] &:= \tbox(a[\sextt\kappa 1]) \\
  \Lambda(\vrho, x, t)[\kappa] &:= \Lambda(\vrho[\kappa], x, t) \\
  \uparrow^T(c)[\kappa] &:= \uparrow^T(c[\kappa]) \\
  \\
  z[\kappa] &:= z \\
  c\ d[\kappa] &:= (c[\kappa])\ (d[\kappa]) \\ 
  \tunbox(k, c)[\kappa] &:= \tunbox(\Ltotal\kappa k, c[\trunc\kappa k]) \\
  \\
  \downarrow^T(a)[\kappa] &:= \downarrow^T(a[\kappa]) \\
  \\
  \vrho[\kappa](0) &:= (\Ltotal\kappa k, \rho[\kappa]) \tag{where $(k, \rho) :=
                     \vrho(0)$} \\
  \vrho[\kappa](1 + n) &:= \trunc\vrho 1[\trunc\kappa{\Ltotal\vrho 1}](n)
\end{align*}
where $\rho[\kappa]$ is defined by mapping $x$ to $\rho(x)[\kappa]$ if $x$ is defined
in $\rho$. Notice that the $1 + n$ case for $\vrho[\kappa]$ does not really decrease
by structure; nonetheless, it helps us to understand what $\vrho[\kappa]$ does. The following lemma
gives a structure-decreasing equation which can be plugged in as the definition, but
it is rather convoluted so we choose to do the otherwise:
\begin{lemma}
  If $\vrho(n) = (m, \rho)$, then $\vrho[\kappa](n) = (\Ltotal{\trunc\kappa{\Ltotal\vrho n}} m, \rho[\trunc\kappa{\Ltotal\vrho n}])$.
\end{lemma}
The right hand side might seem unmotivated at the first glance, but it will make more
sense after we establish some lemmas and notice that $\vrho[\kappa](n) =
(\trunc{\vrho[\kappa]} n)(0) = \trunc\vrho n[\trunc\kappa{\Ltotal\vrho n}](0)$.

We can prove a number of properties for modal transformations.
\begin{lemma}
  The following are properties for $\vone$:
  \begin{itemize}
  \item $\vone ; 1 = \vone$
  \item $\trunc\vone n = \vone$
  \item $\Ltotal\vone n = n$
  \item $\vone \circ \kappa = \kappa$ 
  \item $\kappa \circ \vone = \kappa$ 
  \item $a[\vone] = a$
  \item $c[\vone] = c$
  \item $d[\vone] = d$
  \item $\vrho[\vone] = \vrho$ 
  \end{itemize}
\end{lemma}

The following lemmas describe the general interaction between composition and truncation:
\begin{lemma}
  $\trunc{(\kappa \circ \kappa')} 1 = \trunc\kappa 1 \circ \trunc{\kappa'}{\Ltotal\kappa 1}$
\end{lemma}
\begin{proof}
  By definition.
\end{proof}
\begin{lemma}
  $\trunc{(\kappa \circ \kappa')}n = \trunc\kappa n \circ \trunc{\kappa'}{\Ltotal\kappa n}$
\end{lemma}
\begin{proof}
  Do induction on $n$ and use the previous lemma.
\end{proof}

There is a similar lemma for truncation offset:
\begin{lemma}
  $\Ltotal{\kappa \circ \kappa'}n = \Ltotal{\kappa'}{\Ltotal\kappa n}$
\end{lemma}

From this lemma we can derive the associativity:
\begin{lemma}
  $(\kappa \circ \kappa') \circ \kappa'' = \kappa \circ (\kappa' \circ \kappa'')$
\end{lemma}

We have the following lemma by induction:
\begin{lemma}
  \begin{itemize}
  \item $a[\kappa][\kappa'] = a[\kappa \circ \kappa']$
  \item $c[\kappa][\kappa'] = c[\kappa \circ \kappa']$
  \item $d[\kappa][\kappa'] = d[\kappa \circ \kappa']$
  \item $\vrho[\kappa][\kappa'] = \vrho[\kappa \circ \kappa']$
  \end{itemize}
\end{lemma}

The follow lemmas investigate the interactions between truncation offset, truncation and application. First we
consider the special case where the truncation is $1$. 
\begin{lemma}\labeledit{lem:envs-mon-1}
  $\trunc{\vrho[\kappa]}1 = \trunc\vrho 1[\trunc\kappa{\Ltotal\vrho 1}]$
\end{lemma}
\begin{proof}
  Immediate by definition. 
\end{proof}

We can then generalize the conclusion to all $n$.
\begin{lemma}\labeledit{lem:envs-mon}
  $\trunc{\vrho[\kappa]}n = \trunc\vrho n[\trunc\kappa{\Ltotal\vrho n}]$ 
\end{lemma}
\begin{proof}
  We proceed by induction on $n$.
  \begin{itemize}[label=Case]
  \item $n = 0$, immediate.
  \item $n = 1 + n'$, then we consider its argument $m$:
    \begin{align*}
      \trunc{\vrho[\kappa]}{(1 + n')}(m)
      &= \vrho[\kappa](1 + n' + m) \\
      &= \trunc\vrho 1[\trunc\kappa{\Ltotal\vrho 1}](n' + m) \\
      &= (\trunc{(\trunc\vrho 1[\trunc\kappa{\Ltotal\vrho 1}])}{n'})(m) \\
      &= (\trunc{\trunc{\vrho[\kappa]}{1}}{n'})(m)
      \tag{by \Cref{lem:envs-mon-1}}\\
      &= (\trunc{\vrho[\kappa]}{(1 + n')})(m)
    \end{align*}
  \end{itemize}
\end{proof}

\begin{lemma}\labeledit{lem:envs-L}
  $\Ltotal{\vrho[\kappa]}n = \Ltotal\kappa{\Ltotal\vrho n}$
\end{lemma}
\begin{proof}
  We proceed by induction on $n$.
  \begin{itemize}[label=Case]
  \item $n = 0$, immediate.
  \item $n = 1 + n'$, then
    \begin{align*}
      \Ltotal{\vrho[\kappa]}n
      &= k + \Ltotal{\trunc{\vrho[\kappa]}1}{n'} \tag{where $(k, \_) := \vrho[\kappa](0)$}  \\
      &= k + \Ltotal{\trunc\vrho 1[\trunc\kappa{\Ltotal\vrho 1}]}{n'} \tag{by \Cref{lem:envs-mon-1}} \\
      &= k + \Ltotal{\trunc\kappa{\Ltotal\vrho 1}}{\Ltotal{\trunc\vrho 1}{n'}} \byIH \\
      &= \Ltotal\kappa{\Ltotal\vrho n}
    \end{align*}
  \end{itemize}
\end{proof}

\subsection{Evaluation}

Next we consider the evaluation function, which, given $\vrho$, translate an $\Exp$ to
a $D$.
\begin{align*}
  \intp{\_} &: \Exp \to \Envs \to D \\
  \intp{x}(\vrho) &:= \rho(x) \tag{where $(\_, \rho) := \vrho(0)$} \\
  \intp{\boxit t}(\vrho) &:= \tbox(\intp{t}(\ext(\vrho))) \\
  \intp{\unbox n t}(\vrho) &:= \tunbox \cdot (\Ltotal\vrho n, \intp{t}(\trunc\vrho n)) \\
  \intp{\lambda x. t}(\vrho) &:= \Lambda(\vrho, x, t) \\
  \intp{t\ s}(\vrho) &:= \intp{t}(\vrho) \cdot \intp{s}(\vrho) \\
  \intp{t[\vsigma]}(\vrho) &:= \intp{t}(\intp{\vsigma}(\vrho))
\end{align*}

Here we make use of two partial functions which evaluates $\tbox$ and $\Lambda$,
respectively.
We define the partial unbox as follows:
\begin{align*}
  \tunbox \cdot &: \mathbb{N} \rightharpoonup D \rightharpoonup D \\
  \tunbox \cdot (k, \tbox(a)) &:= a[\sextt\vone k]  \\
    \tunbox \cdot (k, \uparrow^{\square T} c) &:= \uparrow^{T} (\tunbox(k, c))
\end{align*}
We define the partial application as follows:
\begin{align*}
  \_ \cdot \_ &: D \rightharpoonup D \rightharpoonup D \\
  (\Lambda(\vrho, x, t)) \cdot a &:= \intp{t}(\ext(\vrho, x, a)) \\
  (\uparrow^{T \func T'} c) \cdot a &:= \uparrow^{T'} (c\ \downarrow^T(a))
\end{align*}
In the $t[\vsigma]$ case, we need the interpretation of substitutions:
\begin{align*}
    \intp{\_} &: \Substs \to \Envs \to \Envs \\
    \intp{\vect I}(\vrho) &:= \vrho \\
    \intp{\vsigma, t/x}(\vrho) &:= \ext(\intp{\vsigma}(\vrho), x, \intp{t}(\vrho)) \\
    \intp{\wk_x}(\vrho) &:= \drop(\vrho, x) \\
  \intp{\sextt\vsigma n}(\vrho) &:= \ext(\intp{\vsigma}(\trunc\vrho n), \Ltotal\vrho n) \\
  \intp{\vsigma \circ \vdelta}(\vrho) &:= \intp{\vsigma}(\intp{\vdelta}(\vrho))
\end{align*}

In the presheaf model shown in \Cref{sec:presheaf}, unified weakenings commute with
many operations, e.g. evaluation and exponential application, because of functoriality
and naturality. In the untyped domain, something similar exists.
\begin{lemma}\labeledit{lem:unbox.-mon}
  $\tunbox \cdot (k, a)[\kappa] = \tunbox \cdot (\Ltotal\kappa k, a[\trunc\kappa k])$
\end{lemma}
\begin{proof}
  We consider the cases where the equation is defined:
  \begin{itemize}[label=Case]
  \item $a = \tbox(b)$, then
    \begin{align*}
      \tunbox \cdot (k, a)[\kappa]
      &= b[\sextt\vone k][\kappa] \\
      &= b[\sextt{\trunc\kappa k}{\Ltotal\kappa k}] \\
      &= b[\sextt{\trunc\kappa k} 1][\sextt{\vone}{\Ltotal\kappa k}] \\
      &= \tunbox \cdot (\Ltotal\kappa k, b[\sextt{\trunc\kappa k} 1]) \\
      &= \tunbox \cdot (\Ltotal\kappa k, \tbox(b)[\trunc\kappa k])
    \end{align*}
  \item $a = \uparrow^{\square T} c$, then
    \begin{align*}
      \tunbox \cdot (k, a)[\kappa]
      &= \uparrow^T (\tunbox(k, c)[\kappa]) \\
      &= \uparrow^T (\tunbox(\Ltotal\kappa k, c[\trunc\kappa k])) \\
      &= \tunbox \cdot (\Ltotal\kappa k, (\uparrow^{\square T} c)[\trunc\kappa k])
    \end{align*}
  \end{itemize}
\end{proof}

In general, lemmas like this should characterize the interaction between untyped modal
transformations and operations.
\begin{lemma}[Naturality Equalities]$\;$
  \begin{itemize}
  \item $\intp{t}(\vrho[\kappa]) = \intp{t}(\vrho)[\kappa]$
  \item $\intp{\vsigma}(\vrho[\kappa]) = \intp{\vsigma}(\vrho)[\kappa]$
  \item $(a \cdot b)[\kappa] = (a[\kappa]) \cdot (b[\kappa])$
  \end{itemize}
\end{lemma}
These three equations including the previous one are collectively called
\emph{naturality equalities}. Similar to the presheaf model, we will also need these
equalities to esstablish the completeness proof. Unlike the presheaf model, in which
naturality has been built in as part of various natural transformations, in the
untyped domain model, we have to explicitly list and prove them.
\begin{proof}
  These three equalities must be proved mutually. They are proved by analyzing each
  defined case and apply IH whenever necessary.
\end{proof}

\subsection{Readback Functions}\labeledit{sec:st:rb}

After evaluating an $\Exp$ to $D$, we have already got the corresponding $\beta$
normal form of the term in $D$. We need a one last step to read from $D$ back to
normal form and do the $\eta$ expansion at the same time to obtain a $\beta\eta$
normal form:
\begin{align*}
  \Rnf &: (\N \rightharpoonup \Var) \rightharpoonup D^{\Nf} \rightharpoonup \Nf \\
  \Rnf_{\alpha} (\downarrow^{\square T} (a))
       &:= \boxit \Rnf_{\alpha} (\downarrow^{T} (\tunbox \cdot (1, a)))
  \\
  \Rnf_{\alpha} (\downarrow^{S \func T} (a))
       &:= \lambda x. \Rnf_{\alpha[z \mapsto x]} (\downarrow^T (a \cdot \uparrow^S(z)))
  \tag{where $z := \tnext(\alpha)$}\\
  \Rnf_{\alpha} (\downarrow^B (\uparrow^B (c)))
       &:= \Rne_{\alpha} (c) \\
  \\
  \Rne &: (\N \rightharpoonup \Var) \rightharpoonup D^{\Ne} \rightharpoonup \Ne \\
  \Rne_{\alpha}(z)
       &:= \alpha(z)
  \tag{if $\alpha$ is undefined at $z$, we assign some fixed default $\Var$}\\
  \Rne_{\alpha}(c\ d)
       &:= \Rne_{\alpha} (c)\ \Rnf_{\alpha}(d) \\
  \Rne_{\alpha}(\tunbox(k, c))
       &:= \unbox{k}{\Rne_{\alpha}(c)}
\end{align*}
In the readback function, a function $\alpha$ is maintained. It is a function keeping
track of the correspondence between syntactic and domain-level
variables. Conceptually, a domain variable $z$ is a natural number remembering the
\emph{location} of the syntactic variable in the context stack it corresponds
to. Concretely, there are two solutions:
\begin{enumerate}
\item $z$ is the $z$'th variable in its context. That is, if $\vGamma; (\Gamma, x : T,
  \Gamma')$ where $x$ and $z$ corresponds, then $z = |\Gamma|$. 
\item $z$ is the $z$'th variable in the whole context stack. That is, if $\vGamma; (\Gamma, x : T,
  \Gamma')$ where $x$ and $z$ corresponds, then $z = ||\vGamma|| + |\Gamma|$, where
  \begin{align*}
    ||\vGamma|| := \sum_{\Gamma'' \in \vGamma} |\Gamma''|
  \end{align*}
\end{enumerate}
Either solution would work in our NbE algorithm. The first solution seems to map
multiple syntactic variables to the same natural number, e.g. in the context stack
$\epsilon;x : T; y : S$, $\alpha$ would map both
$x$ and $y$ to $0$. This is fine because the $0$s representing $x$ and $y$ are
distinguished by worlds, and variables cannot cross different worlds.

We use $\alpha$ to represent the effect from $\N$ to $\Var$ and the opposite direction is
$\alpha^{-1}$. Technically speaking, both effects are only defined in finite domains, so
they can only be partial functions, but our logical relations will demonstrate that
invalid inputs are never given to the effects.  The $\tnext(\alpha)$ function used
in the function case in $\Rnf$ returns the next
available domain-level variable in $\alpha$. If $\alpha$ is undefined everywhere, then
$\tnext(\alpha) = 0$. Otherwise, $\tnext(\alpha) = \max(\dom(\alpha)) + 1$. We will discuss how
$\alpha$ is implemented when de Bruijn indices are used in the syntactic level in \Cref{sec:st:debruijn}. 

We define the initial
evaluation environment:
\begin{align*}
  \uparrow^{\vGamma} &: \Envs \\
  \uparrow^{\epsilon;\Gamma} &:= \ext(\empenv, x \mapsto \uparrow^T(\alpha^{-1}(x)) \text{
                            if }x : T \in \Gamma) \\
  \uparrow^{\vGamma;\Gamma} &:= \ext(\uparrow^{\vGamma}, x \mapsto \uparrow^T(\alpha^{-1}(x)) \text{
                            if }x : T \in \Gamma)
\end{align*}
where $\alpha^{-1} : \Var \to \N$ is the opposite of $\alpha$.
\begin{definition}
  The NbE algorithm is defined to be
  \begin{align*}
    \nbe_{\vGamma}^T(t) := \Rnf_{\alpha}(\downarrow^T (\intp{t}(\uparrow^{\vGamma})))
  \end{align*}
  where $\alpha$ is defined as above. 
\end{definition}

\subsection{Candidate Space and PER Model}\labeledit{sec:st:per}

We define two PERs over $D$ from which we form a candidate space for each type. For a PER over
$D$, $P$, we write $a \approx b \in P$ for $(a, b) \in P$:
\begin{mathpar}  
  \inferrule*
  {\forall \alpha, \kappa. \Rnf_{\alpha} (\downarrow^T(a[\kappa])) = \Rnf_{\alpha} (\downarrow^T(b[\kappa]))}
  {a \approx b \in \top^T}
  
  \inferrule*
  {\forall \alpha, \kappa. \Rne_{\alpha}(c[\kappa]) = \Rne_{\alpha}(c'[\kappa])}
  {\uparrow^T (c) \approx \uparrow^T (c') \in \bot^T}
\end{mathpar}
$\bot^T$ and $\top^T$ are binary relations over $D$. Furthermore, we can easily verify
that they are partial equivalence relations (PERs). Note that the universal
quantification over $\alpha$ here still maintain a correspondence.

We then interpret the types to a PER model:
\begin{align*}
  \intp{B} &:= \bot^B \\
  \intp{\square T} &:= \{ (a, b) \sep \forall k, \kappa. \tunbox \cdot (k, a[\kappa]) \approx \tunbox
                     \cdot (k, b[\kappa]) \in \intp{T} \} \\
  \intp{S \func T} &:= \{ (a, b) \sep \forall \kappa, a' \approx b' \in \intp{S}.
                     a[\kappa] \cdot a' \approx b[\kappa] \cdot b' \in \intp{T} \} 
\end{align*}

Next we demonstrate an important property of the PER model, realizability.
Realizability of $\intp{T}$ states $\bot_T \subseteq \intp{T} \subseteq \top_T$. The
realizability theorem states that all $\intp{T}$ are realizable.

We first need the following lemma:
\begin{lemma}\labeledit{lem:varbot}
  $\uparrow^T(z) \approx \uparrow^T(z) \in \bot^T$
\end{lemma}
\begin{proof}
  Holds by definition. 
\end{proof}

\subsubsection{Realizability}

\begin{theorem}[Realizability]\labeledit{thm:real}
   For all type $T$, $\bot^T \subseteq \intp{T} \subseteq \top^T$.
\end{theorem}
\begin{proof}
  We do induction on $T$.
  \begin{itemize}[label=Case]
  \item $T= B$, immediate.
  \item $T = \square T'$,
    \begin{itemize}[label=Subcase]
    \item We shall show $\bot^{\square T'} \subseteq \intp{\square T'}$. Our goal is
      to apply $\bot^{T'} \subseteq \intp{T'}$ as IH and thus we will first need to
      construct some terms related by $\bot^{T'}$. 
      \begin{align*}
        H_1: &\ \forall \alpha, \kappa. \Rne_{\alpha}(c[\kappa]) = \Rne_{\alpha}(c'[\kappa]) \tag{assumption} \\
             &\ \text{assume } k, \kappa, \alpha, \kappa' \\
             &\ \Rne_{\alpha}(c[\kappa \circ (\trunc{\kappa'}k)]) = \Rne_{\alpha}(c'[\kappa \circ (\trunc{\kappa'}k)])
               \tag{by $H_1$} \\
             &\ \Rne_{\alpha}(\tunbox(\Ltotal{\kappa'}k, c[\kappa][\trunc{\kappa'}k])) =
               \Rne_{\alpha}(\tunbox(\Ltotal{\kappa'}k, c'[\kappa][\trunc{\kappa'}k]))
               \tag{by congruence}\\
             &\ \Rne_{\alpha}(\tunbox(k, c[\kappa])[\kappa']) =
               \Rne_{\alpha}(\tunbox(k, c'[\kappa])[\kappa'])
               \tag{by definition} \\
             &\ \uparrow^{T'}(\tunbox(k, c[\kappa])) \approx \uparrow^{T'}(\tunbox(k, c'[\kappa])) \in \bot^{T'}
               \tag{by abstraction} \\
             &\ \uparrow^{T'}(\tunbox(k, c[\kappa])) \approx \uparrow^{T'}(\tunbox(k, c'[\kappa])) \in
               \intp{T'}
               \byIH \\
             &\ \uparrow^{\square T'}(c) \approx \uparrow^{\square T'}(c) \in
               \intp{\square T'}
               \tag{by abstraction}
      \end{align*}
      
    \item Now we show $\intp{\square T'} \subseteq \top^{\square T'}$. Assuming
      $a, b : D$, $\alpha$ and $\kappa$, then we shall prove
      $\Rnf_\alpha (\downarrow^{\square T'}(a[\kappa])) = \Rnf_\alpha (\downarrow^{\square
        T'} (b[\kappa]))$:
      \begin{align*}
        &\forall k, \kappa. \tunbox \cdot (k, a[\kappa]) \approx \tunbox \cdot (k, b[\kappa]) \in \intp{T'}
          \tag{assumption} \\
        &\tunbox \cdot (1, a[\kappa]) \approx \tunbox \cdot (1, b[\kappa]) \in \intp{T'} \\
        &\tunbox \cdot (1, a[\kappa]) \approx \tunbox \cdot (1, b[\kappa]) \in
          \top^{T'}
          \byIH \\
        &\Rnf_\alpha(\downarrow^{T'}(\tunbox \cdot (1, a[\kappa]))) =
          \Rnf_\alpha(\downarrow^{T'}(\tunbox \cdot (1, b[\kappa])))
        \tag{by definition}\\
        &\boxit\Rnf_\alpha(\downarrow^{T'}(\tunbox \cdot (1, a[\kappa]))) =
          \boxit\Rnf_\alpha(\downarrow^{T'}(\tunbox \cdot (1, b[\kappa])))
        \tag{by congruence}\\
        &\Rnf_\alpha (\downarrow^{\square T'}(a[\kappa])) = \Rnf_\alpha (\downarrow^{\square T'}(b[\kappa]))
      \end{align*}
    \end{itemize}

  \item $T = S \func T'$, then 
    \begin{itemize}[label=Subcase]
    \item $\bot^{S \func T'} \subseteq \intp{S \func T'}$ is rather straightforward.
      \begin{align*}
        H_2: &\ \forall \alpha, \kappa. \Rne_{\alpha}(c[\kappa]) = \Rne_{\alpha}(c'[\kappa]) \tag{assumption} \\
        &\ \text{assume }\kappa', a \approx b \in \intp{S} \\
        &\ a \approx b \in \top^S \byIH \\
        &\ \forall \alpha, \kappa. \Rne_{\alpha}(c[\kappa' \circ \kappa]) =
          \Rne_{\alpha}(c'[\kappa' \circ \kappa]) \tag{by $H_2$} \\
        &\ \forall \alpha, \kappa. \Rne_{\alpha}((c[\kappa']\ \downarrow^S(a))[\kappa]) =
          \Rne_{\alpha}((c'[\kappa]\ \downarrow^S(b))[\kappa]) \\
        &\ \uparrow^{T'}(c[\kappa']\ \downarrow^S(a)) \approx
          \uparrow^{T'}(c'[\kappa']\ \downarrow^S(b)) \in \bot^{T'}
        \tag{by definition} \\
        &\ \uparrow^{T'}(c[\kappa']\ \downarrow^S(a)) \approx
          \uparrow^{T'}(c'[\kappa']\ \downarrow^S(b)) \in \intp{T'}
          \byIH \\
        &\ \uparrow^{S \func T'}(c) \approx \uparrow^{S \func T'}(c') \in \intp{S
          \func T'} \tag{by abstraction} 
      \end{align*}
      We will only show one of the two equations as they are symmetric. We assume
      $\kappa''$, then 
      \begin{align*}
        \uparrow^{S \func T'}(c)[\kappa'][\kappa''] \cdot a[\kappa'']
        &= \uparrow^{T'}(c[\kappa'][\kappa'']\ \downarrow^S(a[\kappa'']))  \\
        &= \uparrow^{T'}(c[\kappa']\ \downarrow^S(a))[\kappa'']  \\
        &= (\uparrow^{S \func T'}(c[\kappa']) \cdot a)[\kappa'']
      \end{align*}
      Thus they agree. 
      
    \item $\intp{S \func T'} \subseteq \top^{S \func T'}$, assume $a \approx b \in
      \intp{S \func T'}$, $\alpha$ and $\kappa$, then
      \begin{align*}
        &\uparrow^S(z) \approx \uparrow^S(z) \in \bot^S
        \tag{by \Cref{lem:varbot}} \\
        &\uparrow^S(z) \approx \uparrow^S(z) \in \intp{S}
          \byIH \\
        &a[\kappa] \cdot \uparrow^S(z) \approx b[\kappa] \cdot \uparrow^S(z) \in
          \intp{T'} \\
        &a[\kappa] \cdot \uparrow^S(z) \approx b[\kappa] \cdot \uparrow^S(z) \in
          \top^{T'} \byIH \\
        &\Rnf_{\alpha[z \mapsto x]} \downarrow^T (a[\kappa] \cdot \uparrow^S(z)) =
          \Rnf_{\alpha[z \mapsto x]} \downarrow^T (b[\kappa] \cdot \uparrow^S(z)) \\
        &\lambda x. \Rnf_{\alpha[z \mapsto x]} \downarrow^T (a[\kappa] \cdot \uparrow^S(z)) =
          \lambda x. \Rnf_{\alpha[z \mapsto x]} \downarrow^T (b[\kappa] \cdot \uparrow^S(z))
          \tag{by congruence} \\
        &a \approx b \in \top^{S \func T'} \tag{by abstraction}
      \end{align*}
    \end{itemize}    
  \end{itemize}
\end{proof}

\begin{lemma}
  $\intp{T}$ is a PER.
\end{lemma}

Next, we generalize the PER model to both contexts and context stacks.
\begin{align*}
  \intp{\Gamma} &:= \{ (\rho, \rho') \sep \rho(x) \approx \rho'(x) \in \intp{T}
                  \text{ if } x : T \in \Gamma \} \\
  \intp{\epsilon; \Gamma} &:= \{ (\vrho, \vrho') \sep \rho \approx \rho
                   \in \Gamma \text{ where } (k, \rho) := \vrho(0) \tand (k', \rho')
                   := \vrho'(0) \} \\
  \intp{\vGamma; \Gamma} &:= \{ (\vrho, \vrho') \sep k = k' \tand \rho \approx \rho
                   \in \Gamma \tand \trunc\vrho 1 \approx \trunc{\vrho'}1 \in \intp{\vGamma} \text{ where } (k, \rho) := \vrho(0) \tand (k', \rho')
                   := \vrho'(0) \}
\end{align*}

Following lemma holds:
\begin{lemma}
  If $\vrho \approx \vrho' \in \intp{\vGamma}$, then $\Ltotal\vrho n = \Ltotal{\vrho'}n$.
\end{lemma}
\begin{lemma}
  If $\vrho \approx \vrho' \in \intp{\vGamma}$, then $\trunc\vrho n \approx \trunc{\vrho'}n \in \intp{\vGamma|_n}$.
\end{lemma}
\begin{lemma}
  $\intp{\Gamma}$ and $\intp{\vGamma}$ are PERs. 
\end{lemma}

We define the semantic judgments:
\begin{definition}
  We define semantic typing and semantic typing equivalence as follows:
  \begin{align*}
    \msemtyeq{t}{t'}{T} &:=
    \forall \vrho \approx \vrho' \in \intp{\vGamma}. \intp{t}(\vrho) \approx
                          \intp{t'}(\vrho') \in \intp{T} \\
                        & \qquad \qquad \tand \forall \kappa. \intp{t}(\vrho[\kappa]) =
                          \intp{t}(\vrho)[\kappa] \tand \intp{t'}(\vrho'[\kappa]) =
                          \intp{t'}(\vrho')[\kappa] \\
    \msemtyeq{\vsigma}{\vsigma'}{\vDelta} &:=
    \forall \vrho \approx \vrho' \in \intp{\vGamma}. \intp{\vsigma}(\vrho) \approx
                                            \intp{\vsigma'}(\vrho') \in \intp{\vDelta} \\
                        & \qquad \qquad \tand \forall \kappa. \intp{\vsigma}(\vrho[\kappa]) =
                          \intp{\vsigma}(\vrho)[\kappa] \tand \intp{\vsigma'}(\vrho'[\kappa]) =
                          \intp{\vsigma'}(\vrho')[\kappa]
  \end{align*}
  \begin{align*}
    \msemtyp{t}{T} &:= \msemtyeq{t}{t}{T} \\
    \msemtyp{\vsigma}{\vDelta} &:= \msemtyeq{\vsigma}{\vsigma}{\vDelta}
  \end{align*}
\end{definition}
Notice that we add \Cref{conj:intp-mon} in the semantic judgment, so that we can make
sure the conjecture holds for the well-typed terms. We call these equations \emph{naturality equalities}.

\subsubsection{Kripke Structure of PER Model}

It turns out that the PER model $\intp{T}$ is closed under modal transformations. In
fact, we can view $\intp{T}$ as a Kripke relation where the wolds are connected by
modal transformations. This intuition is justified by the following lemma.
\begin{lemma}\labeledit{lem:intpt-resp-mon}
  If $a \approx b \in \intp{T}$, then $a[\kappa] \approx b[\kappa] \in \intp{T}$. 
\end{lemma}
\begin{proof}
  We proceed by induction on $T$.
  \begin{itemize}[label=Case]
  \item $T = B$, immediate because $\kappa$ is closed under composition. 
  \item $T = \square T'$, assume $\kappa$, then
    \begin{align*}
      H_1: &\ \forall k, \kappa'. \tunbox \cdot (k, a[\kappa']) \approx \tunbox \cdot (k, b[\kappa']) \in \intp{T'}
        \tag{by assumption} \\
      &\ \text{assume }k, \kappa'
      \\
      &\ (\tunbox \cdot (k, a[\kappa \circ \kappa'])) \approx (\tunbox \cdot (k, b[\kappa \circ \kappa'])) \in
      \intp{T'}
      \tag{by $H_1$} \\
      &\ (\tunbox \cdot (k, a[\kappa][\kappa'])) \approx (\tunbox \cdot (k, b[\kappa][\kappa'])) \in
        \intp{T'}
        \tag{by definition} \\
      &\ a[\kappa] \approx b[\kappa] \in \intp{\square T'}
      \tag{by abstraction}
    \end{align*}
    
  \item $T = S \func T'$, assume $\kappa$, then
    \begin{align*}
      H_2: &\ \forall \kappa', a' \approx b' \in \intp{S}.
             a[\kappa'] \cdot a' \approx b[\kappa'] \cdot b' \in \intp{T'}
        \tag{by assumption} \\
           &\ \text{assume }\kappa', a' \approx b' \in \intp{S} \\
           &\ a[\kappa \circ \kappa'] \cdot a' \approx b[\kappa \circ \kappa']
             \cdot b' \in \intp{T'}
             \tag{by $H_2$ with $\kappa \circ \kappa'$} \\
           &\ a[\kappa][\kappa'] \cdot a' \approx b[\kappa][\kappa']
             \cdot b' \in \intp{T'} \\
           &\ a[\kappa] \approx b[\kappa] \in \intp{S \func T'}
             \tag{by abstraction}
    \end{align*}
  \end{itemize}
\end{proof}
This lemma holds because untyped modal transformations are closed under composition. 

We can generalize this lemma of types to context stacks:
\begin{lemma}\labeledit{lem:intpvg-resp-mon}
  If $\vrho \approx \vrho' \in \intp{\vGamma}$, then $\vrho[\kappa] \approx \vrho'[\kappa] \in \intp{\vGamma}$.
\end{lemma}

\subsection{Completeness}

In this section, we will cover the semantic equivalence typing rules and the
fundamental theorem which justifies the the completeness theorem.

\subsubsection{PER Rules}

\begin{lemma}
  \begin{mathpar}
    \inferrule*
    {\msemtyeq s t T}
    {\msemtyeq t s T}
  \end{mathpar}
\end{lemma}
\begin{lemma}
  \begin{mathpar}
    \inferrule*
    {\msemtyeq s t T \\ \msemtyeq t u T}
    {\msemtyeq s u T}
  \end{mathpar}
\end{lemma}
\begin{proof}
 We use the fact that $\intp{\vGamma}$ and $\intp{T}$ are PER. 
\end{proof}

\subsubsection{Congruence Rules}

\begin{lemma}
  \begin{mathpar}
    \inferrule*
    {x : T \in \Gamma}
    {\msemtyeq[\vGamma; \Gamma] x x T}
  \end{mathpar}
\end{lemma}
\begin{proof}
  Immediate because we are doing the same lookup.
\end{proof}

\begin{lemma}
  \begin{mathpar}
    \inferrule*
    {\msemtyeq[\vGamma; \cdot]{t}{t'}T}
    {\msemtyeq{\boxit t}{\boxit t'}{\square T}}
  \end{mathpar}
\end{lemma}
\begin{proof}
  \begin{align*}
    H_1: &\ \msemtyeq[\vGamma; \cdot]{t}{t'}T
           \tag{by assumption} \\
    H_2: &\ \vrho \approx \vrho' \in \intp{\vGamma}
           \tag{by assumption} \\
    H_3: &\ \ext(\vrho) \approx \ext(\vrho') \in \intp{\vGamma; \cdot}
           \tag{by definition} \\
         &\ \intp{t}(\ext(\vrho)) \approx \intp{t'}(\ext(\vrho')) \in \intp{T}
           \tag{by $H_1$ and $H_3$} \\
         &\ \text{assume }k, \kappa\\
         &\ \intp{t}(\ext(\vrho))[\sextt\kappa 1][\sextt\vone k] \approx
           \intp{t'}(\ext(\vrho'))[\sextt\kappa 1][\sextt\vone k] \in \intp{T}
           \tag{by \Cref{lem:intpt-resp-mon}} \\
         &\ \forall k, \kappa. \tunbox \cdot (k, \tbox(\intp{t}(\ext(\vrho)))[\kappa]) \approx \tunbox
           \cdot (k, \tbox(\intp{t'}(\ext(\vrho')))[\kappa]) \in \intp{T}
           \tag{by abstraction} \\
         &\ \intp{\boxit t}(\vrho) \approx \intp{\boxit t'}(\vrho') \in \intp{\square
           T}
           \tag{by definition}
  \end{align*}
\end{proof}

\begin{lemma}
  \begin{mathpar}
    \inferrule*
    {\msemtyeq[\vGamma;(\Gamma, x : S)]{t}{t'}T}
    {\msemtyeq[\vGamma; \Gamma]{\lambda x. t}{\lambda x. t'}{S \func T}}
  \end{mathpar}
\end{lemma}
\begin{proof}
  \begin{align*}
    H_1: &\ \msemtyeq[\vGamma;(\Gamma, x : S)]{t}{t'}T
           \tag{by assumption} \\
    H_2: &\ \vrho \approx \vrho' \in \intp{\vGamma; \Gamma}
           \tag{by assumption} \\
         &\ \text{assume }\kappa, a \approx a' \in \intp{S} \\
         &\ \vrho[\kappa] \approx \vrho'[\kappa] \in \intp{\vGamma; \Gamma}
           \tag{by \Cref{lem:intpvg-resp-mon}}  \\
    H_3: &\ \ext(\vrho[\kappa], x, a) \approx \ext(\vrho'[\kappa], x, a') \in
           \intp{\vGamma; (\Gamma, x : S)}
           \tag{by definition} \\
         &\ \intp{t}(\ext(\vrho[\kappa], x, a)) \approx
           \intp{t'}(\ext(\vrho'[\kappa], x, a'))
           \in \intp{T}
           \tag{by $H_1$ and $H_3$} \\
    H_4: &\ \forall \kappa'. \intp{t}(\ext(\vrho[\kappa], x, a)[\kappa']) =
           \intp{t}(\ext(\vrho[\kappa], x, a))[\kappa'] \\
         &\ \tand \intp{t'}(\ext(\vrho'[\kappa], x, a')[\kappa']) = \intp{t'}(\ext(\vrho'[\kappa], x, a'))[\kappa']
           \tag{by $H_1$ and $H_3$} \\
         &\ \Lambda(\vrho, x, t)[\kappa] \cdot a
           \approx
           \Lambda(\vrho', x, t')[\kappa] \cdot a'
           \in \intp{T}
           \tag{by definition} \\
         &\ \intp{\lambda x. t}(\vrho) \approx \intp{\lambda x. t}(\vrho) \in \intp{S
           \func T}
           \tag{by abstraction}
  \end{align*}
\end{proof}

\begin{lemma}
  \begin{mathpar}
    \inferrule
    {\msemtyeq{t}{t'}{S \func T} \\
      \msemtyeq{s}{s'}S}
    {\msemtyeq{t\ s}{t'\ s'}T}
  \end{mathpar}
\end{lemma}
\begin{proof}
  Assuming $\vrho \approx \vrho' \in \intp{\vGamma; \Gamma}$, we demonstrate that
  $\forall \kappa. \intp{t\ s}(\vrho[\kappa']) = \intp{t\ s}(\vrho)[\kappa]$. Assuming
  $\kappa$, we compute:
  \begin{align*}
    \intp{t\ s}(\vrho[\kappa'])
    &= \intp{t}(\vrho[\kappa']) \cdot \intp{s}(\vrho[\kappa']) \\
    &= \intp{t}(\vrho)[\kappa'] \cdot \intp{s}(\vrho)[\kappa']
  \end{align*}
  and
  \begin{align*}
    \intp{t\ s}(\vrho)[\kappa]
    &= (\intp{t}(\vrho) \cdot \intp{s}(\vrho))[\kappa]
  \end{align*}
  Now we need to equate both sides, and that is exactly provided by $\intp{t}(\vrho)
  \approx \intp{t'}(\vrho') \in \intp{S \func T}$. 
\end{proof}

All following are immediate just by expanding the definitions:
\begin{lemma}
  \begin{mathpar}
    \inferrule
    {\msemtyeq{t}{t'}{\square T} \\
      |\vDelta| = n}
    {\msemtyeq[\vGamma; \vDelta]{\unbox n t}{\unbox n t'}{T'}}  

    \inferrule
    {\msemtyeq[\vDelta]{t}{t'}{T} \\
      \msemtyeq{\vsigma}{\vsigma'}{\vDelta}}
    {\msemtyeq{t[\vsigma]}{t'[\vsigma']}T}

    \inferrule
    { }
    {\msemtyeq{\vect I}{\vect I}{\vGamma}}

    \inferrule
    {\msemtyeq{\vsigma}{\vsigma'}{\vGamma'; \Gamma} \\
      \msemtyeq{t}{t'}T}
    {\msemtyeq{\vsigma, x : t}{\vsigma', x : t'}{\vGamma'; (\Gamma, x : T)}}
    
    \inferrule
    { }
    {\msemtyeq[\vGamma; (\Gamma, x : T)]{\wk_x}{\wk_x}{\vGamma;\Gamma}}
  
    \inferrule
    {\msemtyeq{\vsigma}{\vsigma'}{\vDelta} \\ |\vGamma'| = n}
    {\msemtyeq[\vGamma; \vGamma']{\sextt\vsigma n}{\sextt{\vsigma'}n}{\vDelta; \cdot}}

    \inferrule
    {\msemtyeq[\vGamma']{\vsigma}{\vsigma'}{\vGamma''} \\ \msemtyeq{\vdelta}{\vdelta'}{\vGamma'}}
    {\msemtyeq{\vsigma \circ \vdelta}{\vsigma' \circ \vdelta'}{\vGamma''}}
\end{mathpar}
\end{lemma}
\begin{proof}
  Immediate by premises.
\end{proof}

\subsubsection{$\beta$ Rules}

\begin{lemma}
  \begin{mathpar}
    \inferrule*
    {\msemtyp[\vGamma; \cdot]{t}{T} \\
      |\vDelta| = n}
    {\msemtyeq[\vGamma; \vDelta]{\unbox{n}{(\boxit t)}}{t[\sextt{\vect I}n]}{T}}
  \end{mathpar}
\end{lemma}
\begin{proof}
  \begin{align*}
    H_1: &\ \msemtyp[\vGamma; \cdot]{t}{T}
           \tag{by assumption} \\
    H_2: &\ \vrho \approx \vrho' \in \intp{\vGamma; \vDelta}
           \tag{by assumption}
  \end{align*}
  We compute both sides:
  \begin{align*}
    \intp{\unbox n{(\boxit t)}}(\vrho)
    &= \tunbox \cdot (\Ltotal\vrho n, \tbox(\intp{t}(\ext(\trunc\vrho n)))) \\
    &= \intp{t}(\ext(\trunc\vrho n))[\sextt\vone{\Ltotal\vrho n}] \\
    &= \intp{t}(\ext(\trunc\vrho n))[\sextt{\vone}{\Ltotal\vrho n}])
    \tag{by naturality equalities} \\
    &= \intp{t}(\ext(\trunc\vrho n, \Ltotal\vrho n))
      \tag{by definition}
  \end{align*}
  \begin{align*}
    \intp{t[\sextt{\vect I}n]}(\vrho')
    &= \intp{t}(\intp{\sextt{\vect I}n}(\vrho')) \\
    &= \intp{t}(\ext(\trunc{\vrho'}n, \Ltotal{\vrho'}n))
  \end{align*}
  We also know
  \begin{align*}
    &\trunc\vrho n \approx \trunc{\vrho'}n \in \intp{\vGamma}
    \tag{by $H_2$} \\
    &\Ltotal\vrho n = \Ltotal{\vrho'}n
    \tag{by $H_2$} \\
    &\ext(\trunc\vrho n, \Ltotal\vrho n) \approx \ext(\trunc{\vrho'}n, \Ltotal{\vrho'}n) \in \intp{\vGamma; \vDelta}
  \end{align*}
  Thus the target holds by $H_1$. 
\end{proof}

\begin{lemma}
  \begin{mathpar}
    \inferrule*
    {\msemtyp[\vGamma;(\Gamma, x : S)]t T \\
      \msemtyp[\vGamma; \Gamma] s S}
    {\msemtyeq[\vGamma; \Gamma]{(\lambda x. t)\ s}{t[\vect I, s/x]}{T}}
  \end{mathpar}
\end{lemma}
\begin{proof}
  \begin{align*}
    H_1: &\msemtyp[\vGamma;(\Gamma, x : S)]t T
           \tag{by assumption} \\
    H_2: &\msemtyp[\vGamma; \Gamma] s S
           \tag{by assumption} \\
    H_3: &\ \vrho \approx \vrho' \in \intp{\vGamma; \Gamma}
           \tag{by assumption}
  \end{align*}
  We compute both sides:
  \begin{align*}
    \intp{(\lambda x. t)\ s}(\vrho)
    &= \intp{t}(\ext(\vrho, x, \intp{s}(\vrho)))
  \end{align*}
  \begin{align*}
    \intp{t[\vect I, s/x]}(\vrho')
    &= \intp{t}(\intp{\vect I, s/x}(\vrho')) \\
    &= \intp{t}(\ext(\vrho', x, \intp{s}(\vrho')))
  \end{align*}
  We follow the definition and can show
  \begin{align*}
    \ext(\vrho, x, \intp{s}(\vrho)) \approx \ext(\vrho', x, \intp{s}(\vrho')) \in
    \intp{\vGamma; (\Gamma, x : S)}
  \end{align*}
  The target immediately follows.
\end{proof}

\subsubsection{$\eta$ Rules}

\begin{lemma}
  \begin{mathpar}
    \inferrule*
    {\msemtyp{t}{\square T}}
    {\msemtyeq{t}{\boxit{\unbox 1 t}}{\square T}}
  \end{mathpar}
\end{lemma}
\begin{proof}
  \begin{align*}
    H_1:&\ \msemtyp{t}{\square T}
          \tag{by assumption} \\
    H_2: &\ \vrho \approx \vrho' \in \intp{\vGamma}
           \tag{by assumption} \\
    H_3: &\ \intp{t}(\vrho) \approx \intp{t}(\vrho') \in \intp{\square T}
          \tag{by $H_1$}
  \end{align*}
  Since
  \begin{align*}
    \intp{\boxit{(\unbox 1 t)}}(\vrho')
    &= \tbox(\tunbox \cdot (1, \intp{t}(\vrho)))
  \end{align*}
  we shall prove
  \begin{align*}
    \intp{t}(\vrho) \approx \tbox(\tunbox \cdot (1, \intp{t}(\vrho'))) \in \intp{\square T}
  \end{align*}
  We proceed by its definition. Assuming $k$ and $\kappa$, we need to prove
  \begin{align*}
    &\ \tunbox \cdot (k, \intp{t}(\vrho)[\kappa]) \approx (\tunbox \cdot (1,
    \intp{t}(\vrho')))[\sextt\kappa 1][\sextt\vone k] \in \intp{T} \\
    \Leftrightarrow &\ \tunbox \cdot (k, \intp{t}(\vrho)[\kappa]) \approx (\tunbox \cdot (k,
    \intp{t}(\vrho')[\kappa])) \in \intp{T} 
    \tag{by \Cref{lem:unbox.-mon}}
  \end{align*}
  but this holds by $H_3$.
\end{proof}

\begin{lemma}
  \begin{mathpar}
    \inferrule*
    {\msemtyp{t}{S \func T}}
    {\msemtyeq t {\lambda x. (t[\wk_x]\ x)}{S \func T}}
  \end{mathpar}
\end{lemma}
\begin{proof}
  \begin{align*}
    H_1: &\ \msemtyp{t}{S \func T}
           \tag{by assumption} \\
    H_2: &\ \vrho \approx \vrho' \in \intp{\vGamma}
           \tag{by assumption} \\
    H_3: &\ \intp{t}(\vrho) \approx \intp{t}(\vrho') \in \intp{S \func T}
          \tag{by $H_1$}
  \end{align*}
  We reason about the goal, and we shall prove
  \begin{align*}
    \intp{t}(\vrho) \approx \Lambda(\vrho', x, t[\wk_x]\ x) \in \intp{S \func T}
  \end{align*}
  We proceed by its definition. Assuming $\kappa$, $a \approx b \in \intp{S}$, then
  we have the goal as
  \begin{align*}
    \Leftrightarrow&\ \intp{t}(\vrho)[\kappa] \cdot a \approx
                     \intp{t}(\drop(\ext(\vrho'[\kappa], x, b), x)) \cdot b
      \in \intp{S \func T} \\
    \Leftrightarrow&\ \intp{t}(\vrho)[\kappa] \cdot a \approx \intp{t}(\vrho'[\kappa]) \cdot b
      \in \intp{S \func T}
      \tag{by definition} \\
    \Leftrightarrow &\ \intp{t}(\vrho)[\kappa] \cdot a \approx \intp{t}(\vrho')[\kappa] \cdot b
      \in \intp{S \func T}
      \tag{by naturality equalities}
  \end{align*}
  We can conclude the last by $H_3$. 
\end{proof}

\subsubsection{Substitution Application Rules}

The following rules are immediate by following the definitions:
\begin{lemma}
  \begin{mathpar}
    \inferrule
    {\msemtyp t T}
    {\msemtyeq{t[\vect I]}{t}{T}}

    \inferrule
    {\msemtyp[\vGamma']{\vsigma}{\vGamma''} \\ \msemtyp[\vGamma]{\vdelta}{\vGamma'} \\
      \msemtyp[\vGamma'']{t}{T}}
    {\msemtyeq{t[\vsigma \circ \vdelta]}{t[\vsigma][\vdelta]}{T}}
    
    \inferrule
    {\msemtyp{\vsigma}{\vGamma'; \Gamma} \\ \msemtyp t T}
    {\msemtyeq{x[\vsigma, t/x]}{t}{T}}
    
    \inferrule
    {\msemtyp{\vsigma}{\vGamma'; \Gamma} \\ \msemtyp t T' \\ y : T \in \Gamma}
    {\msemtyeq{y[\vsigma, t/x]}{y[\vsigma]}{T}}

    \inferrule
    {y : T \in \Gamma}
    {\msemtyeq[\vGamma; (\Gamma, x : T')]{y[\wk_x]}{y}{T}}
    
    \inferrule
    {\msemtyp[\vGamma; \Gamma]{\vsigma}{\vDelta;\Delta} \\
      \msemtyp[\vDelta; \Delta, x : S]{t}{T}}
    {\msemtyeq[\vGamma; \Gamma]{\lambda x. t[\vsigma]}{\lambda x. (t[(\vsigma \circ
        \wk_x), x / x])}{S \func T}}
    
    \inferrule
    {\msemtyp{\vsigma}{\vDelta} \\
      \msemtyp[\vDelta]{s}{S \func T} \\
      \msemtyp[\vDelta]{t}{S}}
    {\msemtyeq{s\ t[\vsigma]}{(s[\vsigma])\ (t[\vsigma])}{T}}
\end{mathpar}  
\end{lemma}
\begin{proof}
  Immediate by definitions.
\end{proof}

\begin{lemma}
  \begin{mathpar}
    \inferrule
    {\msemtyp{\vsigma}{\vDelta} \\ \msemtyp[\vDelta; \cdot]{t}{T}}
    {\msemtyeq{\boxit t[\vsigma]}{\boxit{(t[\sextt\vsigma 1])}}{\square T}}
  \end{mathpar}
\end{lemma}
\begin{proof}
  \begin{align*}
    H_1: &\ \msemtyp{\vsigma}{\vDelta}
           \tag{by assumption} \\
    H_2: &\ \msemtyp[\vDelta; \cdot]{t}{T}
           \tag{by assumption} \\
    H_3: &\ \vrho \approx \vrho' \in \intp{\vGamma}
           \tag{by assumption} 
  \end{align*}
  We compute both sides:
  \begin{align*}
    \intp{\boxit t[\vsigma]}(\vrho)
    &= \tbox(\intp{t}(\ext(\intp{\vsigma}(\vrho))))
  \end{align*}
  \begin{align*}
    \intp{\boxit{(t[\sextt\vsigma 1])}}(\vrho')
    &= \tbox(\intp{t}(\intp{\sextt\vsigma 1}(\ext(\vrho')))) \\
    &= \tbox(\intp{t}(\ext(\intp{\vsigma}(\vrho'))))
      \tag{by definition}
  \end{align*}
  The goal follows immediately by the premises and \Cref{lem:intpt-resp-mon}. 
\end{proof}

To prove the congruence for $\tunbox$, we need two helper lemmas:
\begin{lemma}\labeledit{lem:L-intp-vsigma-vrho}
  $\Ltotal{\intp{\vsigma}(\vrho)} n = \Ltotal\vrho{\Ltotal\vsigma n}$
\end{lemma}
\begin{proof}
  We proceed by induction on $n$ and $\vsigma$. We analyze the valid cases by
  inverting $\intp{\vsigma}(\vrho)$ and $\Ltotal\vsigma n$.
  \begin{itemize}[label=Case]
  \item $n = 0$, immediate.
  \item $n = 1 + n'$, then
    \begin{itemize}[label=Subcase]
    \item $\vsigma = \vect I$ or $\vsigma = \wk_x$, immediate.
    \item $\vsigma = \vsigma', t/x$, then
      \begin{align*}
        \Ltotal{\intp{\vsigma}(\vrho)} n
        &= \Ltotal{\intp{\vsigma'}(\vrho)}n \\
        &= \Ltotal\vrho{\Ltotal{\vsigma'} n}
          \byIH \\
        &= \Ltotal\vrho{\Ltotal\vsigma n}
      \end{align*}
      
    \item $\vsigma = \sextt{\vsigma'}m$, then
      \begin{align*}
        \Ltotal{\intp{\vsigma}(\vrho)} n
        &= \Ltotal\vrho m + \Ltotal{\intp{\vsigma'}(\trunc \vrho m)}{n'} \\
        &= \Ltotal\vrho m + \Ltotal{\trunc \vrho m}{\Ltotal{\vsigma'}{n'}}
        \byIH \\
        &= \Ltotal{\vrho}{m + \Ltotal{\vsigma'}{n'}}
        \tag{by \Cref{lem:L-add-envs}} \\
        &= \Ltotal\vrho{\Ltotal\vsigma n}
      \end{align*}
      
    \item $\vsigma = \vsigma' \circ \vdelta$, then
      \begin{align*}
        \Ltotal{\intp{\vsigma}(\vrho)} n
        &= \Ltotal{\intp{\vsigma'}(\intp{\vdelta}(\vrho))}n \\
        &= \Ltotal{\intp{\vdelta}(\vrho)}{\Ltotal{\vsigma'} n}
          \byIH \\
        &= \Ltotal\vrho{\Ltotal{\vdelta}{\Ltotal{\vsigma'} n}}
          \byIH \\
        &= \Ltotal\vrho{\Ltotal{\vsigma' \circ \vdelta} n} \\
        &= \Ltotal\vrho{\Ltotal\vsigma n}
      \end{align*}
    \end{itemize}
  \end{itemize}
\end{proof}

\begin{lemma}\labeledit{lem:trunc-intp-vsigma-vrho}
  $\trunc{\intp{\vsigma}(\vrho)}n = \intp{\trunc\vsigma n}(\trunc\vrho{\Ltotal\vsigma n})$
\end{lemma}
\begin{proof}
  The same as above, we proceed by induction on $n$ and $\vsigma$. We analyze the
  valid cases by inverting $\intp{\vsigma}(\vrho)$ and $\trunc\vsigma n$.
  \begin{itemize}[label=Case]
  \item $n = 0$, immediate.
  \item $n = 1 + n'$, then
    \begin{itemize}[label=Subcase]
    \item $\vsigma = \vect I$ or $\vsigma = \wk_x$, immediate.
    \item $\vsigma = \vsigma', t/x$, then
      \begin{align*}
        \trunc{\intp{\vsigma}(\vrho)}n
        &= \trunc{\intp{\vsigma'}(\vrho)}n \\
        &= \intp{\trunc{\vsigma'}n}(\trunc\vrho{\Ltotal{\vsigma'} n})
          \byIH \\
        &= \intp{\trunc\vsigma n}(\trunc\vrho{\Ltotal\vsigma n})
      \end{align*}
      
    \item $\vsigma = \sextt{\vsigma'}m$, then
      \begin{align*}
        \trunc{\intp{\vsigma}(\vrho)}n
        &= \trunc{\intp{\vsigma'}(\trunc \vrho m)}n \\
        &= \intp{\trunc{\vsigma'}n}(\trunc{\trunc \vrho m}{\Ltotal{\vsigma'} n})
          \byIH \\
        &= \intp{\trunc{\vsigma'}n}(\trunc\vrho{(m + \Ltotal{\vsigma'} n)}) \\
        &= \intp{\trunc\vsigma n}(\trunc\vrho{\Ltotal\vsigma n})
      \end{align*}
    \item $\vsigma = \vsigma' \circ \vdelta$, then
      \begin{align*}
        \trunc{\intp{\vsigma}(\vrho)}n
        &= \trunc{\intp{\vsigma'}(\intp{\vdelta}(\vrho))}n \\
        &= \intp{\trunc{\vsigma'}n}(\trunc{\intp{\vdelta}(\vrho)}{\Ltotal{\vsigma'} n})
          \byIH \\
        &= \intp{\trunc{\vsigma'}n}(\intp{\trunc\vdelta{\Ltotal{\vsigma'} n}}(\trunc\vrho{\Ltotal\vdelta{\Ltotal{\vsigma'} n}}))
          \byIH \\
        &= \intp{\trunc{\vsigma'}n}(\intp{\trunc\vdelta{\Ltotal{\vsigma'} n}}(\trunc\vrho{\Ltotal{\vsigma' \circ \vdelta}n})) \\
        &= \intp{\trunc\vsigma n}(\trunc\vrho{\Ltotal\vsigma n})
      \end{align*}

    \end{itemize}
  \end{itemize}
\end{proof}

\begin{lemma}
  \begin{mathpar}
    \inferrule
    {\msemtyp[\vDelta]{t}{\square T} \\ \msemtyp[\vGamma; \vGamma']{\vsigma}{\vDelta;
        \vDelta'} \\
      |\vDelta'| = n \\ |\vGamma'| = (\Ltotal\vsigma n)}
    {\msemtyeq[\vGamma; \vGamma']{\unbox n t[\vsigma]}{\unbox {\Ltotal\vsigma n}{(t[\trunc\vsigma n])}}{T}}
  \end{mathpar}
\end{lemma}
\begin{proof}
  \begin{align*}
    H_1: &\ \msemtyp[\vDelta]{t}{\square T}
           \tag{by assumption} \\
    H_2: &\ \msemtyp[\vGamma; \vGamma']{\vsigma}{\vDelta; \vDelta'}
           \tag{by assumption} \\
    H_3: &\ \vrho \approx \vrho' \in \intp{\vGamma; \vGamma'}
           \tag{by assumption} 
  \end{align*}
  We compute both sides:
  \begin{align*}
    \intp{\unbox n t[\vsigma]}(\vrho)
    &= \tunbox \cdot (\Ltotal{\intp{\vsigma}(\vrho)} n, \intp{t}(\trunc{\intp{\vsigma}(\vrho)}n))
    \\
    &= \tunbox \cdot (\Ltotal\vrho{\Ltotal\vsigma n},
      \intp{t}(\intp{\trunc\vsigma n}(\trunc\vrho{\Ltotal\vsigma n})))
      \tag{by \Cref{lem:L-intp-vsigma-vrho,lem:trunc-intp-vsigma-vrho}}
  \end{align*}
  \begin{align*}
    \intp{\unbox {\Ltotal\vsigma n}{(t[\trunc\vsigma n])}}(\vrho')
    &= \tunbox \cdot (\Ltotal{\vrho'}{\Ltotal\vsigma n}, \intp{t}(\intp{\trunc\vsigma n}(\trunc{\vrho'}{\Ltotal\vsigma n})))
  \end{align*}
  They agree up to $H_1$. 
\end{proof}

\subsubsection{Substitution Equivalence Rules}

The following rules are immediate following the definitions:
\begin{lemma}
  \begin{mathpar}
    \inferrule
    {\msemtyp{\vsigma}{\vDelta}}
    {\msemtyeq{\vsigma \circ \vect I}{\vsigma}{\vDelta}}

    \inferrule
    {\msemtyp{\vsigma}{\vDelta}}
    {\msemtyeq{\vect I \circ \vsigma}{\vsigma}{\vDelta}}

    \inferrule
    {\msemtyp[\vGamma'']{\vsigma''}{\vGamma'''} \\ \msemtyp[\vGamma']{\vsigma'}{\vGamma''} \\ \msemtyp{\vsigma}{\vGamma'}}
    {\msemtyeq{(\vsigma'' \circ \vsigma') \circ \vsigma}{\vsigma'' \circ (\vsigma' \circ \vsigma)}{\vGamma'''}}

    \inferrule
    {\msemtyp[\vGamma']{\vsigma}{\vGamma''; \Gamma} \\ \msemtyp[\vGamma']{t}{T} \\ \msemtyp{\vdelta}{\vGamma'}}
    {\msemtyeq{\vsigma, t/x \circ \vdelta}{(\vsigma \circ \vdelta), t[\vdelta]/x}{\vGamma''; (\Gamma, x : T)}}

    \inferrule
    {\msemtyp[\vGamma']{\vsigma}{\vGamma; \Gamma} \\ \msemtyp[\vGamma']{t}{T}}
    {\msemtyeq[\vGamma']{\wk_x \circ (\vsigma, t/x)}{\vsigma}{\vGamma; \Gamma}}

    \inferrule
    {\msemtyp[\vGamma']{\vsigma}{\vGamma; (\Gamma, x : T)}}
    {\msemtyeq[\vGamma']{\vsigma}{(\wk_x \circ \vsigma), x[\vsigma]/x}{\vGamma; (\Gamma, x : T)}}
  \end{mathpar}  
\end{lemma}

We outline the proof for the following rules:
\begin{lemma}
  \begin{mathpar}
    \inferrule
    {\msemtyp{\vsigma}{\vGamma'} \\ \msemtyp[\vGamma'']{\vdelta}{\vGamma;\vDelta} \\
      |\vDelta| = n}
    {\msemtyeq[\vGamma'']{(\sextt\vsigma n) \circ \vdelta}{\sextt{(\vsigma \circ (\trunc\vdelta n))}{\Ltotal \vdelta n} }{\vGamma';\cdot}}
  \end{mathpar}
\end{lemma}
\begin{proof}
  \begin{align*}
    H_1: &\ \msemtyp{\vsigma}{\vGamma'}
           \tag{by assumption} \\
    H_2: &\ \msemtyp[\vGamma'']{\vdelta}{\vGamma;\vDelta}
           \tag{by assumption} \\
    H_3: &\ \vrho \approx \vrho' \in \intp{\vGamma''}
           \tag{by assumption}
  \end{align*}
  We compute both sides:
  \begin{align*}
    \intp{(\sextt\vsigma n) \circ \vdelta}(\vrho)
    &= \intp{\vsigma}(\ext(\trunc{\intp{\vdelta}(\vrho)}n), \Ltotal{\intp{\vdelta}(\vrho)}n) \\
    &= \intp{\vsigma}(\ext(\intp{\trunc\vdelta n}(\trunc\vrho{\Ltotal \vdelta n})), \Ltotal\vrho {\Ltotal \vdelta n})
      \tag{by \Cref{lem:L-intp-vsigma-vrho,lem:trunc-intp-vsigma-vrho}}
  \end{align*}
  \begin{align*}
    \intp{\sextt{(\vsigma \circ (\trunc\vdelta n))}{\Ltotal \vdelta n}}(\vrho')
    &= \intp{\vsigma}(\ext(\intp{\trunc\vdelta n}(\trunc{\vrho'}{\Ltotal \vdelta n})), \Ltotal{\vrho'}{\Ltotal \vdelta n})
  \end{align*}
  The goal follows by the premises.
\end{proof}

\begin{lemma}
  \begin{mathpar}
    \inferrule
    {\msemtyp{\vsigma}{\vDelta; \cdot} \\ |\vDelta| > 0}
    {\msemtyeq{\vsigma}{\sextt{\trunc \vsigma 1}{\Ltotal\vsigma 1}}{\vDelta; \cdot}}
  \end{mathpar}
\end{lemma}
\begin{proof}
  \begin{align*}
    H_1: &\ \msemtyp{\vsigma}{\vDelta; \cdot}
         \tag{assumption} \\
    H_2: &\ \vrho \approx \vrho' \in \intp{\vGamma}
           \tag{by assumption}
  \end{align*}
  We evaluate the right hand side:
  \begin{align*}
    \intp{\sextt{\trunc \vsigma 1}{\Ltotal\vsigma 1}}(\vrho')
    &= \ext(\intp{\trunc \vsigma 1}(\trunc{\vrho'}{\Ltotal\vsigma 1}), \Ltotal\vrho{\Ltotal\vsigma 1}) \\
    &= \ext(\trunc{\intp{\vsigma}(\vrho')}1, \Ltotal{\intp{\vsigma}(\vrho')} 1)
      \tag{by \Cref{lem:L-intp-vsigma-vrho,lem:trunc-intp-vsigma-vrho}}
  \end{align*}
  We know
  \begin{align*}
    &\trunc{\intp{\vsigma}(\vrho)}1 \approx \trunc{\intp{\vsigma}(\vrho')}1
    \in \intp{\vDelta} \\
    &\Ltotal{\intp{\vsigma}(\vrho)}1 = \Ltotal{\intp{\vsigma}(\vrho')}1
  \end{align*}
  Since the topmost context of $\vDelta; \cdot$ is $\cdot$, we have finished the check
  and the goal follows immediately.
\end{proof}

\subsubsection{Fundamental Theorem}

We have proved all semantic rules. We can thus conclude the fundamental theorem.
\begin{theorem}
  (fundamental)
  \begin{itemize}
  \item If $\mtyping t T$, then $\msemtyp t T$.
  \item If $\mtyping{\vsigma}{\vDelta}$, then $\msemtyp{\vsigma}{\vDelta}$.
  \item If $\mtyequiv{t}{t'}T$, then $\msemtyeq{t}{t'}T$.
  \item If $\mtyequiv{\vsigma}{\vsigma'}{\vDelta}$, then $\msemtyeq{\vsigma}{\vsigma'}{\vDelta}$.
  \end{itemize}
\end{theorem}

We also can show the initial environment is reflexive in $\intp{\vGamma}$.
\begin{lemma}
  $\uparrow^{\vGamma} \approx \uparrow^{\vGamma} \in \intp{\vGamma}$
\end{lemma}
\begin{proof}
  Immediate by \Cref{lem:varbot} and realizability. 
\end{proof}

\begin{theorem} (completeness)
  If $\mtyequiv{t}{t'}{T}$, then $\nbe^T_{\vGamma}(t) = \nbe^T_{\vGamma}(t')$. 
\end{theorem}
\begin{proof}
  Given $\mtyequiv{t}{t'}T$, we know $\msemtyeq{t}{t'}T$. Combining realizability, we
  can conclude the goal.
\end{proof}

\subsection{Restricted Weakening}\labeledit{sec:st:rweaken}

In \Cref{sec:usubst}, we motivated unified substitutions by first considering the
compositional closure of modal transformations, which turns out has a critical position
in the soundness proof. Luckily, To establish the soundness proof,
we only need to capture the compositional closure of a
special class of modal transformation and ensure that they form a category. The
morphisms in the resulting category are called \emph{restricted weakening}, as they
are semantically a subcategory of unified weakening defined in \Cref{sec:uweaken}.

\begin{definition}
  A substitution $\vsigma$ is a restricted weakening, if it satisfies the following 
  inductively generated predicates:
  \begin{enumerate}
  \item if $\mtyequiv[\vGamma]{\vsigma}{\vect I}{\vGamma}$;
  \item if $\mtyping[\vDelta]{\vdelta}{\vGamma; (\Gamma, x : T)}$ is restricted and
    $\mtyequiv[\vDelta]{\vsigma}{\wk_x \circ \vdelta}{\vGamma; \Gamma}$
    for some $m = |\Delta|$;
  \item if $\mtyping[\vDelta]{\vdelta}{\vGamma}$ is restricted and
    $\mtyequiv[\vDelta; \vDelta']{\vsigma}{\sextt\vdelta n}{\vGamma; \cdot}$ where $n = |\vDelta'|$.
  \end{enumerate}
\end{definition}
This definition essentially characterizes a restricted $\vsigma \approx \wk_{x_1}
\circ \cdots \circ \wk_{x_m} \circ
\sextt{\sextt{\vect I}{k_1}; \cdots}{k_n}$ for some $m$, $n$ and $k_i$. We write $\vsigma : \vGamma \To_r
\vDelta$ to denote $\vsigma$ being restricted. 

We consider some basic properties of restricted weakenings It is easy to see that
$\vect I$ is restricted, as well as $\sextt{\vect I}n$, the modal transformation in the $\beta$
rule for $\square$. If $\vsigma$ is restricted, then $\trunc\vsigma n$ is
also restricted. Moreover, this definition is closed under composition:
\begin{lemma}
  If $\vsigma' : \vGamma' \To_r \vGamma''$ and $\vsigma : \vGamma \To_r \vGamma'$,
  then
  $\vsigma' \circ \vsigma : \vGamma \To_r \vGamma''$.
\end{lemma}
\begin{proof}
  Do induction on the restricted predicate on $\vsigma' : \vGamma' \To_r \vGamma''$ and
  then apply the appropriate equivalence rules. 
\end{proof}
This concludes that restricted weakenings do form a (sub)category. The restrictedness
is important in the soundness proof to simplify the proof. 

\subsection{Gluing Model}

In order to conclude the soundness of the NbE algorithm, we construct a gluing model,
which intuitively relates a syntactic term with its corresponding semantic value. From
this relation, the soundness theorem can be extracted. Moreover, we need the gluing
relation between related syntactic terms and domain values to be invariant under
restricted weakening as in the PER model. We first need to extract an untyped modal
transformation from a unified substitution:
\begin{align*}
  \mt{\_} &: \vGamma \To \vDelta \to \N \to \N \\
  \mt{\vect I} &:= \vone \\
  \mt{\vsigma, t/x} &:= \mt{\vsigma} \\
  \mt{\wk_x} &:= \vone \\
  \mt{\sextt\vsigma n} &:= \sextt{\mt{\vsigma}} n \\
  \mt{\vsigma \circ \vdelta} &:= \mt{\vsigma} \circ \mt{\vdelta}
\end{align*}
This operation simply pulls out the offsets from the input
unified substitution. 
Moreover, for convenience, instead of writing $a[\mt{\vsigma}]$, we write $a[\vsigma]$ as
we know $\vsigma$ must be extracted first to be applied to a domain value. The
notation is justified by the following two lemmas. 

We prove the following properties:
\begin{lemma}\labeledit{lem:L-extract}
  $\Ltotal{\mt{\vsigma}} n = \Ltotal\vsigma n$
\end{lemma}
\begin{proof}
  We do induction on $\vsigma$ and $n$. 
  \begin{itemize}[label=Case]
  \item $n = 0$, immediate.
  \item $n = 1 + n'$, then
    \begin{itemize}[label=Subcase]
    \item $\vsigma = \vect I$, $\vsigma = \vsigma', t/x$ or $\vsigma = \wk_x$ are immediate.
    \item $\vsigma = \sextt{\vsigma'}m$, then
      \begin{align*}
        \Ltotal{\mt{\vsigma}} n
        &= m + \Ltotal{\mt{\vsigma'}}{n'} \\
        &= m + \Ltotal{\vsigma'}{n'}
          \byIH \\
        &= \Ltotal\vsigma n
      \end{align*}
    \item $\vsigma = \vsigma' \circ \vdelta$, then
      \begin{align*}
        \Ltotal{\mt{\vsigma}} n
        &= \Ltotal{\mt{\vdelta}}{\Ltotal{\mt{\vsigma'}} n} \\
        &= \Ltotal{\mt{\vdelta}}{\Ltotal{\vsigma'} n}
          \byIH \\
        &= \Ltotal\vdelta{\Ltotal{\vsigma'} n}
          \byIH \\
        &= \Ltotal\vsigma n
      \end{align*}
    \end{itemize}
  \end{itemize}  
\end{proof}
\begin{lemma}\labeledit{lem:trunc-extract}
  $\mt{\trunc\vsigma n} = \trunc{\mt{\vsigma}}n$
\end{lemma}
\begin{proof}
  We do induction on $\vsigma$ and $n$. 
  \begin{itemize}[label=Case]
  \item $n = 0$, immediate.
  \item $n = 1 + n'$, then
    \begin{itemize}[label=Subcase]
    \item $\vsigma = \vect I$, $\vsigma = \vsigma', t/x$ or
      $\vsigma = \wk_x$ are immediate.
    \item $\vsigma = \sextt{\vsigma'}m$, then
      \begin{align*}
        \mt{\trunc\vsigma n}
        &= \mt{\trunc{\vsigma'}{n'}} \\
        &= \trunc{\mt{\vsigma'}}{n'}
          \byIH \\
        &= \trunc{\sextt{\mt{\vsigma'}} m}{(1 + n')} \\
        &= \trunc{\mt{\vsigma}}n
      \end{align*}
    \item $\vsigma = \vsigma' \circ \vdelta$, then
      \begin{align*}
        \mt{\trunc\vsigma n}
        &= \mt{\trunc{\vsigma'}n} \circ \mt{\trunc\vdelta{\Ltotal{\vsigma'} n}} \\
        &= \trunc{\mt{\vsigma'}}n \circ \trunc{\mt{\vdelta}}{\Ltotal{\vsigma'} n}
          \byIH \\
        &= \trunc{\mt{\vsigma'}}n \circ \trunc{\mt{\vdelta}}{\Ltotal{\mt{\vsigma'}} n}
        \tag{by \Cref{lem:L-extract}} \\
        &= \trunc{\mt{\vsigma}}n
      \end{align*}      
    \end{itemize}
  \end{itemize}
\end{proof}

The gluing model is constructed to realize a typed candidate space:
\begin{align*}
  \underline{T}_{\vGamma} &:= \{ (t, \uparrow^T(c)) \sep \mtyping t T \tand
                            \forall \vsigma : \vDelta \To_r \vGamma.\
                           \mtyequiv[\vDelta]{t[\vsigma]}{\Rne_{\alpha}(c[\vsigma])}T \} \\
  \overline{T}_{\vGamma} &:= \{ (t, a) \sep \mtyping t T \tand \forall
                           \vsigma : \vDelta \To_r \vGamma.\
                           \mtyequiv[\vDelta]{t[\vsigma]}{\Rnf_{\alpha}(\downarrow^T
                           (a[\vsigma]))}T \}
\end{align*}
where $\alpha$ is uniquely determined by the context stack it is working in. We write
$t \sim a \in P$ for $(t, a) \in P$ when $P$ is a gluing relation.  We can thus define
the gluing model:
\begin{align*}
  \glu{B}_{\vGamma} &:= \underline{B}_{\vGamma} \\
  \glu{\square T}_{\vGamma} &:= \{ (t, a) \sep \mtyping{t}{\square T} \tand
                              \forall \vDelta', \vsigma : \vDelta \To_r
                              \vGamma. \unbox{|\vDelta'|}{(t[\vsigma])} \sim \tunbox
                              \cdot (|\vDelta'|, a[\vsigma]) \in \glu{T}_{\vDelta;
                              \vDelta'} \}
  \\
  \glu{S \func T}_{\vGamma} &:= \{ (t, a) \sep \mtyping{t}{S \func T} \tand
                              \forall \vsigma : \vDelta \To_r \vGamma, s \sim b
                              \in \glu{S}_{\vDelta}. t[\vsigma]\ s \sim a[\vsigma]
                              \cdot b \in \glu{T}_{\vDelta} \}
\end{align*}
Notice that the $S \func T$ does not need a naturality equality and the
soundness proof does not rely on it. 

\begin{lemma}
  If $t \sim a \in \glu{T}_{\vGamma}$, then $\mtyping t T$. 
\end{lemma}
\begin{proof}
  Immediate by induction on $T$. 
\end{proof}

\begin{lemma}\labeledit{lem:glu-varbot}
  $x \sim \uparrow^T(\alpha^{-1}(x)) \in
  \underline{T}_{\vGamma; \Gamma, x : T}$
\end{lemma}
\begin{proof}
  Assuming $\vsigma : \vDelta \To_r \vGamma; \Gamma, x : T$, we do induction on
  $\vsigma : \vDelta \To_r \vGamma; \Gamma, x : T$. There can only be two cases:
  $\vsigma \approx \vect I$ or $\vsigma \approx \wk_x \circ \vdelta$ for some restricted
  $\vdelta$. Thus $\vsigma$ can only be local weakening. From this observation, we
  conclude $\mtyequiv[\vDelta]{x[\vsigma]}{\alpha(\alpha^{-1}(x))}T$ holds.
\end{proof}

Similar to the PER model, we also need a realizability theorem.
\begin{theorem}[Realizability]
   $\underline{T}_{\vGamma} \subseteq \glu{T}_{\vGamma} \subseteq
  \overline{T}_{\vGamma}$
\end{theorem}
\begin{proof}
  We proceed by induction on $T$.
  \begin{itemize}[label=Case]
  \item $T = B$, $\underline{B}_{\vGamma} \subseteq \glu{B}_{\vGamma}$ holds
    immediately and $\glu{B}_{\vGamma} \subseteq \overline{B}_{\vGamma}$ is also easy
    because $\Rnf$ of $B$ is just $\Rne$.
    
  \item $T = \square T'$, then
    \begin{itemize}[label=Subcase]
    \item To show $\underline{\square T'}_{\vGamma} \subseteq \glu{\square
        T'}_{\vGamma}$,
      \begin{align*}
        H_1: &\ \forall \vsigma : \vDelta \To_r \vGamma.\
               \mtyequiv[\vDelta]{t[\vsigma]}{\Rne_{\alpha}(c[\vsigma])}{\square T'}
               \tag{by assumption} \\
             &\ \text{assume }\vDelta', \vsigma : \vDelta \To_r \vGamma, k :=
               |\vDelta'|, \vdelta : \vDelta'' \To_r \vDelta; \vDelta'\\
             &\ \trunc\vdelta k : \trunc{\vDelta''}{\Ltotal\vdelta k} \To_r \vDelta \\
             &\ \vsigma \circ \trunc\vdelta k : \trunc{\vDelta''}{\Ltotal\vdelta k} \To_r \vGamma \\
             &\ \mtyequiv[\trunc{\vDelta''}{\Ltotal\vdelta k}]{t[\vsigma \circ (\trunc\vdelta k)]}{\Rne_{\alpha}(c[\vsigma \circ (\trunc\vdelta k)])}{\square
               T'}
               \tag{by $H_1$ and \Cref{lem:trunc-extract}} \\
             &\ \mtyequiv[\vDelta'']{\unbox{\Ltotal\vdelta k}(t[\vsigma][\trunc\vdelta k])}{\unbox{\Ltotal\vdelta k}\Rne_{\alpha}(c[\vsigma][\trunc\vdelta k])}{
               T'}
               \tag{by congruence} \\
             &\ \mtyequiv[\vDelta'']{\unbox k
               (t[\vsigma])[\vdelta]}{\Rne_{\alpha}(\tunbox(k, c[\vsigma])[\vdelta])}{T'}
               \tag{use \Cref{lem:L-extract}} \\
             &\ \unbox k {(t[\vsigma])} \sim \tunbox(k, c[\vsigma]) \in \underline{T'}_{\vDelta; \vDelta'}
               \tag{by abstraction} \\
             &\ \unbox k {(t[\vsigma])} \sim \tunbox(k, c[\vsigma]) \in \glu{T'}_{\vDelta; \vDelta'}
               \byIH \\
             &\ t \sim \uparrow^{\square T'}(c) \in \glu{\square T'}_{\vGamma}
               \tag{by abstraction}
      \end{align*}

    \item To show $\glu{\square T'}_{\vGamma} \subseteq \overline{\square
        T'}_{\vGamma}$,
      \begin{align*}
        H_2:&\ \forall \vDelta', \vsigma : \vDelta \To_r
              \vGamma. \unbox{|\vDelta'|}{(t[\vsigma])} \sim \tunbox
              \cdot (|\vDelta'|, a[\vsigma]) \in \glu{T'}_{\vDelta;
              \vDelta'}
              \tag{by assumption} \\
            &\ \text{assume }\vsigma : \vDelta \To_r \vGamma \\
            &\ \unbox{1}{(t[\vsigma])} \sim \tunbox
              \cdot (1, a[\vsigma]) \in \glu{T'}_{\vDelta; \cdot}
              \tag{by $H_2$, letting $\vDelta' = \cdot$} \\
        H_3: &\ \unbox{1}{(t[\vsigma])} \sim \tunbox
               \cdot (1, a[\vsigma]) \in \overline{T'}_{\vDelta; \cdot}
               \byIH \\
            &\ \mtyequiv[\vDelta; \cdot]{\unbox 1
              {(t[\vsigma])}}{\Rnf_{\alpha}(\downarrow^{T'}(\tunbox \cdot (1,
              a[\vsigma])))}{T'}
              \tag{by $H_3$} \\
            &\ \mtyequiv[\vDelta]{\boxit{(\unbox 1 (t[\vsigma]))}}{\boxit{(\Rnf_{\alpha}(\downarrow^{T'}(\tunbox \cdot (1,
              a[\vsigma]))))}}{\square T'}
              \tag{congruence} \\
            &\ \mtyequiv[\vDelta]{t[\vsigma]}{\Rnf_{\alpha}(\downarrow^{\square
              T'}(a[\vsigma]))}{\square T'}
              \tag{$\eta$ expansion} \\
            &\ t \sim a \in \overline{\square T'}_{\vGamma}
              \tag{by abstraction}
      \end{align*}
    \end{itemize}

  \item $T = S \func T'$, then
    \begin{itemize}[label=Subcase]
    \item To show $\underline{S \func T'}_{\vGamma} \subseteq \glu{S \func
        T'}_{\vGamma}$,
      \begin{align*}
        H_4: &\ \forall \vsigma : \vDelta \To_r \vGamma.\
               \mtyequiv[\vDelta]{t[\vsigma]}{\Rne_{\alpha}(c[\vsigma])}{S \func
               T'}
               \tag{by assumption} \\
             &\ \text{assume } \vsigma : \vDelta \To_r \vGamma, s \sim b \in
               \glu{S}_{\vDelta}, \vsigma' : \vDelta' \To_r \vDelta \\
             &\ \mtyequiv[\vDelta']{t[\vsigma \circ \vsigma']}{\Rne_{\alpha}(c[\vsigma \circ \vsigma'])}{S
               \func T'}
               \tag{by $H_4$} \\
        H_5: &\ s \sim b \in \overline{S}_{\vDelta}
               \byIH \\
             &\ \mtyequiv[\vDelta']{s[\vsigma']}{\Rnf_{\alpha}(\downarrow^S(b[\vsigma']))}{S}
               \tag{by $H_5$} \\
             &\ \mtyequiv[\vDelta']{t[\vsigma \circ \vsigma']\
               s[\vsigma']}{\Rne_{\alpha}(c[\vsigma \circ \vsigma'])\
               \Rnf_{\alpha}(\downarrow^S(b[\vsigma']))}{T'}
               \tag{by congruence} \\
             &\ \mtyequiv[\vDelta']{(t[\vsigma]\ s)[\vsigma']}{\Rne_{\alpha}((c[\vsigma]\ \downarrow^S(b))[\vsigma'])}{T'}\\
             &\ t[\vsigma]\ s \sim \uparrow^{T'}(c[\vsigma]\ \downarrow^S(b)) \in
               \underline{T'}_{\vDelta}
               \tag{by abstraction} \\
             &\ t[\vsigma]\ s \sim \uparrow^{T'}(c[\vsigma]\ \downarrow^S(b)) \in
               \glu{T'}_{\vDelta}
               \byIH \\
             &\ t \sim \uparrow^{S \func T'}(c) \in \glu{S \func T'}_{\vGamma}
               \tag{by abstraction}
      \end{align*}
      
    \item To show $\glu{S \func T'}_{\vGamma} \subseteq \overline{S \func
        T'}_{\vGamma}$,
      \begin{align*}
        H_6:&\ \forall \vsigma : \vDelta \To_r \vGamma, s \sim b
              \in \glu{S}_{\vDelta}. t[\vsigma]\ s \sim a[\vsigma]
              \cdot b \in \glu{T'}_{\vDelta} \\
            &\ \text{assume }\vsigma : \vDelta; \Delta \To_r \vGamma\\
            &\ \text{let }\vsigma' := \vsigma \circ \wk_x : \vDelta; \Delta,x : S \To_r
              \vGamma
              \tag{note $\overline{\vsigma'} = \overline{\vsigma}$}\\
            &\ x \sim z \in \underline{S}_{\vDelta; \Delta, x : S}
              \tag{for some $z$, by \Cref{lem:glu-varbot}} \\
            &\ x \sim z \in \glu{S}_{\vDelta; \Delta, x : S}
              \byIH \\
            &\ t[\vsigma']\ x \sim a[\vsigma] \cdot z \in
              \glu{T'}_{\vDelta; \Delta, x : S}
              \tag{by $H_6$} \\
        H_7: &\ t[\vsigma']\ x \sim a[\vsigma] \cdot z \in
               \overline{T'}_{\vDelta; \Delta, x : S}
               \byIH \\
            &\ \mtyequiv[\vDelta; \Delta, x : S]{t[\vsigma']\
              x}{\Rnf_{\alpha[z \mapsto x]}(\downarrow^{T'}(a[\vsigma] \cdot z))}{T'}
              \tag{by $H_7$} \\
            &\ \mtyequiv[\vDelta; \Delta]{\lambda x. t[\vsigma \circ \wk_x]\
              x}{\lambda x. \Rnf_{\alpha[z \mapsto x]}(\downarrow^{T'}(a[\vsigma] \cdot z))}{S \func T'}
              \tag{by congruence} \\
            &\ \mtyequiv[\vDelta; \Delta]{t[\vsigma]}{\Rnf_{\alpha}(\downarrow^{S \func T'}(a[\vsigma]))}{S \func T'} \\
            &\ t \sim a \in \overline{S \func T'}_{\vGamma}
              \tag{by abstraction}
      \end{align*}
    \end{itemize}

  \end{itemize}
\end{proof}

We need this lemma stating that the gluing model respects syntactic equivalence:
\begin{lemma}\labeledit{lem:glue-resp-equiv'}
  If $t \sim a \in \glu{T}_{\vGamma}$ and $\mtyequiv{t'}{t}{T}$, then
  $t' \sim a \in \glu{T}_{\vGamma}$. 
\end{lemma}
\begin{proof}
  We do induction on $T$.
  \begin{itemize}[label=Case]
  \item $T = B$, immediate by transitivity. 
  \item $T = \square T'$, then
    \begin{align*}
      H_1: &\ \forall \vDelta', \vsigma : \vDelta \To_r \vGamma.
             \unbox{|\vDelta'|}{(t[\vsigma])} \sim \tunbox
             \cdot (|\vDelta'|, a[\vsigma]) \in \glu{T'}_{\vDelta;
             \vDelta'}
             \tag{by assumption} \\
           &\ \text{assume }\vDelta', \vsigma : \vDelta \To_r \vGamma \\
           &\ \mtyequiv[\vDelta; \vDelta']
             {\unbox{|\vDelta'|}{(t'[\vsigma])}}{\unbox{|\vDelta'|}{(t[\vsigma])}}{T'}
             \tag{by congruence} \\
           &\ \unbox{|\vDelta'|}{(t'[\vsigma])} \sim \tunbox
             \cdot (|\vDelta'|, a[\vsigma]) \in \glu{T'}_{\vDelta; \vDelta'}
             \byIH \\
           &\ t' \sim a \in \intp{\square T'}_{\vGamma}
             \tag{by abstraction}
    \end{align*}
    
  \item $T = S \func T'$, then
    \begin{align*}
      H_2:&\ \forall \vsigma : \vDelta \To_r \vGamma, s \sim b
            \in \glu{S}_{\vDelta}. t[\vsigma]\ s \sim a[\vsigma]
            \cdot b \in \glu{T'}_{\vDelta} \\
          &\ \text{assume }\vsigma : \vDelta \To_r \vGamma, s \sim b
            \in \glu{S}_{\vDelta} \\
          &\ \mtyping[\vDelta] s S \\
          &\ \mtyequiv[\vDelta]{t'[\vsigma]\ s}{t[\vsigma]\ s}{T'}
            \tag{by congruence} \\
          &\ t'[\vsigma]\ s \sim a[\vsigma]
            \cdot b \in \glu{T'}_{\vDelta}
            \byIH \\
          &\ t' \sim a \in \intp{S \func T'}_{\vGamma}
            \tag{by abstraction}
    \end{align*}
  \end{itemize}
\end{proof}

We would also like to reveal the Kripke structure in the gluing model:
\begin{lemma}\labeledit{lem:glue-mon'}
  If $t \sim a \in \intp{T}_{\vDelta}$ and $\vsigma : \vGamma \To_r \vDelta$, then
  $t[\vsigma] \sim a[\vsigma] \in \glu{T}_{\vGamma}$. 
\end{lemma}
\begin{proof}
  We analyze $T$. Notice that this proof does not use IH.
  \begin{itemize}[label=Case]
  \item $T = B$, immediate because $\vsigma$ is closed under composition.
  \item $T = \square T'$, then
    \begin{align*}
      H_1: &\ \forall \vDelta'', \vdelta : \vDelta' \To_r \vDelta.
             \unbox{|\vDelta''|}{(t[\vdelta])} \sim \tunbox
             \cdot (|\vDelta''|, a[\vdelta]) \in \glu{T'}_{\vDelta';
             \vDelta''}
             \tag{by assumption} \\
           &\ \text{assume }\vDelta'', \vdelta : \vDelta' \To_r \vGamma \\
           &\ \unbox{|\vDelta''|}{(t[\vsigma \circ \vdelta])} \sim \tunbox
             \cdot (|\vDelta''|, a[\vsigma \circ \vdelta]) \in \glu{T'}_{\vDelta';
             \vDelta''}
             \tag{by $H_1$} \\
           &\ \unbox{|\vDelta''|}{(t[\vsigma][\vdelta])} \sim \tunbox
             \cdot (|\vDelta''|, a[\vsigma][\vdelta]) \in \glu{T'}_{\vDelta';
             \vDelta''}
             \tag{by \Cref{lem:glue-resp-equiv'}} \\
           &\ t[\vsigma] \sim a[\vsigma] \in \glu{\square T'}_{\vGamma}
             \tag{by abstraction}
    \end{align*}
    
  \item $T = S \func T'$, then
    \begin{align*}
      H_2:&\ \forall \vsigma' : \vDelta' \To_r \vDelta, s \sim b
            \in \glu{S}_{\vDelta'}. t[\vsigma']\ s \sim a[\vsigma']
            \cdot b \in \glu{T'}_{\vDelta'} \\
          &\ \text{assume }\vsigma' : \vDelta' \To_r \vGamma, s \sim b
            \in \glu{S}_{\vDelta'} \\
          &\ t[\vsigma \circ \vsigma']\ s \sim a[\vsigma \circ \vsigma']
            \cdot b \in \glu{T'}_{\vDelta'}
            \tag{by $H_2$} \\
          &\ t[\vsigma][\vsigma']\ s \sim a[\vsigma][\vsigma']
            \cdot b \in \glu{T'}_{\vDelta'}
            \tag{by \Cref{lem:glue-resp-equiv'}}\\
          &\ t[\vsigma] \sim a[\vsigma] \in \intp{S \func T'}_{\vGamma}
            \tag{by abstraction}
    \end{align*}
  \end{itemize}
\end{proof}

We generalize gluing to contexts and context stacks:
\begin{align*}
  \glu{\epsilon; \Gamma}_{\vDelta} &:= \{ (\vsigma, \vrho) \sep
                                     \mtyping[\vDelta]{\vsigma}{\epsilon; \Gamma} \tand
                                     \forall x : T \in
                                     \Gamma. x[\vsigma] \sim \rho(x) \in
                                  \glu{T}_{\vDelta} \text{ where }(k, \rho) :=
                                  \vrho(0) \} \\
  \glu{\vGamma; \Gamma}_{\vDelta} &:= \{ (\vsigma, \vrho) \sep
                                    \mtyping[\vDelta]{\vsigma}{\vGamma; \Gamma} \tand
                                    \Ltotal\vsigma 1 = k \\
  & \qquad \qquad \tand \trunc \vsigma 1 \sim \trunc\vrho 1
                                    \in \glu{\vGamma}_{\trunc\vDelta k} \tand \forall x : T \in
                                     \Gamma. x[\vsigma] \sim \rho(x) \in
                                  \glu{T}_{\vDelta}
                           \text{ where }(k, \rho) :=
                           \vrho(0) \}
\end{align*}

We need the follow lemmas to understand operations on $\vsigma \sim \vrho \in
\glu{\vGamma}_{\vDelta}$.
\begin{lemma}\labeledit{lem:glue-L'}
  If $\vsigma \sim \vrho \in \glu{\vGamma}_{\vDelta}$, then $\Ltotal\vsigma n = \Ltotal\vrho n$.
\end{lemma}
\begin{proof}
  Immediate by induction on $n$. 
\end{proof}

\begin{lemma}\labeledit{lem:glue-trunc'}
  If $\vsigma \sim \vrho \in \glu{\vGamma}_{\vDelta}$, then $\trunc\vsigma n \sim \trunc\vrho n \in
  \glu{\trunc\vGamma n}_{\trunc\vDelta{\Ltotal\vsigma n}}$.
\end{lemma}
\begin{proof}
  Immediate by induction on $n$. 
\end{proof}

 \begin{lemma}\labeledit{lem:glue-stack-resp-equiv}
  If $\vsigma \sim \vrho \in \glu{\vGamma}_{\vDelta}$ and $\mtyequiv[\vDelta']{\vsigma'}{\vsigma}{\vGamma}$,  then
  $\vsigma' \sim \vrho \in \glu{\vGamma}_{\vDelta}$. 
\end{lemma}
\begin{proof}
  Immediate by induction on $\vGamma$ and apply \Cref{lem:glue-resp-equiv'} when
  appropriate.
\end{proof}

\begin{lemma}\labeledit{lem:glue-stack-mon'}
  If $\vsigma \sim \vrho \in \glu{\vGamma}_{\vDelta}$, given $\vdelta : \vDelta' \To_r
  \vDelta$,  then
  $\vsigma \circ \vdelta \sim \vrho[\vdelta] \in \glu{\vGamma}_{\vDelta'}$. 
\end{lemma}
\begin{proof}
  Immediate by induction on $\vGamma$ and apply \Cref{lem:glue-resp-equiv',lem:glue-mon'} when
  appropriate.
\end{proof}

Finally, we can define the semantic judgment:
\begin{definition}
  \begin{align*}
    \mSemtyp t T &:= \forall \vsigma \sim \vrho \in
    \glu{\vGamma}_{\vDelta}. t[\vsigma] \sim \intp{t}(\vrho) \in
                   \glu{T}_{\vDelta} \\
    \mSemtyp{\vdelta}{\vGamma'} &:= \forall \vsigma \sim \vrho \in
    \glu{\vGamma}_{\vDelta}. \vdelta \circ \vsigma \sim \intp{\vdelta}(\vrho) \in \glu{\vGamma'}_{\vDelta}
  \end{align*}
\end{definition}

\subsection{Soundness}

\subsubsection{Semantic Typing Rules}

\begin{lemma}
  \begin{mathpar}
    \inferrule*
    {x : T \in \Gamma}
    {\mSemtyp[\vGamma; \Gamma] x T}
  \end{mathpar}
\end{lemma}
\begin{proof}
  Immediate.
\end{proof}

\begin{lemma}
  \begin{mathpar}
    \inferrule*
    {\mSemtyp[\vGamma; \cdot] t T}
    {\mSemtyp{\boxit t}{\square T}}
  \end{mathpar}
\end{lemma}
\begin{proof}
  \begin{align*}
    H_1: &\ \mSemtyp[\vGamma; \cdot] t T
           \tag{by assumption} \\
    H_2: &\ \vsigma \sim \vrho \in \glu{\vGamma}_{\vDelta} 
           \tag{by assumption} \\
         &\ \text{assume }\vDelta'', \vdelta : \vDelta' \To_r \vDelta, \text{let }k :=
           |\vDelta''| \\
    H_3: &\ \sextt\vsigma 1 \sim \ext(\vrho) \in \glu{\vGamma; \cdot}_{\vDelta; \cdot}
           \tag{by definition} \\
    H_4: &\ t[\sextt\vsigma 1] \sim \intp{t}(\ext(\vrho)) \in \glu{T}_{\vDelta; \cdot}
           \tag{by $H_1$ and $H_3$}
  \end{align*}
  By computation, we have
  \begin{align*}
    \unbox k {(\boxit t[\vsigma][\vdelta])}
    & \approx t[\sextt\vsigma 1][\sextt{\vdelta} 1][\sextt{\vect I} k] \\
    & \approx t[\sextt\vsigma 1][\sextt{\vdelta} k] : T
  \end{align*}
  as well as
  \begin{align*}
    \tunbox \cdot (k, \intp{\boxit t}(\vrho)[\vdelta])
    &= \intp{t}(\ext(\vrho))[\sextt{\mt{\vdelta}} 1][\sextt\vone k] \\
    &= \intp{t}(\ext(\vrho))[\sextt{\mt{\vdelta}} k]
  \end{align*}
  Thus $H_4$ is what we want after applying \Cref{lem:glue-resp-equiv',lem:glue-mon'}.
  \begin{align*}
    &\boxit t[\vsigma] \sim \intp{\boxit t}(\vrho) \in \glu{\square T}_{\vDelta}
      \tag{by abstraction} \\
    &\mSemtyp{\boxit t}{\square T} \tag{by abstraction}
  \end{align*}
\end{proof}

\begin{lemma}
  \begin{mathpar}
    \inferrule*
    {\mSemtyp t {\square T} \\ |\vDelta| = n}
    {\mSemtyp[\vGamma; \vDelta]{\unbox n t}{T}}
  \end{mathpar}
\end{lemma}
\begin{proof}
  Immediate by definition of $\glu{\square T}$ and \Cref{lem:glue-trunc'}. 
\end{proof}

\begin{lemma}
  \begin{mathpar}
    \inferrule*
    {\mSemtyp[\vGamma;\Gamma, x : S]t T}
    {\mSemtyp[\vGamma;\Gamma]{\lambda x. t}{S \func T}}
  \end{mathpar}
\end{lemma}
\begin{proof}
  Assume $\vsigma \sim \vrho \in \glu{\vGamma; \Gamma}_{\vDelta}$, $\vDelta'$,
  $\vdelta : \vDelta' \To_r \vDelta$ and $s \sim b \in \glu{S}_{\vDelta'}$. Then we should
  prove that
  \begin{align*}
    (\lambda x. t)[\vsigma][\vdelta]\ s \sim
    \intp{t}(\ext(\vrho[\vdelta], x, b)) \in \glu{T}_{\vDelta'}
  \end{align*}

  We need to construct $\glu{\vGamma;\Gamma, x : S}_{\vDelta'}$ from what we have. We
  claim
  \begin{align*}
    (\vsigma \circ \vdelta), s/x \sim \ext(\vrho[\vdelta], x, b) \in \glu{\vGamma;\Gamma, x : S}_{\vDelta'}
  \end{align*}
  by \Cref{lem:glue-stack-mon'}. Therefore
  \begin{align*}
    t[(\vsigma \circ \vdelta), s/x] \sim \intp{t}(\ext(\vrho[\vdelta], x, b)) \in \glu{T}_{\vDelta'}
  \end{align*}
  Comparing the goal and what we have, we can conlude by \Cref{lem:glue-resp-equiv'}
  using congruence rules and $\beta$ equivalence for $\lambda$.
\end{proof}

The following rules are proved just by following the definitions:
\begin{lemma}
  \begin{mathpar}
    \inferrule
    {\mSemtyp t {S \func T} \\ \mSemtyp s S}
    {\mSemtyp{t\ s}{T}}

    \inferrule
    {\mSemtyp[\vDelta]t T \\ \mSemtyp{\vsigma}{\vDelta}}
    {\mSemtyp{t[\vsigma]}{T}}    

    \inferrule
    { }
    {\mSemtyp{\vect I}{\vGamma}}

    \inferrule
    {\mSemtyp{\vsigma}{\vGamma'; \Gamma} \\ \mSemtyp t T}
    {\mSemtyp{\vsigma, t/x}{\vGamma';(\Gamma, x : T)}}

    \inferrule
    { }
    {\mSemtyp[\vGamma; (\Gamma, x : T)]{\wk_x}{\vGamma;\Gamma}}

    \inferrule
    {\mSemtyp[\vGamma']{\vsigma}{\vGamma''} \\ \mSemtyp{\vdelta}{\vGamma'}}
    {\mSemtyp{\vsigma \circ \vdelta}{\vGamma''}}
  \end{mathpar}
\end{lemma}
\begin{proof}
  Immediate by just applying the premises. 
\end{proof}

At last, we consider the following rule:
\begin{lemma}
  \begin{mathpar}
    \inferrule
    {\mSemtyp{\vsigma}{\vDelta} \\ |\vGamma'| = n}
    {\mSemtyp[\vGamma; \vGamma']{\sextt\vsigma n}{\vDelta; \cdot}}
  \end{mathpar}
\end{lemma}
\begin{proof}
  \begin{align*}
    H_1: &\ \mSemtyp{\vsigma}{\vDelta}
           \tag{by assumption} \\
    H_2: &\ \vdelta \sim \vrho \in \glu{\vGamma; \vGamma'}_{\vDelta'}
           \tag{by assumption} \\
    H_3: &\ \trunc\vdelta n \sim \trunc\vrho n \in \glu{\vGamma}_{\trunc{\vDelta'}{\Ltotal \vdelta n}}
           \tag{by $H_2$} \\
         &\ \vsigma \circ (\trunc\vdelta n) \sim \intp{\vsigma}(\trunc\vrho n)
           \in \glu{\vDelta}_{\trunc{\vDelta'}{\Ltotal \vdelta n}}
           \tag{by $H_1$} \\
         &\ \sextt{(\vsigma \circ (\trunc\vdelta n))}{\Ltotal \vdelta n} \sim
           \ext(\intp{\vsigma}(\trunc\vrho n), \Ltotal\vrho n)
           \in \glu{\vDelta; \cdot}_{\vDelta'}
           \tag{by definition and \Cref{lem:glue-L'}}
  \end{align*}
  Notice that the left hand side has
  \begin{align*}
    \sextt{(\vsigma \circ (\trunc\vdelta n))}{\Ltotal \vdelta n} \approx
    (\sextt\vsigma n) \circ \vdelta : \vDelta; \cdot
  \end{align*}
  Thus we conclude by \Cref{lem:glue-stack-resp-equiv}
  \begin{align*}
    (\sextt\vsigma n) \circ \vdelta \sim
    \intp{\sextt\vsigma n}(\vrho)
    \in \glu{\vDelta; \cdot}_{\vDelta'}
  \end{align*}
  and hence the goal by abstraction.
\end{proof}

\subsubsection{Fundamental Theorem}

Since we have proven the semantic typing rules, we can conclude the fundamental
theorem:
\begin{theorem} (fundamental)
  \begin{itemize}
  \item If $\mtyping t T$, then $\mSemtyp t T$. 
  \item If $\mtyping \vsigma \vDelta$, then $\mSemtyp \vsigma \vDelta$. 
  \end{itemize}
\end{theorem}

We also need to relate $\vect \id$ and $\uparrow^{\vGamma}$:
\begin{lemma}
  $\vect I \sim \uparrow^{\vGamma} \in \glu{\vGamma}_{\vGamma}$.
\end{lemma}
\begin{proof}
  Immediate by induction on $\vGamma$ and apply \Cref{lem:glu-varbot}. 
\end{proof}

\begin{theorem}
  (soundness) If $\mtyping t T$, then $\mtyequiv{t}{\nbe^T_{\vGamma}(t)}T$. 
\end{theorem}
\begin{proof}
  Applying the fundamental theorem, we have $t[\vect I] \sim
  \intp{t}_{\vGamma}(\uparrow^{\vGamma})\in \glu{T}_{\vGamma}$.  We conclude the goal
  by applying \Cref{lem:glue-resp-equiv'} and realizability. 
\end{proof}

\subsection{Details on De Bruijn Indices}\labeledit{sec:st:debruijn}

In the formulation above, $\alpha$ in the readback functions is somewhat mysterious:
indeed, our naming convention is mysterious after all. On the flip side, if we
concretize the naming convention, then we can say more about $\alpha$. In this section,
we consider names as de Bruijn indices.

With de Bruijn indices, given $\vGamma; (\Gamma, x : T, \Delta)$, $x$ is
represented by the number $|\Delta|$ in the syntax.  As discussed in \Cref{sec:st:rb},
in the domain, $x$ corresponds to $|\Gamma|$ in the domain. Therefore, given any
variable $x$ in the current world, its de Bruijn index $n$ and its domain variable $m$ must sum to
$|\Gamma, x : T, \Delta| - 1$. We can use this relation to keep
track of conversion between two representations. This suggests us to rewrite the
readback functions to
\begin{align*}
  \Rnf &: \vect \N \rightharpoonup D^{\Nf} \rightharpoonup \Nf \\
  \Rnf_{\vect n} (\downarrow^{\square T} (a))
       &:= \boxit \Rnf_{\vect n; 0} (\downarrow^{T} (\tunbox \cdot (1, a)))
  \\
  \Rnf_{\vect n; n} (\downarrow^{S \func T} (a))
       &:= \lambda x. \Rnf_{\vect n; n + 1} (\downarrow^T (a \cdot \uparrow^S(n))) \\
  \Rnf_{\vect n} (\downarrow^B (\uparrow^B (c)))
       &:= \Rne_{\vect n} (c) \\
  \\
  \Rne &: \vect \N \rightharpoonup D^{\Ne} \rightharpoonup \Ne \\
  \Rne_{\vect n; n}(z)
       &:= n - z - 1
  \tag{if $n - z - 1 < 0$, we return $0$ instead}\\
  \Rne_{\vect n}(c\ d)
       &:= \Rne_{\vect n} (c)\ \Rnf_{\vect n}(d) \\
  \Rne_{\vect n}(\tunbox(k, c))
       &:= \unbox{k}{\Rne_{\trunc{\vect n} k}(c)}
\end{align*}
Instead of some partial function $\N \rightharpoonup \Var$, we use more concrete $\vect
\N$ to keep track of the length of each context in the context stack. Since there must
be at least one context in the stack, we can assume all $\vect \N$ involved must be
nonempty. We use this knowledge in the function case, in which we increment the
topmost number in the number stack by one, and the variable case, in which we compute
a de Bruijn index from a domain variable. In the $\square T$ case, since we are
adding an empty context to the stack, the number stack also grows by one element which
is $0$, reflecting the newly added context is empty. Symmetrically, in the $\tunbox$
case, we should truncate the number stack in the recursive call to $\vect n|_k$ to
ensure that variables in the same worlds correspond. 

With de Bruijn indices, we can redefine $\nbe$ as
\begin{definition}
  The NbE algorithm is defined to be
  \begin{align*}
    \nbe_{\vGamma}^T(t) := \Rnf_{\tmap(\Gamma \mapsto |\Gamma|, \vGamma)}(\downarrow^T (\intp{t}(\uparrow^{\vGamma})))
  \end{align*}
\end{definition}
In this definition, we can see that indeed every context stack $\vGamma$ uniquely
determines a number stack as $\tmap(\Gamma \mapsto |\Gamma|, \vGamma)$, which indeed
satisfies our earlier claim.

The completeness and soundness proofs stay essentially the same. Minor adjustments are
needed due to change of readback functions, but they should be abstracted away
by realizability theorems.  Our formalization is based on de Bruijn indices and fully
checked by Agda. It can be found at
\url{https://gitlab.com/JasonHuZS/practice/-/tree/master/Unbox}.

\section{\mintslang: A Dependently Typed Extension}

In this section, we extend the type theories discussed in the previous sections with
dependent types. In the next section, we examine its meta-theory. In the sections
after this, we give its NbE algorithm and the completeness and soundness proofs. We
introduce \mintslang, a \textbf{M}odal \textbf{IN}tuitionistic \textbf{T}ype theory.
It is an extension of Martin-L\"of's intuitionistic type theory
with $\square$. It supports large elimination and a full cumulative
universe hierarchy. 
\begin{alignat*}{2}
  i, j, k, l, m, n & && \tag{given natural numbers, $\mathbb{N}$} \\
  x, y & && \tag{variables, $\Var$} \\
  s, t, u, M, S, T &:=&&\ x \tag{Terms, $\Trm$} \\
  & && \sep \Nat \sep \square T \sep 
  \Pi(x : S). T \tag{types}  \\
  & && \sep \Se_i \tag{universe} \\
  & && \sep \ze \sep \su t \sep \elimn{x.M}{s}{x, y.u}{t} \tag{natural numbers $\Nat$} \\
  & && \sep \boxit t \sep \unbox n t
  \tag{box}\\
  & &&\sep \lambda x. t \sep s\ t \tag{dependent functions}  \\
  & &&\sep t[\vsigma] \tag{substitution}  \\
  \vsigma, \vdelta &:= &&\ \vect I \sep \vsigma, t/x \sep \wk_x
  \sep \sextt \vsigma n \sep \vsigma \circ \vdelta \tag{Unified substitution, \Substs} \\
  \Gamma, \Delta, \Phi &:= &&\ \cdot \sep \Gamma, x : T
  \tag{Contexts, \Ctx}\\
  \vGamma, \vDelta &:= &&\ \epsilon \sep \vGamma; \Gamma \tag{Context
    stacks, $\vect{\Ctx}$} \\
  w, W &:= &&\ v \sep \Nat \sep \Se_i \sep \square W \sep \Pi(x : W). W' 
  \tag{Normal form ($\Nf$)}  \\
  & && \sep \ze \sep \su w \sep \boxit w \sep 
  \lambda x. w
  \\
  v, V &:= &&\ x \sep \elimn{x.w}{w'}{x,y.w''}{v} \sep \unbox n v \sep v\ w
  \tag{Neutral form ($\Ne$)}
\end{alignat*}

Since we are working in a dependently typed system, the distinction between terms and
types are gone. We use $S$ and $T$ to represent terms that are intended to be
types. $s$, $t$ and $u$ are terms of some types. We use $M$ exclusively for motives,
which are needed for the elimination principle of object-level natural numbers. In
general, we always use upper letters to represent those intended to be types and lower
letters to represent those intended to be terms. Otherwise, one finds many
resemblances from the simply typed case. We next move on to defining the rules.

In the judgments below, we use
\begin{align*}
  \vdash \vGamma
\end{align*}
to denote well-formedness of the context stack $\vGamma$. We use
\begin{align*}
  \mtyping t T
\end{align*}
to denote a term $t$ has type $T$ in the context stack $\vGamma$. When $T$ is $\Se_n$
for some $n$, it denotes the formation of $t$ as a type. Moreover, we assume $\vGamma$
to have length at least $1$. When stating the well-formedness of a type, we sometimes
do not care about the universe level it resides in, and thus we use
\begin{align*}
  \mjudge T := \mtyping{T}{\Se_i}
\end{align*}
for some $i$. Moreover, 
\begin{align*}
  \mtyequiv{t}{t'}T
\end{align*}
denotes terms $t$ and $t'$ of type $T$ are equivalent in the context stack
$\vGamma$.  Similarly, we define
\begin{align*}
  \mjudge T \approx T' := \mtyequiv{T}{T'}{\Se_i}
\end{align*}
Through this judgment, we can judge whether two types are equivalent. This
can be generalized to the whole context stack:
\begin{mathpar}
  \vdash \vGamma \approx \vGamma'
\end{mathpar}
The following judgments have similar meanings:
\begin{mathpar}
  \mtyping{\vsigma}{\vDelta}

  \mtyequiv{\vsigma}{\vsigma'}{\vDelta} 
\end{mathpar}
except that they are for substitutions.

We sometimes use iterative syntax for $\Pi$ types. That is,
\begin{align*}
  \Pi(x_1 : S_1, \cdots, x_n : S_n). T := \Pi(x_1: S_1). \cdots \Pi(x_n : S_n). T
\end{align*}

\subsection{Well-formedness of Context Stacks}

We use the judgment
\begin{align*}
  &\vdash \vGamma
\end{align*}
to denote that the context stack $\vGamma$ is well-formed. It is
given by the following judgments:
\begin{mathpar}
  \inferrule
  { }
  {\vdash \epsilon; \cdot}

  \inferrule
  {\vdash \vGamma}
  {\vdash \vGamma; \cdot}

  \inferrule
  {\vdash \vGamma; \Gamma  \\ \mtyping[\vGamma; \Gamma]{T}{\Se_i}}
  {\vdash \vGamma; (\Gamma, x : T)}
\end{mathpar}
Note that the base case means that we must start with a singleton stack containing the
empty context. Then we can choose to move to the next world or add a type to the local
world.  These rules capture the assumption we mentioned before that well-formed
context stacks are always non-empty:
\begin{lemma}
  If $\vdash \vGamma$, then $|\vGamma| > 0$. 
\end{lemma}
This means that we can always assume there is a top-level context whenever necessary. 

We also need the equivalence between contexts because this system has explicit substitutions:
\begin{mathpar}
  \inferrule
  { }
  {\vdash \epsilon; \cdot \approx \epsilon; \cdot}

  \inferrule
  {\vdash \vGamma \approx \vDelta}
  {\vdash \vGamma; \cdot \approx \vDelta; \cdot}

  \inferrule
  {\vdash \vGamma; \Gamma \approx \vDelta; \Delta  \\
    \mtyequiv[\vGamma; \Gamma]{T}{T'}{\Se_i} \\\\
    \boxed{\mtyequiv[\vDelta; \Delta]{T}{T'}{\Se_i}} \\
    \boxed{\mtyping[\vGamma; \Gamma]{T}{\Se_i}} \\
    \boxed{\mtyping[\vDelta; \Delta]{T'}{\Se_i}}}
  {\vdash \vGamma; (\Gamma, x : T) \approx \vDelta; (\Delta, x : T')}
\end{mathpar}
We surround some premises by boxes because they are not technically necessary. We add
them here in order to prove some structural properties and after that we will be able
to remove these boxed premises. Some boxed premises are also needed much later when we
establish the soundness proof of the NbE algorithm. We will use them until the end of discussion about
syntactic properties and then they will be removed. 

\subsection{Conversion}

We admit the conversion rule which effectively take quotient on types over
convertibility:
\begin{mathpar}
  \inferrule
  {\mtyping{t}{T} \\ \mtyequiv{T}{T'}{\Se_i}}
  {\mtyping{t}{T'}}

  \inferrule
  {\mtyequiv{t}{t'}{T} \\ \mtyequiv{T}{T'}{\Se_i}}
  {\mtyequiv{t}{t'}{T'}}
\end{mathpar}

\subsection{Substitution}

Following the simply typed case, the dependently typed system also uses explicit
substitutions. This requires truncation and truncation offset defined as in
\Cref{sec:st:subst-calc}. Next we define the rules:
\begin{mathpar}
  \inferrule
  {\mtyping[\vDelta]t T \\ \mtyping{\vsigma}{\vDelta}}
  {\mtyping{t[\vsigma]}{T[\vsigma]}}

  \inferrule
  {\mtyequiv[\vDelta]{t}{t'}T \\ \mtyequiv{\vsigma}{\vsigma'}{\vDelta}}
  {\mtyequiv{t[\vsigma]}{t'[\vsigma']}{T[\vsigma]}}
\end{mathpar}
The following judgments specify well-formed unified substitutions. Adjustments are
made due to the dependently typed nature:
\begin{mathpar}
  \inferrule
  {\vdash \vGamma}
  {\mtyping{\vect I}{\vGamma}}

  \inferrule
  {\vdash \vGamma; \Gamma, x : T}
  {\mtyping[\vGamma; \Gamma, x : T]{\wk_x}{\vGamma; \Gamma}}

  \inferrule
  {\mtyping{\vsigma}{\vGamma'; \Gamma} \\ \mtyping[\vGamma';
    \Gamma]{T}{\Se_i} \\ \mtyping{t}{T[\vsigma]}}
  {\mtyping{\vsigma, t/x}{\vGamma';(\Gamma, x : T)}}

  \inferrule
  {\mtyping{\vsigma}{\vDelta} \\ \vdash \vGamma; \vGamma' \\ |\vGamma'| = n}
  {\mtyping[\vGamma; \vGamma']{\sextt\vsigma n}{\vDelta; \cdot}}

  \inferrule
  {\mtyping[\vGamma']{\vsigma}{\vGamma''} \\ \mtyping{\vdelta}{\vGamma'}}
  {\mtyping{\vsigma \circ \vdelta}{\vGamma''}}
\end{mathpar}
The following are the congruence rules:
\begin{mathpar}
  \inferrule
  {\vdash \vGamma}
  {\mtyequiv{\vect I}{\vect I}{\vGamma}}

  \inferrule
  {\vdash \vGamma; \Gamma, x : T}
  {\mtyequiv[\vGamma; \Gamma, x : T]{\wk_x}{\wk_x}{\vGamma; \Gamma}}

  \inferrule
  {\mtyequiv{\vsigma}{\vsigma'}{\vGamma'; \Gamma} \\ \mtyping[\vGamma';
    \Gamma]{T}{\Se_i} \\ \mtyequiv{t}{t'}{T[\vsigma]}}
  {\mtyequiv{\vsigma, t/x}{\vsigma, t'/x}{\vGamma';(\Gamma, x : T)}}

  \inferrule
  {\mtyequiv{\vsigma}{\vsigma'}{\vDelta} \\ \vdash \vGamma; \vGamma' \\ |\vGamma'| = n}
  {\mtyequiv[\vGamma; \vGamma']{\sextt \vsigma n}{\sextt{\vsigma'} n}{\vDelta; \cdot}}

  \inferrule
  {\mtyequiv[\vGamma']{\vsigma}{\vsigma'}{\vGamma''} \\ \mtyequiv{\vdelta}{\vdelta'}{\vGamma'}}
  {\mtyequiv{\vsigma \circ \vdelta}{\vsigma' \circ \vdelta'}{\vGamma''}}
\end{mathpar}
We also need a conversion rule for substitution:
\begin{mathpar}
  \inferrule
  {\mtyping{\vsigma}{\vDelta} \\ \vdash \vDelta \approx \vDelta'}
  {\mtyping{\vsigma}{\vDelta'}}

  \inferrule
  {\mtyequiv{\vsigma}{\vsigma'}{\vDelta} \\ \vdash \vDelta \approx \vDelta'}
  {\mtyequiv{\vsigma}{\vsigma'}{\vDelta'}}
\end{mathpar}

Application of substitution satisfies the following equivalence:
\begin{mathpar}
  \inferrule
  {\mtyping t T}
  {\mtyequiv{t[\vect I]}{t}{T}}

  \inferrule
  {\mtyping[\vGamma']{\vsigma}{\vGamma''} \\ \mtyping[\vGamma]{\vdelta}{\vGamma'} \\
  \mtyping[\vGamma'']{t}{T}}
  {\mtyequiv{t[\vsigma \circ \vdelta]}{t[\vsigma][\vdelta]}{T[\vsigma \circ \vdelta]}}
\end{mathpar}

Substitution satisfies the PER rules (omitted) and the following algebraic rules:
\begin{mathpar}
  \inferrule
  {\mtyping{\vsigma}{\vDelta}}
  {\mtyequiv{\vsigma \circ \vect I}{\vsigma}{\vDelta}}

  \inferrule
  {\mtyping{\vsigma}{\vDelta}}
  {\mtyequiv{\vect I \circ \vsigma}{\vsigma}{\vDelta}}

  \inferrule
  {\mtyping[\vGamma'']{\vsigma''}{\vGamma'''} \\ \mtyping[\vGamma']{\vsigma'}{\vGamma''} \\ \mtyping{\vsigma}{\vGamma'}}
  {\mtyequiv{(\vsigma'' \circ \vsigma') \circ \vsigma}{\vsigma'' \circ (\vsigma' \circ \vsigma)}{\vGamma'''}}
\end{mathpar}
We next consider composition:
\begin{mathpar}
  \inferrule
  {\mtyping[\vGamma']{\vsigma}{\vGamma''; \Gamma} \\ \mtyping[\vGamma';
    \Gamma]{T}{\Se_i} \\
    \mtyping[\vGamma']{t}{T[\vsigma]} \\ \mtyping{\vdelta}{\vGamma'}}
  {\mtyequiv{\vsigma, t/x \circ \vdelta}{(\vsigma \circ \vdelta), t[\vdelta]/x}{\vGamma''; (\Gamma, x : T)}}

  \inferrule
  {\mtyping{\vsigma}{\vDelta} \\ \mtyping[\vDelta]{T}{\Se_i} \\ \mtyping{t}{T[\vsigma]}}
  {\mtyequiv{\wk_x \circ (\vsigma, t/x)}{\vsigma}{\vDelta}}
  
  \inferrule
  {\mtyping{\vsigma}{\vGamma'} \\ \mtyping[\vGamma'']{\vdelta}{\vGamma;\vDelta} \\
    |\vDelta| = n \\
  \boxed{\vdash \vGamma;\vDelta}}
  {\mtyequiv[\vGamma'']{(\sextt\vsigma n) \circ \vdelta}{\sextt{(\vsigma \circ \trunc\vdelta n)}
      {\Ltotal\vdelta n} }{\vGamma';\cdot}}
\end{mathpar}

At last, we consider the extensionality principles which stipulates how syntax should
interact:
\begin{mathpar}
  \inferrule
  {\mtyping[\vGamma']{\vsigma}{\vGamma; (\Gamma, x : T)}}
  {\mtyequiv[\vGamma']{\vsigma}{(\wk_x \circ \vsigma), x[\vsigma]/x}{\vGamma; (\Gamma, x : T)}}

  \inferrule
  {\mtyping{\vsigma}{\vDelta; \cdot} \\ |\vDelta| > 0}
  {\mtyequiv{\vsigma}{\sextt{\trunc\vsigma 1}{\Ltotal\vsigma 1}}{\vDelta; \cdot}}
\end{mathpar}

\subsection{Variables}

Variables are well-typed if they are bound:
\begin{mathpar}
  \inferrule
  {\vdash \vGamma; \Gamma \\ x : T \in \vGamma; \Gamma}
  {\mtyping[\vGamma; \Gamma]x T}

  \inferrule
  {\vdash \vGamma; \Gamma \\ x : T \in \vGamma; \Gamma}
  {\mtyequiv[\vGamma; \Gamma]x x T}

  \inferrule
  {\vdash \vGamma; (\Gamma, x : T) \\ y : T' \in \vGamma; \Gamma}
  {\mtyequiv[\vGamma; \Gamma]{y[\wk_x]}y{T'[\wk_x]}}
\end{mathpar}
The binding judgment $x : T \in \vGamma; \Gamma$ must be adjusted due to dependent types:
\begin{mathpar}
  \inferrule
  { }
  {x : T[\wk_x] \in \vGamma; \Gamma, x : T}

  \inferrule
  {x : T \in \vGamma; \Gamma}
  {x : T[\wk_y] \in \vGamma; \Gamma, y : T'}
\end{mathpar}
Notice that this judgment never looks up beyond the topmost context, so only variables
in the current world can be directly accessed.

The following specify the application of substitution to variables:
\begin{mathpar}
  \inferrule
  {\mtyping{\vsigma}{\vGamma'; \Gamma} \\ \mtyping[\vGamma';\Gamma]{T}{\Se_i}  \\ \mtyping t T[\vsigma]}
  {\mtyequiv{x[\vsigma, t/x]}{t}{T[\vsigma]}}
 
  \inferrule
  {\mtyping{\vsigma}{\vGamma'; \Gamma} \\ \mtyping[\vGamma';\Gamma]{T'}{\Se_i} \\\\ \mtyping t T'[\vsigma] \\ y : T \in \vGamma';\Gamma}
  {\mtyequiv{y[\vsigma, t/x]}{y[\vsigma]}{T[\vsigma]}}
\end{mathpar}

\subsection{$\Pi$ Types}

$\Pi$ types generalize function spaces by allowing return types to depend on the input
arguments.  For each type to be defined from now on, we will formulate the formation,
construction and elimination rules, as well as the equivalence rules.
\begin{mathpar}
  \inferrule
  {\mtyping[\vGamma; \Gamma]{S}{\Se_i} \\ \mtyping[\vGamma; (\Gamma, x : S)]{T}{\Se_i}}
  {\mtyping[\vGamma; \Gamma]{\Pi(x : S). T}{\Se_{i}}}
  
  \inferrule
  {\boxed{\mtyping[\vGamma; \Gamma]{S}{\Se_i}} \\ \mtyping[\vGamma; (\Gamma, x : S)]{t}{T}}
  {\mtyping[\vGamma; \Gamma]{\lambda x. t}{\Pi(x : S). T}}

  \inferrule
  {\boxed{\mtyping[\vGamma; \Gamma]{S}{\Se_i}} \\
    \boxed{\mtyping[\vGamma; (\Gamma, x : S)]{T}{\Se_i}} \\
    \mtyping{t}{\Pi(x : S). T} \\ \mtyping{s}{S}}
  {\mtyping{t\ s}{T[\vect I, s/x]}}
\end{mathpar}
Since we will be working on cumulative universes, it is sufficient to stipulate that $S$
and $T$ have the same universe level in the formation rule. 

Now we define the equivalence rules. Among the congruence rules, the following one is
interesting as it requires a redundant premise:
\begin{mathpar}
  \inferrule
  {\boxed{\mtyping[\vGamma'; \Gamma]{S}{\Se_i}} \\
    \mtyequiv[\vGamma'; \Gamma]{S}{S'}{\Se_i} \\ \mtyequiv[\vGamma'; (\Gamma, x : S)]{T}{T'}{\Se_i}}
  {\mtyequiv[\vGamma'; \Gamma]{\Pi(x : S). T}{\Pi(x : S'). T'}{\Se_i}}
\end{mathpar}
We omit the rest of congruence rules. Moreover, we
have $\beta$ and $\eta$ equivalences:
\begin{mathpar}
  \inferrule
  {\boxed{\mtyping[\vGamma; \Gamma]{S}{\Se_i}} \\
    \boxed{\mtyping[\vGamma; (\Gamma, x : S)]{T}{\Se_i}} \\\\
    \mtyping[\vGamma; (\Gamma, x : S)]{t}{T} \\ \mtyping[\vGamma;\Gamma]{s}{S}}
  {\mtyequiv[\vGamma; \Gamma]{(\lambda x. t)\ s}{t[\vect I, x : s]}{T[\vect I, x : s]}}

  \inferrule
  {\boxed{\mtyping[\vGamma; \Gamma]{S}{\Se_i}} \\
    \boxed{\mtyping[\vGamma; (\Gamma, x : S)]{T}{\Se_i}} \\\\
    \mtyping{t}{\Pi(x : S). T}}
  {\mtyequiv{t}{\lambda x. t[\wk_x]\ x}{\Pi(x : S). T}}
\end{mathpar}
Given $\mtyping[\vGamma; \Gamma]{\vsigma}{\vDelta;\Delta}$ and $\mtyping[\vDelta;
\Delta]T{\Se_i}$, we define
\begin{align*}
  q(\vsigma, x) &: \vGamma; (\Gamma, x : T[\vsigma]) \To \vDelta;(\Delta, x : T) \\
  q(\vsigma, x) &:= (\vsigma \circ \wk_x), x / x
\end{align*}
So we can stipulate that substitutions apply as follows:
\begin{mathpar}
  \inferrule
  {\mtyping[\vGamma; \Gamma]{\vsigma}{\vDelta;\Delta} \\\\
    \mtyping[\vDelta;\Delta]{S}{\Se_i} \\ \mtyping[\vDelta;(\Delta, x : S)]{T}{\Se_i}}
  {\mtyequiv[\vGamma; \Gamma]{\Pi(x : S). T[\vsigma]}{\Pi(x :
      S[\vsigma]).(T[q(\vsigma, x)])}{\Se_{i}}}

  \inferrule
  {\mtyping[\vGamma; \Gamma]{\vsigma}{\vDelta;\Delta} \\
    \mtyping[\vDelta; \Delta, x : S]{t}{T}}
  {\mtyequiv[\vGamma; \Gamma]{\lambda x. t[\vsigma]}{\lambda x. (t[q(\vsigma, x)])}{\Pi(x : S). T[\vsigma]}}

  \inferrule
  {\boxed{\mtyping[\vDelta; \Delta]{S}{\Se_i}} \\
    \boxed{\mtyping[\vDelta; (\Delta, x : S)]{T}{\Se_i}} \\
    \mtyping{\vsigma}{\vDelta;\Delta} \\
    \mtyping[\vDelta;\Delta]{s}{\Pi(x : S). T} \\
    \mtyping[\vDelta;\Delta]{t}{S}}
  {\mtyequiv{s\ t[\vsigma]}{(s[\vsigma])\ (t[\vsigma])}{T[\vsigma, x : t[\vsigma]]}}
\end{mathpar}

\subsection{Natural Numbers}

We use natural numbers as a representative for algebraic data types. Generalization to
actual inductive data types should follow very closely to MLTT. 
\begin{mathpar}
  \inferrule
  {\vdash \vGamma}
  {\mtyping \Nat \Se_i}

  \inferrule
  {\vdash \vGamma}
  {\mtyping \ze \Nat}

  \inferrule
  {\mtyping t \Nat}
  {\mtyping{\su t}\Nat}

  \inferrule
  {\mtyping[\vGamma; (\Gamma, x : \Nat)]{M}{\Se_i} \\ \mtyping[\vGamma; \Gamma]{s}{M[\vect I, \ze/x]} \\
    \mtyping[\vGamma; (\Gamma, x : \Nat, y : M)]{u}{M[(\wk_x \circ \wk_y), \su x/x]} \\ \mtyping[\vGamma; \Gamma] t \Nat}
  {\mtyping[\vGamma; \Gamma]{\elimn {x.M} s {x,y.u} t}{M[\vect I, t/x]}}
\end{mathpar}

One might have detected from the formation rule that we are forming a cumulative
universe. This allows us to reason about a simpler type theory. We will expand the
discussion when talking about the universes.  Again, we omit the
congruence rules as they are immediate. The $\beta$ equivalence rules are
\begin{mathpar}
  \inferrule
  {\mtyping[\vGamma; (\Gamma, x : \Nat)]{M}{\Se_i} \\ \mtyping[\vGamma; \Gamma]{s}{M[\vect I, \ze/x]} \\
    \mtyping[\vGamma; (\Gamma, x : \Nat, y : M)]{u}{M[(\wk_x \circ \wk_y), \su x/x]}}
  {\mtyequiv[\vGamma; \Gamma]{\elimn {x.M} s {x,y.u} \ze}{s}{M[\vect I, \ze/x]}}  

  \inferrule
  {\mtyping[\vGamma; (\Gamma, x : \Nat)]{M}{\Se_i} \\ \mtyping[\vGamma; \Gamma]{s}{M[\vect I, \ze/x]} \\
    \mtyping[\vGamma; (\Gamma, x : \Nat, y : M)]{u}{M[(\wk_x \circ \wk_y), \su x/x]} \\ \mtyping[\vGamma; \Gamma] t \Nat}
  {\mtyequiv[\vGamma; \Gamma]{\elimn{x.M}s{x,y.u}{(\su t)}}{u[\vect I, t/x,
      \elimn {x.M} s {x,y.u} t/y]} {M[\vect I, \su t/x]}}
\end{mathpar}

Next we define applications of substitutions:
\begin{mathpar}
  \inferrule
  {\mtyping{\vsigma}{\vDelta}}
  {\mtyequiv{\Nat[\vsigma]}{\Nat}\Se_i}

  \inferrule
  {\mtyping{\vsigma}{\vDelta}}
  {\mtyequiv{\ze[\vsigma]}\ze \Nat}

  \inferrule
  {\mtyping{\vsigma}{\vDelta} \\ \mtyping[\vDelta] t \Nat}
  {\mtyequiv{\su t[\vsigma]}{\su{(t[\vsigma])}}\Nat}

  \inferrule
  {\mtyping{\vsigma}{\vDelta;\Delta} \\
    \mtyping[\vDelta; (\Delta, x : \Nat)]{M}{\Se_i} \\ \mtyping[\vDelta; \Delta]{s}{M[\vect I, \ze/x]} \\
    \mtyping[\vDelta; (\Delta, x : \Nat, y : M)]{u}{M[(\wk_x \circ \wk_y), \su x/x]} \\ \mtyping[\vDelta; \Delta] t \Nat}
  {\mtyequiv[\vGamma]{\elimn {x.M} s {x,y.u} t[\vsigma]}{\elimn
      {x.M[q(\vsigma, x)]} {(s[\vsigma])}
      {x,y.u[q(q(\vsigma, x), y)]}{(t[\vsigma])}}{M[\vsigma, t[\vsigma]/x]}}
\end{mathpar}

\subsection{Universes}

Universes are where types reside. They are in a predicative hierarchy and universes
are indexed by meta-level natural numbers, so that Russell's paradox does not
occur. In this note, we use cumulative universes, so that every type in a lower level
also resides in a higher level. In contrast, non-cumulativity ensures every term must
have a unique type. This property is not true with cumulativity. We work with
cumulativity primarily because then in the normalization algorithm we do not have to
always keep track of universe levels.
\begin{mathpar}
  \inferrule
  {\vdash \vGamma}
  {\mtyping{\Se_n}{\Se_{n + 1}}}

  \inferrule
  {\mtyping{T}{\Se_i}}
  {\mtyping{T}{\Se_{i + 1}}}

  \inferrule
  {\mtyequiv{T}{T'}{\Se_i}}
  {\mtyequiv{T}{T'}{\Se_{i + 1}}}

  \inferrule
  {\mtyping{\vsigma}{\vDelta}}
  {\mtyequiv{\Se_n[\vsigma]}{\Se_n}{\Se_{n + 1}}}
\end{mathpar}

\subsection{Modality}

So far, we were still in MLTT, as we only operate in the top-level context in the
stack. The necessity modality allows us to manipulate the context stack depending on
the range of \tunbox levels.
\begin{mathpar}
  \inferrule
  {\mtyping[\vGamma; \cdot]{T}{\Se_i}}
  {\mtyping{\square T}{\Se_i}}

  \inferrule
  {\mtyping[\vGamma; \cdot]{t}{T}}
  {\mtyping{\boxit t}{\square T}}

  \inferrule
  {\boxed{\mtyping[\vGamma; \cdot]{T}{\Se_i}} \\ \mtyping{t}{\square T} \\ \vdash \vGamma; \vDelta \\ |\vDelta| = n}
  {\mtyping[\vGamma; \vDelta]{\unbox n t}{T[\sextt{\vect I} n]}}
\end{mathpar}
The equivalence rules are
\begin{mathpar}
  \inferrule
  {\boxed{\mtyping[\vGamma; \cdot]{T}{\Se_i}} \\ \mtyping[\vGamma; \cdot]{t}{T} \\ \vdash \vGamma; \vDelta \\ |\vDelta| = n}
  {\mtyequiv[\vGamma; \vDelta]{\unbox n {(\boxit t)}}{t[\sextt{\vect I} n]}{T[\sextt{\vect I} n]}}  

  \inferrule
  {\boxed{\mtyping[\vGamma; \cdot]{T}{\Se_i}} \\ \mtyping[\vGamma]{t}{\square T}}
  {\mtyequiv{t}{\boxit{(\unbox 1 t)}}{\square T}}  
\end{mathpar}
Substitution-related rules are
\begin{mathpar}
  \inferrule
  {\mtyping[\vDelta; \cdot]{T}{\Se_i} \\ \mtyping{\vsigma}{\vDelta}}
  {\mtyequiv{\square T[\vsigma]}{\square (T[\sextt\vsigma 1])}{\Se_i}}

  \inferrule
  {\mtyping[\vDelta; \cdot]{t}{T} \\ \mtyping{\vsigma}{\vDelta}}
  {\mtyequiv{\boxit t[\vsigma]}{\boxit {(t[\sextt\vsigma 1])}}{\square T [\vsigma]}}

  \inferrule
  {\boxed{\mtyping[\vDelta; \cdot]{T}{\Se_i}} \\
    \mtyping[\vDelta]{t}{\square T} \\ |\vDelta'| = n \\ \mtyping{\vsigma}{\vDelta; \vDelta'}}
  {\mtyequiv{\unbox n t[\vsigma]}{\unbox{\Ltotal\vsigma n}{(t[\trunc\vsigma
        n])}}{T[\sextt{\trunc\vsigma n}{\Ltotal \vsigma n}]}}
\end{mathpar}

\section{Syntactic Properties of \mintslang}

In this section, we examine some syntactic properties for \mintslang. We will prove
presupposition at the end of this section to ensure we are working on a sensible type
theory. We first start with some small properties that are required to prove presupposition.

\subsection{Basic Properties}

\begin{lemma}
  If $\vdash \vGamma \approx \vGamma'$, then $|\vGamma| = |\vGamma'|$. 
\end{lemma}

\begin{lemma}
  If $\vdash \vGamma \approx \vGamma'$ and $n < |\vGamma|$, then $\vdash \trunc\vGamma
  n \approx \trunc{\vGamma'}n$.
\end{lemma}
\begin{proof}
  Induction on $n$. 
\end{proof}

\begin{lemma}\labeledit{lem:dt:ty-lookup}
  If $\vdash \vGamma$ and $x : T \in \vGamma$, then $\mjudge T$. 
\end{lemma}
\begin{proof}
  We do induction on $x : T \in \vGamma$ and invert $\vdash \vGamma$.
\end{proof}

\begin{lemma}\labeledit{lem:dt:tyequiv-lookup}
  If $\vdash \vGamma \approx \vGamma'$, $x : T \in \vGamma$, then there exists $T'$,
  such that $x : T' \in \vGamma'$ and $\mjudge{T \approx T'}$.
\end{lemma}
\begin{proof}
  We do induction on $x : T \in \vGamma$ and invert $\vdash \vGamma \approx \vGamma'$,
  so we find $x : T' \in \vGamma'$ and apply congruence.
\end{proof}

\begin{lemma}\labeledit{lem:dt:csequiv-psup}
  If $\vdash \vGamma \approx \vGamma'$, then $\vdash \vGamma$ and $\vdash \vGamma'$. 
\end{lemma}
\begin{proof}
  Induction on $\vdash \vGamma \approx \vGamma'$.
\end{proof}

\begin{lemma}[Reflexivity]
  If $\mtyping t T$, then $\mtyequiv{t}{t}{T}$. 
\end{lemma}
\begin{proof}
  Induction on $\mtyping t T$.
\end{proof}

\begin{lemma}[Reflexivity]
  If $\mtyping{\vsigma}{\vDelta}$, then $\mtyequiv{\vsigma}{\vsigma}{\vDelta}$.
\end{lemma}
\begin{proof}
  Induction on $\mtyping{\vsigma}{\vDelta}$.
\end{proof}

\begin{lemma}[Cumulativity]\labeledit{lem:dt:cumu}
  If $\mtyping{T}{\Se_i}$ and $i \le j$, then $\mtyping{T}{\Se_j}$.
\end{lemma}
\begin{proof}
  Since $i \le j$, that means $j = i + k$ for some $k$. We do induction on $k$. 
\end{proof}

\begin{lemma}[Cumulativity]\labeledit{lem:dt:eq-cumu}
  If $\mtyequiv{T}{T'}{\Se_i}$ and $i \le j$, then $\mtyequiv{T}{T'}{\Se_j}$.
\end{lemma}
\begin{proof}
  Proved similarly to above.
\end{proof}

\begin{lemma}\labeledit{lem:dt:pi-inv}
  If $\mjudge[\vGamma; \Gamma]{\Pi(x : S). T}$, then $\mjudge[\vGamma; \Gamma]S$ and
  $\mjudge[\vGamma; (\Gamma, x : S)] T$.
\end{lemma}
\begin{proof}
  Induction on $\mjudge[\vGamma; \Gamma]{\Pi(x : S). T}$.
\end{proof}

\begin{lemma}\labeledit{lem:dt:sq-inv}
  If $\mjudge{\square T}$, then $\mjudge[\vGamma; \cdot]{T}$.
\end{lemma}
\begin{proof}
  Induction on $\mjudge{\square T}$.
\end{proof}

We examine the properties of truncation and truncation offset:
\begin{lemma}
  If $\mtyping \vsigma \vDelta$ and $n < |\vDelta|$, then $\Ltotal \vsigma n <
  |\vGamma|$. 
\end{lemma}
\begin{proof}
  Induction on $n$ and $\mtyping \vsigma \vDelta$.
\end{proof}

\begin{lemma}
  If $\mtyping \vsigma \vDelta$ and $n < |\vDelta|$, then 
  $\mtyping[\trunc\vGamma{\Ltotal \vsigma n}]{\trunc \vsigma n}{\trunc \vDelta n}$.
\end{lemma}
\begin{proof}
  Induction on $n$ and $\mtyping \vsigma \vDelta$.
\end{proof}

\begin{lemma}
  If $\mtyping{\vsigma}{\vDelta}$ and $n + m < |\vDelta|$, then
  $\Ltotal\vsigma {n + m} = \Ltotal \vsigma n + \Ltotal{\trunc \vsigma n} m$.
\end{lemma}
\begin{proof}
  Induction on $n$ and $\mtyping{\vsigma}{\vDelta}$.
\end{proof}

\begin{lemma}
  If $\mtyping{\vsigma}{\vDelta}$ and $n + m < |\vDelta|$, then
  $\trunc{\trunc \vsigma n}m = \trunc\vsigma{n+m}$.
\end{lemma}
\begin{proof}
  Induction on $n$ and $\mtyping{\vsigma}{\vDelta}$.
\end{proof}

\begin{lemma}
  If $\mtyequiv{\vsigma}{\vsigma'}\vDelta$ and $n < |\vDelta|$, then $\Ltotal \vsigma n =
  \Ltotal \vsigma' n$.
\end{lemma}
\begin{lemma}
  If $\mtyequiv{\vsigma}{\vsigma'}\vDelta$ and $n < |\vDelta|$, then
  $\mtyequiv[\trunc\vGamma{\Ltotal \vsigma n}]{\trunc \vsigma n}{\trunc{\vsigma'}n}{\trunc \vDelta n}$.
\end{lemma}
\begin{proof}
  Induction on $n$ and $\mtyequiv{\vsigma}{\vsigma'}\vDelta$. Most cases are
  immediate. We only consider two cases:
  \begin{itemize}[label=Case]
  \item
    \begin{mathpar}
      \inferrule
      {\mtyping{\vsigma}{\vGamma'} \\ \mtyping[\vGamma'']{\vdelta}{\vGamma;\vDelta} \\
        |\vDelta| = m}
      {\mtyequiv[\vGamma'']{(\sextt\vsigma m) \circ \vdelta}{\sextt{(\vsigma \circ (\trunc\vdelta m))}{\Ltotal\vdelta m} }{\vGamma';\cdot}}
    \end{mathpar}
    \begin{align*}
      \Ltotal{(\sextt\vsigma m) \circ \vdelta}{1 + n}
      &= \Ltotal\vdelta {m + \Ltotal \vsigma n} \\
      &= \Ltotal\vdelta m +  \Ltotal{\trunc \vdelta m}{\Ltotal \vsigma n} \\
      &= \Ltotal{\sextt{(\vsigma \circ (\trunc \vdelta m))}{\Ltotal\vdelta m}}{1 + n}
    \end{align*}
    Moreover, 
    \begin{align*}
      \trunc{((\sextt\vsigma m) \circ \vdelta)}{1 + n} = \trunc{(\vsigma \circ \trunc
      \vdelta m)} n
      = \trunc{(\sextt{(\vsigma \circ \trunc \vdelta m)}{\Ltotal\vdelta m})}{1+n}
    \end{align*}
    We also know
    \begin{mathpar}
      \inferrule
      {\inferrule
        {\mtyping{\vsigma}{\vGamma'} \\ \mtyping[\trunc{\vGamma''}{\Ltotal\vdelta m}]{\trunc \vdelta m}{\vGamma}}
        {\mtyping[\trunc{\vGamma''}{\Ltotal\vdelta m}]{\vsigma \circ \trunc \vdelta m}{\vGamma'}}}
      {\mtyping[\trunc{\trunc{\vGamma''}{\Ltotal\vdelta m}}{\Ltotal{\vsigma \circ \trunc \vdelta m} n}]{\trunc{(\vsigma \circ \trunc \vdelta m)}n}{\trunc{\vGamma'}n}}
    \end{mathpar}
    so we can apply reflexivity for the equivalence judgment.  At last, we shall
    examine the truncation size of $\vGamma''$ is $\Ltotal\vdelta{m + \Ltotal \vsigma n}$.
    We have
    \begin{align*}
      \Ltotal\vdelta m + \Ltotal{\vsigma \circ \trunc \vdelta m}n
      &= \Ltotal\vdelta m + \Ltotal{\trunc \vdelta m}{\Ltotal \vsigma n} \\
      &= \Ltotal\vdelta{m + \Ltotal \vsigma n}
    \end{align*}
    
  \item
    \begin{mathpar}
      \inferrule
      {\mtyping{\vsigma}{\vDelta; \cdot} \\ |\vDelta| > 0}
      {\mtyequiv{\vsigma}{\sextt{\trunc \vsigma 1}{\Ltotal\vsigma 1}}{\vDelta; \cdot}}
    \end{mathpar}
    \begin{align*}
      \Ltotal{\sextt{\trunc \vsigma 1}{\Ltotal\vsigma 1}}{1 + n}
      = \Ltotal\vsigma 1 + \Ltotal{\trunc \vsigma 1} n = \Ltotal\vsigma{1 + n}
    \end{align*}
    By the following equation, the other goal is immediate:
    \begin{align*}
      \trunc\vsigma{1 + n} = \trunc{\trunc \vsigma 1}n = \trunc{(\sextt{\trunc \vsigma 1}{\Ltotal\vsigma 1})}{1+n}
    \end{align*}
  \end{itemize}
\end{proof}

\subsection{Context Stack Conversion}

\begin{lemma}[Context stack conversion]\labeledit{lem:dt:cs-equiv}
  If $\vdash \vGamma \approx \vDelta$, and
  \begin{itemize}
  \item if $\mtyping t T$, then $\mtyping[\vDelta]t T$;
  \item if $\mtyequiv{t}{t'}T$, then $\mtyequiv[\vDelta]{t}{t'}T$;
  \item if $\mtyping{\vsigma}{\vDelta}$, then $\mtyping[\vDelta]{\vsigma}{\vDelta}$;
  \item if $\mtyequiv{\vsigma}{\vsigma'}{\vDelta'}$, then $\mtyequiv[\vDelta]{\vsigma}{\vsigma'}{\vDelta'}$.
  \end{itemize}
\end{lemma}
\begin{proof}
  We do induction on each premise derivation.
  \begin{itemize}[label=Case]
  \item $\vGamma = \vGamma'; \Gamma$ and $\vDelta = \vDelta'; \Delta$
    \begin{mathpar}
      \inferrule
      {\vdash \vGamma'; \Gamma \\ x : T \in \vGamma'; \Gamma}
      {\mtyping[\vGamma'; \Gamma]x T}
    \end{mathpar}
    \begin{align*}
      & \vdash \vDelta'; \Delta
        \tag{by \Cref{lem:dt:csequiv-psup}} \\
      & x : T' \in \vDelta';\Delta \tand \mjudge[\vGamma';\Gamma]{T \approx T'}
        \tag{by \Cref{lem:dt:tyequiv-lookup}} \\
      & \mjudge[\vDelta';\Delta]{T \approx T'}
        \byIH
    \end{align*}
    \begin{mathpar}
      \inferrule*
      {\inferrule*
        {\vdash \vDelta'; \Delta \\ x : T' \in \vDelta'; \Delta}
        {\mtyping[\vDelta'; \Delta]x T'} \\
        \inferrule*
        {\mtyequiv[\vDelta';\Delta]{T}{T'}{\Se_i}}
        {\mtyequiv[\vDelta';\Delta]{T'}{T}{\Se_i}}}
      {\mtyping[\vDelta'; \Delta]x T}
    \end{mathpar}
    
  \item $\vGamma = \vGamma'; \Gamma$ and $\vDelta = \vDelta'; \Delta$
    \begin{mathpar}
      \inferrule
      {\mtyping[\vGamma'; \Gamma]{S}{\Se_i} \\ \mtyping[\vGamma'; (\Gamma, x : S)]{t}{T}}
      {\mtyping[\vGamma'; \Gamma]{\lambda x. t}{\Pi(x : S). T}}
    \end{mathpar}
    \begin{align*}
      & \mtyping[\vDelta'; \Delta]{S}{\Se_i}
        \byIH \\
      & \vdash \vGamma'; (\Gamma, x : S) \approx \vDelta; (\Delta, x : S) \\
      & \mtyping[\vDelta'; (\Delta, x : S)]{t}{T}
        \byIH \\
      & \mtyping[\vDelta'; \Delta]{\lambda x. t}{\Pi(x : S). T}
    \end{align*}
    
  \item
    \begin{mathpar}
      \inferrule
      {\mtyping[\vGamma; \cdot]{T}{\Se_i}}
      {\mtyping{\square T}{\Se_i}}
    \end{mathpar}
    \begin{align*}
      & \vdash \vGamma; \cdot \approx \vDelta; \cdot \\
      & \mtyping[\vDelta; \cdot]{T}{\Se_i}
        \byIH \\
      & \mtyping[\vDelta]{\square T}{\Se_i}
    \end{align*}
    
  \item
    \begin{mathpar}
      \inferrule
      {\mtyping[\vGamma; \cdot]{t}{T}}
      {\mtyping{\boxit t}{\square T}}
    \end{mathpar}
    \begin{align*}
      & \vdash \vGamma; \cdot \approx \vDelta; \cdot \\
      & \mtyping[\vDelta; \cdot]{t}{T}
        \byIH \\
      & \mtyping[\vDelta]{\boxit t}{\square T}
    \end{align*}
    
  \item $\vGamma = \vGamma'; \vGamma''$ and $\vDelta = \vDelta'; \vDelta''$
    \begin{mathpar}
      \inferrule
      {\mtyping[\vGamma']{t}{\square T} \\ \vdash \vGamma'; \vGamma'' \\ |\vGamma''| = n}
      {\mtyping[\vGamma'; \vGamma'']{\unbox n t}{T[\sextt{\vect I} n]}}
    \end{mathpar}
    \begin{align*}
      & \vdash \vDelta'; \vDelta''
        \tag{by \Cref{lem:dt:csequiv-psup}} \\
      & \vdash \vGamma' \approx \vDelta' \\
      & \mtyping[\vDelta']{t}{\square T}
        \byIH \\
      & \mtyping[\vDelta'; \vDelta'']{\unbox n t}{T[\sextt{\vect I}n]}
    \end{align*}
    
  \item $\vGamma = \vGamma'; \Gamma$ and $\vDelta = \vDelta'; \Delta$
    \begin{mathpar}
      \inferrule
      {\mtyping[\vGamma'; \Gamma]{S}{\Se_i} \\ \mtyping[\vGamma'; (\Gamma, x : S)]{t}{T} \\ \mtyping[\vGamma';\Gamma]{s}{S}}
      {\mtyequiv[\vGamma'; \Gamma]{(\lambda x. t)\ s}{t[\vect I, s/x]}{T[\vect I, s/x]}}
    \end{mathpar}
    \begin{align*}
      & \mtyping[\vDelta'; \Delta]{S}{\Se_i}
      \byIH \\
      & \mtyping[\vDelta';\Delta]{s}{S}
        \byIH \\
      & \vdash \vGamma'; (\Gamma, x : S) \approx \vDelta'; (\Delta, x : S) \\
      & \mtyping[\vDelta'; (\Delta, x : S)]{t}{T}
        \byIH \\
      & \mtyequiv[\vDelta'; \Delta]{(\lambda x. t)\ s}{t[\vect I, s/x]}{T[\vect I, s/x]}
    \end{align*}
    
  \item $\vGamma = \vGamma'; \vGamma''$ 
    \begin{mathpar}
      \inferrule
      {\mtyping[\vGamma'; \cdot]{t}{T} \\ |\vGamma''| = n}
      {\mtyequiv[\vGamma'; \vGamma'']{\unbox n {(\boxit t)}}{t[\sextt{\vect I}n]}{T[\sextt{\vect I}n]}}
    \end{mathpar}
    \begin{align*}
      & \vdash \vGamma' \approx \trunc \vDelta n \\
      & \vdash \vGamma'; \cdot \approx \trunc \vDelta n; \cdot \\
      & \mtyping[\trunc \vDelta n; \cdot]{t}{T}
        \byIH \\
      & \mtyequiv[\vDelta]{\unbox n {(\boxit t)}}{t[\sextt{\vect I}n]}{T[\sextt{\vect I}n]}
    \end{align*}
    
  \item
    \begin{mathpar}
      \inferrule
      {\mtyping[\vGamma]{t}{\square T}}
      {\mtyequiv{t}{\boxit{(\unbox 1 t)}}{\square T}}  
    \end{mathpar}
    immediate by IH.
    
  \item
    \begin{mathpar}
      \inferrule
      {\vdash \vGamma}
      {\mtyping{\vect I}{\vGamma}}
    \end{mathpar}
    We have
    \begin{mathpar}
      \inferrule*
      {\inferrule*
        {\inferrule
          { }
          {\vdash \vDelta}\text{(\Cref{lem:dt:csequiv-psup})}}
        {\mtyping[\vDelta]{\vect I}{\vDelta}} \\
        \inferrule*
        {\vdash \vGamma \approx \vDelta}
        {\vdash \vDelta \approx \vGamma}}
      {\mtyping[\vDelta]{\vect I}{\vGamma}}
    \end{mathpar}
  \item $\vGamma = \vGamma'; \vGamma''$ 
    \begin{mathpar}
      \inferrule
      {\mtyping[\vGamma']{\vsigma}{\vDelta'} \\ \vdash \vGamma'; \vGamma'' \\ |\vGamma''| = n}
      {\mtyping[\vGamma'; \vGamma'']{\sextt\vsigma n}{\vDelta'; \cdot}}
    \end{mathpar}
    \begin{align*}
      & \vdash \vDelta
        \tag{by \Cref{lem:dt:csequiv-psup}} \\
      & \vdash \vGamma' \approx \trunc \vDelta n \\
      & \mtyping[\trunc \vDelta n]{\vsigma}{\vDelta'}
        \byIH \\
      & \mtyping[\vDelta]{\sextt\vsigma n}{\vDelta'; \cdot}
    \end{align*}
  \end{itemize}
\end{proof}

As a corollary, we prove the following theorem:
\begin{lemma}
  $\vdash \vGamma \approx \vGamma'$ is a PER. 
\end{lemma}
\begin{proof}
  Because $\mtyequiv{t}{t'}{T}$ is a PER and $\vdash \vGamma \approx \vGamma'$ is
  defined pointwise. Apply context stack conversion in transitivity.
\end{proof}

\subsection{Presupposition}

\begin{lemma}[presupposition]
  \begin{itemize}
  \item If $\mtyping t T$, then $\vdash \vGamma$ and $\mjudge T$.
  \item If $\mtyequiv{t}{t'}T$, then $\vdash \vGamma$, $\mtyping t T$,
    $\mtyping{t'}{T}$ and $\mjudge T$. 
  \item If $\mtyping{\vsigma}{\vDelta}$, then $\vdash \vGamma$ and $\vdash \vDelta$. 
  \item If $\mtyequiv{\vsigma}{\vsigma'}{\vDelta}$, then $\vdash \vGamma$,
    $\mtyping{\vsigma}{\vDelta}$, $\mtyping{\vsigma'}{\vDelta}$ and $\vdash \vDelta$. 
  \end{itemize}
\end{lemma}
\begin{proof}
  We do induction on each premise derivation.
  \begin{itemize}[label=Case]
  \item $\vGamma = \vGamma'; \Gamma$
    \begin{mathpar}
      \inferrule
      {\vdash \vGamma'; \Gamma \\ x : T \in \vGamma'; \Gamma}
      {\mtyping[\vGamma'; \Gamma]x T}
    \end{mathpar}
    We obtain $\mjudge[\vGamma';\Gamma]T$ by \Cref{lem:dt:ty-lookup}.

  \item $\vGamma = \vGamma'; \Gamma$
    \begin{mathpar}
      \inferrule
      {\mtyping[\vGamma'; \Gamma]{S}{\Se_i} \\ \mtyping[\vGamma'; (\Gamma, x : S)]{t}{T}}
      {\mtyping[\vGamma'; \Gamma]{\lambda x. t}{\Pi(x : S). T}}
    \end{mathpar}
    \begin{align*}
      & \vdash \vGamma'; \Gamma
        \byIH \\
      & \mjudge[\vGamma'; (\Gamma, x : S)] T
        \byIH \\
      & \mtyping[\vGamma'; \Gamma]{S}{\Se_j} \tand \mtyping[\vGamma'; (\Gamma, x : S)]{T}{\Se_j}
        \tag{for some $j$ by \Cref{lem:dt:cumu}} \\
      & \mtyping[\vGamma'; \Gamma]{\Pi(x : S).T}{\Se_j}
    \end{align*}
    
  \item
    \begin{mathpar}
      \inferrule
      {\mtyping{t}{\Pi(x : S). T} \\ \mtyping{s}{S}}
      {\mtyping{t\ s}{T[\vect I, s/x]}}
    \end{mathpar}
    \begin{align*}
      & \vdash \vGamma
        \byIH \\
      & \mjudge{\Pi(x : S). T}
        \byIH \\
      & \mjudge[\vGamma'; (\Gamma, x : S)] T
        \tag{by \Cref{lem:dt:pi-inv}, where $\vGamma = \vGamma'; \Gamma$} \\
      & \mjudge[\vGamma'; \Gamma]{T[\vect I, s/x]}
    \end{align*}
    
  \item
    \begin{mathpar}
      \inferrule
      {\mtyping[\vGamma; \cdot]{t}{T}}
      {\mtyping{\boxit t}{\square T}}
    \end{mathpar}
    \begin{align*}
      & \vdash \vGamma; \cdot
        \byIH \\
      & \vdash \Gamma \\
      & \mjudge[\vGamma; \cdot]T
        \byIH \\
      & \mjudge{\square T}
    \end{align*}
    
  \item $\vGamma = \vGamma'; \vGamma''$
    \begin{mathpar}
      \inferrule
      {\mtyping[\vGamma']{t}{\square T} \\ \vdash \vGamma'; \vGamma'' \\ |\vGamma''| = n}
      {\mtyping[\vGamma'; \vGamma'']{\unbox n t}{T[\sextt{\vect I}n]}}
    \end{mathpar}
    \begin{align*}
      & \mjudge[\vGamma'; \cdot]T
        \tag{by \Cref{lem:dt:sq-inv}} \\
      & \mjudge[\vGamma'; \vGamma'']{T[\sextt{\vect I}n]}
    \end{align*}

  \item $\vGamma = \vGamma'; \Gamma$
    \begin{mathpar}
      \inferrule
      {\mtyping[\vGamma'; \Gamma]{S}{\Se_i} \\ \mtyequiv[\vGamma'; (\Gamma, x : S)]{t}{t'}{T}}
      {\mtyequiv[\vGamma'; \Gamma]{\lambda x. t}{\lambda x. t'}{\Pi(x : S). T}}
    \end{mathpar}
    \begin{align*}
      H_1: & \vdash \vGamma'; (\Gamma, x : S)
             \byIH  \\
           & \mtyping[\vGamma'; (\Gamma, x : S)]{t}{T}
             \byIH \\
           & \mtyping[\vGamma'; (\Gamma, x : S)]{t'}{T}
             \byIH \\
           & \mjudge[\vGamma'; (\Gamma, x : S)]{T}
             \byIH \\
           & \mtyping[\vGamma'; \Gamma]{\lambda x. t}{\Pi(x : S).T}
             \tand \mtyping[\vGamma'; \Gamma]{\lambda x. t'}{\Pi(x : S).T} \\
           & \vdash \vGamma'; \Gamma \tand \mjudge[\vGamma'; \Gamma]S
             \tag{by $H_1$} \\
           & \mjudge[\vGamma'; \Gamma]{\Pi(x : S). T}
    \end{align*}
    
  \item $\vGamma = \vGamma'; \Gamma$
    \begin{mathpar}
      \inferrule
      {\mtyping[\vGamma'; \Gamma]{S}{\Se_i} \\ 
        \mtyequiv[\vGamma'; \Gamma]{S}{S'}{\Se_i} \\ \mtyequiv[\vGamma'; (\Gamma, x : S)]{T}{T'}{\Se_i}}
      {\mtyequiv[\vGamma'; \Gamma]{\Pi(x : S). T}{\Pi(x : S'). T'}{\Se_i}}
    \end{mathpar}
    \begin{align*}
      & \vdash \vGamma'; \Gamma \tand \mtyping[\vGamma'; \Gamma]{S}{\Se_i}
        \tand \mtyping[\vGamma'; \Gamma]{S'}{\Se_i}
        \byIH \\
      & \mtyping[\vGamma'; (\Gamma, x : S)]{T}{\Se_i} \tand
        \mtyping[\vGamma'; (\Gamma, x : S)]{T'}{\Se_i}
        \byIH \\
      & \vdash \vGamma; (\Gamma, x : S) \approx \vGamma; (\Gamma, x : S')
        \tag{by definition} \\
      & \mtyping[\vGamma'; (\Gamma, x : S')]{T'}{\Se_i}
        \tag{by \Cref{lem:dt:cs-equiv}} \\
      & \mtyping[\vGamma'; \Gamma]{\Pi(x : S). T}{\Se_i} \tand
        \mtyping[\vGamma'; \Gamma]{\Pi(x : S'). T'}{\Se_i}
    \end{align*}

  \item
    \begin{mathpar}
      \inferrule
      {\mtyping[\vGamma; \cdot]{t}{T} \\ \vdash \vGamma; \vDelta \\ |\vDelta| = n}
      {\mtyequiv[\vGamma; \vDelta]{\unbox n {(\boxit t)}}{t[\sextt{\vect I}n]}{T[\sextt{\vect I}n]}}  
    \end{mathpar}
    Immediate.
    
  \item
    \begin{mathpar}
      \inferrule
      {\mtyping[\vGamma]{t}{\square T}}
      {\mtyequiv{t}{\boxit{(\unbox 1 t)}}{\square T}}  
    \end{mathpar}
    Notice we have
    \begin{align*}
      \mtyequiv[\vGamma; \cdot]{\vect I}{\vect I; ()}{\vGamma; \cdot}
    \end{align*}
    This allows us to conclude $\mtyping{\boxit{(\unbox 1 t)}}{\square T}$.
    
  \item $\vGamma = \vGamma'; \vGamma''$
    \begin{mathpar}
      \inferrule
      {\mtyping[\vGamma']{\vsigma}{\vDelta} \\ \vdash \vGamma'; \vGamma'' \\ |\vGamma''| = n}
      {\mtyping[\vGamma'; \vGamma'']{\sextt\vsigma n}{\vDelta; \cdot}}
    \end{mathpar}
    We only need to prove $\vdash \vDelta; \cdot$ which we can get from
    $\vdash \vDelta$ after applying IH to $\mtyping[\vGamma']{\vsigma}{\vDelta}$.
    
  \item
    \begin{mathpar}
      \inferrule
      {\mtyping[\vGamma'']{\vsigma}{\vGamma'} \\ \mtyping{\vdelta}{\vGamma'';\vDelta} \\
        |\vDelta| = m}
      {\mtyequiv{(\sextt\vsigma m) \circ \vdelta}{\sextt{(\vsigma \circ \trunc \vdelta m)}
          {\Ltotal\vdelta m}}{\vGamma';\cdot}}
    \end{mathpar}
    \begin{align*}
      & \vdash \vGamma' \byIH \\
      & \vdash \vGamma \byIH \\
      & \mtyping[\trunc\vGamma{\Ltotal\vdelta m}]{\vdelta}{\vGamma''}
    \end{align*}
    We can derive the goal based on these judgments and the premises. 
  \end{itemize}
\end{proof}

\subsection{Removing Boxed Premises}

With presupposition, we are able to remove those boxed premises. The systems with and
without those premises can be proved equivalent.
\begin{theorem}
  We can remove the boxed premises and the resulting system remains equivalent to the
  one before the removal. 
\end{theorem}
We do this by showing equivalence of typing and equivalence judgments. We apply
presupposition whenever necessary. This theorem also implies the resulting system also
has presupposition.

\section{Untyped Domain and Evaluation for Dependent Types}

In this section, we define the untyped domain in which we will operate and define the
NbE algorithm:
\begin{alignat*}{2}
  z & && \tag{Domain variables, $\N$} \\
  a, b, A &:=&&\ \Nd \sep \squared A \sep Pi(A, x.T, \vrho) \sep \Ud_i 
  \tag{Domain terms, $D$}\\
  & && \sep \zed \sep \sud a \sep \Lambda(x.t, \vrho) \sep \tbox(a) \sep \uparrow^A(c) \\
  c, C &:= &&\ z \sep \rec{x.M}{a}{(x_1, x_2.t)}{c}{\vrho} \sep c\ d \sep \tunbox(n, c)
  \tag{Neutral domain terms, $D^{\Ne}$} \\
  d, D &:= &&\ \downarrow^A(a)
  \tag{Normal domain terms, $D^{\Nf}$} \\
  \Env &:= &&\ \Var \rightharpoonup D \\
  \Envs &:= &&\ \N \to \N \times \Env \\
  \rho & && \tag{Local evaluation environment, $\Env$} \\
  \vect{\rho} & && \tag{(Global) evaluation environment, $\Envs$} 
\end{alignat*}
Similar to the syntactic case, we use capital cases for those values intended to be
types.  Compared to the simply typed case, the syntax for reflection $\uparrow^A(c)$
and reification $\downarrow^A(a)$ marks their types using semantic values. In $Pi$,
$\Lambda$ and $\textsf{rec}$, we capture the global evaluation environment.

We inherit all definitions of evaluation environments in \Cref{sec:st:domain} as $\vrho$ has the same structure. 

\subsection{Untyped Modal Transformations}

Untyped modal transformations, or just \UnMoTs, is ranged over $\kappa$ and modeled by the
meta function space $\N \to \N$, following the simply typed case. It is responsible
for shifting the $\tunbox$ levels of domain neutral values.  We inherit all definitions and lemmas in
\Cref{sec:st:ut-mtrans}, again due to the same structure of evaluation environments.  

We next define application of modal transformation to $D$ and $\Envs$:
\begin{align*}
  \Nd[\kappa] &:= \Nd \\
  \squared A[\kappa] &:= \squared(A[\sextt\kappa 1]) \\
  Pi(A, x.T, \vrho)[\kappa] &:= Pi(A[\kappa], x.T, \vrho[\kappa]) \\
  \Ud_i[\kappa] &:= \Ud_i \\
  \zed[\kappa] &:= \zed \\
  \sud a[\kappa] &:= \sud{a[\kappa]} \\
  \tbox(a)[\kappa] &:= \tbox(a[\sextt\kappa 1]) \\
  \Lambda(x.t, \vrho)[\kappa] &:= \Lambda(x.t, \vrho[\kappa]) \\
  \uparrow^A(c)[\kappa] &:= \uparrow^{A[\kappa]}(c[\kappa]) \\
  \\
  z[\kappa] &:= z \\
  \rec{x.M}{a}{(x_1,x_2.t)}{c}{\vrho}[\kappa] &:= \rec{x.M}{a[\kappa]}{(x_1,x_2.t)}{c[\kappa]}{\vrho[\kappa]} \\
  c\ d[\kappa] &:= (c[\kappa])\ (d[\kappa]) \\ 
  \tunbox(n, c)[\kappa] &:= \tunbox(\Ltotal\kappa n, c[\trunc\kappa n]) \\
  \\
  \downarrow^A(a)[\kappa] &:= \downarrow^{A[\kappa]}(a[\kappa]) \\
  \\
  \vrho[\kappa](0) &:= (\Ltotal\kappa n, \rho[\kappa]) \tag{where $(n, \rho) :=
                     \vrho(0)$} \\
  \vrho[\kappa](1 + n) &:= \trunc\vrho 1[\trunc\kappa{\Ltotal\vrho 1}](n)
\end{align*}
where $\rho[\kappa]$ is defined by mapping $x$ to $\rho(x)[\kappa]$ if $x$ is defined
in $\rho$.

\subsection{Evaluation}

Next we consider the evaluation function, which, given $\vrho$, translate an $\Exp$ to
a $D$.
\begin{align*}
  \intp{\_} &: \Exp \to \Envs \to D \\
  \intp{\Se_i}(\vrho) &:= \Ud_i \\
  \intp{\Nat}(\vrho) &:= \Nd \\
  \intp{\square T}(\vrho) &:= \squared (\intp{T}(\ext(\vrho))) \\
  \intp{\Pi(x : S). T}(\vrho) &:= Pi(\intp{S}(\vrho), x.T, \vrho) \\
  \intp{x}(\vrho) &:= \rho(x) \tag{where $(\_, \rho) := \vrho(0)$} \\
  \intp{\ze}(\vrho) &:= \zed \\
  \intp{\su t}(\vrho) &:= \sud{\intp{t}(\vrho)} \\
  \intp{\elimn{x.M}{s}{x, y.u}{t}}(\vrho) &:= \trec \cdot (x.M,
                                            \intp{s}(\vrho), (x,y.u), \intp{t}(\vrho),
                                            \vrho) \\
  \intp{\boxit t}(\vrho) &:= \tbox(\intp{t}(\ext(\vrho))) \\
  \intp{\unbox n t}(\vrho) &:= \tunbox \cdot (\Ltotal\vrho n, \intp{t}(\trunc\vrho n)) \\
  \intp{\lambda x. t}(\vrho) &:= \Lambda(x.t, \vrho) \\
  \intp{t\ s}(\vrho) &:= \intp{t}(\vrho) \cdot \intp{s}(\vrho) \\
  \intp{t[\vsigma]}(\vrho) &:= \intp{t}(\intp{\vsigma}(\vrho))
\end{align*}

Here we make use of three partial functions which evaluates $\tbox$, $\Lambda$ and
object numbers, respectively.  We define the partial $\tunbox$ as follows:
\begin{align*}
  \tunbox \cdot &: \N \to D \rightharpoonup D \\
  \tunbox \cdot (n, \tbox(a)) &:= a[\sextt\vone n]  \\
    \tunbox \cdot (n, \uparrow^{\squared A}(c)) &:= \uparrow^{A[\sextt\vone n]} (\tunbox(n, c))
\end{align*}
We define the partial application as follows:
\begin{align*}
  \_ \cdot \_ &: D \rightharpoonup D \rightharpoonup D \\
  (\Lambda(x.t, \vrho)) \cdot a &:= \intp{t}(\ext(\vrho, x, a)) \\
  (\uparrow^{Pi(A, x.T, \vrho)}(c)) \cdot a &:= \uparrow^{\intp{T}(\ext(\vrho, x, a))} (c\ \downarrow^A(a))
\end{align*}
We define the partial recursion as follows:
\begin{align*}
  \trec \cdot &: \Trm \to D \to \Trm \to D \to \Envs \rightharpoonup D \\
  \trec \cdot (x.M, a, (x_1, x_2.t), \zed, \vrho)
              &:= a \\
  \trec \cdot (x.M, a, (x_1, x_2.t), \sud b, \vrho)
              &:= \intp{t}(\ext(\vrho, x_1, b, x_2, \trec \cdot (x.M, a, (x_1, x_2.t), b, \vrho))) \\
  \trec \cdot (x.M, a, (x_1, x_2.t), \uparrow^{\Nd}(c), \vrho)
              &= \uparrow^{\intp{M}(\ext(\vrho, x, \uparrow^{\Nd}(c)))}(\rec{x.M}{a}{(x_1,x_2.t)}{c}{\vrho})
\end{align*}
In the $t[\vsigma]$ case, we need the interpretation of substitutions:
\begin{align*}
    \intp{\_} &: \Substs \to \Envs \to \Envs \\
    \intp{\vect I}(\vrho) &:= \vrho \\
    \intp{\wk_x}(\vrho) &:= \drop(\vrho, x) \\
    \intp{\vsigma, t/x}(\vrho) &:= \ext(\intp{\vsigma}(\vrho), x, \intp{t}(\vrho)) \\
    \intp{\sextt\vsigma n}(\vrho) &:= \ext(\intp{\vsigma}(\trunc\vrho n), \Ltotal\vrho n) \\
  \intp{\vsigma \circ \vdelta}(\vrho) &:= \intp{\vsigma}(\intp{\vdelta}(\vrho))
\end{align*}

We can prove the following just like in the simply typed case:
\begin{lemma}\labeledit{lem:dt:unbox.-mon}
  $\tunbox \cdot (n, a)[\kappa] = \tunbox \cdot (\Ltotal\kappa n, a[\trunc\kappa n])$
\end{lemma}

Just as in the simply typed case, we also establish the following
naturality lemmas:
\begin{lemma}\labeledit{conj:dt:.-mon}
  $(a \cdot b)[\kappa] = (a[\kappa]) \cdot (b[\kappa])$
\end{lemma}
\begin{lemma}\labeledit{conj:dt:rec.-mon}
  $\trec \cdot (\tau, a, \tau', b)[\kappa] = \trec \cdot (\tau[\kappa], a[\kappa], \tau'[\kappa], b[\kappa])$
\end{lemma}
\begin{lemma}\labeledit{conj:dt:intp-mon}
  \begin{itemize}
  \item $\intp{t}(\vrho[\kappa]) = \intp{t}(\vrho)[\kappa]$
  \item $\intp{\vsigma}(\vrho[\kappa]) = \intp{\vsigma}(\vrho)[\kappa]$
  \end{itemize}
\end{lemma}

\subsection{Readback Functions}

After evaluating an $\Exp$ to $D$, we have already got the corresponding $\beta$
normal form of the term in $D$. We need a one last step to read from $D$ back to
normal form and do the $\eta$ expansion at the same time to obtain a $\beta\eta$
normal form:
\begin{align*}
  \Rnf &: (\N \rightharpoonup \Var) \rightharpoonup D^{\Nf} \rightharpoonup \Nf \\
  \Rnf_{\alpha} (\downarrow^{\Ud_i} (A))
       &:= \Rty_\alpha(A) \\
  \Rnf_{\alpha} (\downarrow^{\uparrow^A(c)} (\uparrow^{A'}(c')))
       &:= \Rne_\alpha(c') \\
  \Rnf_\alpha(\downarrow^\Nd(\zed))
       &:= \ze \\
  \Rnf_\alpha(\downarrow^\Nd(\sud a))
       &:= \su{\Rnf_\alpha(\downarrow^\Nd(a))} \\
  \Rnf_{\alpha} (\downarrow^{\squared A} (a))
       &:= \boxit \Rnf_{\alpha} (\downarrow^{A} (\tunbox \cdot (1, a)))
  \\
  \Rnf_{\alpha} (\downarrow^{Pi(A, x.T, \vrho)} (a))
       &:= \lambda x. \Rnf_{\alpha[z \mapsto x]} (\downarrow^{\intp{T}(\ext(\vrho, x, \uparrow^A(z)))} (a \cdot \uparrow^A(z)))
         \tag{where $z := \tnext(\alpha)$}\\
  \Rnf_{\alpha} (\downarrow^\Nd (\uparrow^\Nd (c)))
       &:= \Rne_{\alpha} (c) \\
  \\
  \Rne &: (\N \rightharpoonup \Var) \rightharpoonup D^{\Ne} \rightharpoonup \Ne \\
  \Rne_{\alpha}(z)
       &:= \alpha(z)
  \tag{if $\alpha$ is undefined at $z$, we assign some fixed default $\Var$}\\
  \Rne_{\alpha}(c\ d)
       &:= \Rne_{\alpha} (c)\ \Rnf_{\alpha}(d) \\
  \Rne_{\alpha}(\tunbox(n, c))
       &:= \unbox{n}{\Rne_{\alpha}(c)} \\
  \Rne_{\alpha}(\trec(x.M, a, (x_1,x_2.t), c, \vrho))
       &:= \elimn{x.\Rty_{\alpha[z \mapsto
         x]}(A)}{(\Rnf_\alpha(\downarrow^{\intp{M}(\ext(\vrho,x, \zed)))}(a)))\\
       & \qquad }
         {x,y. \Rnf_{\alpha[z \mapsto x, z' \mapsto
         y]}(\downarrow^{\intp{M}(\ext(\vrho,x, \sud b)))}(\intp{t}(\ext(\vrho, x, b,
         y, \uparrow^{A}(z')))))}{\Rne_\alpha(c)}
         \tag{where $b := \uparrow^\Nd(z)$, $z := \tnext(\alpha)$, $z' := \tnext(\alpha[z
         \mapsto x])$, $A = \intp{M}(\ext(\vrho,x, b))$} \\
    \\
  \Rty &: (\N \rightharpoonup \Var) \rightharpoonup D^{\Nf} \rightharpoonup \Nf \\
  \Rty_\alpha(\Ud_i) &:= \Se_i \\
  \Rty_\alpha(\Nd) &:= \Nat \\
  \Rty_\alpha(\squared A) &:= \square \Rty_\alpha(A) \\
  \Rty_\alpha(Pi(A, x.T, \vrho)) &:= \Pi(x : \Rty_\alpha(A)).\Rty_{\alpha[z \mapsto
                          x]}(\intp{T}(\ext(\vrho, x, \uparrow^A(z))))
                          \tag{where $z := \tnext(\alpha)$}\\
  \Rty_\alpha(\uparrow^A(c)) &:= \Rne_\alpha(c)
\end{align*}
$\alpha$ is the same as in \Cref{sec:st:rb}. 

We define the initial evaluation environment:
\begin{align*}
  \uparrow &: \vect{\Ctx} \to \Envs \\
  \uparrow^{\epsilon;\cdot} &:= \empenv \\
  \uparrow^{\vGamma; \cdot} &:= \ext(\uparrow^{\vGamma}) \\
  \uparrow^{\vGamma;(\Gamma, x : T)} &:= \ext(\vrho, x, \uparrow^{\intp{T}(\vrho)}(\alpha^{-1}(x)))
                                       \tag{where $\vrho := \uparrow^{\vGamma; \Gamma}$}
\end{align*}
where $\alpha^{-1} : \Var \to \N$ is the opposite effect of $\alpha$.

\begin{definition}
  For $\mtyping t T$, the NbE algorithm is defined to be
  \begin{align*}
    \nbe_{\vGamma}^T(t) := \Rnf_{\alpha}(\downarrow^{\intp{T}(\uparrow^{\vGamma})} (\intp{t}(\uparrow^{\vGamma})))
  \end{align*}
\end{definition}

\section{Completeness and Dependent Candidate Space}

Similar to the simply typed case, our device for proving completeness is some
candidate space. Due to dependent types, of course, we cannot simply port what is in
the simply typed case, because types can embed arbitrary computation. Thus, what we
need here is a little bit more complex, a dependent candidate space. Nevertheless, we
start with two PERs, $\top$ and $\bot$, the extrema of the space.
\begin{mathpar}  
  \inferrule*
  {\forall \alpha, \kappa. \Rnf_{\alpha} (d[\kappa]) = \Rnf_{\alpha} (d'[\kappa])}
  {d \approx d' \in \top}

  \inferrule*
  {\forall \alpha, \kappa. \Rty_{\alpha} (A[\kappa]) = \Rty_{\alpha} (A'[\kappa])}
  {A \approx A' \in \top}

  \inferrule*
  {\forall \alpha, \kappa. \Rne_{\alpha}(c[\kappa]) = \Rne_{\alpha}(c'[\kappa])}
  {c \approx c' \in \bot}
\end{mathpar}
where $\top \subseteq D^{\Nf} \times D^{\Nf}$, $\top \subseteq D \times D$ ($\top$ is
overloaded because one is meant for values and the other is meant for types) and
$\bot \subseteq D^{\Ne} \times D^{\Ne}$. It is easy to show that they are both PERs.

Now we need one PER for each semantic type, plus types represented by neutral
terms. We begin with $\Nd$, which is the easiest case:
\begin{mathpar}
  \inferrule*
  { }
  {\zed \approx \zed \in Nat}

  \inferrule*
  {a \approx b \in Nat}
  {\sud a \approx \sud b \in Nat}

  \inferrule*
  {c \approx c' \in \bot}
  {\uparrow^\Nd(c) \approx \uparrow^\Nd(c') \in Nat}
\end{mathpar}
Here we use $Nat$ to relate the semantic values that are considered equivalent natural
numbers. 
We can prove that $Nat$ is a PER by consider all cases and use the fact that $\bot$ is
also a PER. 

The next one is values related by a neutral type. We virtually can not do
anything but to related them by $\bot$:
\begin{mathpar}
  \inferrule
  {c \approx c' \in \bot}
  {\uparrow^{A}(c) \approx \uparrow^{A'}(c') \in Neu}
\end{mathpar}
Here the neutral type is represented by a domain neutral value $C$. Effectively,
$\uparrow^{A}(c)$ and $\uparrow^{A'}(c')$ are related, if $c$ and $c'$ are also
related by $\bot$. $A$ and $A'$ have no particular relation. It is immediate that
$Neu$ is a PER.

For the remaining PERs for universes, $\squared$ and $Pi$, we define them
simultaneously by using inductive-recursive definitions:
\begin{mathpar}
  \inferrule
  {C \approx C' \in \bot}
  {\uparrow^{A}(C) \approx \uparrow^{A'}(C') \in \Uc_i}

  \inferrule
  { }
  {\Nd \approx \Nd \in \Uc_i}

  \inferrule
  {j < i}
  {\Ud_j \approx \Ud_j \in \Uc_i}

  \inferrule
  {\forall \kappa. A[\kappa] \approx A'[\kappa] \in \Uc_i}
  {\squared A \approx \squared A' \in \Uc_i}

  \inferrule
  {\forall \kappa. A[\kappa] \approx A'[\kappa] \in \Uc_i \\
    \forall \kappa, a \approx a' \in \El_i(A[\kappa]). \intp{T}(\ext(\vrho[\kappa], x, a)) \approx \intp{T'}(\ext(\vrho'[\kappa], x, a')) \in \Uc_i}
  {Pi(A, x.T, \vrho) \approx Pi(A', x.T', \vrho') \in \Uc_i}
\end{mathpar}
In the $Pi$ case, we use $\El_i$ to compute a relation from $A$, which relates
$a$ and $a'$ and is defined below. By letting  $\El_i(A)$ is defined recursively on $A \approx A' \in \Uc_i$, 
we only need to consider the values quantified by $\Uc_i$. 
\begin{align*}
  \El_i(\uparrow^A(C)) &:= Neu \\
  \El_i(\Nd) &:= Nat \\
  \El_i(\Ud_j) &:= \Uc_j \\
  \El_i(\squared A) &:= \{ (a, b) \sep \forall \kappa, n. \tunbox \cdot (n, a[\kappa]) \approx \tunbox
                     \cdot (n, b[\kappa]) \in \El_i(A[\sextt\kappa n]) \} \\
  \El_i(Pi(A, x.T, \vrho)) &:= \{ (a, b) \sep \forall \kappa, a' \approx b' \in \El_i(A[\kappa]). a[\kappa] \cdot a' \approx b[\kappa] \cdot b' \in
                      \El_i(\intp{T}(\ext(\vrho[\kappa], x, a'))) \}
\end{align*}

In the next section, we verify properties of the definitions we just had.

\subsection{Properties of Candidate Space}

\begin{lemma}
  If $d \approx d' \in \top$, then $d[\kappa] \approx d'[\kappa] \in \top$.
\end{lemma}
\begin{lemma}
  If $A \approx A' \in \top$, then $A[\kappa] \approx A'[\kappa] \in \top$.
\end{lemma}
\begin{lemma}
  If $c \approx c' \in \bot$, then $c[\kappa] \approx c'[\kappa] \in \bot$.
\end{lemma}
\begin{proof}
  Immediate due to universal quantification.
\end{proof}

The following two lemmas are proved mutually:
\begin{lemma}\labeledit{lem:dt:el-resp-u}
  If $A \approx A' \in \Uc_i$, then $a \approx b \in \El_i(A)$ and $a \approx b \in
  \El_i(A')$ are equivalent. 
\end{lemma}
\begin{lemma}\labeledit{lem:u-el-per}
  $\Uc_i$ and if $A \approx A' \in \Uc_i$, $\El_i(A)$ are PERs. 
\end{lemma}
\begin{proof}[Proof of \Cref{lem:dt:el-resp-u}]
  We do induction on $A \approx A' \in \Uc_i$. We omit the easy cases.
  \begin{itemize}[label=Case]
  \item $A = \squared A''$, $A' = \squared A'''$,
    \begin{mathpar}
      \inferrule
      {\forall \kappa. A''[\kappa] \approx A'''[\kappa] \in \Uc_i}
      {\squared A'' \approx \squared A''' \in \Uc_i}
    \end{mathpar}
    \begin{align*}
      & \forall \kappa, n. \tunbox \cdot (n, a[\kappa]) \approx \tunbox
        \cdot (n, b[\kappa]) \in \El_i(A''[\sextt\kappa n])
        \tag{assumption} \\
      & \text{assume }\kappa, n, \\
      & A''[\sextt\kappa n] \approx A'''[\sextt\kappa n] \in \Uc_i \\
      &  \tunbox \cdot (n, a[\kappa]) \approx \tunbox \cdot (n, b[\kappa]) \in
        \El_i(A'''[\sextt\kappa n])
        \byIH \\
      & a \approx b \in \El_i(\squared A''')
        \tag{by abstraction}
    \end{align*}
  \item $A = Pi(A'', x.T'', \vrho'')$, $A' = Pi(A''', x.T''', \vrho''')$
    \begin{mathpar}
      \inferrule
      {\forall \kappa. A''[\kappa] \approx A'''[\kappa] \in \Uc_i \\
        \forall \kappa, a' \approx b' \in \El_i(A''[\kappa]).
        \intp{T''}(\ext(\vrho''[\kappa], x, a')) \approx \intp{T'''}(\ext(\vrho'''[\kappa], x, b')) \in \Uc_i}
      {Pi(A'', x.T'', \vrho'') \approx Pi(A''', x.T''', \vrho''') \in \Uc_i}
    \end{mathpar}
    \begin{align*}
      & \forall \kappa, a' \approx b' \in \El_i(A''[\kappa]). a[\kappa] \cdot a' \approx b[\kappa] \cdot b' \in
        \El_i(\intp{T''}(\ext(\vrho''[\kappa], x, a'))) \tag{by assumption} \\
      & \text{assume }\kappa, a' \approx b' \in \El_i(A'''[\kappa]) \\
      & A''[\kappa] \approx A'''[\kappa] \in \Uc_i \\
      & a' \approx b' \in \El_i(A''[\kappa])
        \byIH \\
      & a[\kappa] \cdot a' \approx b[\kappa] \cdot b' \in \El_i(\intp{T''}(\ext(\vrho''[\kappa], x, a')))
      \\
      & \intp{T''}(\ext(\vrho''[\kappa], x, a')) \approx \intp{T'''}(\ext(\vrho'''[\kappa], x, b')) \in \Uc_i \\
      & a[\kappa] \cdot a' \approx b[\kappa] \cdot b' \in \El_i(\intp{T'''}(\ext(\vrho'''[\kappa], x, b')))
        \byIH \\
      & a \approx b \in \El_i(Pi(A''', x.T''', \vrho'''))
        \tag{by abstraction}
    \end{align*}
  \end{itemize}
  The other direction is proved symmetrically. 
\end{proof}

\begin{proof}[Proof of \Cref{lem:u-el-per}]
  Again we only consider interesting cases. First consider symmetricity. We proceed by
  induction on $A \approx A' \in \Uc_i$:
  \begin{itemize}[label=Case]
  \item $A = Pi(A'', x.T'', \vrho'')$, $A' = Pi(A''', x.T''', \vrho''')$,
    \begin{mathpar}
      \inferrule
      {\forall \kappa. A''[\kappa] \approx A'''[\kappa] \in \Uc_i \\
        \forall \kappa, a' \approx b' \in \El_i(A''[\kappa]).
        \intp{T''}(\ext(\vrho''[\kappa], x, a')) \approx \intp{T'''}(\ext(\vrho'''[\kappa], x, b')) \in \Uc_i}
      {Pi(A'', x.T'', \vrho'') \approx Pi(A''', x.T''', \vrho''') \in \Uc_i}
    \end{mathpar}
    \begin{align*}
      & \text{assume }\kappa, a' \approx b' \in \El_i(A'''[\kappa]) \\
      & A''[\kappa] \approx A'''[\kappa] \in \Uc_i \\
      & a' \approx b' \in \El_i(A''[\kappa])
        \tag{by \Cref{lem:dt:el-resp-u}} \\
      & \intp{T''}(\ext(\vrho''[\kappa], x, a')) \approx \intp{T'''}(\ext(\vrho'''[\kappa], x, b')) \in \Uc_i \\
      & \intp{T'''}(\ext(\vrho'''[\kappa], x, b')) \approx \intp{T''}(\ext(\vrho''[\kappa], x, a')) \in \Uc_i
        \byIH  \\
      & Pi(A''', x.T''', \vrho''') \approx Pi(A'', x.T'', \vrho'') \in \Uc_i
    \end{align*}

    Next we consider symmetry of $\El_i(Pi(A'', x.T'', \vrho''))$. 
    \begin{align*}
      & \forall \kappa, a' \approx b' \in \El_i(A''[\kappa]). a[\kappa] \cdot a' \approx
        b[\kappa] \cdot b' \in \El_i(\intp{T''}(\ext(\vrho''[\kappa], x, a')))
        \tag{by assumption} \\
      & \text{assume }\kappa, b' \approx a' \in \El_i(A''[\kappa]) \\
      & a' \approx b' \in \El_i(A''[\kappa])
        \byIH \\
      & a[\kappa] \cdot a' \approx b[\kappa] \cdot b' \in \El_i(\intp{T''}(\ext(\vrho''[\kappa], x, a'))) \\
      & b[\kappa] \cdot b' \approx a[\kappa] \cdot a' \in \El_i(\intp{T''}(\ext(\vrho''[\kappa], x, a')))
        \byIH \\
      & b[\kappa] \cdot b' \approx a[\kappa] \cdot a' \in \El_i(\intp{T'''}(\ext(\vrho'''[\kappa], x, b')))
        \tag{by \Cref{lem:dt:el-resp-u}} \\
      & b \approx a \in \El_i(Pi(A''', x.T''', \vrho'''))
        \tag{by abstraction} \\
      & b \approx a \in \El_i(Pi(A'', x.T'', \vrho''))
        \tag{by \Cref{lem:dt:el-resp-u}}
    \end{align*}
  \end{itemize}

  Now we consider transitivity. We should show
  \begin{itemize}
  \item If $A \approx A' \in \Uc_i$ and $A' \approx A'' \in \Uc_i$, then $A \approx A''
    \in \Uc_i$. 
  \item If $a \approx a' \in \El_i(A)$ and $a' \approx a'' \in \El_i(A)$, then $a
    \approx a'' \in \El_i(A)$. 
  \end{itemize}
  We do induction on $A \approx A' \in \Uc_i$ and invert $A' \approx A'' \in \Uc_i$. We only consider the the interesting
  cases. 
  \begin{itemize}[label=Case]
  \item  $A = Pi(A_1, x.T_1, \vrho_1)$, $A' = Pi(A_2, x.T_2, \vrho_2)$, $A'' = Pi(A_3, x.T_3, \vrho_3)$
    \begin{mathpar}
      \inferrule
      {\forall \kappa. A_1[\kappa] \approx A_2[\kappa] \in \Uc_i \\
        \forall \kappa, b \approx b' \in \El_i(A_1[\kappa]).
        \intp{T_1}(\ext(\vrho_1[\kappa], x, b)) \approx \intp{T_1}(\ext(\vrho_2[\kappa], x, b')) \in \Uc_i}
      {Pi(A_1, x.T_1, \vrho_1) \approx Pi(A_2, x.T_2, \vrho_2) \in \Uc_i}

      \inferrule
      {\forall \kappa. A_2[\kappa] \approx A_3[\kappa] \in \Uc_i \\
        \forall \kappa, b \approx b' \in \El_i(A_2[\kappa]).
        \intp{T_2}(\ext(\vrho_2[\kappa], x, b)) \approx \intp{T_3}(\ext(\vrho_3[\kappa], x, b')) \in \Uc_i}
      {Pi(A_2, x.T_2, \vrho_2) \approx Pi(A_3, x.T_3, \vrho_3) \in \Uc_i}
    \end{mathpar}
    \begin{align*}
      & \text{assume }\kappa, b \approx b' \in \El_i(A_1[\kappa]) \\
      & \intp{T_1}(\ext(\vrho_1[\kappa], x, b)) \approx \intp{T_1}(\ext(\vrho_2[\kappa], x, b')) \in \Uc_i \\
      & b \approx b' \in \El_i(A_2[\kappa])
        \tag{by \Cref{lem:dt:el-resp-u}}\\
      & b' \approx b' \in \El_i(A_2[\kappa])
        \byIH \\
      & \intp{T_2}(\ext(\vrho_2[\kappa], x, b')) \approx \intp{T_3}(\ext(\vrho_3[\kappa], x, b')) \in \Uc_i \\
      & \intp{T_1}(\ext(\vrho_1[\kappa], x, b)) \approx \intp{T_3}(\ext(\vrho_3[\kappa], x, b')) \in \Uc_i
        \byIH \\
      & Pi(A_1, x.T_1, \vrho_1) \approx Pi(A_3, x.T_3, \vrho_3) \in \Uc_i
        \tag{by abstraction}
    \end{align*}

    We move on to the transitivity of $Pi(A_1, x.T_1, \vrho_1)$.
    \begin{align*}
      & \text{assume }\kappa, b \approx b' \in \El_i(A_1[\kappa]) \\
      & a[\kappa] \cdot b \approx a'[\kappa] \cdot b' \in \El_i(\intp{T_1}(\ext(\vrho_1[\kappa], x, b)))
      \\
      & b' \approx b' \in \El_i(A_1[\kappa])
      \byIH \\
      & a'[\kappa] \cdot b' \approx a''[\kappa] \cdot b' \in \El_i(\intp{T_1}(\ext(\vrho_1[\kappa], x, b'))) \\
      & a'[\kappa] \cdot b' \approx a''[\kappa] \cdot b' \in \El_i(\intp{T_1}(\ext(\vrho_1[\kappa], x, b)))
        \tag{by \Cref{lem:dt:el-resp-u}}\\
      & a[\kappa] \cdot b \approx a''[\kappa] \cdot b' \in \El_i(\intp{T_1}(\ext(\vrho_1[\kappa], x, b)))
        \byIH \\
      & a \approx a'' \in \El_i(Pi(A_1, x.T_1, \vrho))
    \end{align*}
  \end{itemize}
\end{proof}
The previous two lemmas justify that $\El_i(A)$ is define if $A \approx A' \in \Uc_i$ or
$A' \approx A \in \Uc_i$ holds. Thus, notationally we can write $A \in \Uc_i$ to express
the well-defined condition for $\El_i(A)$.

The following two lemmas are needed and can be proved just like in the simply typed
case. 
\begin{lemma}\labeledit{lem:dt:L-intp-vsigma-vrho}
  $\Ltotal{\intp{\vsigma}(\vrho)}n = \Ltotal\vrho{\Ltotal\vsigma n}$
\end{lemma}

\begin{lemma}\labeledit{lem:dt:trunc-intp-vsigma-vrho}
  $\trunc{\intp{\vsigma}(\vrho)}n = \intp{\trunc\vsigma n}(\trunc\vrho{\Ltotal\vsigma n})$
\end{lemma}

\begin{lemma}[Monotonicity]\labeledit{lem:dt:u-el-resp-mon}
  \begin{itemize}
  \item If $A \approx A' \in \Uc_i$, then $A[\kappa] \approx A'[\kappa] \in \Uc_i$.
  \item If $a \approx b \in \El_i(A)$, then $a[\kappa] \approx b[\kappa] \in
    \El_i(A[\kappa])$. 
  \end{itemize}
\end{lemma}
\begin{proof}
  Induction on $A \approx A' \in \Uc_i$. All cases are discharged due to the universal
  quantification over \UnMoTs.
\end{proof}

Next, we establish that the PER model is cumulative. We first show that the model is
cumulative by one step. 
\begin{lemma}[Cumulativity and lowering]\labeledit{lem:dt:per-cumu}
  If $A \approx A' \in \Uc_i$,
  \begin{itemize}
  \item then $A \approx A' \in \Uc_{i+1}$;
  \item if $a \approx b \in \El_{i+1}(A)$, then $a \approx b \in \El_i(A)$.
  \item if $a \approx b \in \El_i(A)$, then $a \approx b \in \El_{i+1}(A)$.
  \end{itemize}
\end{lemma}
\begin{proof}
  We proceed by induction on $A \approx A' \in \Uc_i$. Most cases are
  straightforward. However, in the cases of $Pi$ and $\Lambda$, we are given
  $a' \approx b' \in \El_{i + 1}(A')$ for some $A'$ and need to obtain
  $a' \approx b' \in \El_{i}(A')$ to invoke the premise and use IH. This therefore
  requires us to prove a \emph{lowering} lemma simultaneously, which lowers the
  universe level of $\El$. This allows us to conclude the goal. 
\end{proof}

\begin{lemma}[Cumulativity]
  If $A \approx A' \in \Uc_i$ and $i \le k$,
  \begin{itemize}
  \item then $A \approx A' \in \Uc_k$;
  \item if $a \approx b \in \El_i(A)$, then $a \approx b \in \El_k(A)$.
  \end{itemize}
\end{lemma}
\begin{proof}
  $i \le k$ means $k = i + n$ for some $n$. We do induction on $n$. 
\end{proof}
This lemma shows that $\Uc_i$ indeed models the unvierses cumulatively and thus further
justifies that we can take the limit of $\Uc_i$ and $\El_i$. We write $\Uc$ for
$\Uc_{\infty}$ and $\El$ for $\El_{\infty}$. All previous lemmas about $\Uc_i$ and $\El_i$
hold for $\Uc$ and $\El$.  By construction, we have
\begin{lemma}[subsumption of limits]
  \begin{itemize}
  \item If $A \approx A' \in \Uc_i$, then $A \approx A' \in \Uc$. 
  \item If $a \approx b \in \El_i(A)$, then $a \approx b \in \El(A)$.
  \end{itemize}
\end{lemma}

Now we have $\Uc$ to relate two types and $\El$ to relate two values. We then further
extend this relation to context stack as follows:
\begin{mathpar}
  \inferrule
  { }
  {\vDash \epsilon; \cdot \approx \epsilon; \cdot}

  \inferrule
  {\vDash \vGamma \approx \vGamma'}
  {\vDash \vGamma; \cdot \approx \vGamma'; \cdot}

  \inferrule
  {\vDash \vGamma; \Gamma \approx \vGamma'; \Gamma' \\
  \forall \vrho \approx \vrho' \in \intp{\vGamma; \Gamma}. \intp{T}(\vrho) \approx
  \intp{T'}(\vrho') \in \Uc}
  {\vDash \vGamma; (\Gamma, x : T) \approx \vGamma'; (\Gamma', x : T')}
\end{mathpar}
By induction-recursion, we define the following relation relating two evaluation environments:
\begin{align*}
  \intp{\epsilon; \cdot}
  &:= \{(\vrho, \vrho')\} \\
  \intp{\vGamma; \cdot}
  &:= \{(\vrho, \vrho') \sep \trunc\vrho 1 \approx \trunc{\vrho'} 1 \in \intp{\vGamma} \tand k =
    k' \text{ where }(k, \_) := \vrho(0) \tand (k', \_) := \vrho'(0) \} \\
  \intp{\vGamma; (\Gamma, x : T)}
  &:= \{ (\vrho, \vrho') \sep \drop(\vrho, x) \approx \drop(\vrho', x) \in \intp{\vGamma; \Gamma} \tand
    \rho(x) \approx \rho'(x) \in \El(\intp{T}(\drop(\vrho, x))) \\
  & \qquad \text{ where }(\_, \rho) := \vrho(0) \tand (\_, \rho') := \vrho'(0) \}
\end{align*}
where $\intp{\vGamma}$ is defined by recursion on $\vDash \vGamma \approx \vGamma'$.

Interpretation of context stacks remains closed under modal transformations:
\begin{lemma}\labeledit{lem:dt:intpvg-resp-mon}
  Given $\vDash \vGamma \approx \vGamma'$, if
  $\vrho \approx \vrho' \in \intp{\vGamma}$, then
  $\vrho[\kappa] \approx \vrho'[\kappa] \in \intp{\vGamma}$.
\end{lemma}
\begin{proof}
  Immediate by induction on $\vDash \vGamma \approx \vGamma'$ and \Cref{lem:dt:u-el-resp-mon}. 
\end{proof}

\begin{lemma}\labeledit{lem:dt:intpvg-resp-eq}
  \begin{itemize}
  \item $\vDash \vGamma \approx \vGamma'$ is a PER.
  \item If $\vDash \vGamma \approx \vGamma'$ and
    $\vrho \approx \vrho' \in \intp{\vGamma}$, then
    $\vrho \approx \vrho' \in \intp{\vGamma'}$.
  \end{itemize}
\end{lemma}
\begin{proof}
  Immediate by induction on $\vDash \vGamma \approx \vGamma'$ and \Cref{lem:dt:el-resp-u}.
\end{proof}
\begin{lemma}
  If $\vrho \approx \vrho' \in \intp{\vGamma}$, then $\Ltotal\vrho n = \Ltotal{\vrho'} n$.
\end{lemma}
\begin{lemma}
  If $\vrho \approx \vrho' \in \intp{\vGamma}$, then $\trunc\vrho n \approx
  \trunc{\vrho'} n \in \intp{\trunc \vGamma n}$.
\end{lemma}

We then can define the semantic judgments:
\begin{definition}
  We define semantic typing and semantic typing equivalence as follows:
  \begin{align*}
    \vDash \vGamma &:=\ \vDash \vGamma \approx \vGamma \\
    \msemtyeq{t}{t'}{T} &:= \vDash \vGamma \tand 
                          \forall \vrho \approx \vrho' \in
                          \intp{\vGamma}. \intp{T}(\vrho) \approx \intp{T}(\vrho') \in \Uc \tand \intp{t}(\vrho) \approx
                          \intp{t'}(\vrho') \in \El(\intp{T}(\vrho)) \\
    \msemtyeq{\vsigma}{\vsigma'}{\vDelta} &:=\ \vDash \vGamma \tand \vDash \vDelta \tand
                                            \forall \vrho \approx \vrho' \in \intp{\vGamma}. \intp{\vsigma}(\vrho) \approx
                                            \intp{\vsigma'}(\vrho') \in \intp{\vDelta} \\
    \msemtyp{t}{T} &:= \msemtyeq{t}{t}{T} \\
    \msemtyp{\vsigma}{\vDelta} &:= \msemtyeq{\vsigma}{\vsigma}{\vDelta}
  \end{align*}
\end{definition}
Similar to the simply typed case, naturality equalities are also present in the
semantic judgments. 

\subsection{Realizability}

Similar to the simply typed case, we need realizability to convert the PER model to
the $\top$ PER so that it serves as the final step of the completeness proof. 

\begin{lemma}\labeledit{lem:dt:varbot}
  $z \approx z \in \bot$
\end{lemma}
\begin{proof}
  Immediate by definition. 
\end{proof}

\begin{theorem}[realizability]
  Given $A \approx A' \in \Uc_i$,
  \begin{itemize}
  \item If $c \approx c' \in \bot$, then $\uparrow^A(c) \approx \uparrow^{A'}(c') \in
    \El_i(A)$;
  \item If $a \approx a' \in \El_i(A)$, then $\downarrow^A(a) \approx \downarrow^{A'}(a')
    \in \top$;
  \item $A \approx A' \in \top$.
  \end{itemize}
\end{theorem}
\begin{proof}
  We do induction on $A \approx A' \in \Uc_i$,
  \begin{itemize}[label=Case]
  \item $A = \uparrow^{A''}(C)$, $A' = \uparrow^{A'''}(C')$, 
    \begin{mathpar}
      \inferrule
      {C \approx C' \in \bot}
      {\uparrow^{A''}(C) \approx \uparrow^{A'''}(C') \in \Uc_i}
    \end{mathpar}
    Immediate because the first statement is just identity, and the second and the
    third statements
    follow immediately by expanding definition.
    
  \item $A = A' = \Nd$, 
    \begin{mathpar}
      \inferrule
      { }
      {\Nd \approx \Nd \in \Uc_i}
    \end{mathpar}
    Only need to check the second statement. We further do induction on $a \approx
    a' \in Nat$ and the conclusion follows by expanding the definitions.
    
  \item $A = A' = \Ud_j$, 
    \begin{mathpar}
      \inferrule
      {j < i}
      {\Ud_j \approx \Ud_j \in \Uc_i}
    \end{mathpar}
    All statements are immediate(apply IH when necessary).
    
  \item $A = \squared A''$, $A' = \squared A'''$,
    \begin{mathpar}
      \inferrule
      {\forall \kappa. A''[\kappa] \approx A'''[\kappa] \in \Uc_i}
      {\squared A'' \approx \squared A''' \in \Uc_i}
    \end{mathpar}
    \begin{itemize}[label=Subcase]
    \item
      \begin{align*}
        H_1: &\ c \approx c' \in \bot
               \tag{by assumption} \\
             &\ \text{assume }\kappa, n, \alpha, \kappa' \\
             &\ \Rne_\alpha (c[\kappa \circ (\trunc{\kappa'} n)]) = \Rne_\alpha (c'[\kappa \circ (\trunc{\kappa'} n)])
               \tag{by $H_1$} \\
             &\ \Rne_\alpha (\tunbox(\Ltotal{\kappa'}n, c[\kappa \circ (\trunc{\kappa'} n)])) =
               \Rne_\alpha (\tunbox(\Ltotal{\kappa'}n, c'[\kappa \circ (\trunc{\kappa'} n)]))
               \tag{by congruence} \\
             &\ \Rne_\alpha (\tunbox(n, c[\kappa])[\kappa']) = \Rne_\alpha (\tunbox(n, c'[\kappa])[\kappa'])
               \tag{by definition} \\
             &\ \tunbox(n, c[\kappa]) \approx \tunbox(n, c'[\kappa]) \in \bot
               \tag{by abstraction} \\
             &\ \uparrow^{A''[\sextt\kappa n]}(\tunbox(n, c[\kappa]))
               \approx \uparrow^{A'''[\sextt\kappa k]}(\tunbox(n, c'[\kappa]))
               \in \El_i(A''[\sextt\kappa n])
               \byIH \\
             &\ \tunbox \cdot (n, \uparrow^{\squared A''}(c)[\kappa])
               \approx \tunbox \cdot (n, \uparrow^{\squared A'''}(c')[\kappa])
               \in \El_i(A''[\sextt\kappa n])
               \tag{by definition} \\
             &\ \uparrow^{\squared A''}(c) \approx \uparrow^{\squared A'''}(c') \in
               \El_i(\squared A'')
               \tag{by abstraction} 
      \end{align*}
      
    \item
      \begin{align*}
        H_2: &\ a \approx a' \in \El_i(\squared A'')
               \tag{by assumption} \\
             &\ \text{assume }\alpha, \kappa \\
             &\ \tunbox \cdot (1, a[\kappa]) \approx \tunbox \cdot (1, a'[\kappa]) \in 
               \El_i(A''[\sextt\kappa 1])
               \tag{by $H_2$} \\
        H_3: &\ \downarrow^{A''[\sextt\kappa 1]}(\tunbox \cdot (1, a[\kappa]))
               \approx \downarrow^{A'''[\sextt\kappa 1]}(\tunbox \cdot (1, a'[\kappa]))
               \in \top
               \byIH \\
             &\ \Rnf_\alpha(\downarrow^{A''[\sextt\kappa 1]}(\tunbox \cdot (1, a[\kappa])))
               = \Rnf_\alpha(\downarrow^{A'''[\sextt\kappa 1]}(\tunbox \cdot (1, a'[\kappa])))
               \tag{by $H_3$} \\
             &\ \Rnf_\alpha(\downarrow^{\squared A''[\kappa]}(a[\kappa]))
               = \Rnf_\alpha(\downarrow^{\squared A'''[\kappa]}(a'[\kappa]))
               \tag{by congruence and definition} \\
             &\ \Rnf_\alpha(\downarrow^{\squared A''}(a)[\kappa])
               = \Rnf_\alpha(\downarrow^{\squared A'''}(a')[\kappa])
               \tag{by definition} \\
             &\ \downarrow^{\squared A''}(a) \approx \downarrow^{\squared A''}(a')
               \in \top
               \tag{by abstraction}
      \end{align*}
      
    \item
      \begin{align*}
        &\ \text{assume }\alpha, \kappa \\
        &\ A''[\sextt\kappa 1] \approx A''[\sextt\kappa 1] \in \Uc_i \\
        H_4: &\ A''[\sextt\kappa 1] \approx A'''[\sextt\kappa 1] \in \top
               \byIH \\
        &\ \Rty_\alpha (A''[\sextt\kappa 1]) = \Rty_\alpha (A'''[\sextt\kappa 1])
          \tag{by $H_4$} \\
        &\ \Rty_\alpha (\squared A''[\kappa]) = \Rty_\alpha (\squared A'''[\kappa])
          \tag{by congruence} \\
        &\ \squared A'' \approx \squared A''' \in \top
      \end{align*}
    \end{itemize}
    
  \item $A = Pi(A'', x.T'', \vrho'')$, $A' = Pi(A''', x.T''', \vrho''')$, 
    \begin{mathpar}
      \inferrule
      {\forall \kappa. A''[\kappa] \approx A'''[\kappa] \in \Uc_i \\
        \forall \kappa, a' \approx b' \in \El_i(A''[\kappa]).
        \intp{T''}(\ext(\vrho''[\kappa], x, a')) \approx \intp{T'''}(\ext(\vrho'''[\kappa], x, b')) \in \Uc_i}
      {Pi(A'', x.T'', \vrho'') \approx Pi(A''', x.T''', \vrho''') \in \Uc_i}
    \end{mathpar}
    \begin{itemize}[label=Subcase]
    \item
      \begin{align*}
        H_5:&\ c \approx c' \in \bot
              \tag{by assumption} \\
            &\ \text{assume } \kappa, b \approx b' \in \El_i(A''[\kappa]), \alpha, \kappa'
        \\
            &\ \Rne_\alpha (c[\kappa \circ \kappa']) = \Rne_\alpha (c'[\kappa \circ
              \kappa'])
              \tag{by $H_5$} \\
        H_6: &\ \downarrow^{A''[\kappa]}(b) \approx \downarrow^{A'''[\kappa]}(b') \in \top
              \byIH \\
            &\ \Rnf_\alpha(\downarrow^{A''[\kappa]}(b)[\kappa']) =
              \Rnf_\alpha(\downarrow^{A'''[\kappa]}(b')[\kappa'])
              \tag{by $H_6$} \\
            &\ \Rne_\alpha (c[\kappa][\kappa'])\ \Rnf_\alpha(\downarrow^{A''[\kappa]}(b)[\kappa'])
              = \Rne_\alpha (c'[\kappa][\kappa'])\
              \Rnf_\alpha(\downarrow^{A'''[\kappa]}(b')[\kappa'])
              \tag{by congruence} \\
            &\ \Rne_\alpha ((c[\kappa]\ \downarrow^{A''[\kappa]}(b))[\kappa'])
              = \Rne_\alpha ((c'[\kappa]\ \downarrow^{A'''[\kappa]}(b'))[\kappa'])
              \tag{by definition} \\
            &\ c[\kappa]\ \downarrow^{A''[\kappa]}(b) \approx
              c'[\kappa]\ \downarrow^{A'''[\kappa]}(b')
              \in \bot
              \tag{by abstraction} \\
            &\ \uparrow^{Pi(A'', x.T'', \vrho'')}(c)[\kappa] \cdot b \approx
              \uparrow^{Pi(A''', x.T''', \vrho''')}(c')[\kappa] \cdot b'
              \in \bot
              \tag{by definition} \\
            &\ \uparrow^{Pi(A'', x.T'', \vrho'')}(c)[\kappa] \cdot b \approx
              \uparrow^{Pi(A''', x.T''', \vrho''')}(c')[\kappa] \cdot b'
              \in \El_i(\intp{T''}(\ext(\vrho[\kappa], x, b)))
              \byIH \\
            &\ \uparrow^{Pi(A'', x.T'', \vrho'')}(c)\approx
              \uparrow^{Pi(A''', x.T''', \vrho''')}(c')
              \in \El_i(Pi(A'', x.T'', \vrho''))
      \end{align*}
      
    \item
      \begin{align*}
        H_7: &\ a \approx a' \in \El_i(Pi(A'', x.T'', \vrho''))
               \tag{by assumption} \\
             &\ \text{assume }\alpha, \kappa \\
             &\ z \approx z \in \bot
               \tag{for some $z$, by \Cref{lem:dt:varbot}} \\
             &\ \uparrow^{A''[\kappa]}(z) \approx \uparrow^{A'''[\kappa]}(z) \in
               \El_i(A''[\kappa])
               \byIH \\
             &\ a[\kappa] \cdot \uparrow^{A''[\kappa]}(z)
               \approx a'[\kappa] \cdot \uparrow^{A'''[\kappa]}(z)
               \in \El_i(\intp{T''}(B))
               \tag{by $H_7$, where $B := \intp{T''}(\ext(\vrho''[\kappa], x, \uparrow^{A''[\kappa]}(z)))$} \\
        H_8: &\ \downarrow^{B}(a[\kappa] \cdot \uparrow^{A''[\kappa]}(z))
               \approx \downarrow^{B'}(a'[\kappa] \cdot \uparrow^{A'''[\kappa]}(z))
               \in \top
               \tag{by IH, where $B' := \intp{T'''}(\ext(\vrho'''[\kappa], x, \uparrow^{A'''[\kappa]}(z)))$} \\
             &\ \Rnf_{\alpha[z \mapsto x]}(\downarrow^{B}(a[\kappa] \cdot \uparrow^{A''[\kappa]}(z)))
               =
               \Rnf_{\alpha[z \mapsto x]}(\downarrow^{B'}(a'[\kappa]
               \cdot \uparrow^{A'''[\kappa]}(z)))
               \tag{by $H_8$} \\
             &\ \Rnf_\alpha(\downarrow^{Pi(A''[\kappa],x.T'', \vrho''[\kappa])}(a[\kappa]))
               =
               \Rnf_\alpha(\downarrow^{Pi(A'''[\kappa],x.T''', \vrho'''[\kappa])}(a'[\kappa]))
               \tag{by congruence} \\
             &\ \Rnf_\alpha(\downarrow^{Pi(A'', x.T'', \vrho'')}(a)[\kappa])
               = \Rnf_\alpha(\downarrow^{Pi(A''', x.T''', \vrho''')}(a')[\kappa])
               \tag{by definition} \\
             &\ \downarrow^{Pi(A'', x.T'', \vrho'')}(a) \approx \downarrow^{Pi(A''', x.T''', \vrho''')}(a')
               \in \top
               \tag{by abstraction}
      \end{align*}
      
    \item
      \begin{align*}
        &\ \text{assume }\alpha, \kappa \\
        &\ z \approx z \in \bot
          \tag{for some $z$, by \Cref{lem:dt:varbot}} \\
        &\ \uparrow^{A''[\kappa]}(z) \approx \uparrow^{A'''[\kappa]}(z) \in
          \El_i(A''[\kappa])
          \byIH \\
        &\ B \approx B' \in \Uc_i
        \tag{where $B := \intp{T''}(\ext(\vrho''[\kappa], x,
          \uparrow^{A''[\kappa]}(z)))$ and $B' := \intp{T'''}(\ext(\vrho'''[\kappa], x, \uparrow^{A'''[\kappa]}(z)))$} \\
        &\ \Rty_{{\alpha[z \mapsto x]}} (B) = \Rty_{\alpha[z \mapsto x]}(B')
          \byIH \\
        &\ \Rty_\alpha(\Pi(A'', x.T'', \vrho'')[\kappa]) = \Rty_\alpha(\Pi(A''', x.T''', \vrho''')[\kappa])
          \tag{by congruence} \\
        &\ \Pi(A'', x.T'', \vrho'') \approx \Pi(A''', x.T''', \vrho''') \in \top
          \tag{by abstraction}
      \end{align*}
    \end{itemize}
  \end{itemize}
\end{proof}

\subsection{Completeness}

\subsubsection{Well-formed Context Stacks}

\begin{lemma}
  \begin{mathpar}
    \inferrule
    { }
    {\vDash \epsilon; \cdot \approx \epsilon; \cdot}

    \inferrule
    {\vDash \vGamma \approx \vDelta}
    {\vDash \vGamma; \cdot \approx \vDelta; \cdot}

    \inferrule
    {\vDash \vGamma; \Gamma \approx \vDelta; \Delta  \\ \msemtyeq[\vGamma;
      \Gamma]{T}{T'}{\Se_i}}
    {\vDash \vGamma; (\Gamma, x : T) \approx \vDelta; (\Delta, x : T')}
  \end{mathpar}
\end{lemma}
\begin{proof}
  Immediate. 
\end{proof}

\subsubsection{Conversions}

\begin{lemma}
  \begin{mathpar}
    \inferrule
    {\msemtyeq{t}{t'}{T} \\ \msemtyeq{T}{T'}{\Se_i}}
    {\msemtyeq{t}{t'}{T'}}
  \end{mathpar}
\end{lemma}
\begin{proof}
  Assuming $\vrho \approx \vrho' \in \intp{\vGamma}$, we know $\intp{t}(\vrho) \approx
  \intp{t'}(\vrho') \in \El(\intp{T}(\vrho))$ and $\intp{T}(\vrho) \approx \intp{T'}(\vrho) \in
  \Uc$. We obtain $\intp{t}(\vrho) \approx
  \intp{t'}(\vrho') \in \El(\intp{T'}(\vrho))$ by \Cref{lem:dt:el-resp-u}.

  Naturality equalities are direct consequences. 
\end{proof}

\begin{lemma}
  \begin{mathpar}
    \inferrule
    {\msemtyeq{\vsigma}{\vsigma'}{\vDelta} \\ \vDash \vDelta \approx \vDelta'}
    {\msemtyeq{\vsigma}{\vsigma'}{\vDelta'}}
  \end{mathpar}
\end{lemma}
\begin{proof}
  Immediate by PER and \Cref{lem:dt:intpvg-resp-eq}. 
\end{proof}

\subsubsection{Substitution}

All substitution related rules are immediate as they do essentially the same thing as the simply
typed case.
\begin{lemma}
  \begin{mathpar}
    \inferrule
    {\msemtyeq[\vDelta]{t}{t'}T \\ \msemtyeq{\vsigma}{\vsigma'}{\vDelta}}
    {\msemtyeq{t[\vsigma]}{t'[\vsigma']}{T[\vsigma]}}
  \end{mathpar}
\end{lemma}
\begin{lemma}
  \begin{mathpar}
    \inferrule
    {\vDash \vGamma}
    {\msemtyeq{\vect I}{\vect I}{\vGamma}}

    \inferrule
    {\vDash \vGamma; (\Gamma, x : T)}
    {\msemtyeq[\vGamma; (\Gamma, x : T)]{\wk_x}{\wk_x}{\vGamma; \Gamma}}

    \inferrule
    {\msemtyeq{\vsigma}{\vsigma'}{\vGamma'; \Gamma} \\ \msemtyp[\vGamma';
      \Gamma]{T}{\Se_i} \\ \msemtyeq{t}{t'}{T[\vsigma]}}
    {\msemtyeq{\vsigma, t/x}{\vsigma, t'/x}{\vGamma';(\Gamma, x : T)}}

    \inferrule
    {\msemtyeq{\vsigma}{\vsigma'}{\vDelta} \\ \vDash \vGamma; \vGamma' \\ |\vGamma'| = n}
    {\msemtyeq[\vGamma; \vGamma']{\sextt\vsigma n}{\sextt{\vsigma'} n}{\vDelta; \cdot}}

    \inferrule
    {\msemtyeq[\vGamma']{\vsigma}{\vsigma'}{\vGamma''} \\ \msemtyeq{\vdelta}{\vdelta'}{\vGamma'}}
    {\msemtyeq{\vsigma \circ \vdelta}{\vsigma' \circ \vdelta'}{\vGamma''}}
  \end{mathpar}
\end{lemma}
\begin{lemma}
  \begin{mathpar}
    \inferrule
    {\msemtyp t T}
    {\msemtyeq{t[\vect I]}{t}{T}}

    \inferrule
    {\msemtyp[\vGamma']{\vsigma}{\vGamma''} \\ \msemtyp[\vGamma]{\vdelta}{\vGamma'} \\
      \msemtyp[\vGamma'']{t}{T}}
    {\msemtyeq{t[\vsigma \circ \vdelta]}{t[\vsigma][\vdelta]}{T[\vsigma \circ
        \vdelta]}}
  \end{mathpar}
\end{lemma}
\begin{lemma}
  \begin{mathpar}
    \inferrule
    {\msemtyp{\vsigma}{\vDelta}}
    {\msemtyeq{\vsigma \circ \vect I}{\vsigma}{\vDelta}}

    \inferrule
    {\msemtyp{\vsigma}{\vDelta}}
    {\msemtyeq{\vect I \circ \vsigma}{\vsigma}{\vDelta}}

    \inferrule
    {\msemtyp[\vGamma'']{\vsigma''}{\vGamma'''} \\ \msemtyp[\vGamma']{\vsigma'}{\vGamma''} \\ \msemtyp{\vsigma}{\vGamma'}}
    {\msemtyeq{(\vsigma'' \circ \vsigma') \circ \vsigma}{\vsigma'' \circ (\vsigma' \circ \vsigma)}{\vGamma'''}}
  \end{mathpar}
\end{lemma}
\begin{lemma}
  \begin{mathpar}
    \inferrule
    {\msemtyp[\vGamma']{\vsigma}{\vGamma''; \Gamma} \\ \msemtyp[\vGamma';
      \Gamma]{T}{\Se_i} \\
      \msemtyp[\vGamma']{t}{T[\vsigma]} \\ \msemtyp{\vdelta}{\vGamma'}}
    {\msemtyeq{\vsigma, t/x \circ \vdelta}{(\vsigma \circ \vdelta), t[\vdelta]/x}{\vGamma''; (\Gamma, x : T)}}
    
    \inferrule
    {\msemtyp[\vGamma']{\vsigma}{\vGamma; \Gamma} \\ \msemtyp[\vGamma'; \Gamma]{T}{\Se_i} \\ \msemtyp[\vGamma']{t}{T[\vsigma]}}
    {\msemtyeq[\vGamma']{\wk_x \circ (\vsigma, t/x)}{\vsigma}{\vGamma; \Gamma}}

    \inferrule
    {\msemtyp[\vGamma']{\vsigma}{\vGamma; (\Gamma, x : T)}}
    {\msemtyeq[\vGamma']{\vsigma}{(\wk_x \circ \vsigma), x[\vsigma]/x}{\vGamma; (\Gamma, x : T)}}

    \inferrule
    {\msemtyp{\vsigma}{\vDelta; \cdot} \\ |\vDelta| > 0}
    {\msemtyeq{\vsigma}{\sextt{\trunc \vsigma 1}{\Ltotal\vsigma 1}}{\vDelta; \cdot}}
  \end{mathpar}
\end{lemma}

\begin{lemma}
  \begin{mathpar}
    \inferrule
    {\msemtyp{\vsigma}{\vGamma'} \\ \msemtyp[\vGamma'']{\vdelta}{\vGamma;\vDelta} \\
      |\vDelta| = n}
    {\msemtyeq[\vGamma'']{(\sextt\vsigma n) \circ \vdelta}{\sextt{(\vsigma \circ \trunc\vdelta n)}
        {\Ltotal\vdelta n} }{\vGamma';\cdot}}
  \end{mathpar}
\end{lemma}

\subsubsection{Variables}

\begin{lemma}
  \begin{mathpar}
    \inferrule
    {\vDash \vGamma; \Gamma \\ x : T \in \vGamma; \Gamma}
    {\msemtyeq[\vGamma; \Gamma]x x T}

    \inferrule
    {\vDash \vGamma; (\Gamma, x : T) \\ y : T' \in \vGamma; \Gamma}
    {\msemtyeq[\vGamma; \Gamma]{y[\wk_x]}y{T'[\wk_x]}}

    \inferrule
    {\msemtyp{\vsigma}{\vGamma'; \Gamma} \\ \msemtyp[\vGamma';\Gamma]{T}{\Se_i}  \\ \msemtyp t T[\vsigma]}
    {\msemtyeq{x[\vsigma, t/x]}{t}{T[\vsigma]}}
    
    \inferrule
    {\msemtyp{\vsigma}{\vGamma'; \Gamma} \\ \msemtyp[\vGamma';\Gamma]{T'}{\Se_i} \\ \msemtyp t T'[\vsigma] \\ y : T \in \vGamma';\Gamma}
    {\msemtyeq{y[\vsigma, t/x]}{y[\vsigma]}{T[\vsigma]}}
  \end{mathpar}
\end{lemma}
\begin{proof}
  We rules are immediate by following the definitions. 
\end{proof}

\subsubsection{$\Pi$ Types}

\begin{lemma}
  \begin{mathpar}
    \inferrule
    {\msemtyeq[\vGamma'; \Gamma]{S}{S'}{\Se_i} \\ \msemtyeq[\vGamma'; (\Gamma, x : S)]{T}{T'}{\Se_i}}
    {\msemtyeq[\vGamma'; \Gamma]{\Pi(x : S). T}{\Pi(x : S'). T'}{\Se_i}}
  \end{mathpar}
\end{lemma}
\begin{proof}
  \begin{align*}
    H_1: &\ \msemtyeq[\vGamma'; \Gamma]{S}{S'}{\Se_i}
           \tag{by assumption} \\
    H_2: &\ \msemtyeq[\vGamma'; (\Gamma, x : S)]{T}{T'}{\Se_i}
           \tag{by assumption} \\
    H_3: &\ \vrho \approx \vrho' \in \intp{\vGamma'; \Gamma}
           \tag{by assumption} \\
         &\ \text{assume }\kappa \\
         &\ \intp{S}(\vrho) \approx \intp{S}(\vrho') \in \Uc_i
           \tag{by $H_1$} \\
         &\ \forall \kappa. \intp{S}(\vrho)[\kappa] \approx \intp{S}(\vrho')[\kappa] \in \Uc_i
           \tag{by \Cref{lem:dt:u-el-resp-mon}, then abstraction} \\
         &\ \text{assume }\kappa, a \approx a' \in \El_i(\intp{S}(\vrho)[\kappa]) \\
         &\ a \approx a' \in \El_i(\intp{S}(\vrho[\kappa]))
           \tag{by naturality equality of $S$}\\
         &\ \vrho[\kappa] \approx \vrho'[\kappa] \in \intp{\vGamma'; \Gamma}
           \tag{by $H_3$ and \Cref{lem:dt:intpvg-resp-mon}} \\
    H_4: &\ \ext(\vrho[\kappa], x, a) \approx \ext(\vrho'[\kappa], x, a') \in
           \intp{\vGamma'; (\Gamma, x : S)}
           \tag{by definition} \\
         & \intp{T}(\ext(\vrho[\kappa], x, a))
           \approx \intp{T'}(\ext(\vrho'[\kappa], x, a'))
           \in \Uc_i
           \tag{by $H_2$ and $H_4$} \\
         &\ Pi(\intp{S}(\vrho), x.T, \vrho)
           \approx Pi(\intp{S'}(\vrho), x.T', \vrho')
           \in \Uc_i
           \tag{by abstraction}
  \end{align*}
  Naturality equalities are immediate. 
\end{proof}

\begin{lemma}
  \begin{mathpar}
    \inferrule
    {\msemtyeq[\vGamma; (\Gamma, x : S)]{t}{t'}{T}}
    {\msemtyeq[\vGamma; \Gamma]{\lambda x. t}{\lambda x. t'}{\Pi(x : S). T}}
  \end{mathpar}
\end{lemma}
\begin{proof}
  \begin{align*}
    H_1: &\ \msemtyeq[\vGamma; (\Gamma, x : S)]{t}{t'}{T}
           \tag{by assumption} \\
    H_2: &\ \vrho \approx \vrho' \in \intp{\vGamma'; \Gamma}
           \tag{by assumption} \\
         &\ \text{assume }\kappa, a \approx a' \in \El_i(\intp{S}(\vrho)[\kappa]) \\
         &\ \vrho[\kappa] \approx \vrho'[\kappa] \in \intp{\vGamma'; \Gamma}
           \tag{by $H_2$ and \Cref{lem:dt:intpvg-resp-mon}} \\
    H_3: &\ \ext(\vrho[\kappa], x, a) \approx \ext(\vrho'[\kappa], x, a') \in
           \intp{\vGamma'; (\Gamma, x : S)}
           \tag{by definition} \\
         &\ \intp{T}(\ext(\vrho[\kappa], x, a))
           \approx \intp{T}(\ext(\vrho'[\kappa], x, a'))
           \in \Uc
           \tag{by $H_1$ and $H_3$} \\
         &\ \forall \kappa. \intp{S}(\vrho)[\kappa] \approx \intp{S}(\vrho')[\kappa]
           \in \Uc
           \tag{by inverting $H_1$ and \Cref{lem:dt:u-el-resp-mon}} \\
         &\ \msemjudge[\vGamma; \Gamma]{\Pi(x : S). T}
           \tag{first goal, by abstraction} \\
         &\ \intp{t}(\ext(\vrho[\kappa], x, a))
           \approx \intp{t'}(\ext(\vrho'[\kappa], x, a'))
           \in \El(\intp{T}(\ext(\vrho[\kappa], x, a)))
           \tag{by $H_1$} \\
         &\ \intp{t}(\vrho) \approx \intp{t'}(\vrho') \in \El(\intp{\Pi(x :
           S). T}(\vrho))
           \tag{by abstraction}
  \end{align*}
  Naturality equalities are immediate. 
\end{proof}

\begin{lemma}
  \begin{mathpar}
    \inferrule
    {\msemtyeq{t}{t'}{\Pi(x : S). T} \\ \msemtyeq{s}{s'}{S}}
    {\msemtyeq{t\ s}{t'\ s'}{T[\vect I, x : s]}}
  \end{mathpar}
\end{lemma}
\begin{proof}
  Follow definition. Naturality equalities hold for the same reason given in the
  simply typed case. 
\end{proof}

\begin{lemma}
  \begin{mathpar}
    \inferrule
    {\msemtyp[\vGamma; (\Gamma, x : S)]{t}{T} \\ \msemtyp[\vGamma;\Gamma]{s}{S}}
    {\msemtyeq[\vGamma; \Gamma]{(\lambda x. t)\ s}{t[\vect I, s/x]}{T[\vect I, s/x]}}
  \end{mathpar}
\end{lemma}
\begin{proof}
  Following the definition. $\msemtyp[\vGamma; (\Gamma, x : S)]{t}{T}$ will give us
  the desired relation. 
\end{proof}

\begin{lemma}
  \begin{mathpar}
    \inferrule
    {\msemtyp{t}{\Pi(x : S). T}}
    {\msemtyeq{t}{\lambda x. t[\wk_x]\ x}{\Pi(x : S). T}}
  \end{mathpar}
\end{lemma}
\begin{proof}
  $t$ and $\lambda x. t[p(\vect I, x)]\ x$ are the same modulo evaluation environments
  after evaluation, and they can be related just by the premise. 
\end{proof}

\begin{lemma}
  \begin{mathpar}
    \inferrule
    {\msemtyp[\vGamma; \Gamma]{\vsigma}{\vDelta;\Delta} \\
      \msemtyp[\vDelta;\Delta]{S}{\Se_i} \\ \msemtyp[\vDelta;(\Delta, x : S)]{T}{\Se_i}}
    {\msemtyeq[\vGamma; \Gamma]{\Pi(x : S). T[\vsigma]}{\Pi(x :
        S[\vsigma]).(T[q(\vsigma, x)])}{\Se_{i}}}
  \end{mathpar}
\end{lemma}
\begin{proof}
  \begin{align*}
    H_1: &\ \msemtyp[\vGamma; \Gamma]{\vsigma}{\vDelta;\Delta}
           \tag{by assumption} \\
    H_2: &\ \msemtyp[\vDelta;\Delta]{S}{\Se_i}
           \tag{by assumption} \\
    H_3: &\ \msemtyp[\vDelta;(\Delta, x : S)]{T}{\Se_i}
           \tag{by assumption} \\
    H_4: &\ \vrho \approx \vrho' \in \intp{\vGamma; \Gamma}
           \tag{by assumption}
  \end{align*}
  We evaluate both sides:
  \begin{align*}
    \intp{\Pi(x : S). T[\vsigma]}(\vrho)
    &= Pi(\intp{S}(\intp{\vsigma}(\vrho)), x.T, (\intp{\vsigma}(\vrho)))  \\
    \intp{{\Pi(x : S[\vsigma]).(T[q(\vsigma, x)])}}(\vrho')
    &= Pi(\intp{S}(\intp{\vsigma}(\vrho')), x.T[q(\vsigma, x)], \vrho') 
  \end{align*}
  We want to related them by $\Uc_i$. Assuming $\kappa$, $a \approx a' \in
  \El_i(\intp{S}(\intp{\vsigma}(\vrho))[\kappa])$,
  \begin{align*}
    \intp{T}(\ext(\intp{\vsigma}(\vrho)[\kappa], x, a)) 
    &= \intp{T}(\ext(\intp{\vsigma}(\vrho[\kappa]), x, a))
      \tag{by naturality equality} \\
    \intp{T}(\intp{q(\vsigma, x)}(\ext(\vrho'[\kappa], x, a'))) 
    &= \intp{T}(\ext(\intp{\vsigma}(\vrho'[\kappa]), x, a'))
  \end{align*}
  They only differ by evaluation environments, so we can conclude the goal by applying
  \Cref{lem:dt:intpvg-resp-mon} and then abstraction. Moreover, natural equalities
  hold because of natural equalities of $T$. 
\end{proof}

\begin{lemma}
  \begin{mathpar}
    \inferrule
    {\msemtyp[\vGamma; \Gamma]{\vsigma}{\vDelta;\Delta} \\
      \msemtyp[\vDelta; \Delta, x : S]{t}{T}}
    {\msemtyeq[\vGamma; \Gamma]{\lambda x. t[\vsigma]}{\lambda x. (t[q(\vsigma, x)])}{\Pi(x : S). T[\vsigma]}}
  \end{mathpar}
\end{lemma}
\begin{proof}
  \begin{align*}
    H_1: &\ \msemtyp[\vGamma; \Gamma]{\vsigma}{\vDelta;\Delta}
           \tag{by assumption} \\
    H_2: &\ \msemtyp[\vDelta; \Delta, x : S]{t}{T}
           \tag{by assumption} \\
    H_3: &\ \vrho \approx \vrho' \in \intp{\vGamma; \Gamma}
           \tag{by assumption}
  \end{align*}
  We evaluate both sides:
  \begin{align*}
    \intp{\lambda x. t[\vsigma]}(\vrho)
    &= \Lambda(x.t, \intp{\vsigma}(\vrho)) \\
    \intp{\lambda x. (t[q(\vsigma, x)])}(\vrho)
    &= \Lambda(x.t[q(\vsigma, x)], \vrho)
  \end{align*}
  Assuming $\kappa$, $a \approx a' \in \El_i(\intp{S}(\intp{\vsigma}(\vrho))[\kappa])$,
  we applying both sides to them, and then from $H_2$, we have
  \begin{align*}
    \intp{t}(\ext(\intp{\vsigma}(\vrho)[\kappa], x, a))
    \approx \intp{t}(\ext(\intp{\vsigma}(\vrho')[\kappa], x, a'))
    \in \El(\intp{T}(\ext(\intp{\vsigma}(\vrho)[\kappa], x, a)))
  \end{align*}
  By abstraction, we have
  \begin{align*}
    \intp{\lambda x. t[\vsigma]}(\vrho)
    \approx \intp{\lambda x. (t[q(\vsigma, x)])}(\vrho)
    \in \El(\intp{\Pi(x : S). T[\vsigma]}(\vrho))
  \end{align*}
\end{proof}

\begin{lemma}
  \begin{mathpar}
    \inferrule
    {\msemtyp{\vsigma}{\vDelta} \\
      \msemtyp[\vDelta]{s}{\Pi(x : S). T} \\
      \msemtyp[\vDelta]{t}{S}}
    {\msemtyeq{s\ t[\vsigma]}{(s[\vsigma])\ (t[\vsigma])}{T[\vsigma, t[\vsigma]/x]}}
  \end{mathpar}
\end{lemma}
\begin{proof}
  Apply the logical relation of $Pi$. 
\end{proof}

\subsubsection{Natural Numbers}

\begin{lemma}
  \begin{mathpar}
    \inferrule
    {\vDash \vGamma}
    {\msemtyeq \Nat \Nat \Se_i}

    \inferrule
    {\vDash \vGamma}
    {\msemtyeq{{\ze}}{\ze} \Nat}

    \inferrule
    {\msemtyeq{t}{t'}\Nat}
    {\msemtyeq{\su t}{\su t'}\Nat}
  \end{mathpar}
\end{lemma}
\begin{proof}
  Immediate. 
\end{proof}

We need the following helper lemma to reason about the partial $\trec$ function:
\begin{lemma}\labeledit{lem:dt:rec-per-helper}
  If 
  \begin{itemize}
  \item $\vrho \approx \vrho' \in \intp{\vGamma; \Gamma}$,
  \item $\msemtyeq[\vGamma; (\Gamma, x : \Nat)]{M}{M'}{\Se_i}$,
    
  \item $\msemtyeq[\vGamma; \Gamma]{s}{s'}{M[\vect I, \ze/x]}$,
  \item $\msemtyeq[\vGamma; (\Gamma, x : \Nat, y : M)]{u}{u'}{M[\wk_x \circ \wk_y, \su
      x/x]}$ and,
  \item $a \approx a' \in Nat$,
  \end{itemize}
  then
  $\trec \cdot (x.M, \intp{s}(\vrho), (x,y.u), a, \vrho) \approx \trec
  \cdot (x.M', \intp{s'}(\vrho'), (x,y.u'), a', \vrho') \in
    \El(\intp{M}(\ext(\vrho, x, a)))$
\end{lemma}
\begin{proof}
  We do induction on $a \approx a' \in Nat$.
  \begin{itemize}
  \item $a = a' = \zed$,
    \begin{mathpar}
      \inferrule
      { }
      {\zed \approx \zed \in Nat}
    \end{mathpar}
    then
    \begin{align*}
      \trec \cdot (x.M \lhd \vrho, \intp{s}(\vrho), (x,y.u \lhd \vrho), \zed)
      &= \intp{s}(\vrho) \\
      &\approx \intp{s'}(\vrho') \\
      &= \trec \cdot (x.M' \lhd \vrho', \intp{s'}(\vrho'), (x,y.u' \lhd \vrho'), \zed)
    \end{align*}
    Assuming $\kappa$, we have
    \begin{align*}
      &\ \trec \cdot ((x.M \lhd \vrho), \intp{s}(\vrho), (x,y.u\lhd \vrho),
        \zed)[\kappa] \\
      =&\ \intp{s}(\vrho)[\kappa] \\
      =&\ \intp{s}(\vrho[\kappa])
        \tag{by naturality equality} \\
      =&\ \trec \cdot ((x.M \lhd \vrho[\kappa]), \intp{s}(\vrho[\kappa]), (x,y.u\lhd
   \vrho[\kappa]), \zed[\kappa])
    \end{align*}
    The other equality is proved symmetrically. 

  \item $a = \sud b$ and $a' = \sud{b'}$, 
    \begin{mathpar}
      \inferrule
      {b \approx b' \in Nat}
      {\sud b \approx \sud{b'} \in Nat}
    \end{mathpar}
    By IH, we have
    $\trec \cdot (x.M \lhd \vrho, \intp{s}(\vrho), (x,y.u \lhd \vrho), b) \approx
    \trec \cdot (x.M' \lhd \vrho', \intp{s'}(\vrho'), (x,y.u' \lhd \vrho'), b') \in
    \El(\intp{M}(\ext(\vrho, x, b)))$.
    We let
    \begin{align*}
      r_1 &:= \trec \cdot (x.M \lhd \vrho, \intp{s}(\vrho), (x,y.u \lhd \vrho), b) \\
      r_2 &:= \trec \cdot (x.M' \lhd \vrho', \intp{s'}(\vrho'), (x,y.u' \lhd \vrho'), b')
    \end{align*}
    Due to
    $\msemtyeq[\vGamma; (\Gamma, x : \Nat, y : M)]{u}{u'}{M[\wk_x \circ \wk_y, \su x/x]}$,
    we further have
    \begin{align*}
      \intp{u}(\ext(\vrho, x, b, y, r_1)) \approx \intp{u'}(\ext(\vrho, x, b', y,
      r_2))
      \in \El(\intp{M}(\ext(\vrho, x, \sud b)))
    \end{align*}
    Notice that this is our first goal. 
    
    We further assume $\kappa$. The equalities can be easily established by IH.
    
  \item $a = \uparrow^\Nd(c)$ and $a' = \uparrow^\Nd(c')$, 
    \begin{mathpar}
      \inferrule*
      {c \approx c' \in \bot}
      {\uparrow^\Nd(c) \approx \uparrow^\Nd(c') \in Nat}
    \end{mathpar}
    The first goal is established by first showing the desired values are related by
    $\bot$, and then realizability shows they are also related in the target PER. Assuming $\alpha$, $\kappa$,
    \begin{align*}
      &\ \uparrow^\Nd(z) \approx \uparrow^\Nd(z) \in Nat
        \tag{for some $z$} \\
      &\ \uparrow^{\intp{M}(\ext(\vrho[\kappa], x, \uparrow^\Nd(z)))}(z') \approx
        \uparrow^{\intp{M'}(\ext(\vrho'[\kappa], x, \uparrow^\Nd(z)))}(z') \in
        \El(\intp{M}(\ext(\vrho[\kappa], x, \uparrow^\Nd(z))))
        \tag{for some $z'$} \\
      &\ \intp{M}(\ext(\vrho[\kappa], x, \uparrow^\Nd(z))) \approx \intp{M'}(\ext(\vrho'[\kappa], x, \uparrow^\Nd(z))) \in \Uc_i \\
      &\ \downarrow^{\Ud_i}(\intp{M}(\ext(\vrho[\kappa], x, \uparrow^\Nd(z))))
        \approx \downarrow^{\Ud_i}(\intp{M'}(\ext(\vrho'[\kappa], x, \uparrow^\Nd(z))))
        \in \top
        \tag{by realizability}
    \end{align*}
    and similar for $u$ and $u'$. Since all components have equal normal forms, by
    applying realizability, we get the conclusion.

    The equalities hold by definition. 
  \end{itemize}
\end{proof}
Notice that \Cref{conj:dt:rec.-mon} is proved in this helper lemma given terms are
well-typed and existing setup in the logical relations is sufficient to establish this
fact. 

Rules related to induction on natural numbers are consequence of this helper lemma.
this 
\begin{lemma}
  \begin{mathpar}
    \inferrule
    {\msemtyeq[\vGamma; (\Gamma, x : \Nat)]{M}{M'}{\Se_i} \\ \msemtyeq[\vGamma; \Gamma]{s}{s'}{M[\vect I, \ze/x]} \\
      \msemtyeq[\vGamma; (\Gamma, x : \Nat, y : M)]{u}{u'}{M[\wk_x \circ \wk_y, \su x/x]} \\ \msemtyeq[\vGamma; \Gamma]{t}{t'}\Nat}
    {\msemtyeq[\vGamma; \Gamma]{\elimn {x.M} s {x,y.u} t}{\elimn {x.M'}{s'}{x,y.u'}{t'}}{M[\vect I, t/x]}}
  \end{mathpar}
\end{lemma}
\begin{proof}
  Immediate by \Cref{lem:dt:rec-per-helper}. 
\end{proof}

\begin{lemma}
  \begin{mathpar}
  \inferrule
  {\msemtyp[\vGamma; (\Gamma, x : \Nat)]{M}{\Se_i} \\ \msemtyp[\vGamma; \Gamma]{s}{M[\vect I, \ze/x]} \\
    \msemtyp[\vGamma; (\Gamma, x : \Nat, y : M)]{u}{M[\wk_x \circ \wk_y, \su x/x]}}
  {\msemtyeq[\vGamma; \Gamma]{\elimn {x.M} s {x,y.u} \ze}{s}{M[\vect I, \ze/x]}}  
\end{mathpar}
\end{lemma}
\begin{proof}
  Immediate.
\end{proof}

\begin{lemma}
  \begin{mathpar}
    \inferrule
    {\msemtyp[\vGamma; (\Gamma, x : \Nat)]{M}{\Se_i} \\ \msemtyp[\vGamma; \Gamma]{s}{M[\vect I, \ze/x]} \\
      \msemtyp[\vGamma; (\Gamma, x : \Nat, y : M)]{u}{M[\wk_x \circ \wk_y, \su x/x]} \\ \msemtyp[\vGamma; \Gamma] t \Nat}
    {\msemtyeq[\vGamma; \Gamma]{\elimn{x.M}s{x,y.u}{(\su t)}}{u[\vect I, t/x, 
        \elimn {x.M} s {x,y.u} t/y]} {M[\vect I, \su t/x]}}
  \end{mathpar}
\end{lemma}
\begin{proof}
  Apply \Cref{lem:dt:rec-per-helper} to the recursive case. 
\end{proof}

\begin{lemma}
  \begin{mathpar}
    \inferrule
    {\msemtyp{\vsigma}{\vDelta}}
    {\msemtyeq{\Nat[\vsigma]}{\Nat}\Se_i}

    \inferrule
    {\msemtyp{\vsigma}{\vDelta}}
    {\msemtyeq{\ze[\vsigma]}\ze \Nat}

    \inferrule
    {\msemtyp{\vsigma}{\vDelta} \\ \msemtyp[\vDelta] t \Nat}
    {\msemtyeq{\su t[\vsigma]}{\su{(t[\vsigma])}}\Nat}
  \end{mathpar}
\end{lemma}
\begin{proof}
  Immediate.
\end{proof}

The following additional lemma is need to prove the semantic substitution lemma for
eliminator for natural numbers:
\begin{lemma}\labeledit{lem:dt:rec-per-helper2}
  If
  \begin{itemize}
  \item $\msemtyp{\vsigma}{\vDelta;\Delta}$,
  \item $\msemtyp[\vDelta; (\Delta, x : \Nat)]{M}{\Se_i}$,
  \item $\msemtyp[\vDelta; \Delta]{s}{M[\vect I, \ze/x]}$,
  \item $\msemtyp[\vDelta; (\Delta, x : \Nat, y : M)]{u}{M[\wk_x \circ \wk_y, \su x/x]}$,
  \item $\vrho \in \intp{\vGamma}$, and
  \item $a \in Nat$,
  \end{itemize}
  then
  $\trec \cdot (x.M, \intp{s}(\intp{\vsigma}(\vrho)), (x,y.u), a,
  \intp{\vsigma}(\vrho)) \approx \trec \cdot (x.M[q(\vsigma, x)],
  \intp{s[\vsigma]}(\vrho), (x,y.u[q(q(\vsigma, x), y)]), a, \vrho) \in \El(\intp{M}(\ext(\intp{\vsigma}(\vrho), x, a)))$.
\end{lemma}

\begin{lemma}
  \begin{mathpar}
    \inferrule
    {\msemtyp{\vsigma}{\vDelta;\Delta} \\
      \msemtyp[\vDelta; (\Delta, x : \Nat)]{M}{\Se_i} \\ \msemtyp[\vDelta; \Delta]{s}{M[\vect I, \ze/x]} \\
      \msemtyp[\vDelta; (\Delta, x : \Nat, y : M)]{u}{M[\wk_x \circ \wk_y, \su x/x]} \\ \msemtyp[\vDelta; \Delta] t \Nat}
    {\msemtyeq[\vGamma]{\elimn {x.M} s {x,y.u} t[\vsigma]}{\elimn
        {x.M[q(\vsigma, x)]} {(s[\vsigma])}
        {x,y.u[q(q(\vsigma, x), y)]}{(t[\vsigma])}}{M[\vsigma, t[\vsigma]/x]}}
  \end{mathpar}
\end{lemma}
\begin{proof}
  Expand the definitions and apply
  \Cref{lem:dt:rec-per-helper,lem:dt:rec-per-helper2}. \Cref{lem:dt:rec-per-helper2}
  is needed to relate the results of elimination of natural numbers on both sides. 
\end{proof}

\subsubsection{Universes}

\begin{lemma}
  \begin{mathpar}
    \inferrule
    {\vDash \vGamma}
    {\msemtyeq{\Se_n}{\Se_n}{\Se_{n + 1}}}

    \inferrule
    {\msemtyeq{T}{T'}{\Se_i}}
    {\msemtyeq{T}{T'}{\Se_{i + 1}}}

    \inferrule
    {\msemtyp{\vsigma}{\vDelta}}
    {\msemtyeq{\Se_n[\vsigma]}{\Se_n}{\Se_{n + 1}}}
  \end{mathpar}
\end{lemma}
\begin{proof}
  Immediate by the definitions. 
\end{proof}

\subsubsection{Modality}

\begin{lemma}
  \begin{mathpar}
    \inferrule
    {\msemtyeq[\vGamma; \cdot]{T}{T'}{\Se_i}}
    {\msemtyeq{\square T}{\square T'}{\Se_i}}
  \end{mathpar}
\end{lemma}
\begin{proof}
  Immediate by \Cref{lem:dt:u-el-resp-mon}. 
\end{proof}

\begin{lemma}
  \begin{mathpar}
    \inferrule
    {\msemtyeq[\vGamma; \cdot]{t}{t'}{T}}
    {\msemtyeq{\boxit t}{\boxit t'}{\square T}}
  \end{mathpar}
\end{lemma}
\begin{proof}
  \begin{align*}
    H_1: &\ \msemtyeq[\vGamma; \cdot]{t}{t'}{T}
           \tag{by assumption} \\
    H_2: &\ \vrho \approx \vrho' \in \intp{\vGamma}
           \tag{by assumption} \\
         &\ \ext(\vrho) \approx \ext(\vrho') \in \intp{\vGamma; \cdot}
           \tag{by definition} \\
         &\ \intp{T}(\ext(\vrho)) \approx \intp{T'}(\ext(\vrho')) \in \Uc
           \tag{by $H_1$} \\
         &\ \forall \kappa. \intp{T}(\ext(\vrho))[\kappa] \approx
           \intp{T'}(\ext(\vrho'))[\kappa] \in \Uc
           \tag{by \Cref{lem:dt:u-el-resp-mon}} \\
         &\ \intp{\square T}(\vrho) \approx\intp{\square T}(\vrho') \in \Uc
           \tag{first goal, by definition} \\
    H_3: &\ \intp{t}(\ext(\vrho)) \approx \intp{t'}(\ext(\vrho'))
           \in \El(\intp{T}(\ext(\vrho)))
           \tag{by $H_1$} \\
         &\ \text{assume }\kappa, n, \\
         &\ \intp{t}(\ext(\vrho))[\sextt\kappa n] \approx
           \intp{t'}(\ext(\vrho'))[\sextt\kappa n]
           \in \El(\intp{T}(\ext(\vrho))[\sextt\kappa n])
           \tag{by applying \Cref{lem:dt:u-el-resp-mon} to $H_3$} \\
         &\ \intp{\boxit t}(\vrho) \approx \intp{\boxit t'}(\vrho')
           \in \El(\intp{\square T}(\vrho))
           \tag{by abstraction}
  \end{align*}
\end{proof}

\begin{lemma}
  \begin{mathpar}
    \inferrule
    {\msemtyeq{t}{t'}{\square T} \\ \vDash \vGamma; \vDelta \\ |\vDelta| = n}
    {\msemtyeq[\vGamma; \vDelta]{\unbox n t}{\unbox n t'}{T[\sextt{\vect I} n]}}
  \end{mathpar}
\end{lemma}
\begin{proof}
  \begin{align*}
    H_1 :&\ \msemtyeq{t}{t'}{\square T}
           \tag{by assumption} \\
    H_2: &\ \vDash \vGamma; \vDelta
           \tag{by assumption} \\
    H_3: &\ \vrho \approx \vrho' \in \intp{\vGamma; \vDelta}
           \tag{by assumption} \\
         &\ \trunc \vrho n \approx \trunc{\vrho'}n \in \intp{\vGamma}
           \tag{by $H_3$} \\
    H_4: &\ \intp{t}(\trunc \vrho n) \approx \intp{t'}(\trunc{\vrho'}n)
           \in \El(\intp{\square T}(\trunc \vrho n))
           \tag{by $H_1$} \\
         &\ \tunbox \cdot (\Ltotal\vrho n, \intp{t}(\trunc \vrho n))
           \approx \tunbox \cdot (\Ltotal{\vrho'} n, \intp{t}(\trunc{\vrho'}n))
           \in \El(\intp{T}(\ext(\trunc \vrho n, \Ltotal\vrho n)))
           \tag{by $H_4$} \\
         &\ \intp{\unbox n t}(\vrho)
           \approx \intp{\unbox n t'}(\vrho')
           \in \El(\intp{\square T}(\vrho))
           \tag{by abstraction}
  \end{align*}
\end{proof}

\begin{lemma}
  \begin{mathpar}
    \inferrule
    {\msemtyp[\vGamma; \cdot]{t}{T} \\ \vdash \vGamma; \vDelta \\ |\vDelta| = n}
    {\msemtyeq[\vGamma; \vDelta]{\unbox n {(\boxit t)}}{t[\sextt{\vect I}n]}{T[\sextt{\vect I}n]}}  

    \inferrule
    {\msemtyp[\vGamma]{t}{\square T}}
    {\msemtyeq{t}{\boxit{(\unbox 1 t)}}{\square T}}  
  \end{mathpar}
\end{lemma}
\begin{proof}
  Evaluate both sides and the results only differ by evaluation environments, which
  can be related by definition. 
\end{proof}

\begin{lemma}
    \begin{mathpar}
    \inferrule
    {\msemtyp[\vDelta; \cdot]{T}{\Se_i} \\ \msemtyp{\vsigma}{\vDelta}}
    {\msemtyeq{\square T[\vsigma]}{\square (T[\sextt\vsigma 1])}{\Se_i}}

    \inferrule
    {\msemtyp[\vDelta; \cdot]{t}{T} \\ \msemtyp{\vsigma}{\vDelta}}
    {\msemtyeq{\boxit t[\vsigma]}{\boxit {(t[\sextt\vsigma 1])}}{\square T [\vsigma]}}
  \end{mathpar}
\end{lemma}
\begin{proof}
  Evaluate both sides. We match the evaluation environment and the goal is immediate. 
\end{proof}

\begin{lemma}
  \begin{mathpar}
    \inferrule
    {\msemtyp[\vDelta]{t}{\square T} \\ |\vDelta'| = n \\ \msemtyp{\vsigma}{\vDelta; \vDelta'}}
    {\msemtyeq{\unbox n t[\vsigma]}{\unbox{\Ltotal\vsigma
          n}{(t[\trunc\vsigma n])}}{T[\sextt{\trunc\vsigma n}{\Ltotal{\vsigma}n}]}}
  \end{mathpar}
\end{lemma}
\begin{proof}
  \begin{align*}
    H_1: &\ \msemtyp[\vDelta]{t}{\square T}
           \tag{by assumption} \\ 
    H_2: &\ \msemtyp{\vsigma}{\vDelta; \vDelta'}
           \tag{by assumption} \\ 
    H_3: &\ \vrho \approx \vrho' \in \intp{\vGamma}
           \tag{by assumption} \\
         &\ \intp{\vsigma}(\vrho)
           \approx \intp{\vsigma}(\vrho')
           \in \intp{\vDelta; \vDelta'}
           \tag{by $H_2$} \\
    H_4: &\ \trunc{\intp{\vsigma}(\vrho)} n
           \approx \intp{\trunc\vsigma n}(\trunc{\vrho'}{\Ltotal\vsigma n})
           \in \intp{\vDelta}
           \tag{by \Cref{lem:dt:L-intp-vsigma-vrho,lem:dt:trunc-intp-vsigma-vrho}}
  \end{align*}
  We conclude the goal by applying $H_1$ to $H_4$. 
\end{proof}

\subsection{Fundamental Theorem}

We have proved all semantic rules. We can thus conclude the fundamental theorem.
\begin{theorem}\labeledit{thm:dt:per-fund}
  (fundamental)
  \begin{itemize}
  \item If $\vdash \vGamma$, then $\vDash \vGamma$. 
  \item If $\mtyping t T$, then $\msemtyp t T$.
  \item If $\mtyping{\vsigma}{\vDelta}$, then $\msemtyp{\vsigma}{\vDelta}$.
  \item If $\vdash \vGamma \approx \vGamma'$, then $\vDash \vGamma \approx \vGamma'$. 
  \item If $\mtyequiv{t}{t'}T$, then $\msemtyeq{t}{t'}T$.
  \item If $\mtyequiv{\vsigma}{\vsigma'}{\vDelta}$, then $\msemtyeq{\vsigma}{\vsigma'}{\vDelta}$.
  \end{itemize}
\end{theorem}

We also can show the initial environment is reflexive in $\intp{\vGamma}$.
\begin{lemma}
  $\uparrow^{\vGamma} \approx \uparrow^{\vGamma} \in \intp{\vGamma}$
\end{lemma}
\begin{proof}
  Immediate by \Cref{lem:dt:varbot} and realizability. 
\end{proof}

\begin{theorem} (completeness)
  If $\mtyequiv{t}{t'}{T}$, then $\nbe^T_{\vGamma}(t) = \nbe^T_{\vGamma}(t')$. 
\end{theorem}
\begin{proof}
  Combine fundamental theorem and realizability. 
\end{proof}

\section{Soundness for Dependent Types}

\subsection{Gluing Model}

Following the simply typed case, when constructing a gluing model, it has to be Kripke
w.r.t. restricted weakenings. We continue to use the definition of restricted
weakenings in \Cref{sec:st:rweaken} and inherit its properties, as these properties do
not rely on whether types depend on terms or not.
To define the gluing model, we shall work with the values considered as
types. Luckily, $\Uc_i$ does exactly that. The following relations are defined by
induction on $A \in \Uc_i$.
\begin{definition}
  Given $A \in \Uc_i$,
  \begin{itemize}
  \item $\mgluty T[i] A$ denotes that $T$ is a well-formed type in $\vGamma$ and
    corresponds to a domain value $A$ as a semantic type. 
  \item $\mglutm t T a [i] A$ denotes that $t$ has type $T$ in $\vGamma$ and corresponds
    to a domain value $a$ which has semantic type $A$. 
  \end{itemize}
  They are defined as follows by induction on well-foundedness of $i$ and $A \in \Uc_i$:
  \begin{itemize}
  \item $A = \uparrow^{A'}(C)$, 
    \begin{mathpar}
      \inferrule
      {C \approx C \in \bot}
      {\uparrow^{A'}(C) \approx \uparrow^{A'}(C) \in \Uc_i}
    \end{mathpar}
  \begin{align*}
    \mgluty T[i]\D
    &:= \typing{T}{\Se_i} \tand \forall \vsigma : \vDelta \To_r \vGamma.\  \mtyequiv[\vDelta]{T[\vsigma]}{\Rne_\alpha(C[\vsigma])}{\Se_i} \\
    \mglutm t T {\uparrow^{A''}(c)}[i] {\D}
    &:= c \in \bot \tand \typing T{\Se_i} \tand \typing t T \\
    & \tand \forall \vsigma : \vDelta \To_r \vGamma.\ \mtyequiv[\vDelta]{T[\vsigma]}{\Rne_\alpha(C[\vsigma])}{\Se_i} \tand \\
    & \qquad \mtyequiv[\vDelta]{t[\vsigma]}{\Rne_\alpha(c[\vsigma])}{T[\vsigma]}
  \end{align*}
    
  \item $A = \Nd$,
    \begin{mathpar}
      \inferrule
      { }
      {\Nd \approx \Nd \in \Uc_i}
    \end{mathpar}
    \begin{align*}
      \mgluty T[i]\D
      &:= \mtyequiv{T}{\Nat}{\Se_i} \\
      \mglutm t T a [i] \D
      &:= \mglunat{t}{a} \tand \mtyequiv{T}{\Nat}{\Se_i}
    \end{align*}
    where we need an auxiliary relation to relate syntactic terms of type $\Nat$ and
    semantic values of PER $Nat$:
    \begin{mathpar}
      \inferrule
      {\mtyequiv{t}{\ze}{\Nat}}
      {\mglunat{t}{\zed}}

      \inferrule
      {\mtyequiv{t}{\su t'}{\Nat} \\ \mglunat{t'}{b}}
      {\mglunat{t}{\sud b}}

      \inferrule
      {c \in \bot \\
        \forall \vsigma : \vDelta \To_r \vGamma.\ \mtyequiv{t[\vsigma]}{\Rne_\alpha(c[\vsigma])}{\Nat}}
      {\mglunat{t}{\uparrow^\Nd(c)}}
    \end{mathpar}
    
  \item $A = \Ud_j$, 
    \begin{mathpar}
      \inferrule
      {j < i}
      {\Ud_j \approx \Ud_j \in \Uc_i}
    \end{mathpar}
    \begin{align*}
      \mgluty T[i]{\D}
      &:= \mtyequiv{T}{\Se_j}{\Se_i} \\
      \mglutm t T a[i]{\D}
      &:= \typing t T \tand \mtyequiv{T}{\Se_j}{\Se_i} \tand a \in \Uc_j
        \tand \mgluty t[j]{a} 
    \end{align*}
    
  \item $A = \squared A'$,
    \begin{mathpar}
      \inferrule
      {\forall \kappa. A'[\kappa] \approx A'[\kappa] \in \Uc_i}
      {\squared A' \approx \squared A' \in \Uc_i}
    \end{mathpar}
    \begin{align*}
      \mgluty T[i]{\squared A'}
      &:= \exists T'. \mtyequiv{T}{\square T'}{\Se_i} \\
        & \quad \tand \forall \vDelta',
        \vdash \vDelta; \vDelta',
        \vsigma : \vDelta
        \To_r \vGamma.\ \mgluty[\vDelta; \vDelta']{T'[\sextt\vsigma {|\vDelta'|}]}[i]{A'[\sextt\vsigma {|\vDelta'|}]} \\
      \mglutm t T a[i]{\squared A'}
      &:= \mtyping t T \tand a \in \El_i(\squared A')  \tand \exists T'. \mtyequiv{T}{\square T'}{\Se_i} \tand \\
      &  \qquad \forall \vDelta', \vdash \vDelta; \vDelta', \vsigma : \vDelta \To_r \vGamma. \\
      & \qquad\quad \mglutm[\vDelta;
        \vDelta']{\unbox{|\vDelta'|}{(t[\vsigma])}}{T'[\sextt\vsigma {|\vDelta'|}]}{\tunbox \cdot
        (|\vDelta'|, a[\vsigma])}[i]{A'[\sextt{\mt{\vsigma}}{|\vDelta'|}]}
    \end{align*}
    
  \item $A = Pi(A', x.T', \vrho)$,
    \begin{mathpar}
      \inferrule
      {\forall \kappa. A'[\kappa] \approx A'[\kappa] \in \Uc_i \\ \forall \kappa, a
        \approx a' \in \El_i(A'[\kappa]). \intp{T'}(\ext(\vrho[\kappa], x, a))
        \approx \intp{T'}(\ext(\vrho[\kappa], x, a') \in \Uc_i}
      {Pi(A', x.T', \vrho) \approx Pi(A', x.T', \vrho) \in \Uc_i}
    \end{mathpar}
    \begin{align*}
      \mgluty[\vGamma; \Gamma] T[i]{Pi(A', x.T', \vrho)}
      &:= \exists T_1, T_2. \mtyequiv[\vGamma; \Gamma]{T}{\Pi(x : T_1). T_2}{\Se_i} \tand
        \mtyping[\vGamma;\Gamma,x :T_1]{T_2}{\Se_i} \tand \\
      & \qquad \forall \vsigma : \vDelta \To_r \vGamma;\Gamma.\
        \mgluty[\vDelta]{T_1[\vsigma]}[i]{A'[\vsigma]} \tand \\
      & \qquad \forall
        \mglutm[\vDelta]{s}{T_1[\vsigma]}{b}[i]{A'[\vsigma]}.\ \mgluty[\vDelta]{T_2[\vsigma, 
        s/x]}[i]{\intp{T'}(\ext(\vrho[\vsigma],x, b))} \\
      \mglutm[\vGamma; \Gamma] t T a[i]{Pi(A', x.T', \vrho)}
      &:= \mtyping[\vGamma; \Gamma] t T \tand a \in \El_i(Pi(A', x.T', \vrho)) \\
      & \tand \exists T_1, T_2. \mtyequiv[\vGamma; \Gamma]{T}{\Pi(x : T_1). T_2}{\Se_i} \tand \\
      & \qquad \forall \vsigma : \vDelta \To_r
        \vGamma.\ \mgluty[\vDelta]{T_1[\vsigma]}[i]{A'[\vsigma]} \tand \\
      & \qquad \forall \mglutm[\vDelta]{s}{T_1[\vsigma]}{b}[i]{A'[\vsigma]}. \\
      & \qquad\quad \mglutm[\vDelta]{t[\vsigma]\ s}{T_2[\vsigma, s/x]}{a[\vsigma] \cdot b}[i]{\intp{T'}(\ext(\vrho[\vsigma],x, b))}
    \end{align*}
  \end{itemize}
\end{definition}

\subsection{Properties of Gluing Model}

In the following lemmas, we assume $A \in \Uc_i$ whenever necessary. 


\begin{lemma}
  If $\mgluty T[i]A$, then $\mtyping T \Se_i$. 
\end{lemma}

\begin{lemma}\labeledit{lem:dt:inv-glu-tm}
  If $\mglutm t T a[i]A$, then
  \begin{itemize}
  \item $\mtyping t T$,
  \item $a \in \El_i(A)$;
  \item $\mgluty T[i]A$.
  \end{itemize}
\end{lemma}
\begin{proof}
  Induction on $A \in \Uc_i$ immediately discharges most cases in the first two
  statements, except
  \begin{itemize}[label=Case]
  \item $A = \Nd$,
    \begin{mathpar}
      \inferrule
      { }
      {\Nd \approx \Nd \in \Uc_i}
    \end{mathpar}
    \begin{align*}
      \mglunat{t}{a} \tag{by assumption} 
    \end{align*}
    We do induction on this structure and obtain have $a \in Nat$. 
  \end{itemize}

  In the last statement, we consider two interesting cases:
  \begin{itemize}[label=Case]
  \item $A = \squared A'$,
    \begin{mathpar}
      \inferrule
      {\forall \kappa. A'[\kappa] \approx A'[\kappa] \in \Uc_i}
      {\squared A' \approx \squared A' \in \Uc_i}
    \end{mathpar}
    \begin{align*}
      &\ \mjudge{T \approx \square T'}
        \tag{by assumption} \\
      &\ \forall \vDelta', \vdash \vDelta; \vDelta', \vsigma : \vDelta \To_r \vGamma. \mglutm[\vDelta;
        \vDelta']{\unbox{|\vDelta'|}{(t[\vsigma])}}{T'[\sextt{\vsigma}{|\vDelta'|}]}{\tunbox \cdot
        (|\vDelta'|, a[\vsigma])}[i]{A'[\sextt{\mt{\vsigma}}{|\vDelta'|}]}
        \tag{by assumption} \\
      &\ \forall \vDelta', \vdash \vDelta; \vDelta', \vsigma : \vDelta \To_r \vGamma. \mgluty[\vDelta;
        \vDelta']{T'[\sextt{\vsigma}{|\vDelta'|}]}[i]{A'[\sextt{\mt{\vsigma}}{|\vDelta'|}]}
        \byIH
    \end{align*}
    
  \item $A = Pi(A', x.T', \vrho)$,
    \begin{mathpar}
      \inferrule
      {\forall \kappa. A'[\kappa] \approx A'[\kappa] \in \Uc_i \\ \forall \kappa, a
        \approx a' \in \El_i(A'[\kappa]). \intp{T'}(\ext(\vrho[\kappa], x, a))
        \approx \intp{T'}(\ext(\vrho[\kappa], x, a') \in \Uc_i}
      {Pi(A', x.T', \vrho) \approx Pi(A', x.T', \vrho) \in \Uc_i}
    \end{mathpar}
    Similar to the previous case, we just need to apply IH appropriately. 
  \end{itemize}
\end{proof}

\begin{lemma}[Type conversion]\labeledit{lem:dt:gluty-resp-equiv}
  If $\mgluty T[i]A$ and $\mtyequiv{T}{T'}{\Se_i}$, then $\mgluty{T'}[i]A$.
\end{lemma}
\begin{proof}
  Induction on $i$ and then induction on $A \in \Uc_i$ and apply congruence.
\end{proof}

\begin{lemma}[Type conversion]\labeledit{lem:dt:glutm-resp-equiv-ty}
  If $\mglutm t T a[i]A$ and $\mtyequiv{T}{T'}{\Se_i}$, then $\mglutm t{T'}a[i]A$.
\end{lemma}
\begin{proof}
  Induction on $i$ and then induction on $A \in \Uc_i$, and apply congruence and conversion.
\end{proof}

\begin{lemma}[Term conversion]\labeledit{lem:dt:glutm-resp-equiv-tm}
  If $\mglutm t T a[i]A$ and $\mtyequiv{t}{t'}{T}$, then $\mglutm{t'}{T}a[i]A$.
\end{lemma}
\begin{proof}
  Induction on $A \in \Uc_i$, and apply \Cref{lem:dt:gluty-resp-equiv}, congruence, conversion
  and transitivity. We consider one interesting case: $A = Pi(A', x.T', \vrho)$,
    \begin{mathpar}
      \inferrule
      {\forall \kappa. A'[\kappa] \approx A'[\kappa] \in \Uc_i \\
        \forall \kappa, a
        \approx a' \in \El_i(A'[\kappa]). \intp{T'}(\ext(\vrho[\kappa], x, a))
        \approx \intp{T'}(\ext(\vrho[\kappa], x, a') \in \Uc_i}
      {Pi(A', x.T', \vrho) \approx Pi(A', x.T', \vrho) \in \Uc_i}
    \end{mathpar}
  \begin{align*}
    &\ \mjudge{T \approx \Pi(x : T_1). T_2}
      \tag{by assumption, for some $T_1$ and $T_2$} \\
    &\ \text{assume } \vsigma : \vDelta \To_r \vGamma,
      \mglutm[\vDelta]{s}{S[\vsigma]}{b}[i]{A'[\vsigma]} \\
    &\ \mglutm[\vDelta]{t[\vsigma]\ s}{T_2[\vsigma, s/x]}{a[\vsigma] \cdot
      b}[i]{\intp{T'}(\vrho[\vsigma], x , b)}
      \tag{by assumption} 
  \end{align*}
  \begin{mathpar}
    \inferrule
    {\inferrule
      {\inferrule
        {\mtyequiv{t}{t'}{T}}
        {\mtyequiv[\vDelta]{t[\vsigma]}{t'[\vsigma]}{T[\vsigma]}}}
      {\mtyequiv[\vDelta]{t[\vsigma]}{t'[\vsigma]}{\Pi(x : T_1). T_2[\vsigma]}} \\
      \inferrule
      {\inferrule
        {\Cref{lem:dt:inv-glu-tm}}
        {\mtyping{s}{T_1}}}
      {\mtyequiv{s}{s}{T_1[\vsigma]}}}
    {\mtyequiv[\vDelta]{t[\vsigma]\ s}{t'[\vsigma]\ s}{T_2[\vsigma, s/x]}}
  \end{mathpar}
  Now IH applies. 
\end{proof}

The gluing relations respect equivalence in semantic types:
\begin{lemma}[$\Uc$ conversion]\labeledit{lem:dt:glu-resp-u}
  Given $A \approx A' \in \Uc_i$, 
  \begin{itemize}
  \item $\mgluty T[i]A$ and $\mgluty T[i]A'$ are equivalent;
  \item $\mglutm t T a[i]A$ and $\mglutm t T a[i]{A'}$ are equivalent. 
  \end{itemize}
\end{lemma}
\begin{proof}
  Induction on $A \in \Uc_i$:
  \begin{itemize}[label=Case]
  \item $A = \uparrow^{A''}(C)$, $A' = \uparrow^{A'''}(C')$,
    \begin{mathpar}
      \inferrule
      {C \approx C' \in \bot}
      {\uparrow^{A''}(C) \approx \uparrow^{A'''}(C') \in \Uc_i}
    \end{mathpar}
    Notice that the premise implies $\Rne_\alpha(C[\kappa]) = \Rne_\alpha(C'[\kappa])$ for
    all $\alpha$ and $\kappa$. We simply substitute $\Rne_\alpha(C[\vsigma])$ with
    $\Rne_\alpha(C'[\vsigma])$ and achieve the goal.
    
  \item $A = \squared A''$, $A' = \squared A'''$,
    \begin{mathpar}
      \inferrule
      {\forall \kappa. A''[\kappa] \approx A'''[\kappa] \in \Uc_i}
      {\squared A'' \approx \squared A''' \in \Uc_i}
    \end{mathpar}
    We can resort to IHs and \Cref{lem:dt:el-resp-u} for the goal.
    
  \item $A = Pi(A'', x.T'', \vrho'')$, $A' = Pi(A''', x.T''', \vrho''')$,
    \begin{mathpar}
      \inferrule
      {\forall \kappa. A''[\kappa] \approx A'''[\kappa] \in \Uc_i \\
        \forall \kappa, a
        \approx a' \in \El_i(A''[\kappa]). \intp{T''}(\ext(\vrho''[\kappa], x, a))
        \approx \intp{T'''}(\ext(\vrho'''[\kappa], x, a''') \in \Uc_i}
      {Pi(A'', x.T'', \vrho'') \approx Pi(A''', x.T''', \vrho''') \in \Uc_i}
    \end{mathpar}
    Assuming $\vsigma : \vDelta \To_r \vGamma$,
    $\mglutm[\vDelta]{s}{S[\vsigma]}{b}[i]{A'''[\vsigma]}$,
    \begin{align*}
      &\ \mtyequiv{T}{\Pi(x : T_1). T_2}{\Se_i}
        \tag{by assumption, for some $T_1$ and $T_2$} \\
      &\ \mglutm[\vDelta]{s}{T_1[\vsigma]}{b}[i]{A''[\vsigma]}
        \byIH \\
      &\ \mgluty[\vDelta]{T_2[\vsigma, s/x]}[i]{\intp{T''}(\ext(\vrho''[\vsigma], x, b))}
        \tag{by $\mgluty T[i]{Pi(A'', x.T'', \vrho'')}$} \\
      &\ \mgluty[\vDelta]{T_2[\vsigma, s/x]}[i]{\intp{T'''}(\ext(\vrho'''[\vsigma], x, b))}
        \byIH \\
      &\ \mglutm[\vDelta]{t[\vsigma]\ s}{T_2[\vsigma, s/x]}{a[\vsigma] \cdot
        b}[i]{\intp{T''}(\ext(\vrho''[\vsigma], x, b))}
        \tag{by $\mglutm t T a[i]{Pi(A'', x.T'', \vrho'')}$} \\
      &\ \mglutm[\vDelta]{t[\vsigma]\ s}{T_2[\vsigma, s/x]}{a[\vsigma] \cdot
        b}[i]{\intp{T'''}(\ext(\vrho'''[\vsigma], x, b))}
        \byIH 
    \end{align*}
    We obtain the goals by abstraction. 
  \end{itemize}
\end{proof}

\begin{lemma}[Monotonicity]\labeledit{lem:dt:glu-mon}
  \begin{itemize}
  \item If $\mgluty T[i]A$ and $\vdelta : \vDelta \To_r \vGamma$, then
    $\mgluty[\vDelta]{T[\vdelta]}[i]{A[\vdelta]}$. 
  \item If $\mglutm t T a[i]A$ and $\vdelta : \vDelta \To_r \vGamma$, then
    $\mglutm[\vDelta]{t[\vdelta]}{T[\vdelta]}{a[\vdelta]}[i]{A[\vdelta]}$.
  \end{itemize}
\end{lemma}
\begin{proof}
  We do well-foundned induction on $i$ and then induction on $A \in \Uc_i$, apply
  \Cref{lem:dt:glutm-resp-equiv-ty,lem:dt:glutm-resp-equiv-tm} when necessary.  We
  consider one case: $A = Pi(A', x.T', \vrho)$,
    \begin{mathpar}
      \inferrule
      {\forall \kappa. A'[\kappa] \approx A'[\kappa] \in \Uc_i \\ \forall \kappa, a
        \approx a' \in \El_i(A'[\kappa]). \intp{T'}(\ext(\vrho[\kappa], x, a))
        \approx \intp{T'}(\ext(\vrho[\kappa], x, a') \in \Uc_i}
      {Pi(A', x.T', \vrho) \approx Pi(A', x.T', \vrho) \in \Uc_i}
    \end{mathpar}
  We know
  \begin{align*}
    \mtyequiv{T}{\Pi(x : T_1). T_2}{\Se_i}
  \end{align*}
  for some $T_1$ and $T_2$. Thus we have 
  \begin{align*}
    \mtyequiv[\vDelta]{T[\vdelta]}{\Pi(x : T_1[\vdelta]). (T_2[q(\vdelta, x)])}{\Se_i}
  \end{align*}
  Hence we can discharge the first statement.

  For the second one, \Cref{lem:dt:u-el-resp-mon} ensures
  $a[\vdelta] \in \El_i(Pi(A'[\vdelta], x.T', \vrho[\vdelta]))$. Assuming
  $\vsigma : \vDelta' \To_r \vDelta$ and
  $\mglutm[\vDelta']{s}{S[\vdelta][\vsigma]}{b}[i]{A'[\vdelta][\vsigma]}$, we have
  \begin{align*}
    \mglutm[\vDelta']{s}{S[\vdelta \circ \vsigma]}{b}[i]{A'[\vdelta \circ \vsigma]}
  \end{align*}
  and thus by premise, we have 
  \begin{align*}
    \mglutm[\vDelta']{t[\vdelta \circ \vsigma]\ s}{T_2[\vdelta \circ \vsigma,
    s/x]}{a[\vdelta \circ \vsigma] \cdot b}[i]{\intp{T'}(\ext(\vrho[\vdelta \circ
    \vsigma], x, b))}
  \end{align*}
  This is equivalent to
  \begin{align*}
    \mglutm[\vDelta']{t[\vdelta][\vsigma]\ s}{T_2[q(\vdelta, x)][\vsigma, s/x]}{a[\vdelta][\vsigma] \cdot b}[i]{\intp{T'}(\ext(\vrho[\vdelta][\vsigma], x, b))}
  \end{align*}
  and we conclude the second statement. 
\end{proof}

\subsection{Realizability}

\begin{definition}
  Given $A \in \Uc_i$, 
  \begin{align*}
    \mglutyu T [i] A
    &:= \typing T{\Se_i} \tand A \in \top \tand
      \forall \vsigma : \vDelta \To_r \vGamma.\
      \mtyequiv[\vDelta]{T[\vsigma]}{\Rty_\alpha(A[\vsigma])}{\Se_i}\\
    \mglutmu t T a[i]A
    &:= \mtyping t T \tand \mgluty T[i]A \tand \downarrow^A(a) \in \top \\
    & \qquad \tand \forall \vsigma : \vDelta \To_r \vGamma.\
      \mtyequiv[\vDelta]{t[\vsigma]}{\Rnf_\alpha(\downarrow^A(a)[\vsigma])}{T[\vsigma]} \\
    \mglutmd t T c[i]A
    &:= \mtyping t T \tand \mgluty T[i]A \tand c \in \bot \\
    & \qquad \tand \forall \vsigma : \vDelta \To_r \vGamma.\
      \mtyequiv[\vDelta]{t[\vsigma]}{\Rne_\alpha(c[\vsigma])}{T[\vsigma]} 
  \end{align*}
\end{definition}

\begin{lemma}\labeledit{lem:dt:glu-var}
  If $\mgluty T[i]A$, then $\mglutmd[\vGamma; (\Gamma, x : T)]{x}{T[\wk_x]}{\alpha^{-1}(x)}[i]A$. 
\end{lemma}
\begin{proof}
  Assume $\vsigma : \vDelta \To_r \vGamma; (\Gamma, x : T)$ and then induction on it.
  We know $\vsigma \approx \wk_{x_1} \circ \cdots \circ \wk_{x_m}$ for some
  $m$. That means $\vsigma$ is just some local weakening. Then our goal
  \begin{align*}
    \mtyequiv[\vDelta]{x[\vsigma]}{\alpha(\alpha^{-1}(x))}{T[\wk_x \circ \vsigma]}
  \end{align*}
  is immediate by reflexivity.

  At last, we show $\alpha^{-1}(x) \in \El_i(A)$ by \Cref{lem:dt:varbot} and
  realizability of the PER model. 
\end{proof}

\begin{theorem}[Realizability]
  Given $A \in \Uc_i$, 
  \begin{itemize}
  \item If $\mglutmd t T c[i]A$, then $\mglutm t T{\uparrow^A(c)}[i]A$.
  \item If $\mglutm t T a[i]A$, then $\mglutmu t T a[i]A$. 
  \item If $\mgluty T[i]A$, then $\mglutyu T [i] A$. 
  \end{itemize}
\end{theorem}
\begin{proof}
  We do well-foundned induction on $i$ and then induction on $A \in \Uc_i$.
  \begin{itemize}[label=Case]
  \item $A = \uparrow^{A'}(C)$,
    \begin{mathpar}
      \inferrule
      {C \approx C \in \bot}
      {\uparrow^{A'}(C) \approx \uparrow^{A'}(C) \in \Uc_i}
    \end{mathpar}
    All statements are immediate by definitions. 

  \item $A = \Nd$,
    \begin{mathpar}
      \inferrule
      { }
      {\Nd \approx \Nd \in \Uc_i}
    \end{mathpar}
    We shall prove the second statement. We do induction on $\mglunat t a$. Given
    $\vsigma : \vDelta \To_r \vGamma$, we should prove
    $\mtyequiv[\vDelta]{t[\vsigma]}{\Rnf_\alpha(\downarrow^\Nd(a)[\vsigma])}{\Nat}$.
    \begin{itemize}[label=Subcase]
    \item
      \begin{mathpar}
        \inferrule
        {\mtyequiv{t}{\ze}{\Nat}}
        {\mglunat{t}{\zed}}
      \end{mathpar}
      Immediate.
      
    \item
      \begin{mathpar}
        \inferrule
        {\mtyequiv{t}{\su t'}{\Nat} \\ \mglunat{t'}{b}}
        {\mglunat{t}{\sud b}}
      \end{mathpar}
      \begin{align*}
        & \mtyequiv[\vDelta]{t'[\vsigma]}{\Rnf_\alpha(\downarrow^\Nd(b)[\vsigma])}{\Nat}
          \byIH \\
        & \mtyequiv[\vDelta]{t[\vsigma]}{\Rnf_\alpha(\downarrow^\Nd(a)[\vsigma])}{\Nat}
          \tag{by congruence}
      \end{align*}
      
    \item
      \begin{mathpar}
        \inferrule
        {c \in \bot \\
          \forall \vsigma : \vDelta \To_r \vGamma.\ \mtyequiv{t[\vsigma]}{\Rne_\alpha(c[\vsigma])}{\Nat}}
        {\mglunat{t}{\uparrow^\Nd(c)}}
      \end{mathpar}
      Immediate by definition. 
    \end{itemize}
    
  \item $A = \Ud_j$,
    \begin{mathpar}
      \inferrule
      {j < i}
      {\Ud_j \approx \Ud_j \in \Uc_i}
    \end{mathpar}
    To prove that if $\mglutmd t T c[i]{\Ud_j}$, then $\mglutm t
    T{\uparrow^A(c)}[i]{\Ud_j}$, notice that the premise gives
    \begin{align*}
      \forall \vsigma : \vDelta \To_r \vGamma.\
      \mtyequiv[\vDelta]{t[\vsigma]}{\Rne_\alpha(c[\vsigma])}{\Se_j}
    \end{align*}
    This proposition, when given $\vect I : \vGamma \To_r \vGamma$, is just $\mgluty t[j]{\uparrow^A(c)}$ and is our
    goal. 
    
  \item $A = \squared A'$,
    \begin{mathpar}
      \inferrule
      {\forall \kappa. A'[\kappa] \approx A'[\kappa] \in \Uc_i}
      {\squared A' \approx \squared A' \in \Uc_i}
    \end{mathpar}
    \begin{itemize}
    \item
      \begin{align*}
        &\ \mtyping{T \approx \square T'}{\Se_i}
          \tag{by assumption} \\ 
        H_1: &\ \forall \vsigma : \vDelta \To_r \vGamma.\
               \mtyequiv[\vDelta]{t[\vsigma]}{\Rne_\alpha(c[\vsigma])}{T[\vsigma]}
               \tag{by assumption} \\
        &\ \text{assume }\vDelta',  \vdash \vDelta; \vDelta', \vsigma : \vDelta \To_r \vGamma, \vdelta :
          \vDelta'' \To_r \vDelta; \vDelta', \text{let }n := |\vDelta'| \\
        &\ \mtyequiv[\trunc{\vDelta''}{\Ltotal\vdelta n}]{t[\vsigma \circ \trunc \vdelta n]}{\Rne_\alpha(c[\vsigma \circ \trunc \vdelta n])}{T[\vsigma \circ \trunc \vdelta n]}
          \tag{by $H_1$} \\
        &\ \text{let }n' := \Ltotal\vdelta n \\
        &\ \mtyequiv[\vDelta'']{\unbox{n'}(t[\vsigma \circ
          \trunc \vdelta n])}{\Rne_\alpha(\tunbox(n', c[\vsigma \circ
          \trunc \vdelta n]))}{T'[\sextt{(\vsigma \circ \trunc \vdelta n)}{n'}]}
          \tag{by congruence and conversion} \\
        &\ \mtyequiv[\vDelta'']{\unbox{n}(t[\vsigma])[\vdelta]}{\Rne_\alpha(\tunbox(n,
          c[\vsigma])[\vdelta])}{T'[\sextt\vsigma n][\vdelta]}
          \tag{by congruence and conversion} \\
        &\ \mglutmd[\vDelta; \vDelta']{\unbox{n}(t[\vsigma])}{T'[\sextt\vsigma
          n]}{\tunbox(n, c[\vsigma])}[i]{A'[\sextt\vsigma n]}
          \tag{by abstraction} \\
        &\ \mglutm t T{\downarrow^A(c)}[i]A
          \tag{by abstraction}
      \end{align*}
      
    \item
      \begin{align*}
        &\ \mjudge{T \approx \square T'}
          \tag{by assumption} \\ 
        H_2: &\ \forall \vDelta', \vdash \vDelta; \vDelta', \vsigma : \vDelta \To_r \vGamma. \\
        &\ \mglutm[\vDelta; \vDelta']{\unbox{|\vDelta'|}{(t[\vsigma])}}{T'[\sextt\vsigma{|\vDelta'|}]}{\tunbox \cdot (|\vDelta'|,
          a[\vsigma])}[i]{A'[\sextt{\mt{\vsigma}}{|\vDelta'|}]} 
          \tag{by assumption} \\
        &\ \text{assume }\vsigma : \vDelta \To_r \vGamma \\
        &\ \mglutm[\vDelta; \cdot]{\unbox 1{(t[\vsigma])}}{T'[\sextt\vsigma 1]}
          {\tunbox \cdot (1,
          a[\vsigma])}[i]{A'[\sextt{\mt{\vsigma}} 1]}
          \tag{by $H_2$} \\
        H_3: &\ \mglutmu[\vDelta; \cdot]{\unbox 1{(t[\vsigma])}}{T'[\sextt\vsigma 1]}
          {\tunbox \cdot (1,
          a[\vsigma])}[i]{A'[\sextt{\mt{\vsigma}} 1]}
          \byIH \\
        &\ \mtyequiv[\vDelta; \cdot]{\unbox
          1{(t[\vsigma])}}{\Rnf_\alpha(\downarrow^{A'[\sextt{\mt{\vsigma}} 1]}(\tunbox
          \cdot (1, a[\vsigma])))}{T'[\sextt\vsigma 1]}
          \tag{by $H_3$} \\
        &\ \mtyequiv[\vDelta]{t[\vsigma]}{\Rnf_\alpha(\downarrow^{\squared
          A'[\vsigma]}(a[\vsigma]))}{\square T'[\vsigma]}
          \tag{by congruence and conversion} \\
        &\ \mtyequiv[\vDelta]{t[\vsigma]}{\Rnf_\alpha(\downarrow^{\squared
          A'}(a)[\vsigma])}{T[\vsigma]}
          \tag{by definition} \\
        &\ \mglutmu t T a[i]A
          \tag{by abstraction}
      \end{align*}
      
    \item
      \begin{align*}
        &\ \mjudge{T \approx \square T'}
          \tag{by assumption} \\ 
        &\ \text{assume }\vsigma : \vDelta \To_r \vGamma \\
        H_4: &\ \forall \vDelta', \vsigma : \vDelta
               \To_r \vGamma.\ \mgluty[\vDelta; \vDelta']{T'[\sextt\vsigma{|\vDelta'|}]}[i]{A'[\sextt\vsigma{|\vDelta'|}]}
               \tag{by assumption} \\
        &\ \mgluty[\vDelta; \cdot]{T'[\sextt\vsigma 1]}[i]{A'[\sextt\vsigma 1]}
          \tag{by $H_4$} \\
        &\ \mtyequiv[\vDelta; \cdot]{T'[\sextt\vsigma 1]}{\Rty_\alpha(A'[\sextt\vsigma 1])}{\Se_i}
          \byIH \\
        &\ \mtyequiv[\vDelta]{\square T'[\vsigma]}{\Rty_\alpha(\squared
          A'[\vsigma])}{\Se_i}
          \tag{by congruence} \\
        &\ \mtyequiv[\vDelta]{T[\vsigma]}{\Rty_\alpha(\squared
          A'[\vsigma])}{\Se_i}
          \tag{by congruence}
      \end{align*}
    \end{itemize}
    
  \item $A = Pi(A', x.T', \vrho')$,
    \begin{mathpar}
      \inferrule
      {\forall \kappa. A'[\kappa] \approx A'[\kappa] \in \Uc_i \\
        \forall \kappa, a \approx a' \in
        \El_i(A'[\kappa]). \intp{T'}(\ext(\vrho'[\kappa], x, a))
        \approx \intp{T'}(\ext(\vrho'[\kappa], x, a')) \in \Uc_i}
      {Pi(A', x.T', \vrho') \approx Pi(A', x.T', \vrho') \in \Uc_i}
    \end{mathpar}
    \begin{itemize}
    \item
      \begin{align*}
        &\ \mtyequiv{T}{\Pi(x:S). T'}{\Se_i}
          \tag{by assumption, for some $S$ and $T'$} \\
        H_5: &\ \forall \vsigma : \vDelta \To_r \vGamma.\
               \mtyequiv[\vDelta]{t[\vsigma]}{\Rne_\alpha(c[\vsigma])}{T[\vsigma]}
               \tag{by assumption} \\
        &\ \text{assume }\vsigma : \vDelta \To_r \vGamma,
          \mglutm[\vDelta]{s}{S[\vsigma]}{b}[i]{A'[\vsigma]}, \vdelta : \vDelta' \To_r \vDelta \\
        &\ \mtyequiv[\vDelta']{t[\vsigma \circ \vdelta]}{\Rne_\alpha(c[\vsigma \circ \vdelta])}{T[\vsigma \circ \vdelta]}
          \tag{by $H_5$} \\
        H_6: &\ \mglutmu[\vDelta]{s}{S[\vsigma]}{b}[i]{A'[\vsigma]}
          \byIH \\
        &\
          \mtyequiv[\vDelta']{s[\vdelta]}{\Rnf_\alpha(\downarrow^{A'[\vsigma]}(b)[\vdelta])}{S[\vsigma][\vdelta]}
          \tag{by $H_6$} \\
        &\ \mtyequiv[\vDelta']{t[\vsigma \circ \vdelta]\
          s[\vdelta]}{\Rne_\alpha(c[\vsigma \circ \vdelta])\
          \Rnf_\alpha(\downarrow^{A'[\vsigma]}(b)[\vdelta])}
          {T'[\vsigma \circ \vdelta; s[\vdelta]/x]}
          \tag{by congruence and conversion} \\
        &\ \mtyequiv[\vDelta']{(t[\vsigma]\ s)[\vdelta]}
          {\Rne_\alpha((c[\vsigma]\ \downarrow^{A'[\vsigma]}(b))[\vdelta])}
          {T'[\vsigma; s/x][\vdelta]}
          \tag{by congruence and conversion} \\
        &\ \mglutmd[\vDelta']{(t[\vsigma]\ s)}{T'[\vsigma; s/x]}
          {c[\vsigma]\ \downarrow^{A'[\vsigma]}(b)}[i]{\intp{T'}(\ext(\vrho'[\vsigma], x, b))}
          \tag{by abstraction} \\
        &\ \mglutm[\vDelta']{(t[\vsigma]\ s)}{T'[\vsigma; s/x]}
          {c[\vsigma]\ \downarrow^{A'[\vsigma]}(b)}[i]{\intp{T'}(\ext(\vrho'[\vsigma], x, b))}
          \byIH \\
        &\ \mglutm t T a[i]A
          \tag{by abstraction}
      \end{align*}

    \item
      \begin{align*}
        &\ \mtyequiv{T}{\Pi(x : S). T'}{\Se_i}
          \tag{by assumption, for some $S$ and $T'$} \\
        H_7: &\ \forall \vsigma : \vDelta \To_r
               \vGamma, \mglutm[\vDelta]{s}{S[\vsigma]}{b}[i]{A'[\vsigma]}. \\
        & \qquad \mglutm[\vDelta]{t[\vsigma]\ s}{T'[\vsigma, s/x]}{a[\vsigma] \cdot
          b}[i]{\intp{T'}(\ext(\vrho'[\vsigma], x, b))}
          \tag{by assumption} \\
        &\ \text{assume }\vsigma : \vDelta;\Delta \To_r \vGamma \\
        &\ \text{let }\vsigma' := \vsigma \circ \wk_x : \vDelta;(\Delta, x :
          S[\vsigma]) \To_r \vGamma \\
        &\ \mglutmd[\vDelta; (\Delta, x : S[\vsigma])]{x}{S[\vsigma']}z[i]{A'[\vsigma]}
          \tag{by \Cref{lem:dt:glu-var}, for some $z$} \\
        &\ \mglutm[\vDelta; (\Delta, x : S[\vsigma])]{x}{S[\vsigma']}{\uparrow^{A'[\vsigma]}(z)}[i]{A'[\vsigma]}
          \byIH \\
        &\ \text{let }b := \uparrow^{A'[\vsigma]}(z) \\
        &\ \mglutm[\vDelta; (\Delta, x : S[\vsigma])]{t[\vsigma']\ x}{T'[\vsigma', x/x]}{a[\vsigma] \cdot
          b}[i]{\intp{T'}(\ext(\vrho'[\vsigma], x, b))}
          \tag{by $H_7$}\\
        H_8: &\ \mglutmu[\vDelta; (\Delta, x : S[\vsigma])]{t[\vsigma']\ x}{T'[\vsigma', x/x]}{a[\vsigma] \cdot
               b}[i]{\intp{T'}(\ext(\vrho'[\vsigma], x, b))}
               \byIH \\
        &\ \mtyequiv[\vDelta; (\Delta, x : S[\vsigma])]{t[\vsigma']\
          x}{\Rnf_\alpha(\downarrow^{\intp{T'}(\ext(\vrho'[\vsigma], x, b))}(a[\vsigma] \cdot
          b))}{T'[q(\vsigma, x)]}
          \tag{by $H_8$} \\
        &\ \mtyequiv[\vDelta; \Delta]{t[\vsigma]}{\Rnf_\alpha(\downarrow^{Pi(A', x.T', \vrho')[\vsigma]}(a[\vsigma]))}{\Pi(x : S).T'[\vsigma]}
          \tag{by congruence and conversion} \\
        &\ \mtyequiv[\vDelta; \Delta]{t[\vsigma]}{\Rnf_\alpha(\downarrow^{Pi(A', x.T', \vrho')}(a)[\vsigma])}{\Pi(x : S).T'[\vsigma]}
          \tag{by definition} \\
        &\ \mglutmu t T a[i]A
          \tag{by abstraction}
      \end{align*}
      
    \item
      \begin{align*}
        &\ \mtyequiv{T}{\Pi(x : S). T'}{\Se_i}
          \tag{by assumption, for some $S$ and $T'$} \\
        H_9: &\ \forall \vsigma : \vDelta \To_r \vGamma, 
               \mglutm[\vDelta]{s}{S[\vsigma]}{b}[i]{A'[\vsigma]}.\ \mgluty[\vDelta]{T'[\vsigma, s/x]}[i]{\intp{T'}(\ext(\vrho'[\vsigma], x, b))} \\
        &\ \text{assume }\vsigma : \vDelta;\Delta \To_r \vGamma \\
        &\ \text{let }\vsigma' := \vsigma \circ \wk_x : \vDelta;(\Delta, x :
          S[\vsigma]) \To_r \vGamma \\
        &\ \mglutm[\vDelta; (\Delta, x : S[\vsigma])]{x}{S[\vsigma']}{b}[i]{A'[\vsigma]}
          \tag{by IH, let $b := \uparrow^{A'[\vsigma]}(z)$} \\
        &\ \mgluty[\vDelta; (\Delta, x : S[\vsigma])]{T'[q(\vsigma, x)]}[i]{\intp{T'}(\ext(\vrho'[\vsigma], x, b))}
          \tag{by $H_9$} \\
        &\ \mtyequiv[\vDelta; (\Delta, x :
          S[\vsigma])]{T'[q(\vsigma, x)]}{\Rty_\alpha(\intp{T'}(\ext(\vrho'[\vsigma], x, b)))}{\Se_i}
          \byIH \\
        &\ \mtyequiv[\vDelta; \Delta]{\Pi(x :
          S[\vsigma]).(T'[q(\vsigma, x)])}{\Rty_\alpha(Pi(A', x.T', \vrho')[\vsigma])}{\Se_i}
          \tag{by congruence} \\
        &\ \mtyequiv[\vDelta; \Delta]{\Pi(x :
          S).T'[\vsigma]}{\Rty_\alpha(Pi(A', x.T', \vrho')[\vsigma])}{\Se_i}
      \end{align*}
    \end{itemize}
    
  \end{itemize}
\end{proof}

From realizability, we see that the two syntactic types glued by the same value are
equivalent.
\begin{lemma}\labeledit{lem:ty-glu-d}
  If $\mgluty T[i]A$ and $\mgluty{T'}[i]A$, then $\mtyequiv{T}{T'}{\Se_i}$. 
\end{lemma}
\begin{proof}
  By realizability, we have
  \begin{align*}
    T \approx T[\vect I] \approx \Rty_\alpha(A) \approx T'[\vect I] \approx T'
  \end{align*}
\end{proof}

\subsection{Cumulativity}

Similar to the completeness case (\Cref{lem:dt:per-cumu}), when proving cumulativity,
we also need to mutually prove a lowering theorem. This again is because in the
dependent function case, the input is in a contravariant position. When we lift a
level in the output, we are only given an input on a higher level. Unlike in the
completeness case, where cumulativity and lowering are rather uniform, in the gluing
model, lowering does not come for free: since the gluing model also stores syntactic
information, we must have enough syntactic information in order to prove
lowering. More precisely, we prove cumulativity as follows:
\begin{lemma}[Cumulativity and lowering]
  Given $A \in \Uc_i$, 
  \begin{itemize}
  \item if $\mgluty T[i]A$, then $\mgluty T[1 + i]A$;
  \item if $\mglutm t T a [i] A$, then $\mglutm t T a [1 + i] A$;
  \item if $\mglutm t T a [1 + i] A$ and $\mgluty T[i]A$, then $\mglutm t T a [i] A$.
  \end{itemize}
\end{lemma}
The first two statements are intuitive cumulativity. The last one is lowering. Notice
that in order to lower $\mglutm t T a [1 + i] A$ by one level, we must know that $T$
and $A$ are related one level lower (i.e. the premise $\mgluty T[i]A$).
\begin{proof}
  We proceed by induction on $A \in \Uc_i$. The trickiest cases is $Pi$. When proving
  cumulativity, we must apply lowering to the argument so that the input level is
  lowered by one, using which we apply IH; vice versa, we apply cumulativity when
  proving lowering, where we must increase the input level by one so that we can use
  it to apply IH. 
\end{proof}

Once we establish one-step cumulativity, we can generalize it to any greater level:
\begin{lemma}[Cumulativity]
  If $i \le j$ and $\mgluty T[i]A$ and $\mglutm t T a[i]A$, resp., then $\mgluty T[j] A$ and
  $\mglutm t T a[j]A$, resp..
\end{lemma}

Due to cumulativity, we can do the same as the PER models, where we hide the levels by
taking limits of the gluing model. That is
\begin{align*}
  \mgluty T A &:= \mgluty T[\infty]A \\
  \mglutm t T a A &:= \mglutm t T a[\infty]A 
\end{align*}
These limits are well defined because if $T$ and $A$ are related by the limit, then
cumulativity ensures that they must be related by some concrete level in natural
numbers (but we do not have to be concerned about its actual level). Similarly, if $t$
and $a$ are related by the limit, they are actually related at some level. 

\subsection{Gluing of Substitutions}

We generalize $\mglutm t T a A$ to a gluing relation between substitutions and
evaluation environments.
\begin{definition}
  $\mglutms \vsigma \vDelta \vrho$ relates $\vsigma$ and $\vrho$ and is defined by
  induction on $\vDelta$:
  \begin{align*}
    \mglutms \vsigma{\epsilon; \cdot} \vrho
    &:= \mtyping \vsigma \epsilon; \cdot \\
    \mglutms \vsigma{\vDelta; \cdot} \vrho
    &:= \mtyping{\vsigma}{\vDelta; \cdot} \\
    &\tand \exists \vsigma', n.\ \mtyequiv[\trunc\vGamma n]{\trunc \vsigma 1}{\vsigma'}{\vDelta}
      \tand \Ltotal\vsigma 1 = \Ltotal\vrho 1 = n \tand \\
    & \qquad \mglutms[\trunc \vGamma n]{\vsigma'}{\vDelta}{\trunc \vrho 1} \\
    \mglutms \vsigma{\vDelta; (\Delta, x : T)} \vrho 
    &:= \mtyping{\vsigma}{\vDelta; (\Delta, x : T)} \\
    & \tand \exists \vsigma', t.\ \mtyequiv{\wk_x \circ \vsigma}{\vsigma'}{\vDelta; \Delta}
      \tand \mtyequiv{x[\vsigma]}{t}{T[\vsigma']} \tand \\
    & \qquad \intp{T}(\drop(\vrho,x)) \in \Uc \tand \\
    & \qquad \mglutm{t}{T[\vsigma']}{\rho(x)}{\intp{T}(\drop(\vrho,x))} \tand \\
    & \qquad \mglutms{\vsigma'}{\vDelta; \Delta}{\drop(\vrho, x)}
      \tag{where $(\_, \rho) := \vrho(0)$}
  \end{align*}
\end{definition}

\begin{lemma}\labeledit{lem:dt:glue-L}
  If $\mglutms{\vsigma}{\vDelta}{\vrho}$, then $\Ltotal\vsigma n = \Ltotal\vrho n$.
\end{lemma}
\begin{lemma}\labeledit{lem:dt:glue-trunc}
  If $\mglutms{\vsigma}{\vDelta}{\vrho}$, then $\mglutms[\trunc\vGamma{\Ltotal\vsigma n}]{\trunc \vsigma n}{\trunc \vDelta n}{\trunc \vrho n}$.
\end{lemma}
\begin{proof}
  Immediate by induction on $n$ and $\vDelta$. 
\end{proof}

\begin{lemma}
  If $\mglutms{\vsigma}{\vDelta}{\vrho}$, then $\mtyping \vsigma \vDelta$.
\end{lemma}
\begin{proof}
  Immediate.
\end{proof}

\begin{lemma}\labeledit{lem:dt:glu-stack-vrho}
  If $\mglutms{\vsigma}{\vDelta}{\vrho}$, then $\vrho \in \intp{\vDelta}$. 
\end{lemma}
\begin{proof}
  Notice that $\mtyping \vsigma \vDelta$ and thus $\msemtyp \vsigma \vDelta$ and $\vDash \vDelta$, so $\vrho \in
  \intp{\vDelta}$ is a well-defined statement. We do induction on $\vDelta$.
  \begin{itemize}[label=Case]
  \item $\vDelta = \epsilon; \cdot$ or $\vDelta = \vDelta' ;\cdot$, then it is
    immediate. Apply IH when necessary. 
  \item $\vDelta = \vDelta'; (\Delta, x : T)$, then
    \begin{align*}
      H_1: &\ \mglutm{t}{T[\vsigma']}{\rho(x)}{\intp{T}(\drop(\vrho,x))}
      \tag{by assumption} \\
      H_2: &\ \mglutms{\vsigma'}{\vDelta; \Delta}{\drop(\vrho, x)}
             \tag{by assumption} \\
           &\ \rho(x) \in \El(\intp{T}(\drop(\vrho,x)))
             \tag{by $H_1$ and \Cref{lem:dt:inv-glu-tm}} \\
           &\ \drop(\vrho, x) \in \intp{\vDelta; \Delta}
             \byIH \\
           &\ \vrho \in \intp{\vDelta}
             \tag{by definition}
    \end{align*}
  \end{itemize}
\end{proof}

 \begin{lemma}\labeledit{lem:dt:glue-stack-conv}
  If $\mglutms{\vsigma}{\vDelta}{\vrho}$ and $\vdash \vDelta \approx \vDelta'$,  then
  $\mglutms{\vsigma}{\vDelta'}{\vrho}$. 
\end{lemma}
\begin{proof}
  We do induction on $\vdash \vDelta \approx \vDelta'$.
  \begin{itemize}[label=Case]
  \item
    \begin{mathpar}
      \inferrule
      { }
      {\vdash \epsilon; \cdot \approx \epsilon; \cdot}
    \end{mathpar}
    Immediate.
    
  \item $\vDelta = \vDelta''; \cdot$, $\vDelta' = \vDelta'''; \cdot$, 
    \begin{mathpar}
      \inferrule
      {\vdash \vDelta'' \approx \vDelta'''}
      {\vdash \vDelta''; \cdot \approx \vDelta'''; \cdot}
    \end{mathpar}
    Apply conversion when necessary. Finally,
    \begin{align*}
      H_1: &\ \mglutms[\trunc \vGamma n]{\vsigma'}{\vDelta}{\trunc \vrho 1}
             \tag{by assumption, for some $\vsigma'$ and $n$} \\
           &\ \mglutms[\trunc \vGamma n]{\vsigma'}{\vDelta'}{\trunc \vrho 1}
             \tag{by $H_1$} \\
           &\ \mglutms{\vsigma}{\vDelta'}{\vrho}
    \end{align*}
    
  \item $\vDelta = \vDelta''; (\Delta, x : T)$, $\vDelta' = \vDelta'''; (\Delta', x : T')$,
    \begin{mathpar}
      \inferrule
      {\vdash \vDelta''; \Delta \approx \vDelta'''; \Delta'  \\ \mtyequiv[\vDelta''; \Delta]{T}{T'}{\Se_i}}
      {\vdash \vDelta''; (\Delta, x : T) \approx \vDelta'''; (\Delta', x : T')}
    \end{mathpar}
    Interestingly, we need to apply the fundamental theorem of the PER model:
    \begin{align*}
      H_2: &\ \vdash \vDelta''; \Delta \approx \vDelta'''; \Delta'
             \tag{by assumption} \\
      H_3: &\ \mtyequiv[\vDelta''; \Delta]{T}{T'}{\Se_i}
             \tag{by assumption} \\
      H_4: &\ \vDash \vDelta''; \Delta \approx \vDelta'''; \Delta'
             \tag{by \Cref{thm:dt:per-fund}} \\
      H_5: &\ \msemtyeq[\vDelta''; \Delta]{T}{T'}{\Se_i}
             \tag{by \Cref{thm:dt:per-fund}} \\
           &\ \vrho \in \intp{\vDelta}
             \tag{by \Cref{lem:dt:glu-stack-vrho}} \\
           &\ \vrho \in \intp{\vDelta'}
             \tag{by \Cref{lem:dt:intpvg-resp-eq}} \\
           &\ \drop(\vrho, x) \in \intp{\vDelta'''; \Delta'}
             \tag{by definition} \\
           &\ \intp{T}(\drop(\vrho, x)) \approx \intp{T'}(\drop(\vrho, x)) \in \Uc_i
             \tag{by $H_3$} \\
           &\ \intp{T'}(\drop(\vrho, x)) \in \Uc
             \tag{PER property} \\
           &\ \mglutm{t}{T[\vsigma']}{\rho(x)}{\intp{T'}(\drop(\vrho,x))}
             \tag{by \Cref{lem:dt:glu-resp-u}} \\
           &\ \mglutms{\vsigma'}{\vDelta'''; \Delta'}{\drop(\vrho, x)}
             \byIH \\
           &\ \mglutms{\vsigma}{\vDelta'}{\vrho}
             \tag{by definition}
    \end{align*}
  \end{itemize}
\end{proof}

 \begin{lemma}\labeledit{lem:dt:glue-stack-resp-equiv}
  If $\mglutms{\vsigma}{\vDelta}{\vrho}$ and $\mtyequiv{\vdelta}{\vsigma}{\vDelta}$,  then
  $\mglutms{\vdelta}{\vDelta}{\vrho}$. 
\end{lemma}
\begin{proof}
  We do induction on $\vDelta$. In each case, we apply presupposition and properties of
    $L$ and truncation.
\end{proof}

\begin{lemma}\labeledit{lem:dt:glue-stack-mon}
  If $\mglutms{\vsigma}{\vDelta}{\vrho}$ and $\vdelta : \vGamma' \To_r \vGamma$, then
  $\mglutms[\vGamma']{\vsigma \circ \vdelta}{\vDelta}{\vrho[\vdelta]}$.
\end{lemma}
\begin{proof}
  Immediate by induction on $\vDelta$.
  \begin{itemize}[label=Case]
  \item $\vDelta = \epsilon; \cdot$, immediate.
  \item $\vDelta = \vDelta'; \cdot$, then
    \begin{align*}
      H_1: &\ \mtyequiv[\trunc \vGamma n]{\trunc \vsigma 1}{\vsigma'}{\vDelta}
             \tag{by assumption, for some $\vsigma'$ and $n$}\\
           &\ \Ltotal\vsigma 1 = \Ltotal\vrho 1 = n
             \tag{by assumption} \\
      H_2: &\ \mglutms[\trunc \vGamma n]{\vsigma'}{\vDelta'}{\trunc \vrho 1}
             \tag{by assumption} \\
           &\ \mtyequiv[\trunc{\vGamma'}{\Ltotal\vdelta n}]{\trunc \vsigma 1 \circ
             \trunc \vdelta n}{\vsigma' \circ \trunc \vdelta n}{\vDelta}
             \tag{by $H_1$} \\
           &\ \mtyequiv[\trunc{\vGamma'}{\Ltotal\vdelta n}]{\trunc{(\vsigma \circ
             \vdelta)}1}{\vsigma' \circ \trunc \vdelta n}{\vDelta} \\
           &\ \Ltotal{(\vsigma \circ \vdelta)}1 = \Ltotal\vdelta{\Ltotal\vsigma 1} = \Ltotal\vdelta
             {\Ltotal\vrho 1} = \Ltotal{\vrho[\vdelta]}1 = \Ltotal\vdelta n
    \end{align*}
    Thus we let the existentials in the goal to be $\vsigma \circ \trunc \vdelta n$ and
    $\Ltotal\vdelta n$.
    \begin{align*}
      &\ \mglutms[\trunc{\vGamma'}{\Ltotal\vdelta n}]{\vsigma' \circ
        \trunc \vdelta n}{\vDelta'}{\trunc \vrho 1[\trunc \vdelta n]}
        \byIH \\
      &\ \mglutms[\trunc{\vGamma'}{\Ltotal\vdelta n}]{\vsigma' \circ
        \trunc \vdelta n}{\vDelta'}{\trunc{\vrho[\vdelta]}1}
        \tag{by definition}  \\
      &\ \mglutms{\vdelta}{\vDelta}{\vrho}
        \tag{by abstraction}
    \end{align*}
    
  \item $\vDelta = \vDelta'; (\Delta, x : T)$, then
    \begin{align*}
      H_3: &\ \mtyequiv{\wk_x \circ \vsigma}{\vsigma'}{\vDelta; \Delta}
      \\
      H_4: &\ \mtyequiv{x[\vsigma]}{t}{T[\vsigma']}
             \tag{by assumption, for some $\vsigma'$ and $t$} \\
      H_5: &\ \intp{T}(\drop(\vrho,x)) \in \Uc
             \tag{by assumption} \\
      H_6: &\ \mglutm{t}{T[\vsigma']}{\rho(x)}{\intp{T}(\drop(\vrho,x))}
             \tag{by assumption} \\
      H_7: &\ \mglutms{\vsigma'}{\vDelta; \Delta}{\drop(\vrho, x)}
             \tag{by assumption}
    \end{align*}
    Notice that
    \begin{align*}
      &\ \mtyequiv[\vGamma']{\wk_x \circ \vsigma \circ \vdelta}{\vsigma' \circ
        \vdelta}{\vDelta; \Delta}
      \\
      &\ \mtyequiv[\vGamma']{x[\vsigma \circ \vdelta]}{t[\vdelta]}{T[\vsigma' \circ \vdelta]}
    \end{align*}
    so we let the existentials in the goal be $\vsigma' \circ \vdelta$ and
    $t[\vdelta]$.
    \begin{align*}
      &\ \intp{T}(\drop(\vrho[\vdelta],x)) \in \Uc
        \tag{by \Cref{lem:dt:u-el-resp-mon}} \\
      &\ \mglutm{t[\vdelta]}{T[\vsigma' \circ \vdelta]}{\rho[\vdelta](x)}{\intp{T}(\drop(\vrho[\vdelta],x))}
        \tag{by \Cref{lem:dt:glu-mon,lem:dt:glutm-resp-equiv-ty}} \\
      &\ \mglutms[\vGamma']{\vsigma' \circ \vdelta}{\vDelta;
        \Delta}{\drop(\vrho[\vdelta], x)}
        \byIH \\
      &\ \mglutms{\vdelta}{\vDelta}{\vrho}
        \tag{by abstraction}
    \end{align*}
  \end{itemize}
\end{proof}

Finally, we can define the semantic judgment:
\begin{definition}
  \begin{align*}
    \mSemtyp t T &:= \forall \mglutms[\vDelta]{\vsigma}{\vGamma}{\vrho}.\
    \mglutm[\vDelta]{t[\vsigma]}{T[\vsigma]}{\intp{t}(\vrho)}{\intp{T}(\vrho)} \\
    \mSemtyp{\vdelta}{\vGamma'} &:= \forall
                                  \mglutms[\vDelta]{\vsigma}{\vGamma}{\vrho}.\
                                  \mglutms[\vDelta]{\vdelta \circ \vsigma}{\vGamma'}{\intp{\vdelta}(\vrho)}
  \end{align*}
\end{definition}

\subsection{Semantic Judgments}

\begin{lemma}
  \begin{mathpar}
    \inferrule
    {\mSemtyp{t}{T} \\ \mtyequiv{T}{T'}{\Se_i}}
    {\mSemtyp{t}{T'}}
  \end{mathpar}
\end{lemma}
\begin{proof}
  \begin{align*}
    H_1: &\ \mSemtyp{t}{T}
           \tag{by assumption} \\
    H_2: &\ \mtyequiv{T}{T'}{\Se_i}
           \tag{by assumption} \\
    H_3: &\ \mglutms[\vDelta]{\vsigma}{\vGamma}{\vrho}
           \tag{by assumption} \\
    H_4: &\ \mglutm[\vDelta]{t[\vsigma]}{T[\vsigma]}{\intp{t}(\vrho)}{\intp{T}(\vrho)}
           \tag{by $H_1$ and $H_3$}\\
         &\ \mtyequiv[\vDelta]{T[\vsigma]}{T'[\vsigma]}{\Se_i}
    \\
    H_5: &\ \mglutm[\vDelta]{t[\vsigma]}{T'[\vsigma]}{\intp{t}(\vrho)}{\intp{T}(\vrho)}
           \tag{by \Cref{lem:dt:glutm-resp-equiv-ty}} \\
    H_6: &\ \msemtyeq{T}{T'}{\Se_i}
           \tag{by \Cref{thm:dt:per-fund}} \\
         &\ \vrho \in \intp{\vGamma}
           \tag{by \Cref{lem:dt:glu-stack-vrho}}\\
         &\ \intp{T}(\vrho) \approx \intp{T'}(\vrho) \in \Uc_i
           \tag{by $H_6$} \\
         &\ \mglutm[\vDelta]{t[\vsigma]}{T'[\vsigma]}{\intp{t}(\vrho)}{\intp{T'}(\vrho)}
           \tag{by $H_5$ and \Cref{lem:dt:glu-resp-u}}
  \end{align*}
\end{proof}

\begin{lemma}
  \begin{mathpar}
    \inferrule
    {\mSemtyp{\vsigma}{\vDelta} \\ \vdash \vDelta \approx \vDelta'}
    {\mSemtyp{\vsigma}{\vDelta'}}
  \end{mathpar}
\end{lemma}
\begin{proof}
  Similar to the previous lemma, bu apply \Cref{lem:dt:glue-stack-conv} instead.
\end{proof}

\begin{lemma}
  \begin{mathpar}
    \inferrule
    {\mSemtyp[\vDelta]t T \\ \mSemtyp{\vsigma}{\vDelta}}
    {\mSemtyp{t[\vsigma]}{T[\vsigma]}}
  \end{mathpar}
\end{lemma}
\begin{proof}
  \begin{align*}
    H_1: &\ \mSemtyp[\vDelta]t T
           \tag{by assumption} \\
    H_2: &\ \mSemtyp{\vsigma}{\vDelta}
           \tag{by assumption} \\
    H_3: &\ \mglutms[\vDelta']{\vdelta}{\vGamma}{\vrho}
           \tag{by assumption} \\
         &\ \mglutms[\vDelta']{\vsigma \circ \vdelta}{\vDelta}{\intp{\vsigma}(\vrho)}
           \tag{by $H_3$} \\
         &\ \mglutm[\vDelta]{t[\vsigma \circ \vdelta]}{T[\vsigma \circ \vdelta]}{\intp{t}(\intp{\vsigma}(\vrho))}{\intp{T}(\intp{\vsigma}(\vrho))}
           \tag{by $H_1$}\\
         &\
           \mglutm[\vDelta]{t[\vsigma][\vdelta]}{T[\vsigma][\vdelta]}{\intp{t}(\intp{\vsigma}(\vrho))}{\intp{T}(\intp{\vsigma}(\vrho))}
           \tag{by \Cref{lem:dt:glutm-resp-equiv-ty,lem:dt:glutm-resp-equiv-tm}}
  \end{align*}
\end{proof}

The following lemmas are immediate by following definitions.
\begin{lemma}
  \begin{mathpar}
    \inferrule
    {\vdash \vGamma}
    {\mSemtyp{\vect I}{\vGamma}}

    \inferrule
    {\mSemtyp{\vsigma}{\vGamma'; \Gamma} \\ \mSemtyp[\vGamma';
      \Gamma]{T}{\Se_i} \\ \mSemtyp{t}{T[\vsigma]}}
    {\mSemtyp{\vsigma, t/x}{\vGamma';(\Gamma, x : T)}}

    \inferrule
    {\vdash \vGamma'; (\Gamma, x : T)}
    {\mSemtyp[\vGamma'; (\Gamma, x : T)]{\wk_x}{\vGamma';\Gamma}}
    
    \inferrule
    {\mSemtyp[\vGamma']{\vsigma}{\vGamma''} \\ \mSemtyp{\vdelta}{\vGamma'}}
    {\mSemtyp{\vsigma \circ \vdelta}{\vGamma''}}
  \end{mathpar}
\end{lemma}

\begin{lemma}
  \begin{mathpar}
    \inferrule
    {\mSemtyp{\vsigma}{\vDelta} \\ |\vGamma'| = n}
    {\mSemtyp[\vGamma; \vGamma']{\sextt\vsigma n}{\vDelta; \cdot}}
  \end{mathpar}
\end{lemma}
\begin{proof}
  \begin{align*}
    H_1: &\ \mSemtyp{\vsigma}{\vDelta}
           \tag{by assumption} \\
    H_2: &\ \mglutms[\vDelta']{\vdelta}{\vGamma; \vGamma'}{\vrho}
           \tag{by assumption} \\
         &\ \mglutms[\trunc{\vDelta'}{\Ltotal\vdelta n}]{\trunc \vdelta n}{\vGamma}{\trunc \vrho n}
           \tag{by \Cref{lem:dt:glue-trunc}} \\
         &\ \mglutms[\trunc{\vDelta'}{\Ltotal\vdelta n}]{\vsigma \circ
           \trunc \vdelta n}{\vDelta}{\intp{\vsigma}(\trunc \vrho n)}
           \tag{by $H_1$} \\
         &\ \mglutms[\vDelta']{\sextt{\vsigma \circ
           \trunc \vdelta n}{\Ltotal\vdelta n}}{\vDelta;
           \cdot}{\ext(\intp{\vsigma}(\trunc \vrho n), \Ltotal\vdelta n)}
           \tag{by definition} \\
         &\ \mglutms[\vDelta']{(\sextt\vsigma n) \circ
           \vdelta}{\vDelta;
           \cdot}{\ext(\intp{\vsigma}(\trunc \vrho n), \Ltotal\vdelta n)}
           \tag{by \Cref{lem:dt:glue-stack-resp-equiv}}
  \end{align*}
\end{proof}

\begin{lemma}
  \begin{mathpar}
    \inferrule
    {x : T \in \vGamma; \Gamma}
    {\mSemtyp[\vGamma; \Gamma]x T}
  \end{mathpar}
\end{lemma}
\begin{proof}
  \begin{align*}
    H_1: &\ x : T \in \vGamma; \Gamma
           \tag{by assumption} \\
    H_2: &\ \mglutms[\vDelta]{\vsigma}{\vGamma; \Gamma}{\vrho}
           \tag{by assumption}
  \end{align*}
  We proceed by induction on $H_1$.
  \begin{itemize}[label=Case]
  \item $T = T'[\wk_x]$ and $\Gamma = \Gamma', x : T'$, 
    \begin{mathpar}
      \inferrule
      { }
      {x : T'[\wk_x] \in \vGamma; (\Gamma', x : T')}
    \end{mathpar}
    Inverting $H_2$, we have
    \begin{align*}
      & \mtyequiv[\vDelta]{\wk_x \circ \vsigma}{\vsigma'}{\vGamma; \Gamma'}
        \tag{for some $\vsigma'$} \\
      & \mtyequiv[\vDelta]{x[\vsigma]}{t}{T'[\vsigma']}
        \tag{ for some  and $t$}\\
      & \mglutm[\vDelta]{t}{T'[\vsigma']}{\rho(x)}{\intp{T'}(\drop(\vrho, x))}
        \tag{where $(\_, \rho) := \vrho(0)$}
    \end{align*}
    We can conclude
    \begin{align*}
      \mglutm[\vDelta]{x[\vsigma]}{T'[\wk_x][\vsigma]}{\rho(x)}{\intp{T'[\wk_x]}(\vrho)}
    \end{align*}
    by \Cref{lem:dt:glutm-resp-equiv-ty,lem:dt:glutm-resp-equiv-tm}.
    
  \item $T = T'[\wk_x]$ and $\Gamma = \Gamma', s/x$, 
    \begin{mathpar}
      \inferrule
      {x : T' \in \vGamma; \Gamma'}
      {x : T'[\wk_y] \in \vGamma; \Gamma', y : S}
    \end{mathpar}
    Inverting $H_2$, we have
    \begin{align*}
      & \mtyequiv[\vDelta]{\wk_y \circ \vsigma}{\vsigma'}{\vGamma; \Gamma'}
        \tag{for some $\vsigma'$} \\
      & \mglutms[\vDelta]{\vsigma'}{\vGamma; \Gamma'}{\drop(\vrho, y)}
    \end{align*}
    By IH, we have
    \begin{align*}
      \mglutm[\vDelta]{x[\vsigma']}{T'[\vsigma']}{\rho(x)}{\intp{T'}(\drop(\vrho, y))}
    \end{align*}
    and then
    \begin{align*}
      \mglutm[\vDelta]{x[\wk_y][\vsigma]}{T'[\wk_y][\vsigma]}{\rho(x)}{\intp{T'[\wk_y]}(\vrho)}
    \end{align*}
    by \Cref{lem:dt:glutm-resp-equiv-ty,lem:dt:glutm-resp-equiv-tm}.
  \end{itemize}
\end{proof}

\subsubsection{$\Pi$ Types}

\begin{lemma}
  \begin{mathpar}
    \inferrule
    {\mSemtyp[\vGamma; \Gamma]{S}{\Se_i} \\ \mSemtyp[\vGamma; (\Gamma, x : S)]{T}{\Se_i}}
    {\mSemtyp[\vGamma; \Gamma]{\Pi(x : S). T}{\Se_{i}}}
  \end{mathpar}
\end{lemma}
\begin{proof}
  \begin{align*}
    H_1: &\ \mSemtyp[\vGamma; \Gamma]{S}{\Se_i}
           \tag{by assumption} \\
    H_2: &\ \mSemtyp[\vGamma; (\Gamma, x : S)]{T}{\Se_i}
           \tag{by assumption} \\           
    H_3: &\ \mglutms[\vDelta]{\vsigma}{\vGamma; \Gamma}{\vrho}
           \tag{by assumption} \\
    H_4: &\ \mglutm[\vDelta]{S[\vsigma]}{\Se_i}{\intp{S}(\vrho)}{\Ud_i}
           \tag{by $H_1$ and \Cref{lem:dt:glutm-resp-equiv-ty}} \\
         &\ \vrho \in \intp{\vGamma; \Gamma}
           \tag{by \Cref{lem:dt:glu-stack-vrho}} \\
    H_5: &\ \msemtyp[\vGamma; \Gamma]{\Pi(x : S). T}{\Se_{i}}
           \tag{by \Cref{thm:dt:per-fund}} \\
         &\ Pi(\intp{S}(\vrho), x.T, \vrho) \in \Uc_i
           \tag{by $H_5$}
  \end{align*}
  The goal is
  \begin{align*}
    \mgluty[\vDelta]{\Pi(x : S).T[\vsigma]}{Pi(\intp{S}(\vrho), x.T, \vrho)}
  \end{align*}
  First, we know
  \begin{align*}
    \mjudge[\vDelta]{\Pi(x : S).T[\vsigma] \approx \Pi(x : S[\vsigma]).(T[q(\vsigma, x)])}
  \end{align*}
  This discharges the existential. Assuming $\vdelta : \vDelta' \To_r \vDelta$ and
  $\mglutm[\vDelta']{s}{S[\vsigma][\vdelta]}{a}{\intp{S}(\vrho)[\vdelta]}$, we have
  \begin{align*}
    &\ \mglutm[\vDelta']{s}{S[\vsigma\circ\vdelta]}{a}{\intp{S}(\vrho[\vdelta])}
      \tag{by \Cref{lem:dt:glutm-resp-equiv-ty}} \\
    &\ \mglutms[\vDelta']{\vsigma \circ \vdelta}{\vGamma; \Gamma}{\vrho[\vdelta]}
      \tag{by \Cref{lem:dt:glue-stack-mon}} \\
    &\ \mglutms[\vDelta']{(\vsigma \circ \vdelta), s/x}{\vGamma; (\Gamma, x :
      S)}{\ext(\vrho[\vdelta], x, a)}
      \tag{by definition} \\
    &\ \mglutm[\vDelta']{T[(\vsigma \circ \vdelta), s/x]}{\Se_i}{\intp{T}(\ext(\vrho[\vdelta], x, a))}{\Ud_i}
      \tag{by $H_2$} \\
    &\ \mgluty[\vDelta']{T[(\vsigma \circ \vdelta), s/x]}{\intp{T}(\ext(\vrho[\vdelta], x, a))}
      \tag{by definition} \\
    &\ \mgluty[\vDelta']{T[q(\vsigma, x)][\vdelta, s/x]}{\intp{T}(\ext(\vrho[\vdelta], x, a))}
      \tag{by \Cref{lem:dt:gluty-resp-equiv}}
  \end{align*}
  By abstraction, we arrive at the goal as desired. 
\end{proof}

\begin{lemma}
  \begin{mathpar}
    \inferrule
    {\mSemtyp[\vGamma; \Gamma]{S}{\Se_i} \\ \mSemtyp[\vGamma; (\Gamma, x : S)]{t}{T}}
    {\mSemtyp[\vGamma; \Gamma]{\lambda x. t}{\Pi(x : S). T}}
  \end{mathpar}
\end{lemma}
\begin{proof}
  \begin{align*}
    H_1: &\ \mSemtyp[\vGamma; \Gamma]{S}{\Se_i}
           \tag{by assumption} \\           
    H_2: &\ \mSemtyp[\vGamma; (\Gamma, x : S)]{t}{T}
           \tag{by assumption} \\
    H_3: &\ \mglutms[\vDelta]{\vsigma}{\vGamma; \Gamma}{\vrho}
           \tag{by assumption}
  \end{align*}
  The goal is
  \begin{align*}
    \mglutm[\vDelta]{\lambda x. t[\vsigma]}{\Pi(x : S).T[\vsigma]}{\Lambda(x.t , 
    \vrho)}{Pi(\intp{S}(\vrho), x.T, \vrho)}
  \end{align*}
  Most obligations are easily discharged, given existentials as $S[\vsigma]$ and
  $T[q(\vsigma, x)]$. We will need to apply the fundamental theorem of the PER model
  (\Cref{thm:dt:per-fund}). Assuming $\vdelta : \vDelta' \To_r \vDelta$ and
  $\mglutm[\vDelta']{s}{S[\vsigma][\vdelta]}{a}{\intp{S}(\vrho)[\vdelta]}$, we have
  \begin{align*}
    &\ \mglutm[\vDelta']{s}{S[\vsigma\circ\vdelta]}{a}{\intp{S}(\vrho[\vdelta])}
      \tag{by \Cref{lem:dt:glutm-resp-equiv-ty}} \\
    &\ \mglutms[\vDelta']{\vsigma \circ \vdelta}{\vGamma; \Gamma}{\vrho[\vdelta]}
      \tag{by \Cref{lem:dt:glue-stack-mon}} \\
    &\ \mglutms[\vDelta']{(\vsigma \circ \vdelta), s/x}{\vGamma; (\Gamma, x :
      S)}{\ext(\vrho[\vdelta], x, a)}
      \tag{by definition} \\
    &\ \mglutm[\vDelta']{t[(\vsigma \circ \vdelta), s/x]}{T[(\vsigma \circ \vdelta),
      s/x]}{\intp{t}(\ext(\vrho[\vdelta], x, a))}{\intp{T}(\ext(\vrho[\vdelta], x,
      a))}
      \tag{by $H_2$} \\
    &\ \mglutm[\vDelta']{(\lambda x. t)[\vsigma][\vdelta]\ s}{T[q(\vsigma, x)][\vdelta,
      s/x]}{(x.t, \vrho)[\vdelta](a)}{(x.T, \vrho)[\vdelta](a)}
      \tag{by \Cref{lem:dt:glutm-resp-equiv-ty,lem:dt:glutm-resp-equiv-tm}}
  \end{align*}
  By abstraction, we arrive at the goal as desired. 
\end{proof}

\begin{lemma}
  \begin{mathpar}
    \inferrule
    {\mSemtyp[\vGamma; (\Gamma, x : S)]{T}{\Se_i} \\ \mSemtyp[\vGamma; \Gamma]{t}{\Pi(x : S). T} \\ \mSemtyp[\vGamma; \Gamma]{s}{S}}
    {\mSemtyp[\vGamma; \Gamma]{t\ s}{T[\vect I, s/x]}}
  \end{mathpar}
\end{lemma}
\begin{proof}
  \begin{align*}
    H_1: &\ \mSemtyp[\vGamma; (\Gamma, x : S)]{T}{\Se_i}
           \tag{by assumption} \\
    H_2: &\ \mSemtyp[\vGamma; \Gamma]{t}{\Pi(x : S). T}
           \tag{by assumption} \\
    H_3: &\ \mSemtyp[\vGamma; \Gamma]{s}{S}
           \tag{by assumption} \\
    H_4: &\ \mglutms[\vDelta]{\vsigma}{\vGamma; \Gamma}{\vrho}
           \tag{by assumption} \\
    H_5: &\ \mglutm[\vDelta]{t[\vsigma]}{\Pi(x :
           S). T[\vsigma]}{\intp{t}(\vrho)}{Pi(\intp{S}(\vrho), x.T, \vrho)}
           \tag{by $H_2$} \\
    H_6: &\ \mglutm[\vDelta]{s[\vsigma]}{S[\vsigma]}
           {\intp{s}{\vrho}}{\intp{S}(\vrho)}
           \tag{by $H_3$} \\
         &\ \mgluty[\vDelta]{S[\vsigma]}{\intp{S}(\vrho)}
           \tag{by \Cref{lem:dt:inv-glu-tm}} \\
         &\ \mgluty[\vDelta]{S'[\vsigma]}{\intp{S}(\vrho)}
           \tag{by $H_5$ for some $S'$} \\
         &\ \mjudge[\vDelta]{S[\vsigma] \approx S'[\vsigma]}
           \tag{by \Cref{lem:ty-glu-d}} \\
    H_7: &\ \mglutm[\vDelta]{s[\vsigma]}{S'[\vsigma]}
           {\intp{s}{\vrho}}{\intp{S}(\vrho)}
           \tag{by $H_6$ and \Cref{lem:dt:glutm-resp-equiv-ty}} \\
    H_8: &\ \mglutm[\vDelta]{t[\vsigma]\ s[\vsigma]}
           {T'[\vsigma, s[\vsigma]/x]}{\intp{t}(\vrho) \cdot
           \intp{s}(\vrho)}{\intp{T}(\ext(\vrho, x, \intp{s}(\vrho)))}
           \tag{by $H_5$ and $H_7$} \\
         &\ \mglutm[\vDelta]{(t\ s)[\vsigma]}
           {T'[\vsigma, s[\vsigma]/x]}{\intp{t\ s}(\vrho)}{\intp{T[\vect I, s/x]}(\vrho)}
           \tag{by \Cref{lem:dt:glutm-resp-equiv-ty,lem:dt:glutm-resp-equiv-tm}} \\
         &\ \mglutms[\vDelta]{\vsigma, s[\vsigma]/x}{\vGamma; (\Gamma, x :
           S)}{\ext(\vrho, x, \intp{s}(\vrho))}
           \tag{by $H_4$ and $H_6$} \\
         &\ \mglutm[\vDelta]{T[\vsigma, s[\vsigma]/x]}{\Se_i}{\intp{T}(\ext(\vrho, x,
           \intp{s}(\vrho)))}{\Ud_i}
           \tag{by $H_1$} \\
         &\ \mgluty[\vDelta]{T[\vsigma, s[\vsigma]/x]}{\intp{T}(\ext(\vrho, x,
           \intp{s}(\vrho)))}
           \tag{by definition} \\
         &\ \mgluty[\vDelta]{T'[\vsigma, s[\vsigma]/x]}{\intp{T}(\ext(\vrho, x,
           \intp{s}(\vrho)))}
           \tag{by $H_8$} \\
         &\ \mjudge[\vDelta]{T[\vsigma, s[\vsigma]/x] \approx T'[\vsigma, 
           s[\vsigma]/x]}
           \tag{by \Cref{lem:ty-glu-d}} \\
         &\ \mglutm[\vDelta]{(t\ s)[\vsigma]}
           {T[\vsigma, s[\vsigma]/x]}{\intp{t\ s}(\vrho)}{\intp{T[\vect I, s/x]}(\vrho)}
           \tag{by $H_8$} \\
         &\ \mglutm[\vDelta]{(t\ s)[\vsigma]}
           {T[\vect I, s/x][\vsigma]}{\intp{t\ s}(\vrho)}{\intp{T[\vect I, s/x]}(\vrho)}
  \end{align*}
  By abstraction, this concludes the goal.
\end{proof}

\subsubsection{Natural Numbers}

\begin{lemma}
  \begin{mathpar}
    \inferrule
    { }
    {\mSemtyp \Nat \Se_i}

    \inferrule
    { }
    {\mSemtyp \ze \Nat}

    \inferrule
    {\mSemtyp t \Nat}
    {\mSemtyp{\su t}\Nat}
  \end{mathpar}
\end{lemma}
\begin{proof}
  Immediate.
\end{proof}

We need a lemma in order to handle recursion. It effectively keeps track of evaluation
and shows that the gluing is maintained.
\begin{lemma}
  If
  \begin{itemize}
  \item $\mglutms[\vDelta]{\vsigma}{\vGamma; \Gamma}{\vrho}$,
  \item $\forall \mglutm[\vDelta]{t'}{\Nat}{a}{\Nd}.\ \mgluty[\vDelta]{M[\vsigma, 
      t'/x]}{\intp{M}(\ext(\vrho, x, a))}$,
  \item $\mglutm[\vDelta]{s[\vsigma]}{M[\vsigma, 
      \ze/x]}{\intp{s}(\vrho)}{\intp{M}(\ext(\vrho, x, \zed))}$,
  \item $\forall \mglutm[\vDelta]{t'}{\Nat}{a}{\Nd},
    \mglutm[\vDelta]{u'}{M[\vsigma, t'/x]}{b}{\intp{M}(\ext(\vrho, x, a))}.\
    \mglutm[\vDelta]{u[\vsigma, t'/x, u'/y]}{M[\vsigma, \su
      t'/x]}{\intp{u}(\ext(\vrho, x, a, y, b))}{\intp{M}(\ext(\vrho, x, \sud{a}))}$
  \item $\mglunat[\vDelta]t a$,
  \end{itemize}
  then
  $\mglutm[\vDelta]{\elimn {x.M[q(\vsigma, x)]}{(s[\vsigma])}{x,y.u[q(q(\vsigma, x), y)]}t}{M[\vsigma, t/x]}
  {\trec \cdot (x.M ,  \vrho, \intp{s}(\vrho), (x,y.u, \vrho),
    a)}{\intp{M}(\ext(\vrho, x, a))}$.
\end{lemma}
\begin{proof}
  We do induction on $\mglunat[\vDelta]t a$.
  \begin{itemize}[label=Case]
  \item $a = \zed$,
    \begin{mathpar}
      \inferrule
      {\mtyequiv[\vDelta]{t}{\ze}{\Nat}}
      {\mglunat[\vDelta]{t}{\zed}}
    \end{mathpar}
    We have 
    \begin{align*}
      \mtyequiv[\vDelta]{\elimn {x.M[q(\vsigma, x)]}{(s[\vsigma])}{x,y.u[q(q(\vsigma,
      x), y)]} \ze}{s[\vsigma]}{M[\vsigma, \ze/x]}
    \end{align*}
    so the goal is discharged by
    $\mglutm[\vDelta]{s[\vsigma]}{M[\vsigma,
      \ze/x]}{\intp{s}(\vrho)}{\intp{M}(\ext(\vrho, x, \zed))}$ modulo
    \Cref{lem:dt:glutm-resp-equiv-ty,lem:dt:glutm-resp-equiv-tm}.
    
  \item $a = \sud{a'}$,
    \begin{mathpar}
      \inferrule
      {\mtyequiv[\vDelta]{t}{\su t'}{\Nat} \\ \mglunat[\vDelta]{t'}{a'}}
      {\mglunat[\vDelta]{t}{\sud b}}
    \end{mathpar}
    We have
    \begin{align*}
      \mtyequiv[\vDelta]{\elimn {x.M[q(\vsigma, x)]}{(s[\vsigma])}{x,y.u[q(q(\vsigma,
      x), y)]}{t}}{\\{u[\vsigma, t'/x, \elimn {x.M[q(\vsigma,
      x)]}{(s[\vsigma])}{x,y.u[q(q(\vsigma, x), y)]}{t'}/y]}}M[\vsigma, t/x]
    \end{align*}
    
    We first have
    \begin{align*}
      \mglutm[\vDelta]{t'}{\Nat}{a'}{\Nd}
    \end{align*}
    Then by IH, we have
    $\mglutm[\vDelta]{\elimn {x.M[q(\vsigma, x)]}{(s[\vsigma])}{x,y.u[q(q(\vsigma, x),
        y)]}{t'}}{M[\vsigma, x : t']} {\trec \cdot (x.M ,  \vrho, \intp{s}(\vrho),
      (x,y.u, \vrho), a')}{\intp{M}(\ext(\vrho, x, a'))}$. We plug this in the $u$
    case and then we are done modulo
    \Cref{lem:dt:glutm-resp-equiv-ty,lem:dt:glutm-resp-equiv-tm}.
    
  \item $a = \uparrow^\Nd(c)$,
    \begin{mathpar}
      \inferrule
      {c \in \bot \\
        \forall \vdelta : \vDelta' \To_r \vDelta.\ \mtyequiv[\vDelta']{t[\vdelta]}{\Rne_\alpha(c[\vdelta])}{\Nat}}
      {\mglunat[\vDelta]{t}{\uparrow^\Nd(c)}}
    \end{mathpar}
  \end{itemize}
  We prove
  $\mglutmd[\vDelta]{\elimn {x.M[q(\vsigma, x)]}{(s[\vsigma])}{x,y.u[q(q(\vsigma, x),
      y)]}t}{M[\vsigma, t/x]} {\trec(x.M ,  \vrho, \intp{s}(\vrho),
    (x,y.u, \vrho), c)}{\intp{M}(\ext(\vrho, x, a))}$ instead and apply
  realizability to obtain the goal.
  
  To prove the subgoal, we assume $\vdelta : \vDelta' \To_r \vDelta$, and then prove the
  subgoal by just following the definitions and realizability. 
\end{proof}

\begin{lemma}
  \begin{mathpar}
    \inferrule
    {\mSemtyp[\vGamma; (\Gamma, x : \Nat)]{M}{\Se_i} \\
      \mSemtyp[\vGamma; \Gamma]{s}{M[\vect I, \ze/x]} \\
      \mSemtyp[\vGamma; (\Gamma, x : \Nat, y : M)]{u}{M[\wk_x \circ \wk_y, \su x/x]} \\
      \mSemtyp[\vGamma; \Gamma] t \Nat}
    {\mSemtyp[\vGamma; \Gamma]{\elimn {x.M} s {x,y.u} t}{M[\vect I, t/x]}}
  \end{mathpar}
\end{lemma}
\begin{proof}
  Assuming $\mglutms[\vDelta]{\vsigma}{\vGamma; \Gamma}{\vrho}$, we extract
  $\mglunat[\vDelta]{t[\vsigma]}{\intp{t}(\vrho)}$. Then we apply the previous lemma
  and obtain the goal up to
  \Cref{lem:dt:glutm-resp-equiv-ty,lem:dt:glutm-resp-equiv-tm}.
\end{proof}

\subsubsection{Universes}

\begin{lemma}
  \begin{mathpar}
    \inferrule
    { }
    {\mSemtyp{\Se_n}{\Se_{n + 1}}}

    \inferrule
    {\mSemtyp{T}{\Se_i}}
    {\mSemtyp{T}{\Se_{i + 1}}}
  \end{mathpar}
\end{lemma}
\begin{proof}
  Immediate.
\end{proof}

\subsubsection{Modality}

\begin{lemma}
  \begin{mathpar}
    \inferrule
    {\mSemtyp[\vGamma; \cdot]{T}{\Se_i}}
    {\mSemtyp{\square T}{\Se_i}}
  \end{mathpar}
\end{lemma}
\begin{proof}
  Immediate by \Cref{lem:dt:glu-mon}. 
\end{proof}

\begin{lemma}
  \begin{mathpar}
    \inferrule
    {\mSemtyp[\vGamma; \cdot]{t}{T}}
    {\mSemtyp{\boxit t}{\square T}}
  \end{mathpar}
\end{lemma}
\begin{proof}
  \begin{align*}
    H_1: &\ \mSemtyp[\vGamma; \cdot]{t}{T}
           \tag{by assumption} \\
    H_2: &\ \mglutms[\vDelta]{\vsigma}{\vGamma}{\vrho}
           \tag{by assumption} \\
         &\ \text{assume }\vDelta'', \vdash \vDelta'; \vDelta'', \vdelta : \vDelta' \To_r \vDelta, \text{let }n =
           |\vDelta''| \\
    H_3: &\ \mglutms[\vDelta'; \vDelta'']{\sextt{(\vsigma \circ \vdelta)} n}{\vGamma;
           \cdot}{\ext(\vrho[\vdelta], n)}
           \tag{by \Cref{lem:dt:glue-stack-mon}} \\
         &\ \mglutm[\vDelta'; \vDelta'']{t[\sextt{(\vsigma \circ \vdelta)} n]}
           {T[\sextt{(\vsigma \circ \vdelta)} n]}
           {\intp{t}(\ext(\vrho[\vdelta], n))}{\intp{T}(\ext(\vrho[\vdelta], n))}
           \tag{by $H_1$ and $H_3$} \\
         &\ \mglutm[\vDelta'; \vDelta'']{t[\sextt\vsigma 1][\sextt\vdelta n]}
           {T[\sextt\vsigma 1][\sextt\vdelta n]}
           {\intp{t}(\ext(\vrho))[\sextt\vdelta n]}{\intp{T}(\ext(\vrho))[\sextt\vdelta n]}
           \tag{by \Cref{lem:dt:glutm-resp-equiv-tm,lem:dt:glutm-resp-equiv-ty}} \\
         &\ \mglutm[\vDelta'; \vDelta'']{\unbox{n}{(\boxit t[\vsigma][\vdelta])}}
           {T[\sextt\vsigma 1][\sextt\vdelta n]}
           {\tunbox \cdot (n, \intp{\boxit
           t}(\vrho)[\vdelta])}{\intp{T}(\ext(\vrho))[\sextt\vdelta n]}
           \tag{by \Cref{lem:dt:glutm-resp-equiv-tm}} \\
         &\ \mglutm[\vDelta]{\boxit t[\vsigma]}
           {\square T[\vsigma]}
           {\intp{\boxit t}(\vrho)}{\intp{\square T}(\vrho)}
           \tag{by abstraction}
  \end{align*}
\end{proof}

\begin{lemma}
  \begin{mathpar}
    \inferrule
    {\mSemtyp[\vGamma; \cdot]{T}{\Se_i} \\ \mSemtyp{t}{\square T} \\ |\vDelta| = n}
    {\mSemtyp[\vGamma; \vDelta]{\unbox n t}{T[\sextt{\vect I}n]}}
  \end{mathpar}
\end{lemma}
\begin{proof}
  \begin{align*}
    H_1: &\ \mSemtyp{t}{\square T}
           \tag{by assumption} \\
    H_2: &\ \mglutms[\vDelta']{\vsigma}{\vGamma;\vDelta}{\vrho}
           \tag{by assumption} \\
         &\ \mglutms[\trunc{\vDelta'}{\Ltotal\vsigma n}]{\trunc \vsigma n}{\vGamma}{\trunc \vrho n}
           \tag{by \Cref{lem:dt:glue-trunc}} \\
    H_3: &\ \mglutm[\trunc{\vDelta'}{\Ltotal\vsigma n}]
           {t[\trunc \vsigma n]}{\square T[\trunc \vsigma n]}
           {\intp{t}(\trunc \vrho n)}
           {\intp{\square T}(\trunc \vrho n)}
           \tag{by $H_1$ and $H_2$} \\
         &\ \mglutm[\vDelta']
           {\unbox{\Ltotal\vsigma n}{(t[\trunc \vsigma n])}}{T[\sextt{\trunc \vsigma n}{\Ltotal\vsigma n}]}
           {\tunbox \cdot (\Ltotal\vrho n, \intp{t}(\trunc \vrho n))}
           {\intp{T}(\ext(\trunc \vrho n, \Ltotal\vrho n))}
           \tag{by $H_3$} \\
         &\ \mglutm[\vDelta']
           {\unbox{n}{t}[\vsigma]}{T[\sextt{\vect I} n][\vsigma]}
           {\intp{\unbox n t}(\vrho)}
           {\intp{T[\sextt{\vect I} n]}(\vrho)}
           \tag{by definition and \Cref{lem:dt:glutm-resp-equiv-tm,lem:dt:glutm-resp-equiv-ty}}
  \end{align*}
\end{proof}

\subsection{Fundamental Theorem}

Since we have proven the semantic typing rules, we can conclude the fundamental
theorem:
\begin{theorem}[Fundamental]
  \begin{itemize}
  \item If $\mtyping t T$, then $\mSemtyp t T$. 
  \item If $\mtyping \vsigma \vDelta$, then $\mSemtyp \vsigma \vDelta$. 
  \end{itemize}
\end{theorem}
In the theorem, we need to work in a slightly adjusted system with some extra premises
for strong IHs. We can discharge these premises once the theorem is established. 

We also need to relate $\vect I$ and $\uparrow^{\vGamma}$:
\begin{lemma}
  $\mglutms{\vect I}{\vGamma}{\uparrow^{\vGamma}}$
\end{lemma}
\begin{proof}
  Immediate by induction on $\vGamma$ and apply \Cref{lem:dt:glu-var}. 
\end{proof}

\begin{theorem}[Soundness]
  If $\mtyping t T$, then $\mtyequiv{t}{\nbe^T_{\vGamma}(t)}T$. 
\end{theorem}
\begin{proof}
  Applying the fundamental theorem, we have
  $\mglutm{t[\vect I]}{T[\vect
    I]}{\intp{t}(\uparrow^{\vGamma})}{\intp{T}(\uparrow^{\vGamma})}$.  We conclude the
  goal by applying \Cref{lem:dt:glutm-resp-equiv-ty,lem:dt:glutm-resp-equiv-tm} and realizability.
\end{proof}

\bibliographystyle{ACM-Reference-Format}


\end{document}